\documentclass{CCCG2026-tex-template/cccg26}
\usepackage{tikz}
\usetikzlibrary{calc, intersections} 
\usepackage{graphicx} 
\usepackage[toc,page]{appendix}
 \usepackage{lineno}
\usepackage{subcaption}
\usepackage{graphicx}
\usepackage{comment}
\usepackage{enumitem}

\usepackage{amssymb} 

\usepackage{xcolor}
\usepackage{subcaption}

\newtheorem{definition}{Definition}
\usepackage{float}
\title{An Exact Generalized k-Cell Decomposition }

\author{
    Yeganeh Bahoo\footnotemark[2] \and 
    Sajad Saeedi\thanks{University College London, London, UK, \texttt{s.saeedi@ucl.ac.uk}} \and 
    Roni Sherman\thanks{Toronto Metropolitan University, Toronto, Canada, \texttt{\{bahoo, roni.sherman\}@torontomu.ca}}
}
\index{Bahoo, Yeganeh}
\index{Saeedi, Sajad}
\index{Sherman, Roni}

\begin{document}

\maketitle

{\def\thefootnote{} \footnotetext{We acknowledge the support of the Natural Sciences and Engineering Research Council of Canada (NSERC).}}
\setcounter{footnote}{0} 

\begin{abstract}
This paper introduces an exact $k$-cell decomposition for visibility planning in polygonal environments for agents equipped with $k$-modems, devices that can see  through up to $k$ walls. Unlike prior decompositions that may include redundant partition lines, our proposed method ensures that visibility events (appear, disappear, merge, and split) are guaranteed to occur on every line of the decomposition. By eliminating these redundancies, we achieve an $O(n^4)$ complexity , representing a potentially quadratic improvement over the previous best $O(k^2n^4)$ result. This decomposition explicitly identifies the locations of all critical visibility events and extends to polygons with holes. It has practical applications in 
tasks such as optimal pursuit-evasion under $k$-visibility and agent counting in invisible regions.
\end{abstract}

\section{Introduction}

The ability to detect and capture intruders in constrained environments is a fundamental challenge in areas such as surveillance, search and rescue, autonomous robotics, and wireless sensor networks~\cite{O’Rourke_1998,deBerg2008}. In these applications, it is critical to develop strategies that ensure complete exploration of a space, particularly when visibility is limited by structural barriers such as walls. Traditional pursuit-evasion scenarios assume direct line-of-sight between a searcher and an intruder, but in many practical settings, devices or agents are equipped with sensing capabilities that allow them to see through a limited number of obstacles. Such devices are referred to as $k$-modems~\cite{Fabia2009}. This generalization introduces the concept of $k$-visibility, that is, the ability to see through $k$ walls, enabling agents to detect intruders even when direct visibility is not feasible. Designing reliable motion strategies under $k$-visibility is essential for improving situational awareness and guaranteeing successful detection in partially obstructed environments. 

Despite its importance, pursuit-evasion under $k$-visibility introduces several geometric and algorithmic challenges~\cite{sfw,kim2026inverse}. As a pursuer equipped with a $k$-modem moves through a polygonal environment, the visibility region may undergo sudden combinatorial changes, known as geometric events, including appear, disappear, split, and merge ~\cite{yu2011shadow}. These events complicate the tracking process, as an evader can exploit transient blind spots to avoid detection. Predicting when and where these events occur is computationally non-trivial and requires a deep understanding of visibility dynamics. Moreover, constructing a representation of the environment that guarantees the stability of visibility regions during movement is difficult, particularly as $k$ increases. 

To address this, we build upon the concept of visibility-based cell decomposition introduced by Guibas et al.~\cite{Guiba1997}, who partitioned polygons to facilitate searching for intruders under $0$-visibility. This approach was later extended to $2$-visibility by Bahoo et al.~\cite{Bahoo2013} and subsequently generalized to $k$-visibility in ~\cite{bahoo2025generalizedkcelldecompositionvisibility}. While the decomposition in ~\cite{bahoo2025generalizedkcelldecompositionvisibility} ensures that the combinatorial representation of the shadow remains invariant within each cell, it is not "exact" - many of its partition lines do not correspond to actual visibility events, leading to unnecessary overhead. The main contribution of this work is the development of an exact $k$-visibility cell decomposition for 2D polygonal environments. By identifying and utilizing only the lines where visibility events are guaranteed to occur, we achieve a more sparse and efficient decomposition with $O(n^4)$ complexity, which is an improvement over the existing method~\cite{bahoo2025generalizedkcelldecompositionvisibility} . Applications include intruder detection~\cite{Guiba1997} and counting of agents in invisible regions~\cite{yu2011shadow} (i.e it is possible to use the same algorithms as described in those papers on the decomposition presented here).

The paradigm of $k$-visibility builds on a rich foundation of geometric illumination and visibility computation. Early work by Aichholzer et al.~\cite{Aichholzer2014} and Fabila-Monroy et al.~\cite{Fabia2009} on modem illumination established foundational bounds for guarding environments with $k$-transmitters. Bahoo et al.~\cite{Bahoo2020, bahoo2019time} later introduced efficient methods to compute $k$-visibility regions from stationary points, while other studies~\cite{k-starshaped} focused on identifying a polygon's $k$-kernel. Recent advancements have further extended the framework to address edge visibility~\cite{bahoo2023segment}, watchman routes~\cite{brotzner2024k}, and $M$-guarding problems~\cite{bahoo2025m}.

—
In this paper we focus on tracing exactly where 
topological events occur to the shadow structure confined strictly inside the polygon as an 
agent moves through the space. Rather than dynamically computing the entire $k$-visibility polygon 
or the full geometric area of the shadow region itself, we map out the discrete critical boundaries 
where the combinatorial layout of these interior occlusions experiences a structural transition.

First, we introduce the necessary preliminary definitions, followed by our proposed cell decomposition. The complexity of this decomposition is analyzed in Section 4. In Section 5, we extend our approach to polygons with holes. Finally, we present our conclusions and discuss directions for future work.

%
\section{Preliminaries}
\label{section:events}

We model an agent as a point robot in a simple polygon $P$ (possibly with holes) equipped with a $k$-modem, a device that can see behind $k$ walls. This means a point $p \in P$ is $k$-visible to the agent at position $q$ if the segment $pq$ intersects at most $k$ edges of $P$. A maximally connected invisible region for a point robot $p$ is a shadow component; For example, $S_1$ in Figure~\ref{fig:Appear}. 

As the agent moves continuously around $P$, the combinatorial structure of its shadow changes only when it crosses specific partition lines. These lines induce a cell decomposition of $P$, where the number and connectivity of shadow components remain invariant within each cell. We consider four primary visibility events (Figure~\ref{fig:fourEvents}):

\begin{itemize}
    \item Appear: a new shadow component emerges. For instance, an agent with $0$-visibility going from $a$ in Figure~\ref{fig:Disappear} to $b$ in Figure~\ref{fig:Appear} by crossing the purple partition line causes an appear event (a new shadow component $S_{1}$ appears).
    \item Disappear: a shadow component vanishes. Consider an agent moving from $b$ in Figure~\ref{fig:Appear} to $a$ in Figure~\ref{fig:Disappear}, and $S_1$ disappears. 
    
    \item Merge: two shadow components unite into one. For example, the agent at $b$ in Figure~\ref{fig:Split} moves to $a$ in Figure~\ref{fig:Merge}, and the shadow components $S_{2}$ and $S_{3}$ merge into one shadow component $S_{4}$.
    \item Split: a shadow component splits into two. See Figure~\ref{fig:Merge} to Figure ~\ref{fig:Split}.
\end{itemize}.

 \begin{figure}
\centering
\begin{subfigure}[b]{.49\linewidth}
\includegraphics[width=\linewidth]{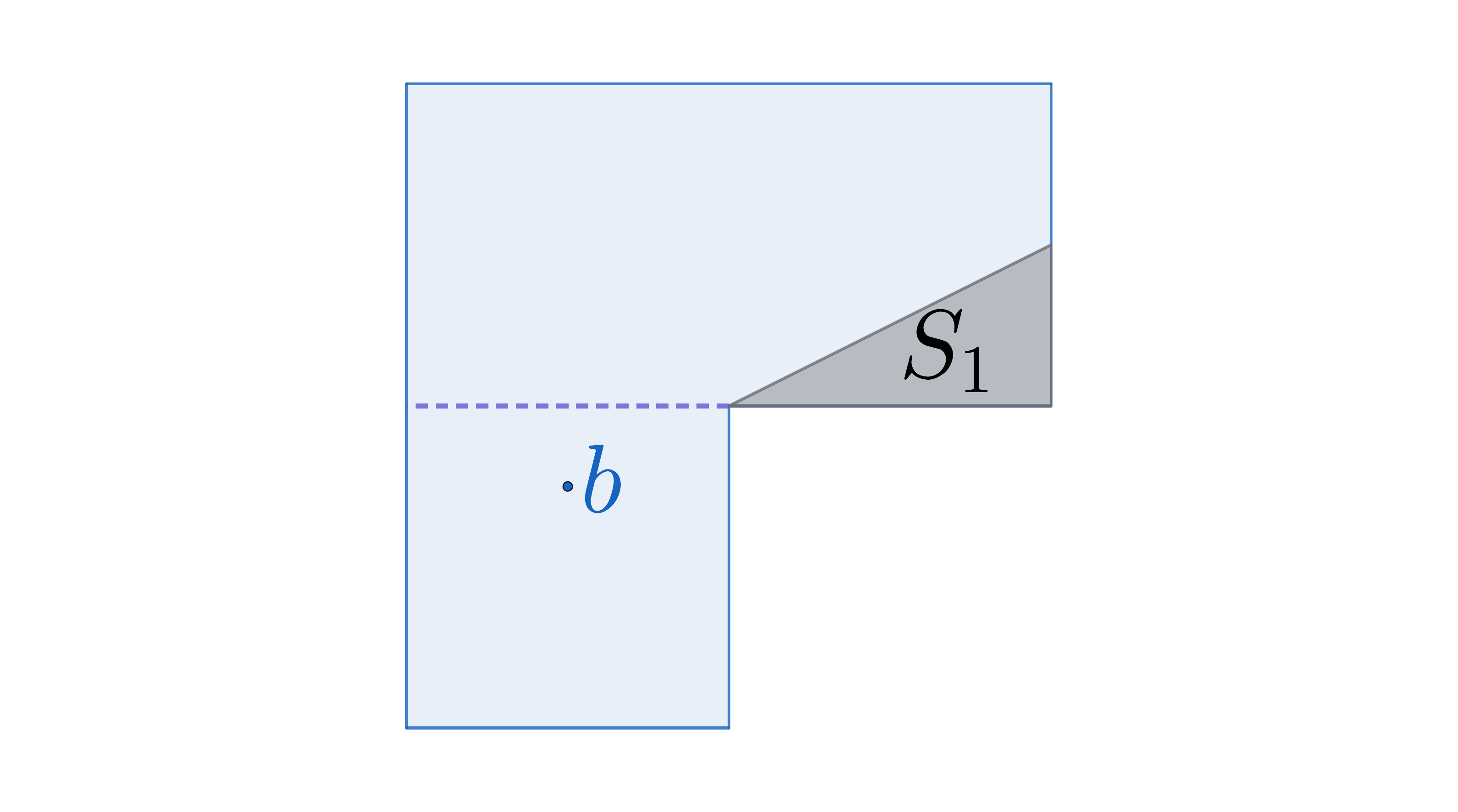}
\caption{}\label{fig:Appear}
\end{subfigure}
\begin{subfigure}[b]{.49\linewidth}
\includegraphics[width=\linewidth]{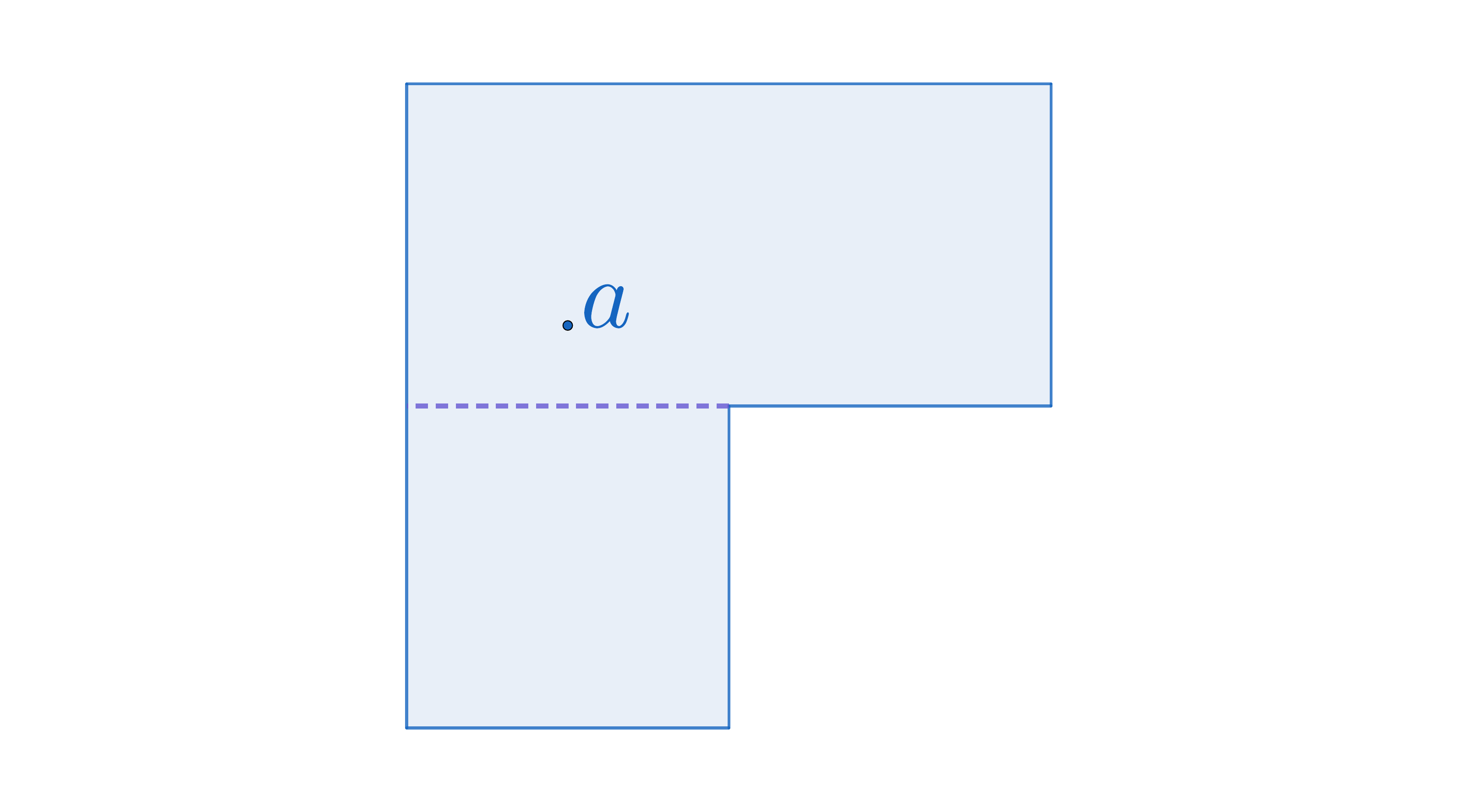}
\caption{}\label{fig:Disappear}
\end{subfigure}
\begin{subfigure}[b]{.49\linewidth}
\includegraphics[width=\linewidth]{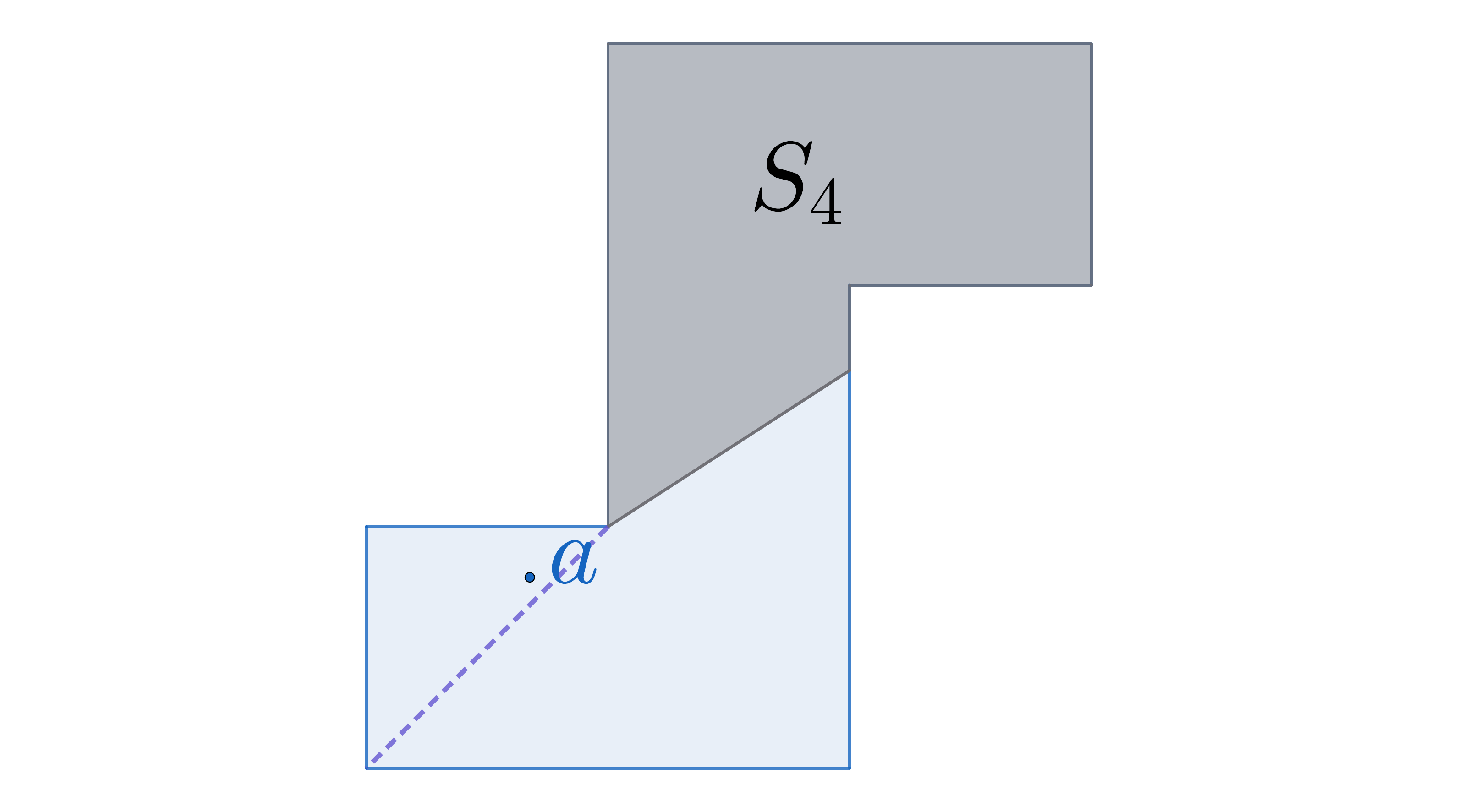}
\caption{}\label{fig:Merge}
\end{subfigure}
\begin{subfigure}[b]{.49\linewidth}
\includegraphics[width=\linewidth]{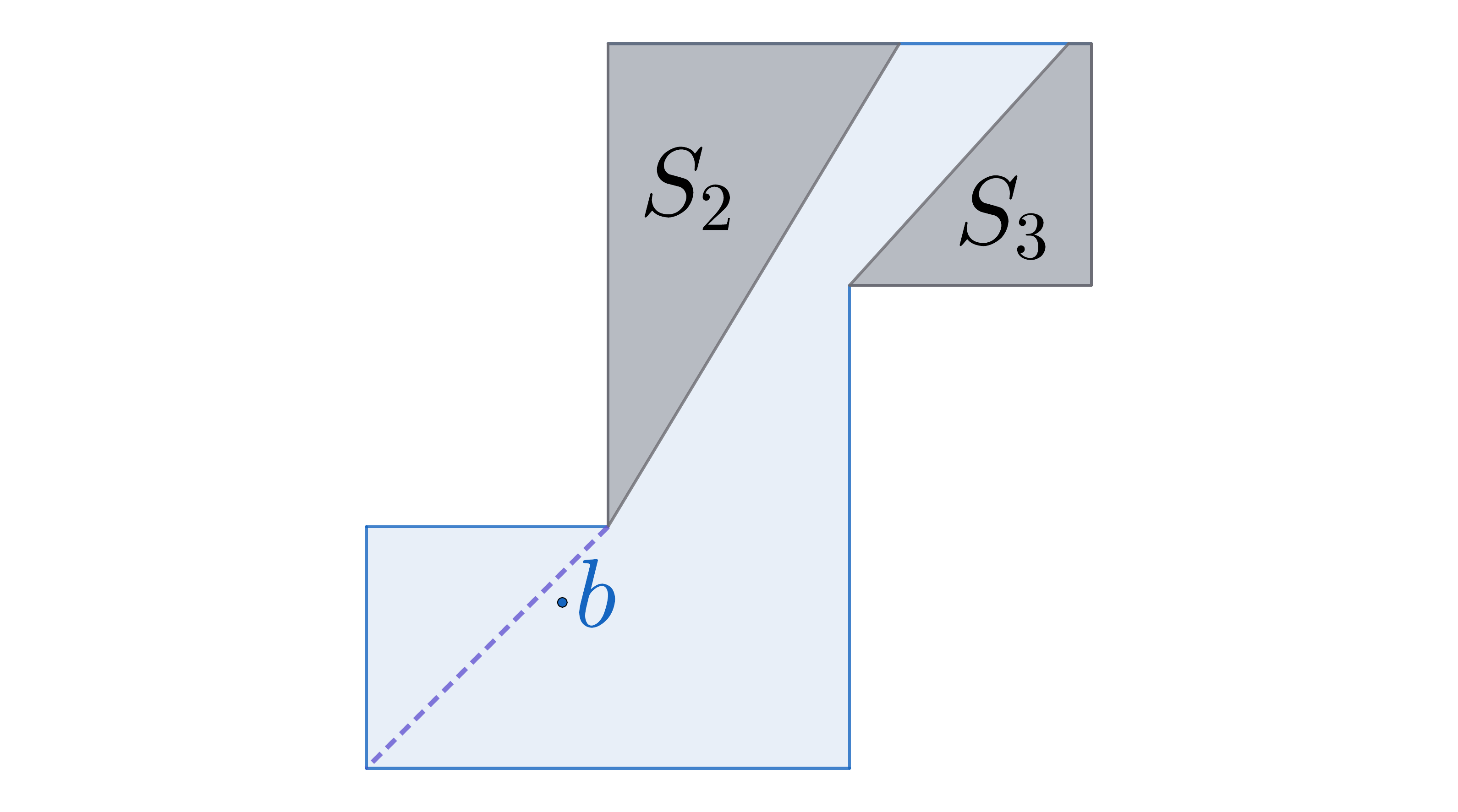}
\caption{}\label{fig:Split}
\end{subfigure}

\caption{The four events for $k=0$. From $(a)$ to $(b)$ is disappear, from $(b)$ to $(a)$ is appear. From $(d)$ to $(c)$ is merge. From $(c)$ to $(d)$ is split.}
\label{fig:fourEvents}
\end{figure}

Visibility events occur when the agent's $k$-visibility boundary undergoes a combinatorial change, which only happens when the agent crosses the line $\ell_g$ defined by two vertices, $v_{1}$ and $v_{2}$. To define different cases, we first need the following definition:

\begin{definition} [Critical Vertex]
   Let $\ell{g}$ be a line through a vertex $a$ and a point $b$. vertex $a$ is critical to $b$ when both of its edges lie to one side of $\ell{g}$, as in Figure~\ref{fig:critical-vertex}. Two vertices are said to be mutually critical when each is critical to the other.
\end{definition}

\begin{figure}
\centering
\includegraphics[width=.3\textwidth]{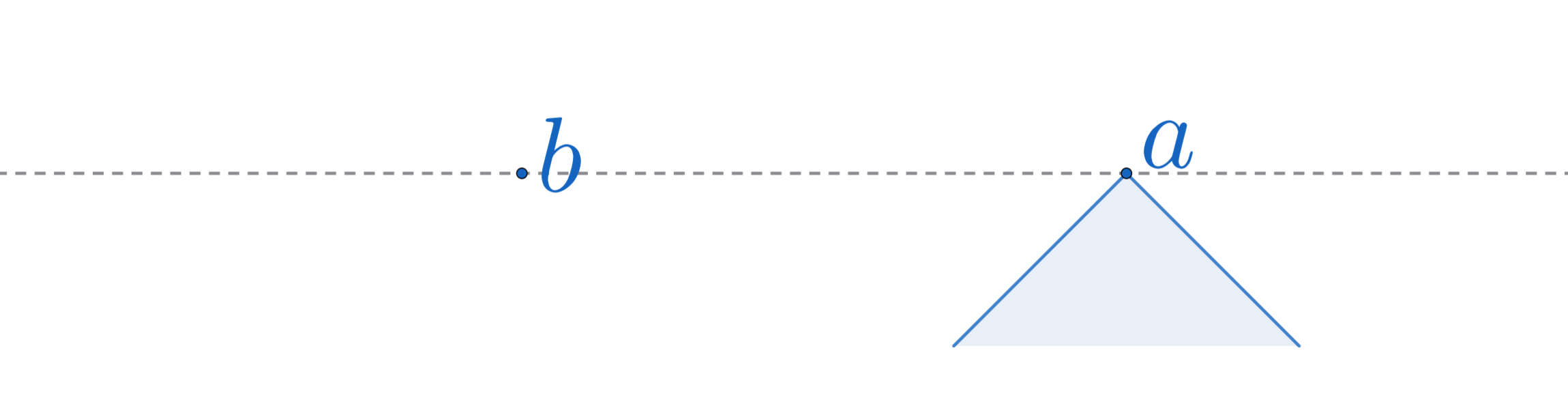}
\caption{A vertex $a$ critical to a point $b$.}
\label{fig:critical-vertex}
\end{figure}

We assume that no more than two vertices can be on the same line and be critical to each other.

 \begin{figure}
\centering
\begin{subfigure}[b]{.49\linewidth}
\includegraphics[width=\linewidth]{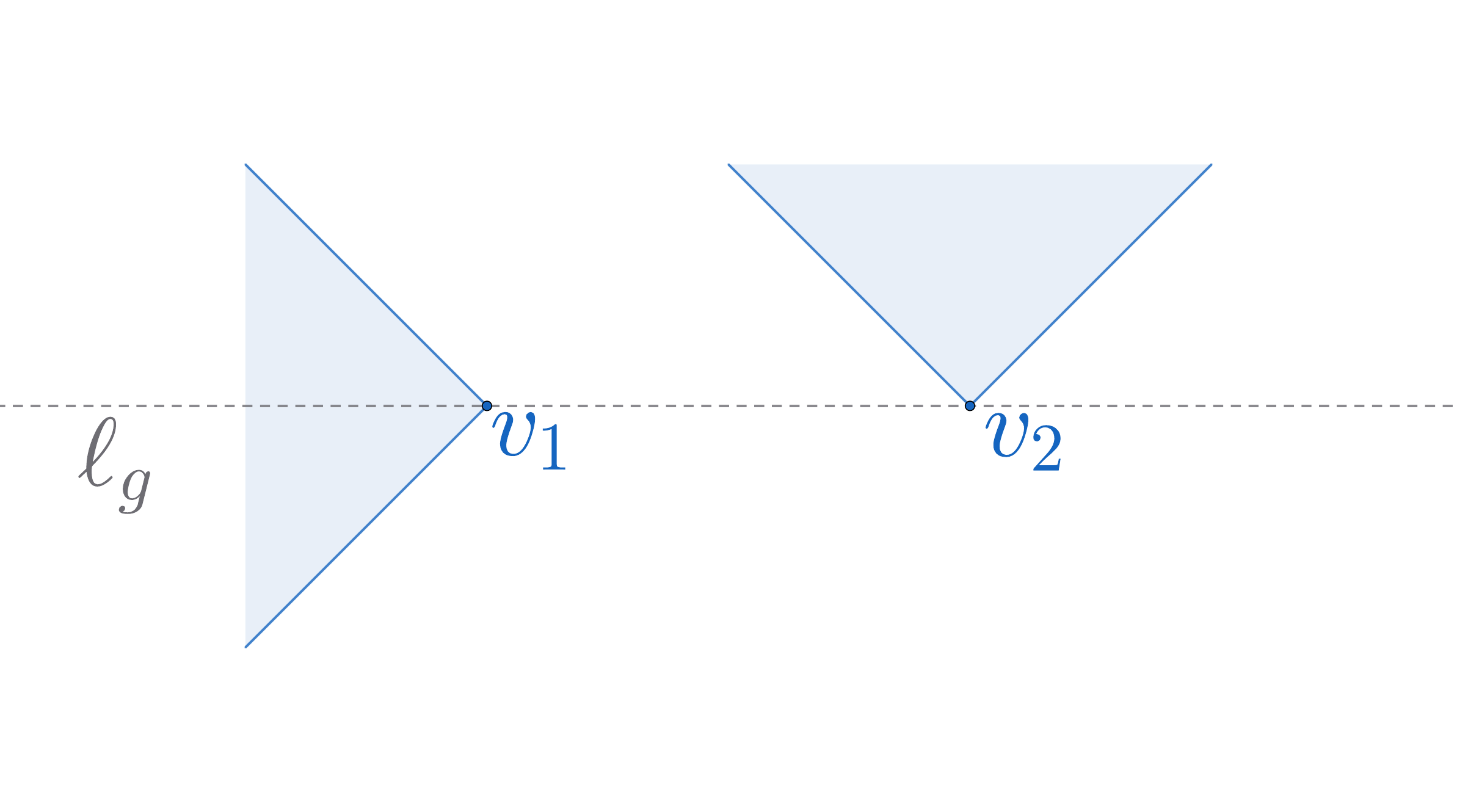}
\caption{No Event}\label{fig:NoEventCase1}
\end{subfigure}
\begin{subfigure}[b]{.49\linewidth}
\includegraphics[width=\linewidth]{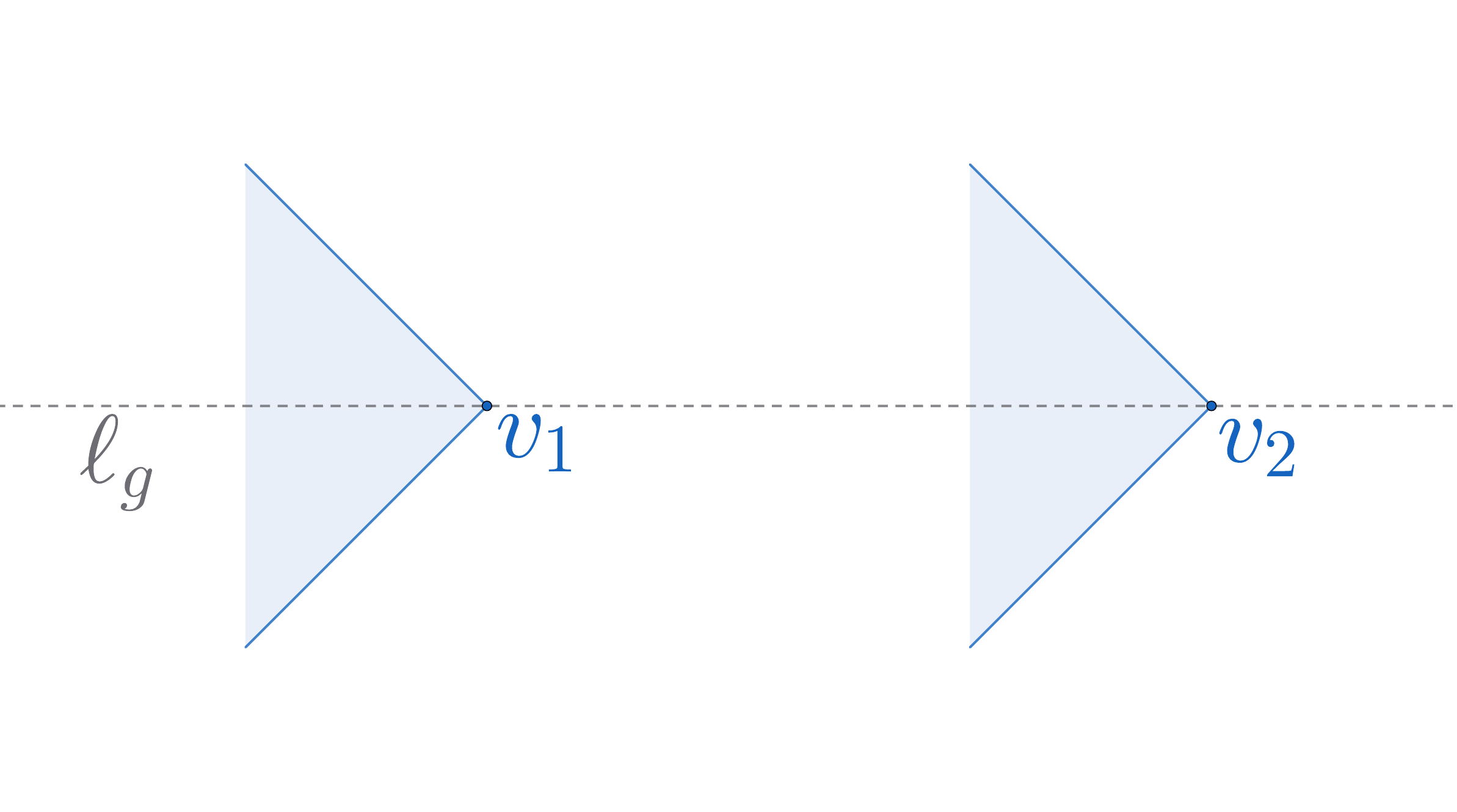}
\caption{No Event}\label{fig:NoEventCase2}
\end{subfigure}

\caption{Orientation of two vertices where there is no event.} \label{fig:noEvent}
\end{figure}

\section{Proposed Cell Decomposition}

To define the decomposition we have the following theorem:

\begin{theorem}
    A change to the combinatorial structure of the shadow is only possible when two vertices $v_{1}$ and $v_{2}$ are mutually critical. 
\end{theorem}

\begin{proof}
    Assume $v_{1}$ and $v_{2}$ are not mutually critical (See Figure~\ref{fig:noEvent}). Crossing $\ell_g$ does not alter the number of walls obstructing the agent's view at $v_{2}$ or beyond, and so the combinatorial structure of the shadow remains invariant. 

\end{proof}

By Theorem 1, it is sufficient to only consider pairs of vertices that are critical for each other when creating the decomposition. Considering only mutually critical vertices, we reduce the cases that must be considered to eight. These are categorized by vertex types, convex or reflex, and whether the two edges of each vertex lie on the same or opposite sides of $\ell_g$. These cases can be shown as follows:

\begin{itemize}[noitemsep]

    \item CCS (Convex - Convex - Same)
    \item CCO (Convex - Convex - Opposite)
    \item CRS (Convex - Reflex - Same)
    \item CRO (Convex - Reflex - Opposite)
    \item RCS (Reflect - Convex - Same)
    \item RCO (Reflex - Convex - Opposite)
    \item RRS (Reflex - Reflex - Same)
    \item RRO (Reflex - Reflex - Opposite)
\end{itemize}

Additionally, we consider four special cases (CC, CR, RC, RR) which occur when $v_{1}$ and $v_{2}$ are adjacent. The shadow behaves differently in these cases as opposed to the basic eight. 

For brevity, the full analysis of the CCO case is provided here, while all other cases  are detailed in the Appendix (\ref{appendix:CCS}-\ref{appendix:CC-SC}). A table (Table~\ref{table:all-cases}) summarizing all the cases can also be found in the appendix.

To determine when a line is necessary for the decomposition,
we analyze how the shadow behaves when an agent passes the line through pairs of mutually critical vertices, $v_1$ and $v_2$. Assume that $v_1$ is situated to the left of $v_2$. Let $\ell_g$ be the line that passes through $v_1$ and $v_2$; shown as the gray dashed line in the figures (for example, see Figure~\ref{fig:CCO-genericA0}). For an agent at a given position, we define $\ell_b$ as the ray emanating from the agent through $v_1$, and $\ell_r$ as the ray through $v_2$, the dashed blue and red lines in the figures (Figure~\ref{fig:CCO-genericA0} -~\ref{fig:CCO-genericB6} and those in the appendix) respectively. The line $\ell_g$ intersects the polygon at points $x_{1}$ (to the left of $v_1$) and $x_{4}$ (to the right of $v_2$). let $x_{2}$ be an intersection of $\ell_g$ with the polygon immediately following $x_{1}$, and $x_{3}$ the intersection immediately preceding $x_{4}$. We call the segment $x_{1}x_{2}$ a partition line if crossing it triggers a visibility event at $v_2$ or in the immediate neighbourhood of $x_{3}x_{4}$ (along the line $x_3x_4$). 

\begin{definition}
    A partition line is a line where a change (appear, disappear, merge, or split event) to the combinatorial structure of the shadow occurs. 
\end{definition}
$Z$ denotes the number of walls intersected by $\ell_g$ between $x_1$ and $v_2$, and $W$ denotes the number of walls intersected by $\ell_g$ between $x_1$ and $x_4$. Note that both $Z$ and $W$ exclude the edges incident to $v_1$ and $v_2$ themselves.

Figures~\ref{fig:CCO-genericA0} to~\ref{fig:CCO-genericB6} illustrate specific basic cases where $k$ ranges from $0$ to $6$. However, these examples generalize to any arbitrary value of $k$ based on the relative difference between the obstruction counts ($Z, W$) and the visibility threshold. For example, the condition $Z = k - 1$ for CCO triggers the same event whether $k=2$ and $Z=1$, or $k=99$ and $Z=98$. In contrast, if $Z \geq k+3$, the region at $v_2$ remains completely in shadow, while for $Z \leq k-9$, the region is entirely visible; therefore, in both extremes, no visibility event occurs at $v_{2}$. A similar argument holds for $W$ and $x_3x_4$. 

Each sub-case is presented as a pair of diagrams (for instance, Figure~\ref{fig:CCO-genericA2} and Figure~\ref{fig:CCO-genericB2}) showing the agent in position $a$ (above $\ell_g$) and position $b$ (below $\ell_g$). At $a$, obstruction counts for $v_2$ and $x_3x_4$ are shown relative to the regions defined by rays $\ell_r$ and $\ell_b$. At $b$, the rays $\ell_b$ and $\ell_r$ swap relative orientations, and the resulting change in obstruction counts determines which of the four events (appear, disappear, merge, or split) occurs. Note that unlike $Z$ and $W$, these obstruction counts include the edges incident to $v_1$ and $v_2$.

Events occurring beyond $x_4$ are omitted because they are  equivalent to previously defined sub-cases. For example, for CCO, $W = k - 2$, (Figure~\ref{fig:CCO-genericA4} and Figure~\ref{fig:CCO-genericB4}) there is a merge/split event after $x_{4}$, but this event corresponds to $W = k$ (Figure~\ref{fig:CCO-genericA2} and Figure~\ref{fig:CCO-genericB2}). Together, all diagrams capture all possible values for $Z$ and $W$. Furthermore, the decomposition is invariant under simultaneous increases in obstruction and the value of $k$: adding two walls and increasing the threshold to $k+2$ results in the same visibility transitions as the original $k$-visibility configuration.

While $Z$ and $W$ are shown together in the figures (Figure~\ref{fig:CCO-genericA0} -~\ref{fig:CCO-genericB6} and those in the appendix), they are independent parameters. A single configuration might, for example, satisfy $Z = k-1$ and $W = k-2$ simultaneously. 


\label{CCOMainBody}



\subsection{Convex-Convex-Opposite}
\label{section:cco}

\begin{lemma}
\label{lemma:CCO-main}
Consider two critical vertices that are of the form CCO ($v_{1}$ and $v_{2}$), When $Z = k-1$ (Figure~\ref{fig:CCO-genericA2} and Figure~\ref{fig:CCO-genericB2}), an appear/disappear event occurs at $v_{2}$; When $W=k$, a merge/split event occurs at $x_{3}x_{4}$ (Figure~\ref{fig:CCO-genericA2} and Figure~\ref{fig:CCO-genericB2}). When $W=k-2$, an appear/disappear event occurs at $x_{3}x_{4}$ (Figure~\ref{fig:CCO-genericA4} and Figure~\ref{fig:CCO-genericB4}). No event occurs for any other $Z$ or $W$.

\end{lemma}

\begin{proof}
    See Figure~\ref{fig:CCO-genericA0} to Figure~\ref{fig:CCO-genericB6}. 
    When $Z= k-1$, a shadow appears at $v_{2}$ when the agent crosses $x_{1}x_{2}$ from $b$ to $a$ as crossing this line $v_{1}$ becomes an obstacle for visibility to see $v_{2}$. Analogously, the disappear event occurs when going from $a$ to $b$ for the shadow that includes $v_{2}$. See Figure~\ref{fig:CCO-genericA2} and Figure~\ref{fig:CCO-genericB2}. \\
    When $W=k$, a merge event occurs at $x_{3}x_{4}$ when the agent moves from $b$ to $a$ because $v_{2}$ and $v_{1}$ block line of sight. A split event occurs when going from $a$ to $b$. 
    
    When $W=k-2$, an appear event occurs at $x_{3}x_{4}$ when moving from $b$ to $a$, because $v_{1}$ and $v_{2}$ both obstruct line of sight, and a disappear event occurs when moving from $a$ to $b$.

For $Z \geq k + 3$, $v_{2}$ and its immediate surrounding (i.e locally incident region in $P$) is entirely in shadow. For $W \geq k + 4$, $x_{3}x_{4}$ and its immediate surroundings are entirely in shadow.

For $Z \leq k - 7$, $v_{2}$ is and its surrounding is entirely visible. For $W \leq k - 6$, $x_{3}x_{4}$ and its surrounding is entirely visible. 
    
\end{proof}


 \begin{figure}[H]
\centering
\includegraphics[width=\linewidth]{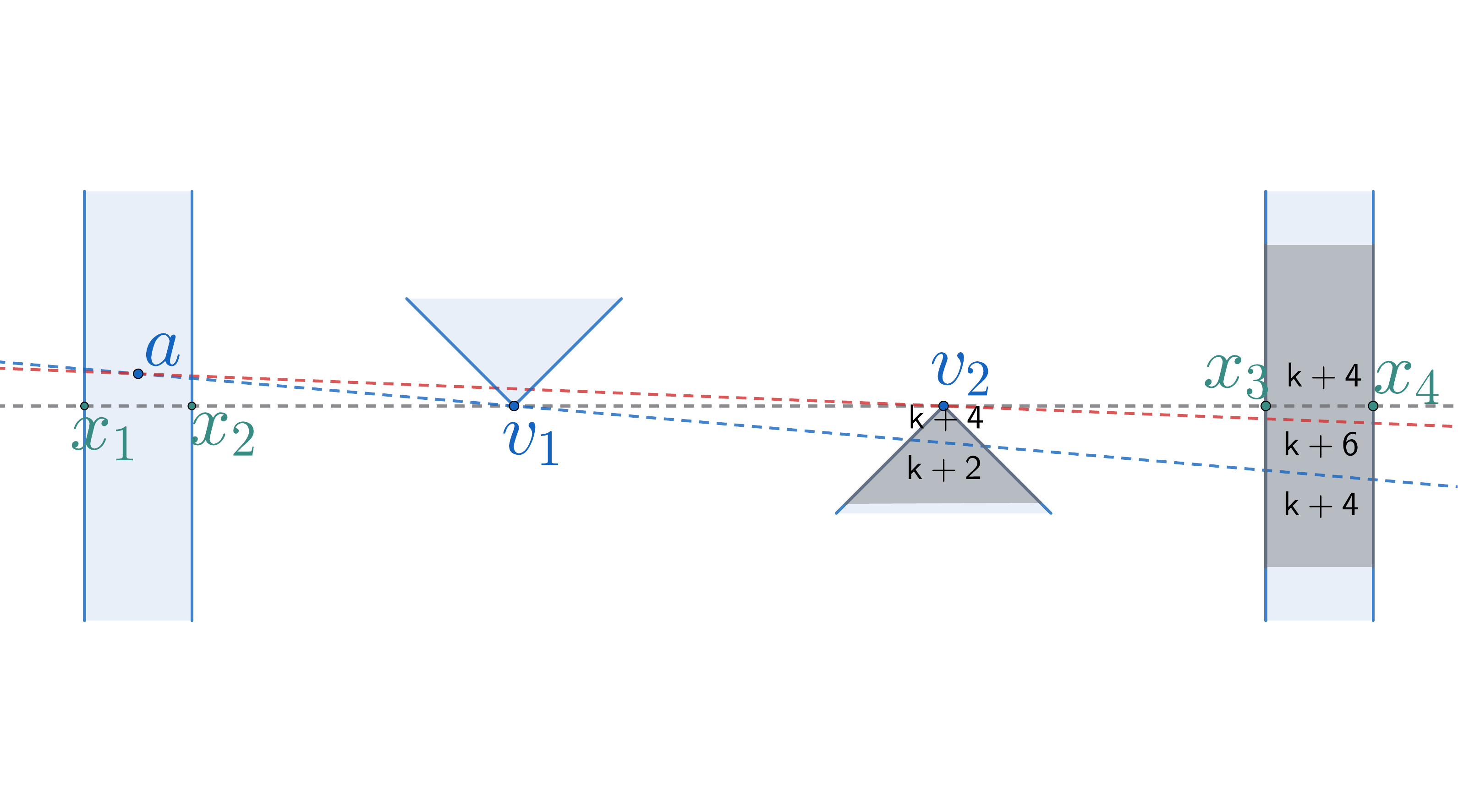}
\caption{CCO; Above $\ell_{g}$; $Z = k + 1$, $W = k + 2$. The dashed gray, red, and blue lines are $\ell_{g}$, $\ell_{r}$, and $\ell_{b}$ respectively.}\label{fig:CCO-genericA0}
\end{figure}
\begin{figure}[H]
\includegraphics[width=\linewidth]{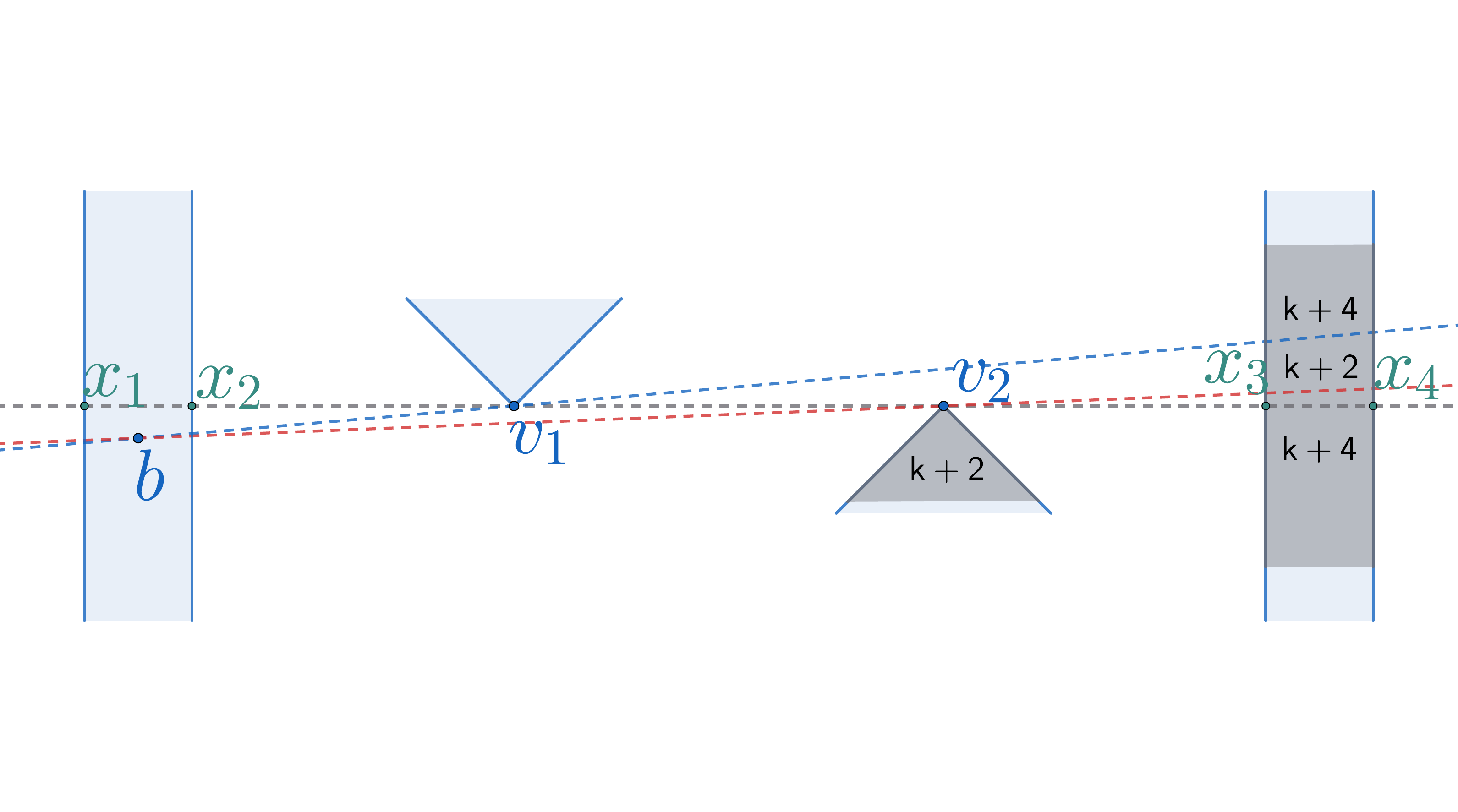}
\caption{ CCO; Below $\ell_{g}$. $Z = k + 1$, $W = k + 2$.} \label{fig:CCO-genericB0} 
\end{figure}


 \begin{figure}[H]
\centering
\includegraphics[width=\linewidth]{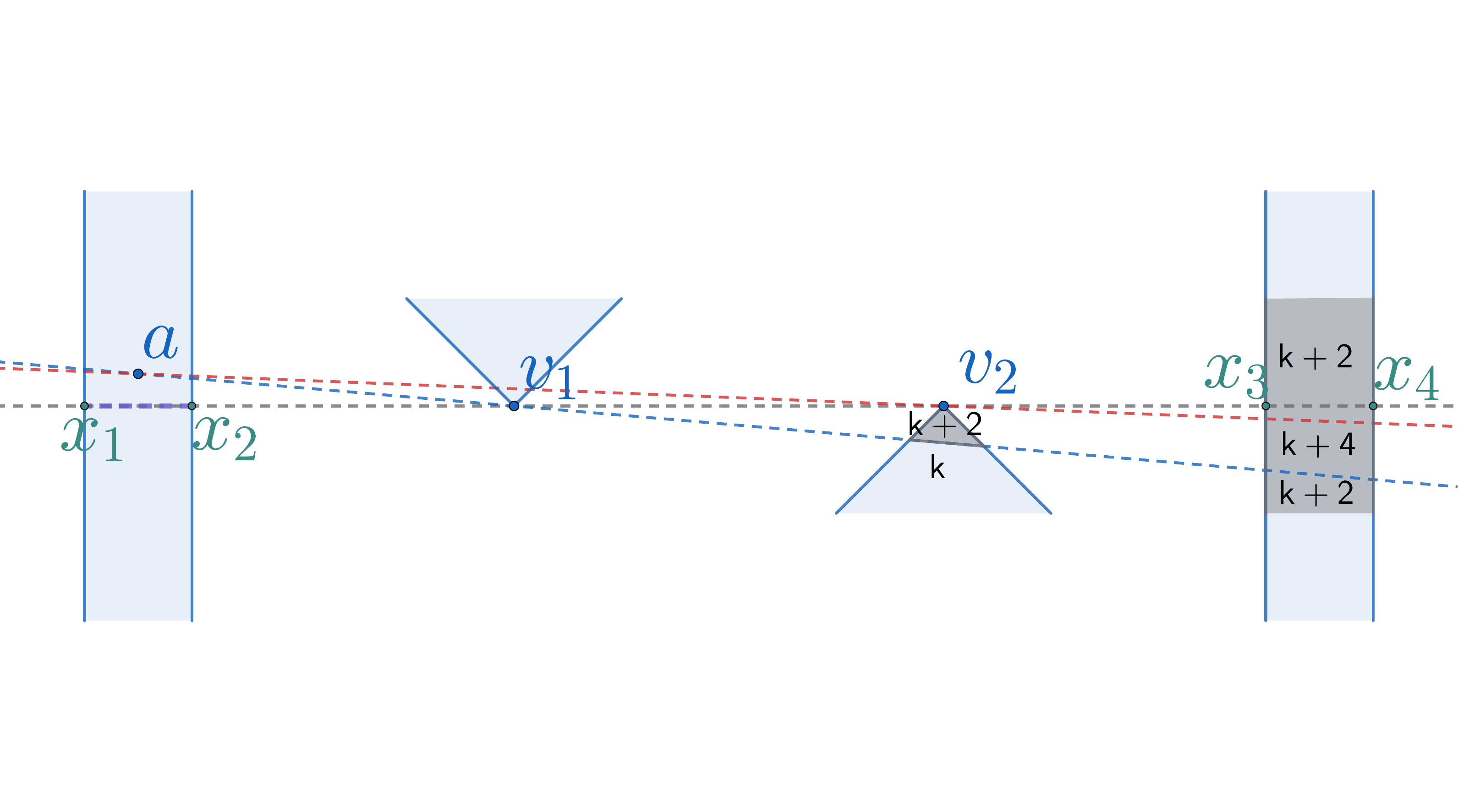}
\caption{CCO; Above $\ell_{g}$. $Z = k - 1$ (appear/disappear), $W = k$ (merge/split).}\label{fig:CCO-genericA2}
\end{figure}
\begin{figure}[H]
\includegraphics[width=\linewidth]{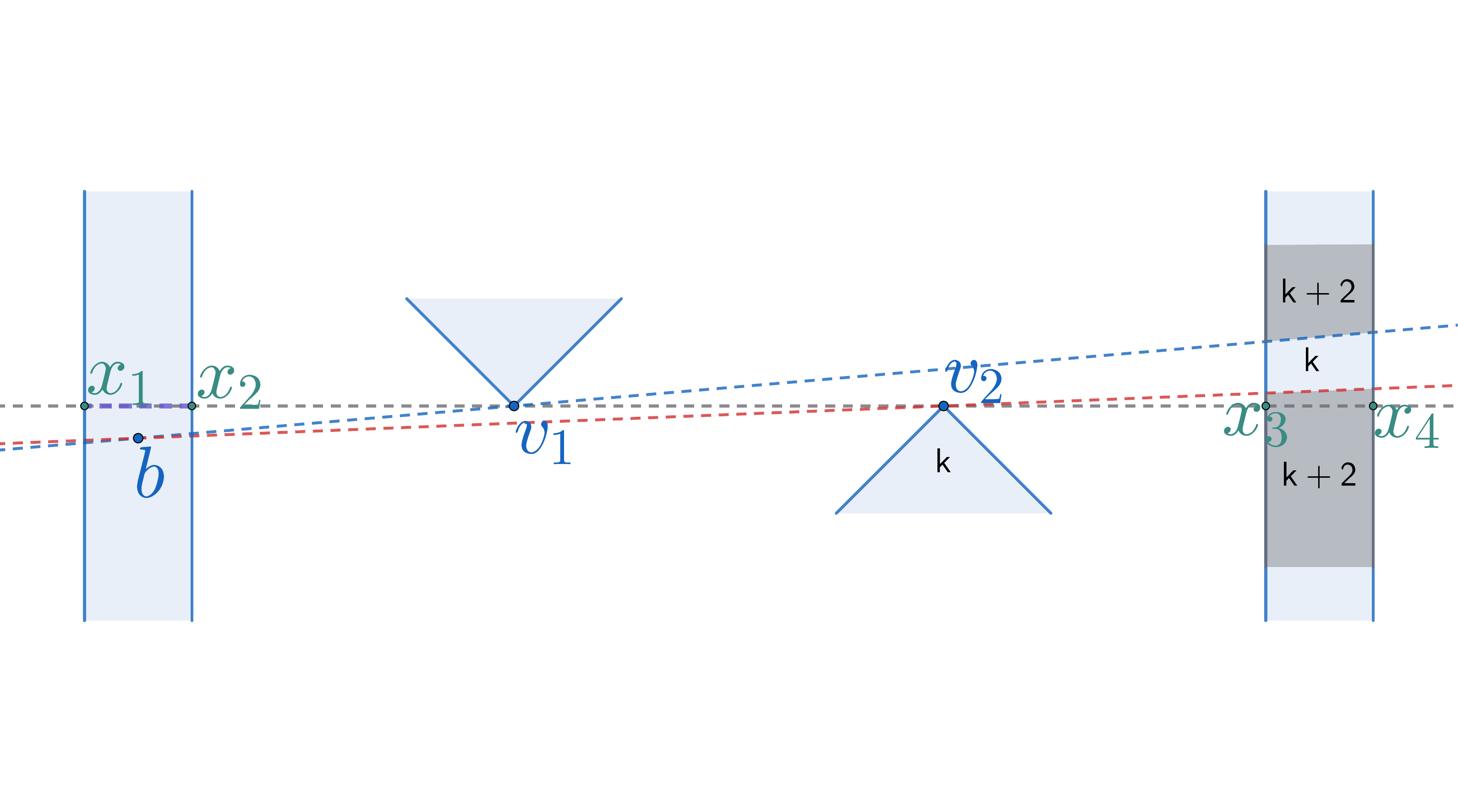}
\caption{CCO; Below $\ell_{g}$; $Z = k - 1$ (appear/disappear), $W = k$ (merge/split).}\label{fig:CCO-genericB2}
\end{figure}

    


 \begin{figure}[H]
\centering
\includegraphics[width=\linewidth]{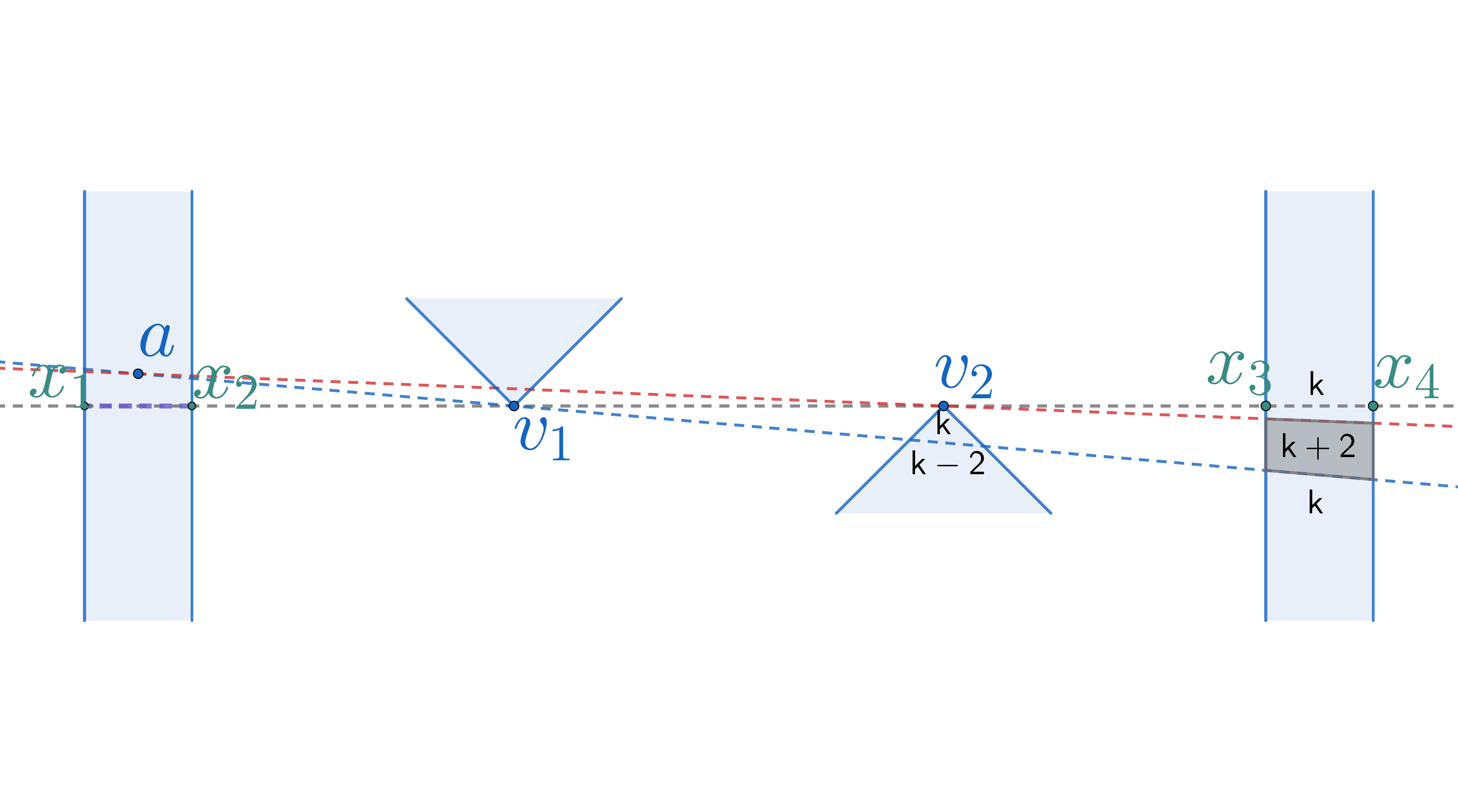}
\caption{CCO; Above $\ell_{g}$ $Z = k - 3$, $W = k - 2$ (appear/disappear) .}\label{fig:CCO-genericA4}
\end{figure}
 \begin{figure}[H]
\includegraphics[width=\linewidth]{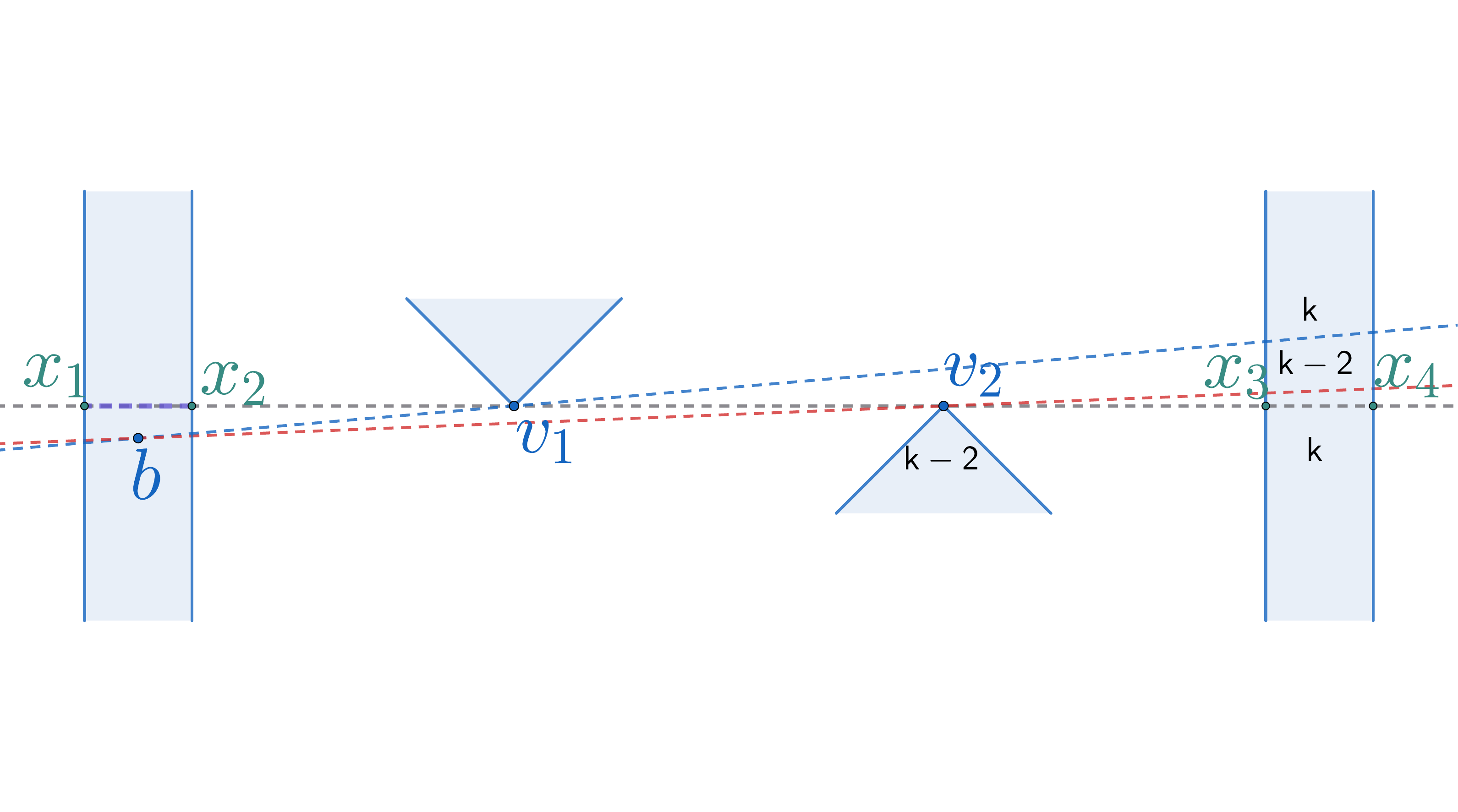}
\caption{CCO; Below $\ell_{g}$; $Z = k - 3$, $W = k - 2$ (appear/disappear).}\label{fig:CCO-genericB4}
\end{figure}


 \begin{figure}[H]
\centering
\includegraphics[width=\linewidth]{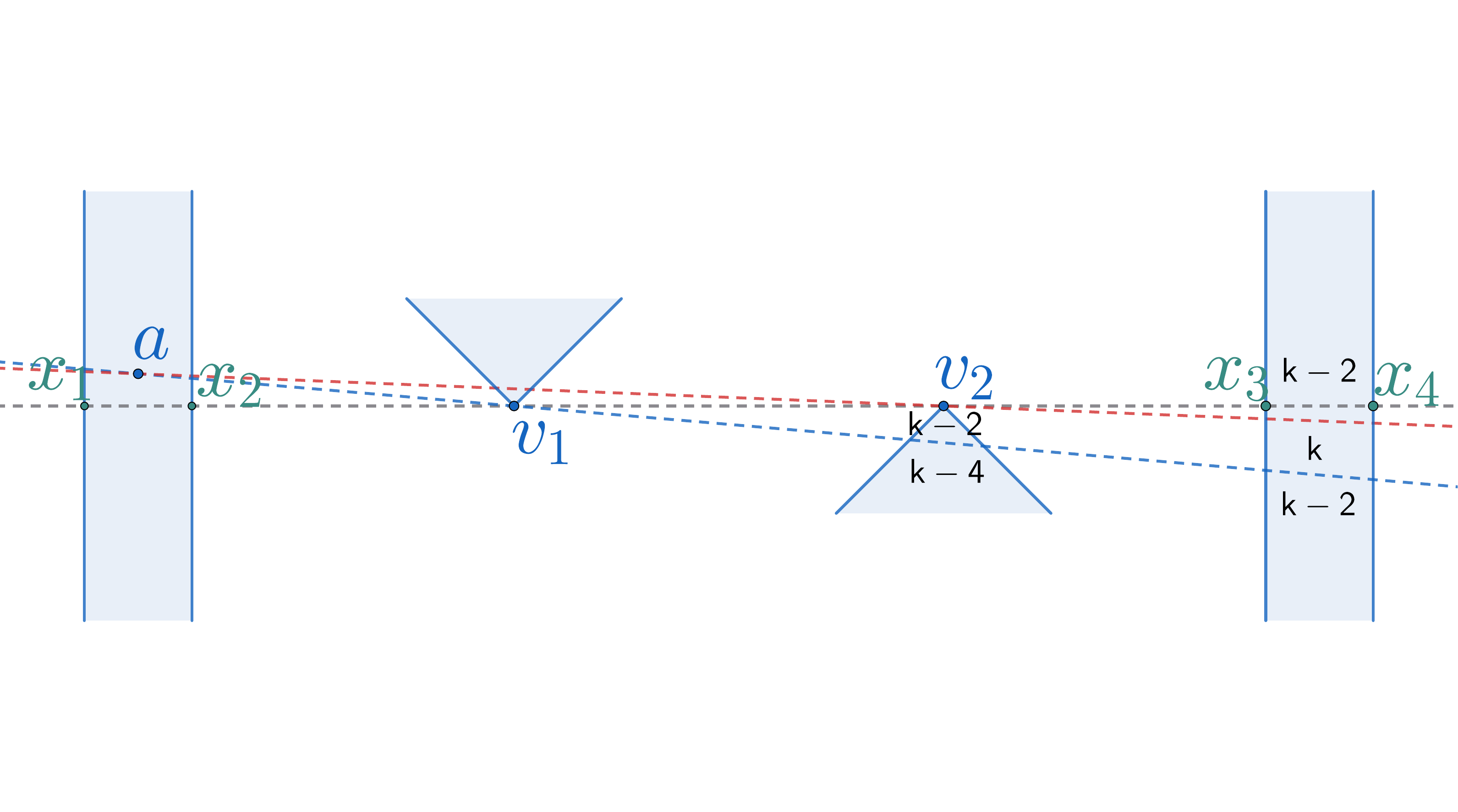}
\caption{CCO; Above $\ell_{g}$; $Z = k - 5$, $W = k - 4$.}\label{fig:CCO-genericA6}
\end{figure}
\begin{figure}[H]
\includegraphics[width=\linewidth]{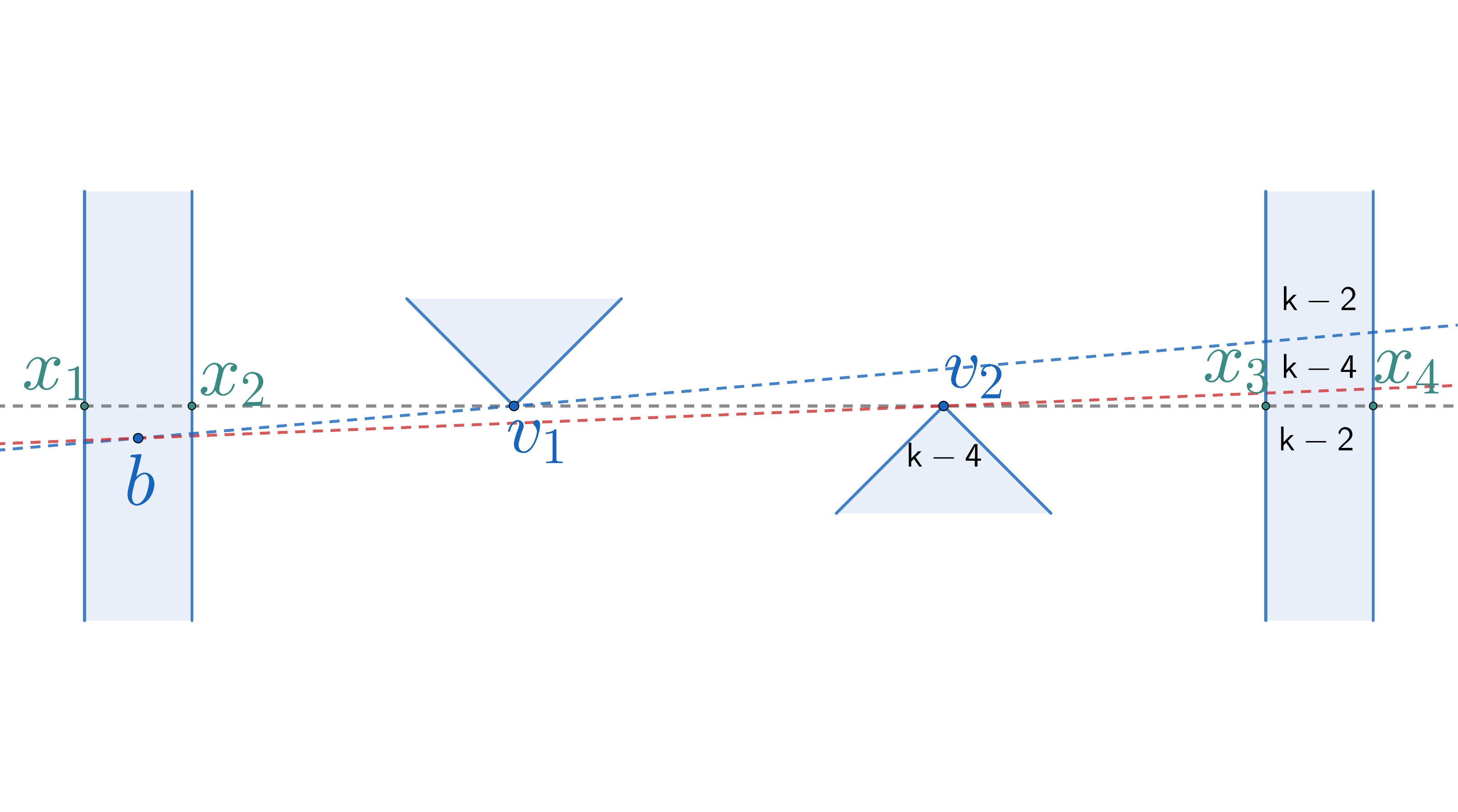}
\caption{CCO; Below $\ell_{g}$; $Z = k - 5$, $W = k - 4$.}\label{fig:CCO-genericB6}
\end{figure}



\subsection{Other Cases}


    The possible combinatorial changes which can occur in each case are covered in the following lemmas (due to lack of space, the proofs are moved to the appendix; a table summarizing all cases when an event occurs is also in the appendix):
    \begin{lemma}
    \label{lemma:CCS-main}
        No event occurs for CCS for any $Z$ or $W$. 
    \end{lemma}
    \begin{proof}
        See Appendix~\ref{section:ccs}.
    \end{proof}
    \begin{lemma}
    \label{lemma:CRO-main}
    For CRO, when $Z = k$, a merge/split event occurs at $v_{2}$ (Figure~\ref{fig:CRS-generic2}). If $Z = k - 2$, an appear/disappear event occurs at $v_{2}$ (Figure~\ref{fig:CRS-generic4}). Additionally, if $W = k$, a merge/split event occurs at $x_{3}x_{4}$ (Figure~\ref{fig:CRS-generic4}). Also, if $W = k - 2$, an appear/disappear event occurs at $x_{3}x_{4}$. Finally, when $Z = k - 4$, an appear/disappear event occurs at $x_{3}x_{4}$ (Figure~\ref{fig:CRS-generic6}). 

\end{lemma}
\begin{proof}
    See Appendix~\ref{appendix:CRO}.
\end{proof}
\begin{lemma}
\label{lemma:CRS-main}
    For CRS, when $Z = k$, there is a merge/split event at $v_{2}$. No event occurs for all other $Z$ and $W$. 
\end{lemma}
\begin{proof}
    See Appendix~\ref{appendix:CRS}.
\end{proof}
\begin{lemma} 
No event occurs for RCS for any $Z$ or $W$. 

\end{lemma}
\begin{proof}
    See Appendix~\ref{appendix:RCS}.
\end{proof}
\begin{lemma} 
\label{lemma:RCO-main}
For RCO, when $Z = k - 1$ (Figure~\ref{fig:RCO-generic-2}), an appear/disappear event occurs at $v_{2}$. If $W = k$, a merge/split event occurs at $x_{3}x_{4}$ (Figure~\ref{fig:RCO-generic-2}). Additionally, if $W=k-3$, an appear/disappear event occurs at $x_{3}x_{4}$ (Figure~\ref{fig:RCO-generic-4}). No event occurs for any other $Z$ or $W$. 

\end{lemma}

\begin{proof}
    See Appendix~\ref{appendix:RCO}.
\end{proof}

\begin{lemma} 
\label{lemma:RRS-main}
For RRS, When $Z = k$ (Figure~\ref{fig:RRS-generic-0}), a merge/split event occurs at $v_{2}$. No event occurs for any other $Z$ or $W$. 

\end{lemma}

\begin{proof}
    See Appendix~\ref{appendix:RRS}.
\end{proof}
\begin{lemma} 
\label{lemma:RRO-main}
For RRO, when $Z = k$ (Figure~\ref{fig:RRO-generic-0}), a merge/split event occurs at $v_{2}$. If $Z = k - 2$, an appear/disappear event occurs at $v_{2}$ (Figure~\ref{fig:RRO-generic-2}). Additionally, if $W = k$, a merge/split event occurs at $x_{3}x_{4}$ (Figure~\ref{fig:RRO-generic-2}). Also, if $W = k - 2$, an appear/disappear event occurs at $x_{3}x_{4}$ (Figure~\ref{fig:RRO-generic-4}). No event occurs for any other $Z$ or $W$. 

\end{lemma}
\begin{proof}
    See Appendix~\ref{appendix:RRO}.
\end{proof}
\begin{lemma} 
\label{lemma:RC-main}
 For RC-SC (Reflex-Convex Special Case), when $Z=k$, there is an appear/disappear event at $v_{2}$. (Note: there can be walls between $v_{1}$ and $v_{2}$). If $W = k-1$, an appear/disappear event occurs at $x_{3}x_{4}$. No event occurs for any other $Z$ or $W$. 

\end{lemma}

\begin{proof}
    See Appendix~\ref{appendix:RC-SC}.
\end{proof}
\begin{lemma} 
\label{lemma:RR-main}
  For RR-SC (Reflex - Reflex Special Case), No event occurs for any $Z$ or $W$. 

\end{lemma}
\begin{proof}
    See Appendix~\ref{appendix:CR-SC}.
\end{proof}
\begin{lemma} 
  For CR-SC (Convex - Reflex Special Case), when $Z = k - 1$, an appear/disappear event occurs at $v_{2}$. If $W = k - 1$, an appear/disappear event occurs at $x_{3}X_{4}$. No event occurs for any other $Z$ or $W$.

\end{lemma}
\begin{proof}
    See Appendix~\ref{appendix:CR-SC}.
\end{proof}
\begin{lemma} 
\label{lemma:CC-main}
   For CC-SC (Convex - Convex Special Case), No event occurs for any $Z$ or $W$. 

\end{lemma}

\begin{proof}
    See Appendix~\ref{appendix:CC-SC}.
\end{proof}

\section{Complexity}
\begin{theorem}
The proposed cell decomposition has complexity $O(n^4)$
\end{theorem}

\begin{proof} In the worst case, every pair of vertices in $P$ is mutually critical, generating $O(n^2)$ partition lines. The arrangement of these lines results in at most $O(n^4)$ cells.
\end{proof}

The best previous decomposition by Bahoo et al. achieved $O(k^2n^4)$. Our approach achieves $O(n^4)$ by considering only lines where visibility events are guaranteed to occur. When $k = n$, this is  a quadratic improvement of $O(n^2)$. Furthermore, our method generates significantly fewer partition lines in practice (see Figure~\ref{fig:requiredLines}). It must be noted that the cells in this decomposition are not necessarily convex.

 \begin{figure}
\centering
\begin{subfigure}[b]{.49\linewidth}
\includegraphics[width=\linewidth]{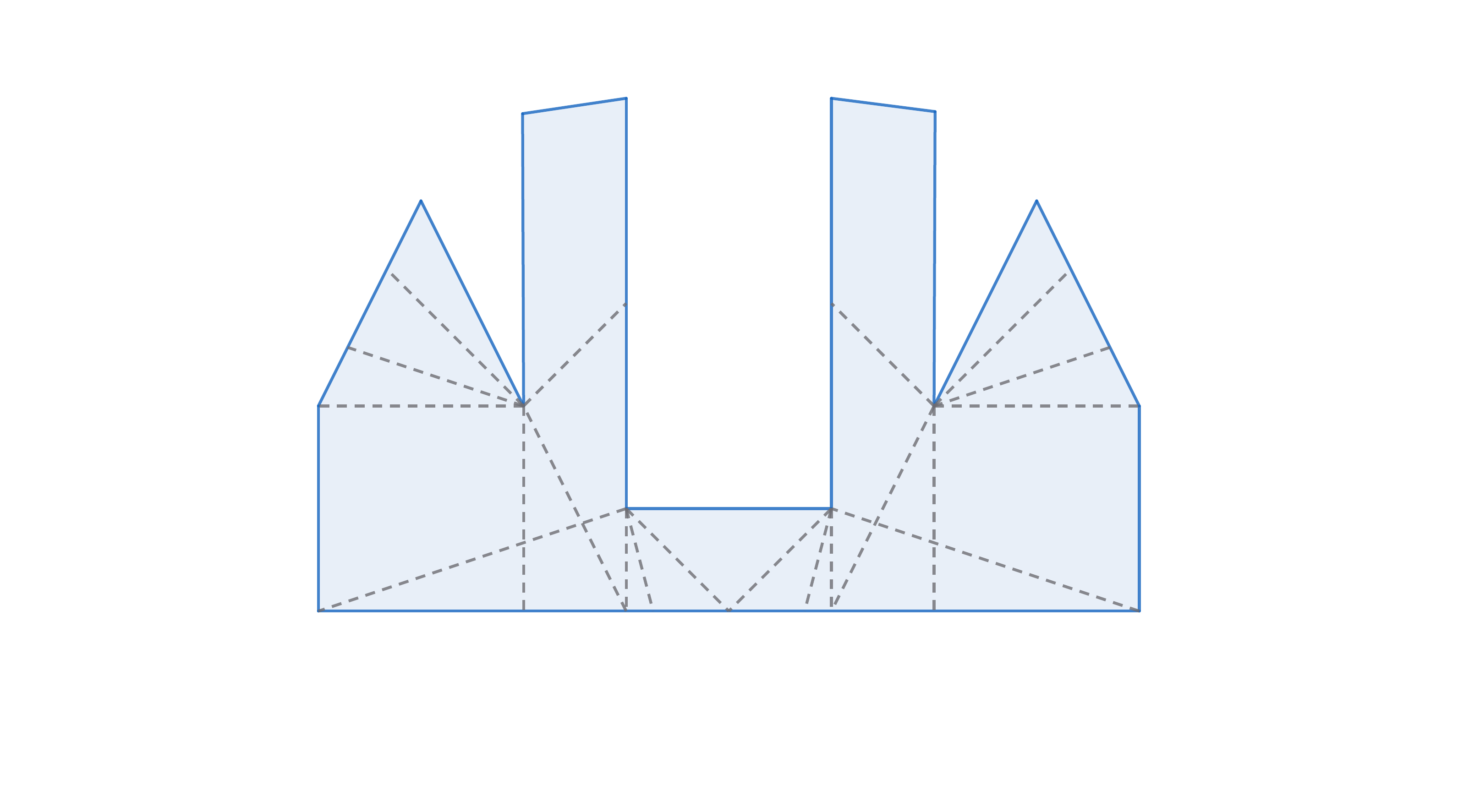}
\caption{Method in~\cite{bahoo2025generalizedkcelldecompositionvisibility}}\label{fig:required2}
\end{subfigure}
\begin{subfigure}[b]{.49\linewidth}
\includegraphics[width=\linewidth]{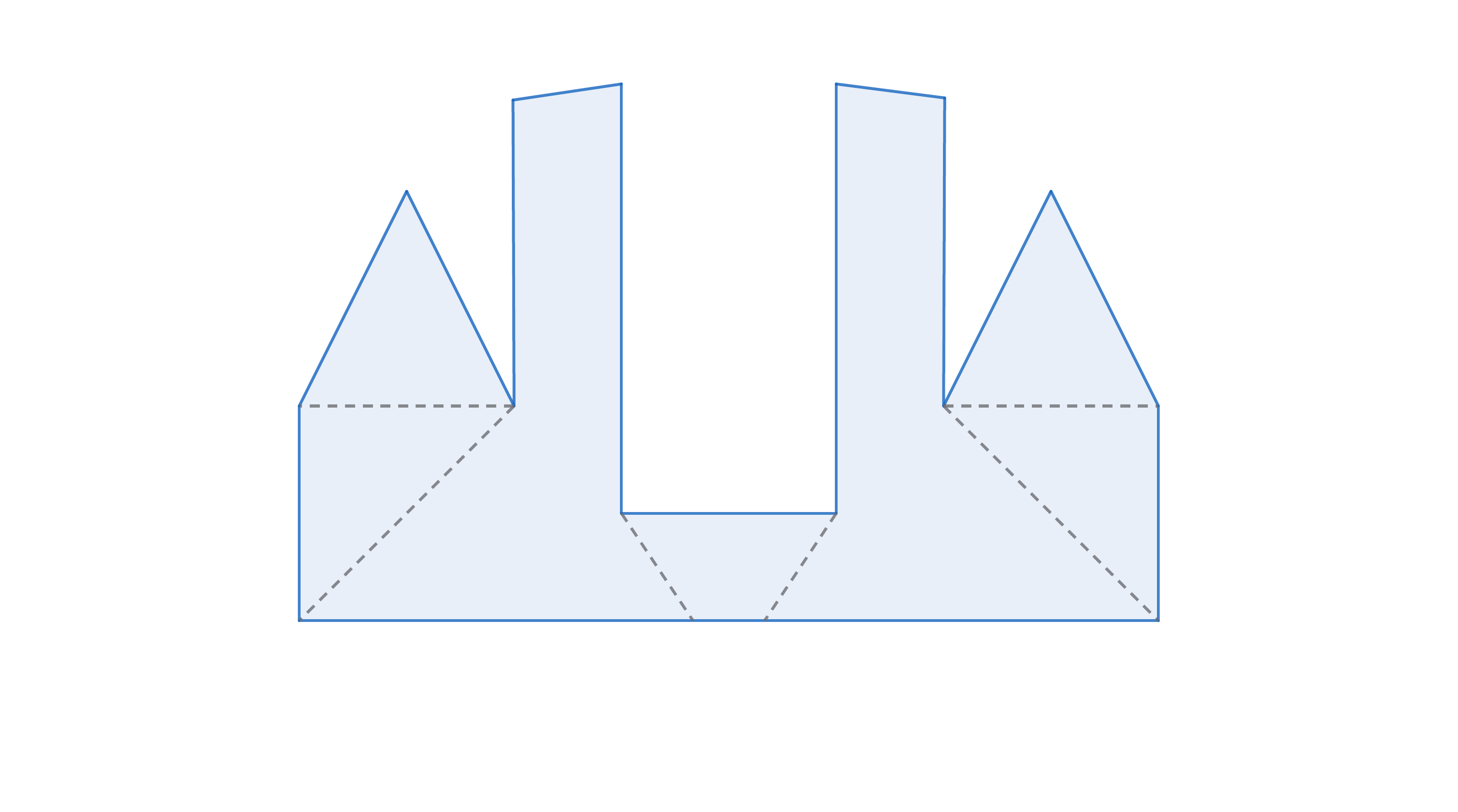}
\caption{Our Method}\label{fig:required1}
\end{subfigure}

\caption{Partition lines of the decomposition for our method and method in~\cite{bahoo2025generalizedkcelldecompositionvisibility} (for $2$-visibility).}
\label{fig:requiredLines}
\end{figure}

\section{Polygons with Holes}

In polygons with holes, a shadow component can ``wrap around" a hole, remaining a single connected component even after a split occurs (see agent at $b$ in Figure~\ref{fig:poly-holes-V2} moving to $a$ in Figure~\ref{fig:poly-holes-V1} and shadow $S_1$ around the purple hole). Consequently, an actual split on a hole only occurs if the shadow is ``broken" by at least two ``split" events that touch the same hole. To maintain an exact decomposition and prune redundant lines corresponding to ``nullified" split events, we use the following procedure: 
\begin{itemize}
    \item For any partition line $\ell_{p1}$ (segment $s_1e_1$; where $s_1$ is the start of the partition line and $e_1$ is the end of the partition line) corresponding to a potential split event that touches a hole, identify its first intersection $i_2$ with another partition line $\ell_{p2}$ whose split event touches the same hole.
    \item When the agent crosses the line between the vertex $s_1$ and the intersection $i_2$, the shadow remains connected around the hole, nullifying the split event. Beyond $i_2$, the second partition line $\ell_{p2}$ (segment $s_2e_2$ breaks this connectivity, allowing an actual split to occur (See agent at $b$ in Figure~\ref{fig:poly-holes-V4} moving to $a$ in Figure~\ref{fig:poly-holes-V3})
    \item Remove the redundant segment $s_1i_2$ from the decomposition. If $\ell_{p1}$ has no such intersection before reaching the hole, the entire segment $s_1e_1$ is removed.
\end{itemize}
This logic applies to both convex and concave holes and, by symmetry, correctly handles merge events.

 \begin{figure}
\centering
\begin{subfigure}[b]{.49\linewidth}
\includegraphics[width=\linewidth]{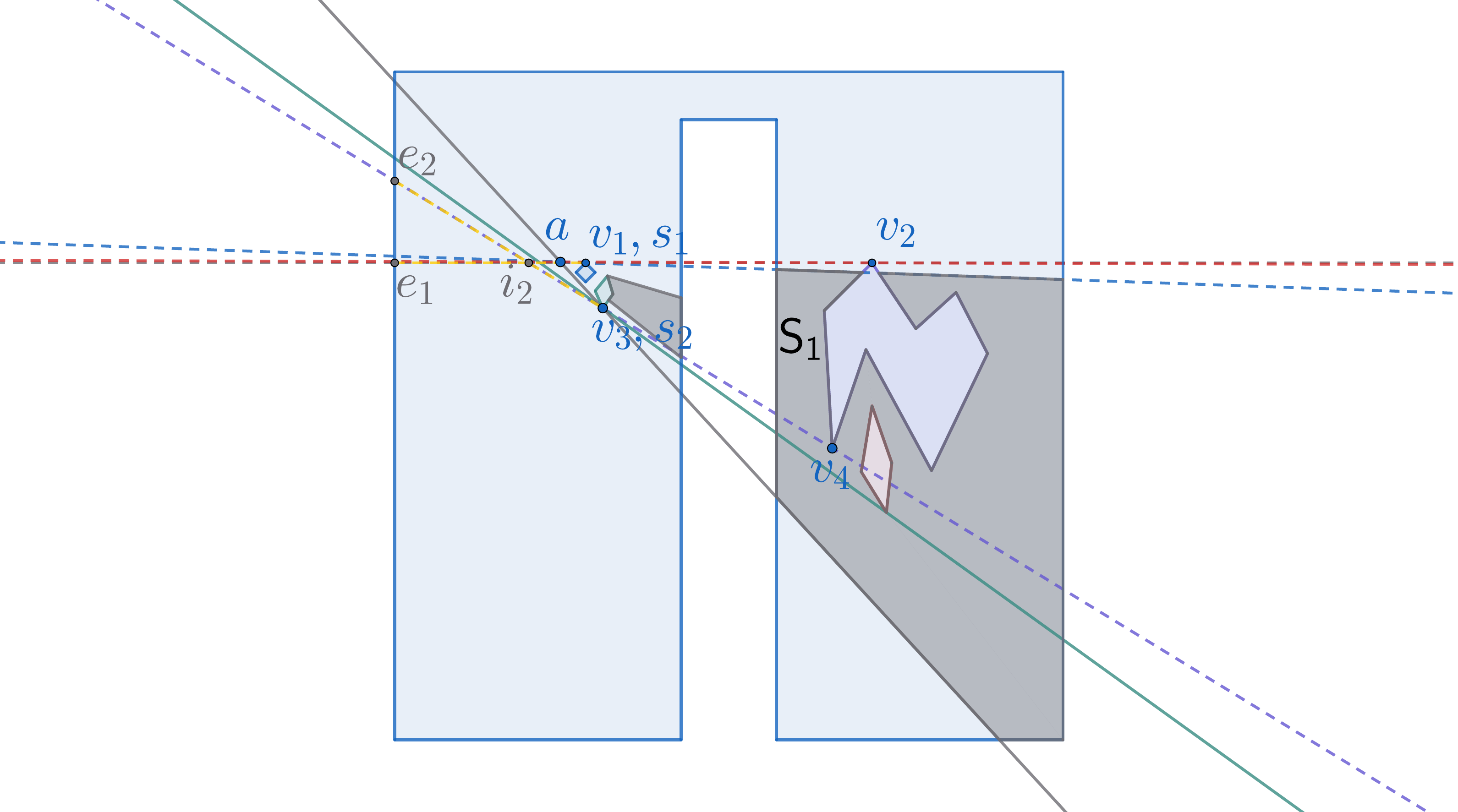}
\caption{Above $l_{p1}$ and above $l_{p2}$}\label{fig:poly-holes-V1}
\end{subfigure}
\begin{subfigure}[b]{.49\linewidth}
\includegraphics[width=\linewidth]{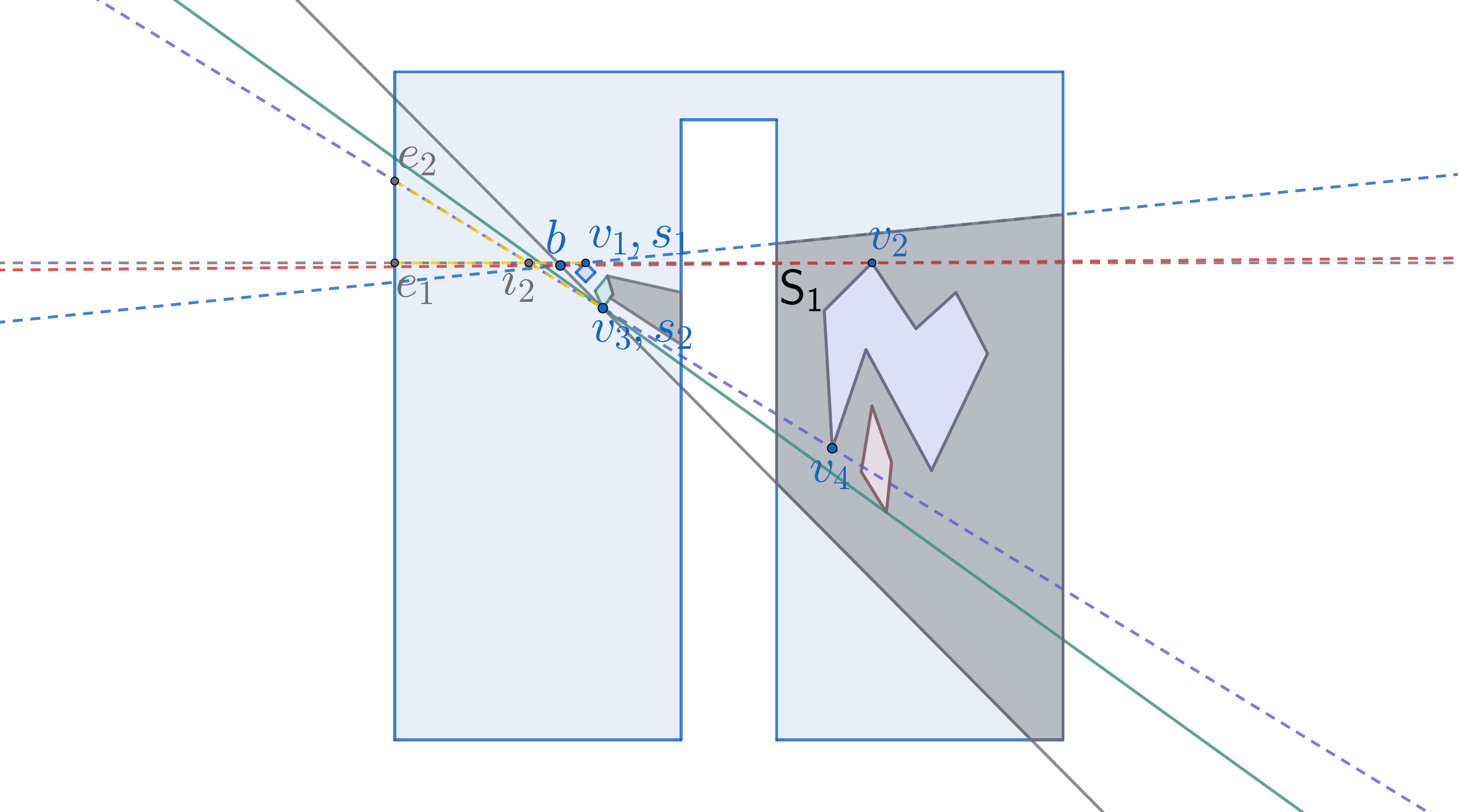}
\caption{Below $l_{p1}$ and above $l_{p2}$}\label{fig:poly-holes-V2}
\end{subfigure}
\begin{subfigure}[b]{.49\linewidth}
\includegraphics[width=\linewidth]{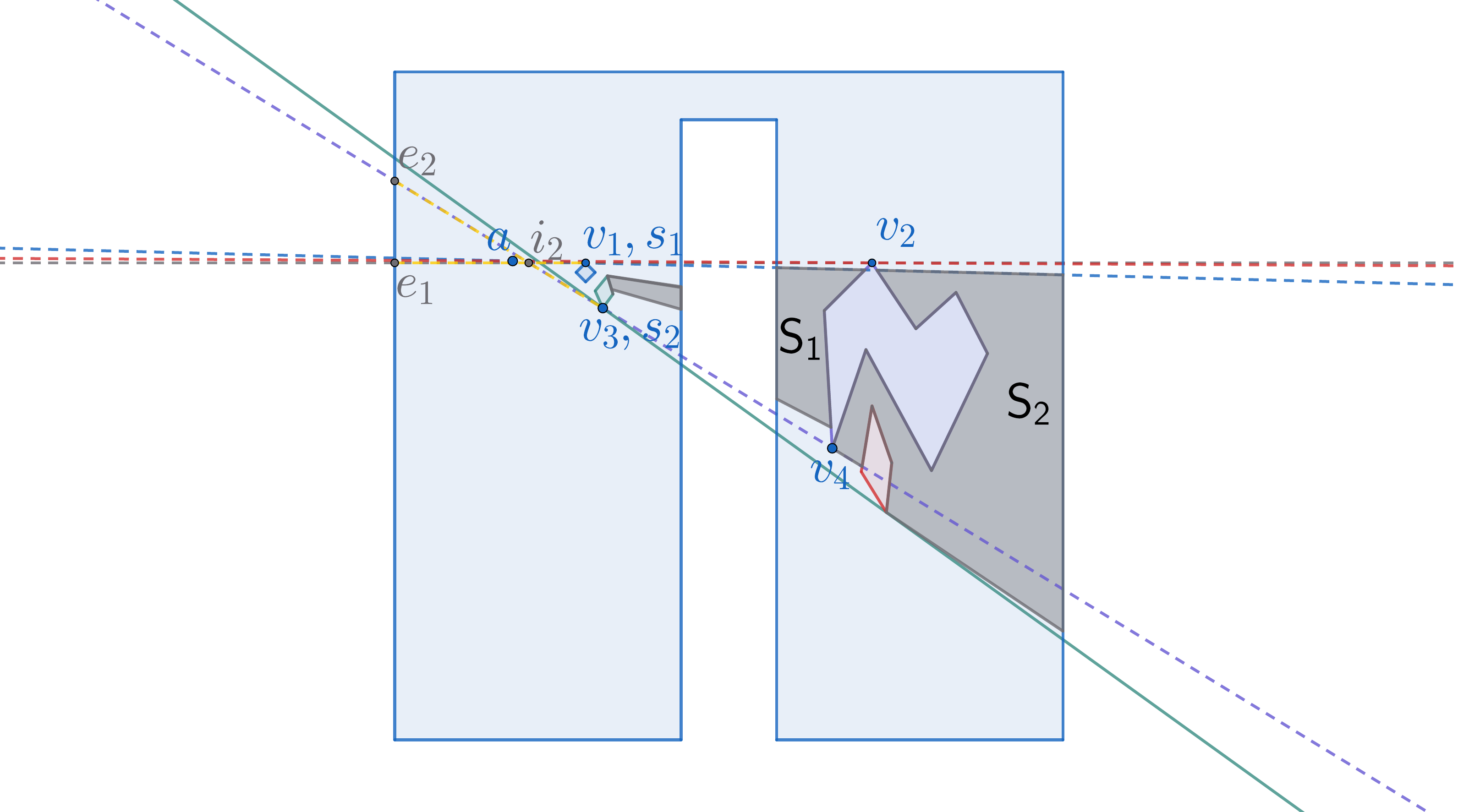}
\caption{Above $l_{p1}$ and below $l_{p2}$}\label{fig:poly-holes-V3}
\end{subfigure}
\begin{subfigure}[b]{.49\linewidth}
\includegraphics[width=\linewidth]{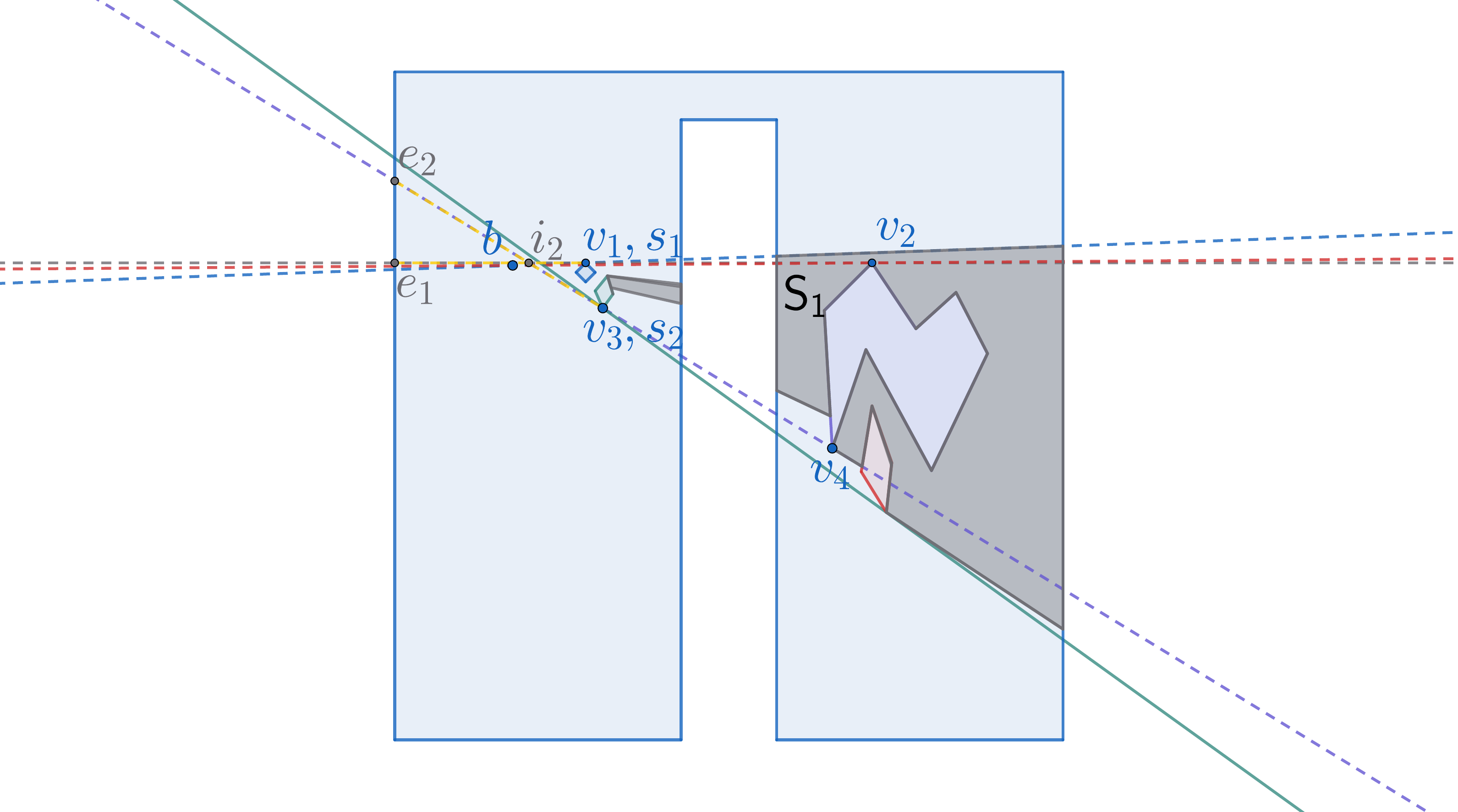}
\caption{Below $l_{p1}$ and below $l_{p2}$}\label{fig:poly-holes-V4}
\end{subfigure}

\caption{A polygon with holes.}
\label{fig:poly-holes-2}
\end{figure}
\section{Future Work and Conclusion}

This paper presents an exact $O(n^4)$ cell decomposition for $k$-visibility in polygonal environments. By analyzing eight fundamental and four special configurations of vertices, we identify precisely where visibility events (i.e. appear, disappear, merge, and split) for shadow components occur. Our approach achieves a potentially quadratic improvement over the previous $O(k^2n^4)$ complexity by eliminating redundant partition lines. Furthermore, the cell decomposition presented in this work can be slightly modified to efficiently query the $k$-visibility polygon of a point.

\bibliographystyle{abbrv}

\bibliography{bibliography}

\newpage
\appendix
\section{CCS}
\label{appendix:CCS}

\subsection{Convex Convex Same (CCS)}
\label{section:ccs}
\begin{lemma} 
\label{lemma:CCS}
No event occurs for CCS for any $Z$ or $W$.

\end{lemma}
\begin{proof}
    See Figure~\ref{fig:CCS-generic-0} to Figure~\ref{fig:CCS-generic-8}. Looking at the counts of when the agent is at $a$ and when the agent is at $b$, it can be seen that no event occurs for any $Z$ or $W$.

For $Z \geq k + 3$, $v_{2}$ and its surroundings (i.e locally incident region in $P$) is entirely in shadow. For $W \geq k + 4$, $x_{3}x_{4}$ and its surroundings is entirely in shadow.

For $Z \leq k - 9$, $v_{2}$ and its surroundings is entirely visible. For $W \leq k - 8$, $x_{3}x_{4}$ and its surroundings is entirely visible. 
\end{proof}
 \begin{figure}[H]
\centering
\begin{subfigure}[b]{.49\linewidth}
\includegraphics[width=\linewidth]{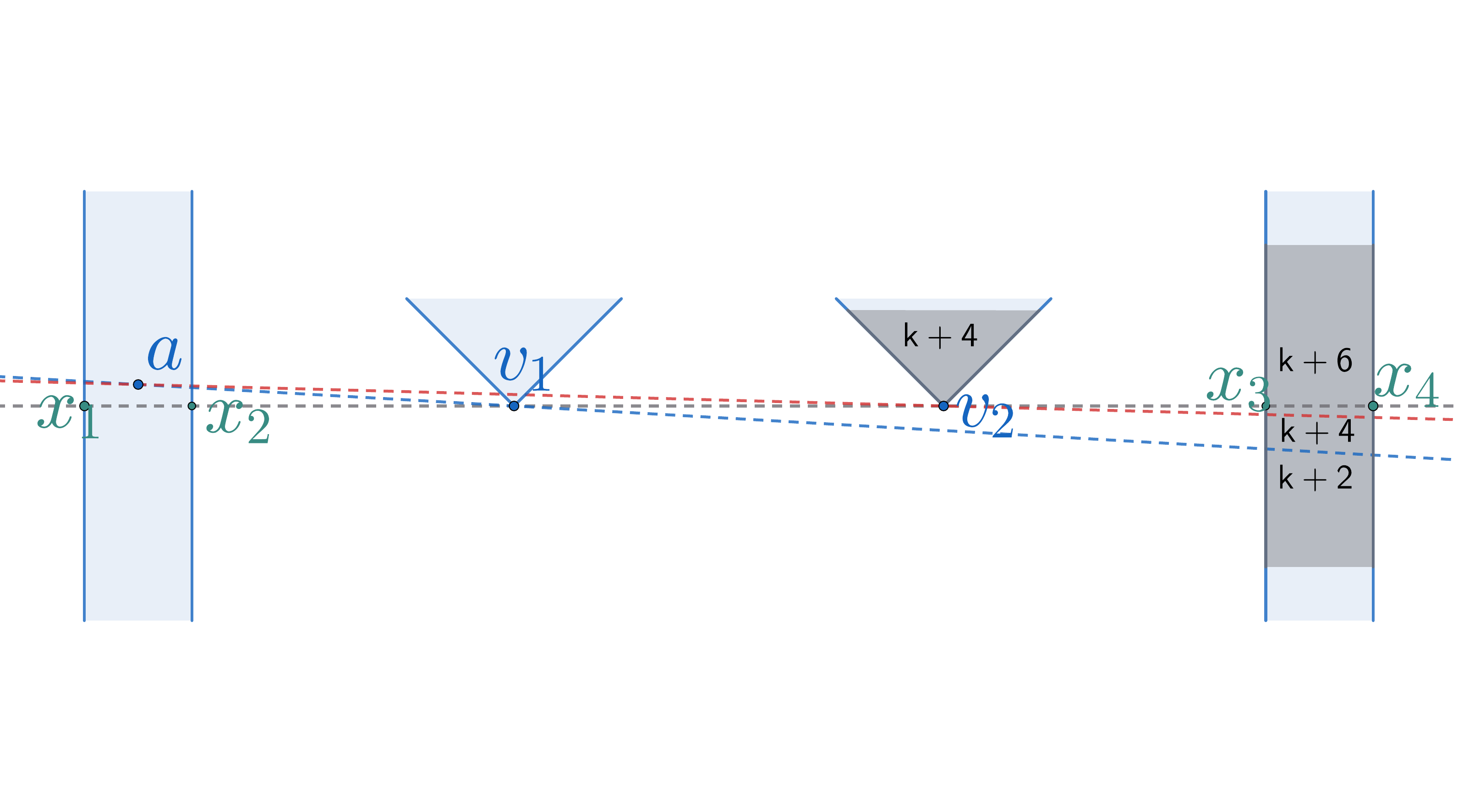}
\caption{Above $\ell_{g}$}\label{fig:CCS-genericA0}
\end{subfigure}
\begin{subfigure}[b]{.49\linewidth}
\includegraphics[width=\linewidth]{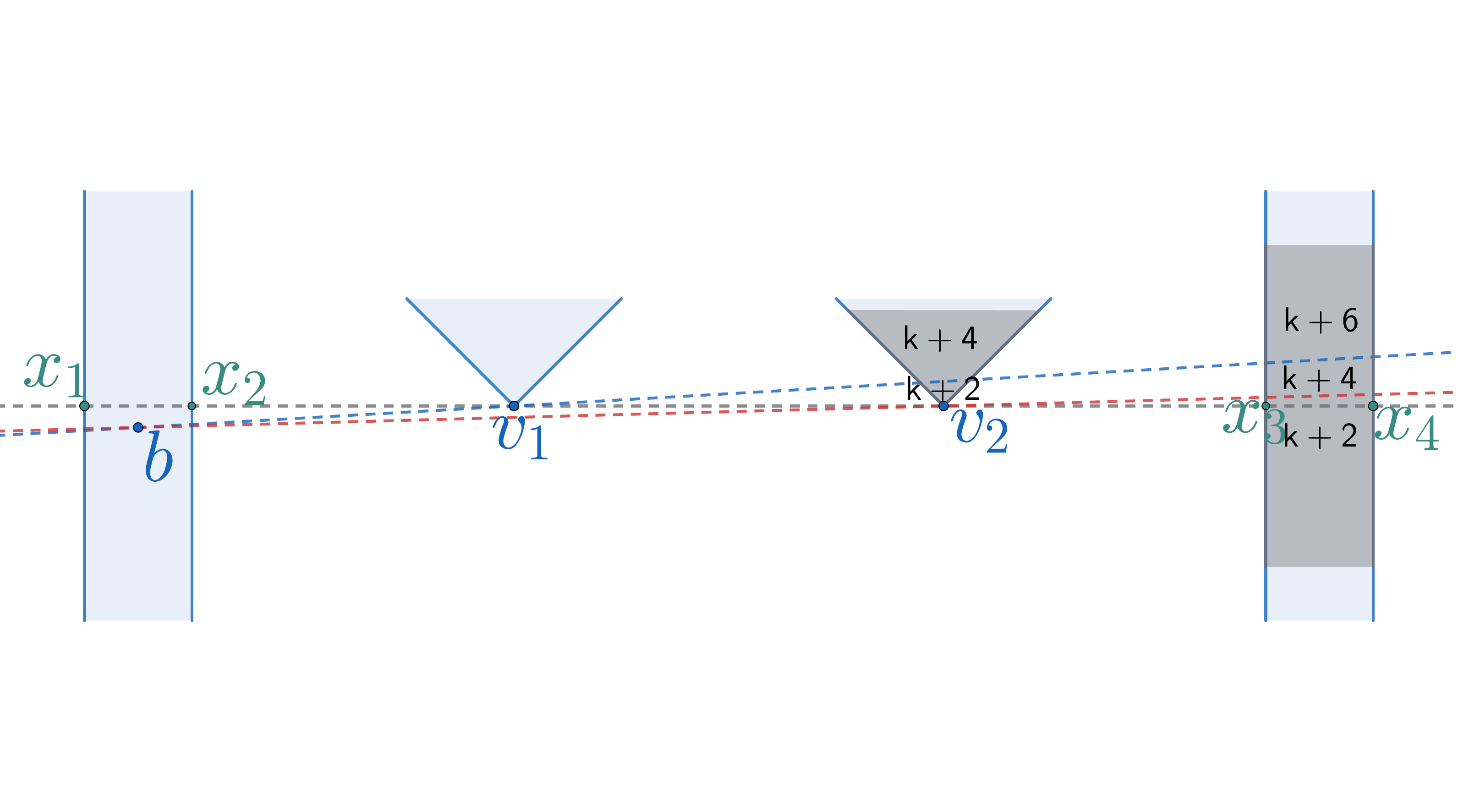}
\caption{Below $\ell_{g}$}\label{fig:CCS-genericB0}
\end{subfigure}

\caption{CCS; $Z = k + 1$, $W = k + 2$.}
\label{fig:CCS-generic-0}
\end{figure}

 \begin{figure}[H]
\centering
\begin{subfigure}[b]{.49\linewidth}
\includegraphics[width=\linewidth]{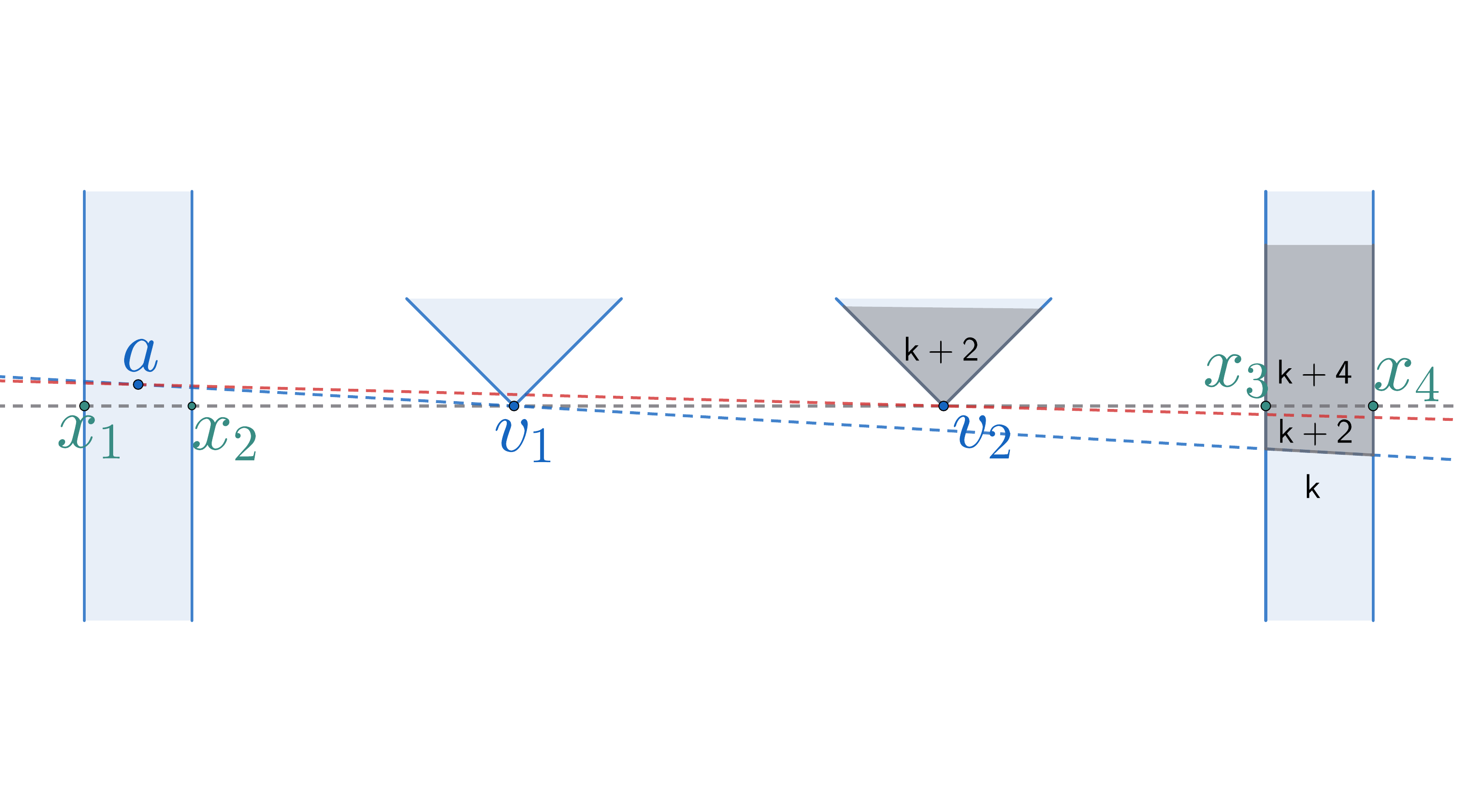}
\caption{Above $\ell_{g}$}\label{fig:CCS-genericA2}
\end{subfigure}
\begin{subfigure}[b]{.49\linewidth}
\includegraphics[width=\linewidth]{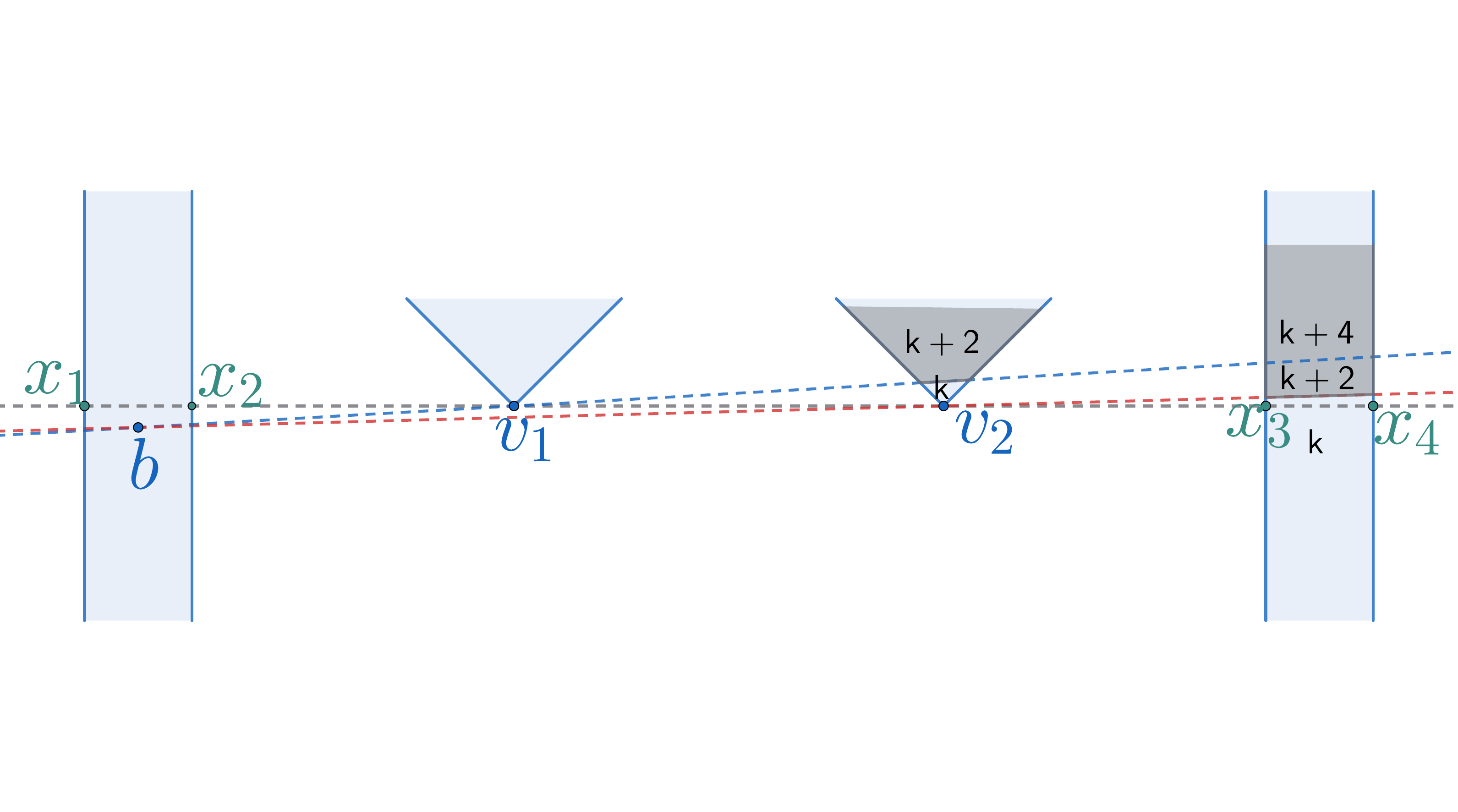}
\caption{Below $\ell_{g}$}\label{fig:CCS-genericB2}
\end{subfigure}

\caption{CCS; $Z = k - 1$, $W = k$.}
\label{fig:CCS-generic-2}
\end{figure}

 \begin{figure}[H]
\centering
\begin{subfigure}[b]{.49\linewidth}
\includegraphics[width=\linewidth]{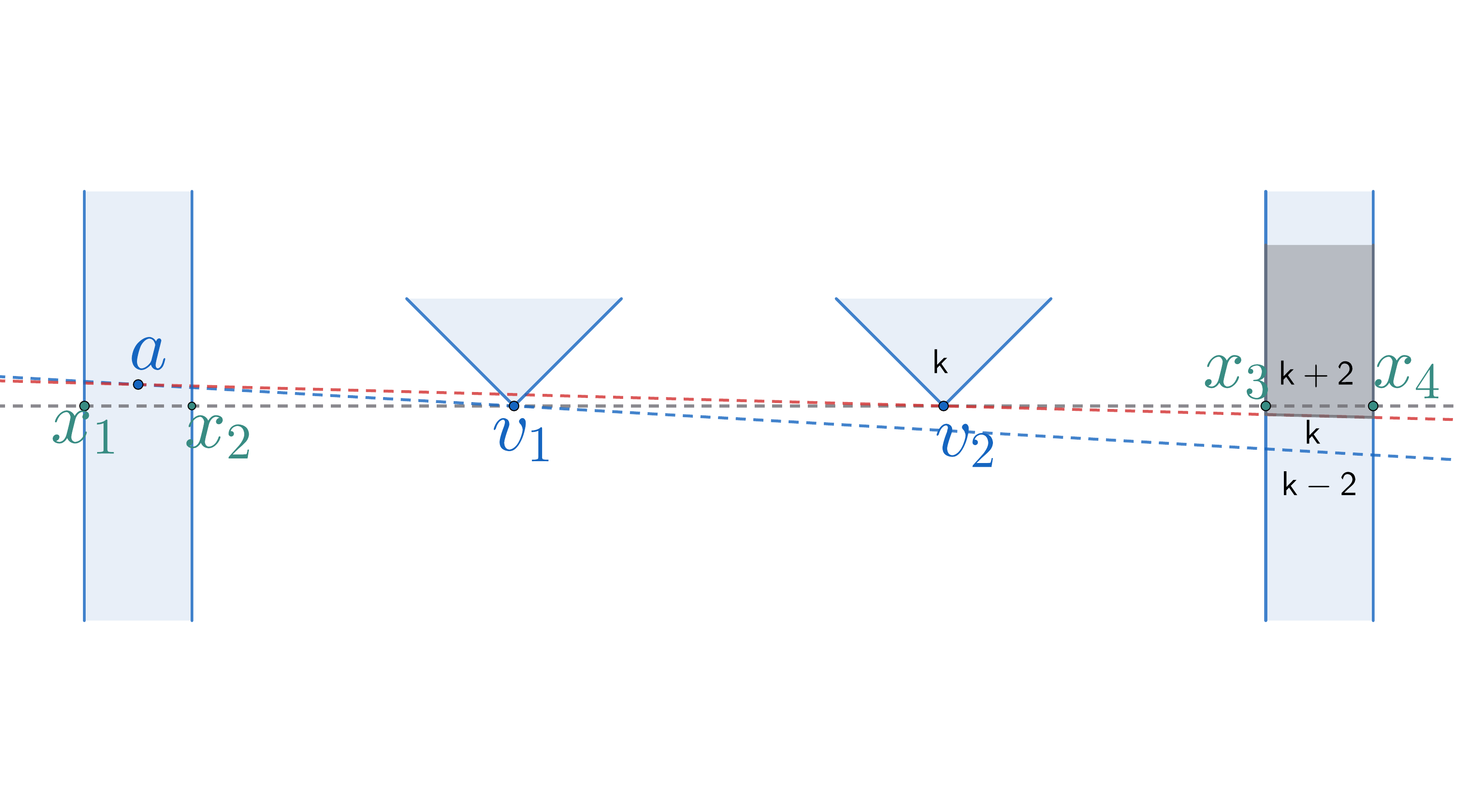}
\caption{Above $\ell_{g}$}\label{fig:CCS-genericA4}
\end{subfigure}
\begin{subfigure}[b]{.49\linewidth}
\includegraphics[width=\linewidth]{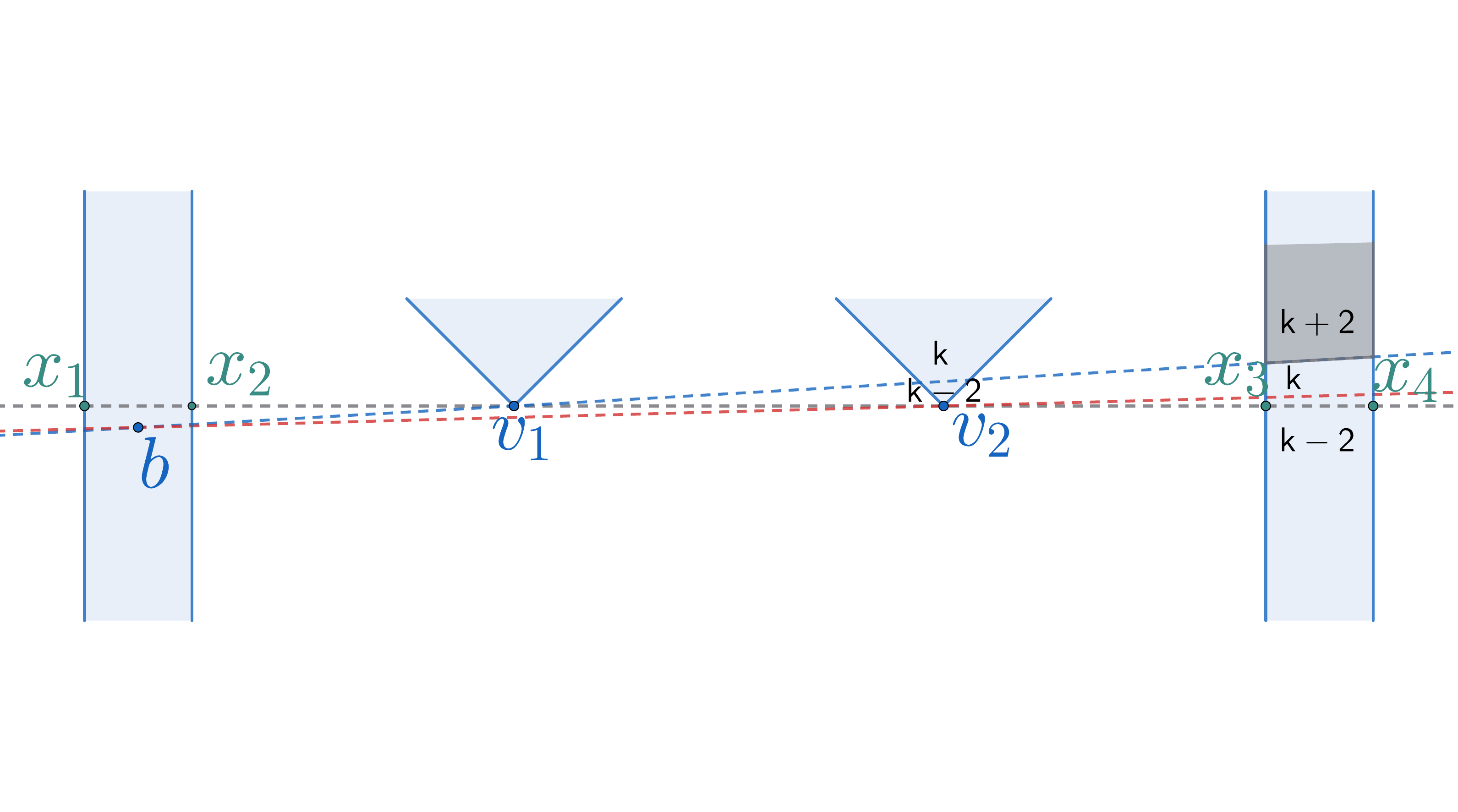}
\caption{Below $\ell_{g}$}\label{fig:CCS-genericB4}
\end{subfigure}

\caption{CCS; $Z = k - 3$, $W = k - 2$.}
\label{fig:CCS-generic-4}
\end{figure}

 \begin{figure}[H]
\centering
\begin{subfigure}[b]{.49\linewidth}
\includegraphics[width=\linewidth]{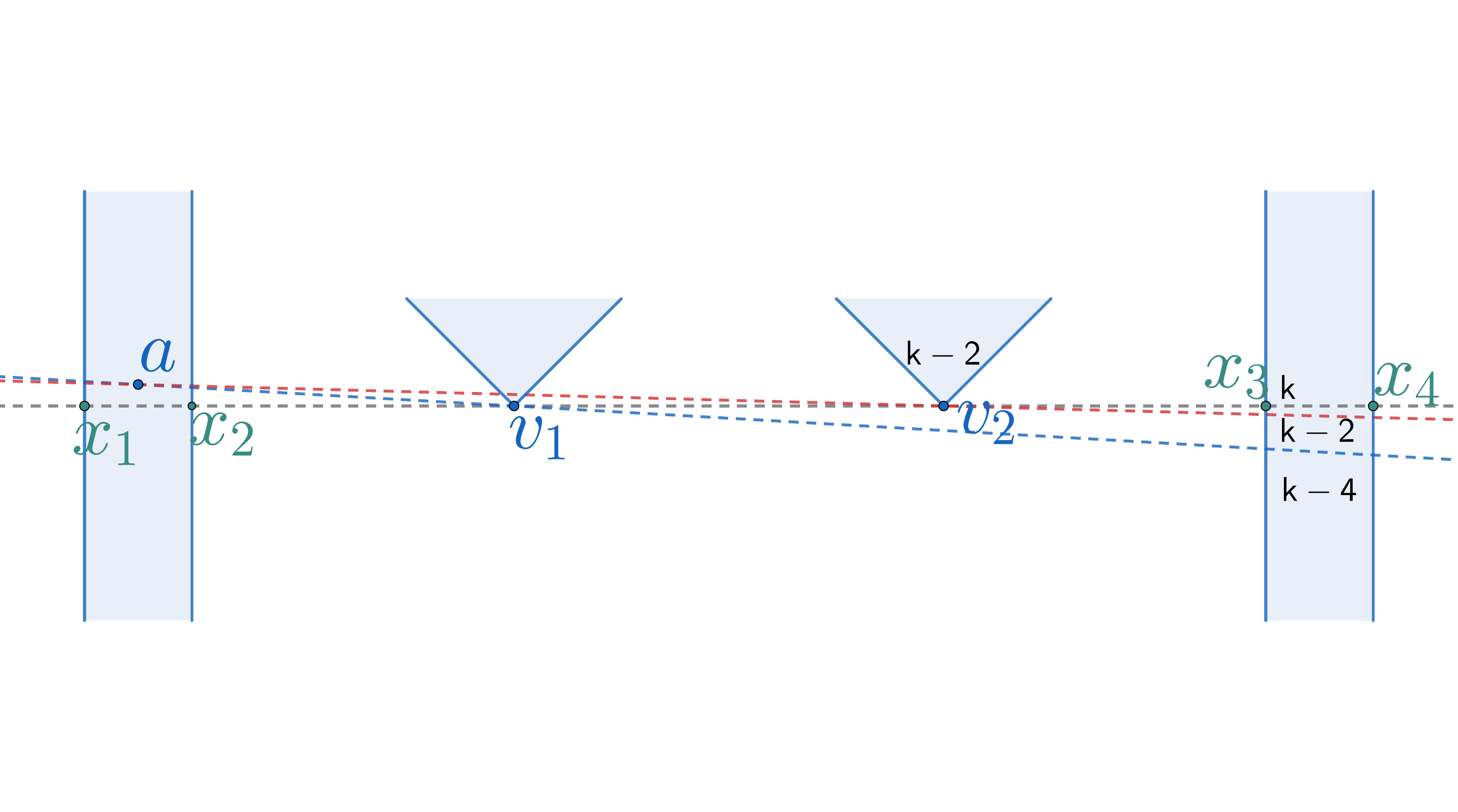}
\caption{Above $\ell_{g}$}\label{fig:CCS-genericA6}
\end{subfigure}
\begin{subfigure}[b]{.49\linewidth}
\includegraphics[width=\linewidth]{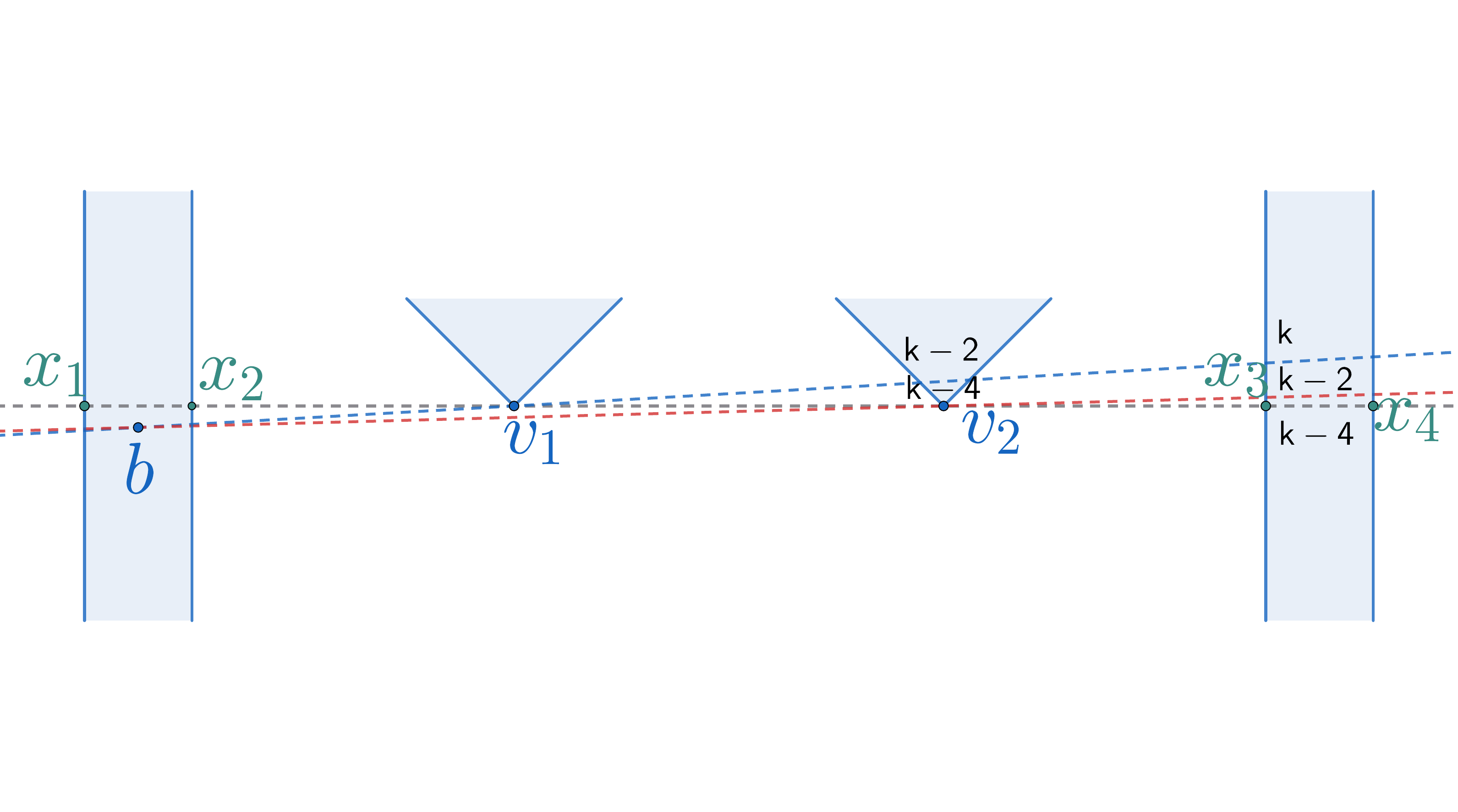}
\caption{Below $\ell_{g}$}\label{fig:CCS-genericB6}
\end{subfigure}

\caption{CCS; $Z = k - 5$, $W = k - 4$.}
\label{fig:CCS-generic-6}
\end{figure}

 \begin{figure}[H]
\centering
\begin{subfigure}[b]{.49\linewidth}
\includegraphics[width=\linewidth]{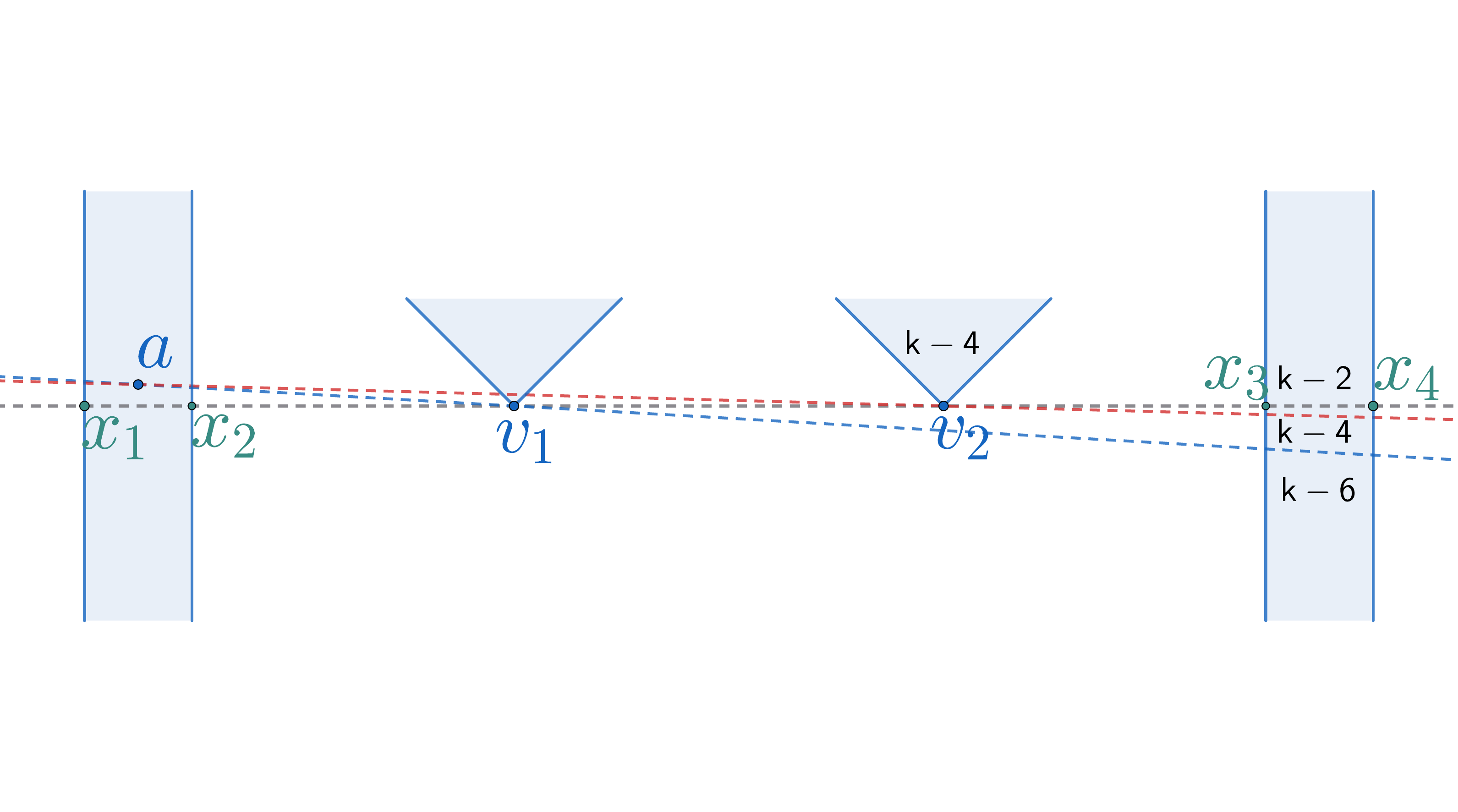}
\caption{Above $\ell_{g}$}\label{fig:CCS-genericA8}
\end{subfigure}
\begin{subfigure}[b]{.49\linewidth}
\includegraphics[width=\linewidth]{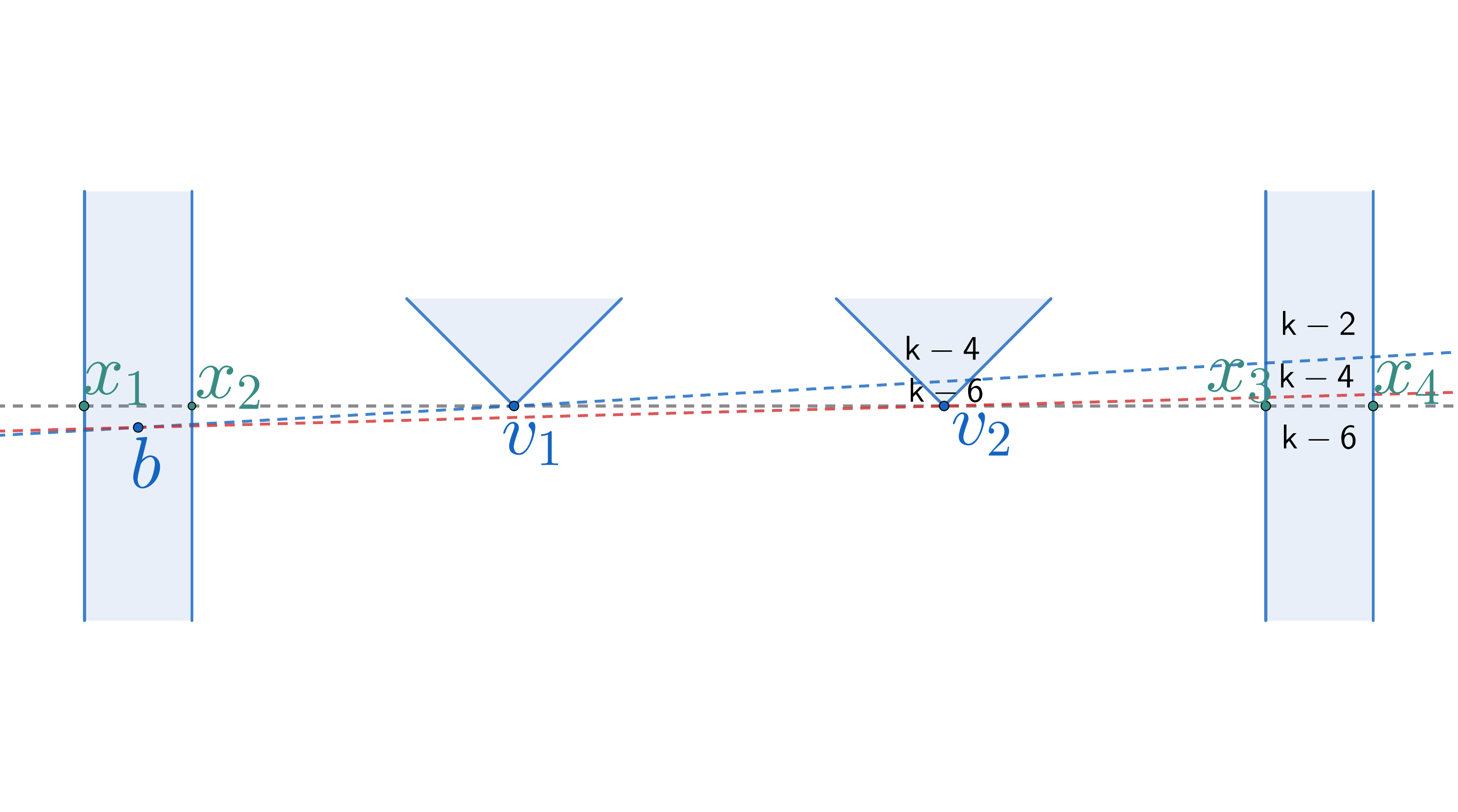}
\caption{Below $\ell_{g}$}\label{fig:CCS-genericB8}
\end{subfigure}

\caption{CCS; $Z = k - 7$, $W = k - 6$.}
\label{fig:CCS-generic-8}
\end{figure}

\section{CCO}
\label{appendix:CCO}
\subsection{Convex Convex Opposite (CCO)}

See section ~\ref{CCOMainBody}.

\section{CRO}
\label{appendix:CRO}
\subsection{Convex Reflex Opposite (CRO)}

\begin{lemma}
\label{lemma:CRO}
    When $Z = k$, a merge/split event occurs at $v_{2}$ (Figure~\ref{fig:CRS-generic2}). If $Z = k - 2$, an appear/disappear event occurs at $v_{2}$ (Figure~\ref{fig:CRS-generic4}). Also, if $W = k$, a merge/split event occurs at $x_{3}x_{4}$(Figure~\ref{fig:CRS-generic4}). Additionally, if $W = k - 2$, an appear/disappear event occurs at $x_{3}x_{4}$(Figure~\ref{fig:CRS-generic6}). 

\end{lemma}
\begin{proof}
    See Figure~\ref{fig:CRS-generic0} - Figure~\ref{fig:CRS-generic8}

    For $Z \geq k+4$, $v_{2}$ and its surroundings is entirely in shadow. For $W \geq k+6$, $x_{3}x_{4}$ and its surroundings is entirely in shadow. 

For $Z \leq k - 8$, $v_{2}$ and its surroundings is entirely visible. For $W \leq k - 6$, $x_{3}x_{4}$ and its surroundings is entirely visible. 
    
\end{proof}

 \begin{figure}[H]
\centering
\begin{subfigure}[b]{.49\linewidth}
\includegraphics[width=\linewidth]{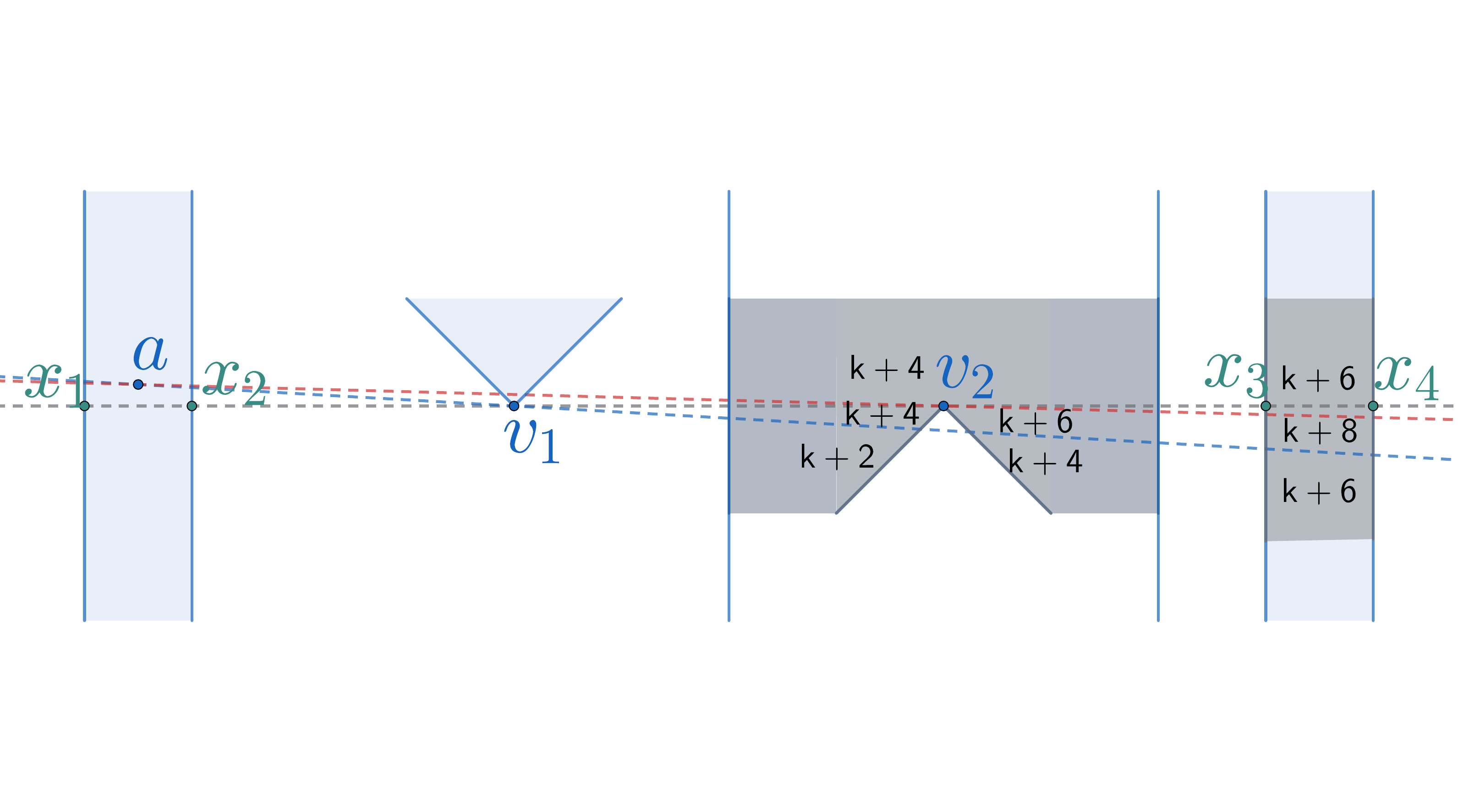}
\caption{Above $\ell_{g}$}\label{fig:CRS-genericA0}
\end{subfigure}
\begin{subfigure}[b]{.49\linewidth}
\includegraphics[width=\linewidth]{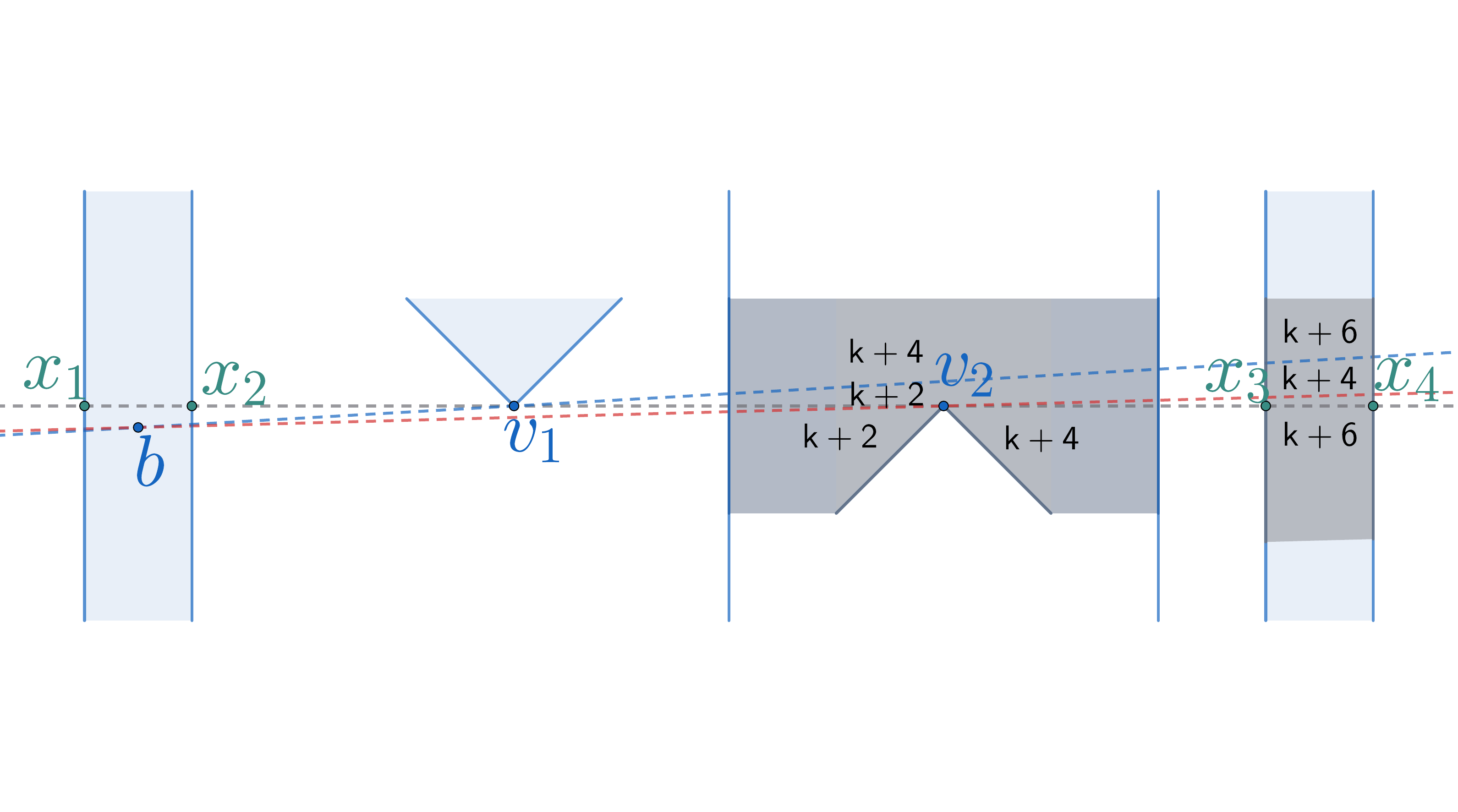}
\caption{Below $\ell_{g}$}\label{fig:CRS-genericB0}
\end{subfigure}

\caption{CRO; $Z = k+2$ and $W = k+4$.}
\label{fig:CRS-generic0}
\end{figure}

 \begin{figure}[H]
\centering
\begin{subfigure}[b]{.49\linewidth}
\includegraphics[width=\linewidth]{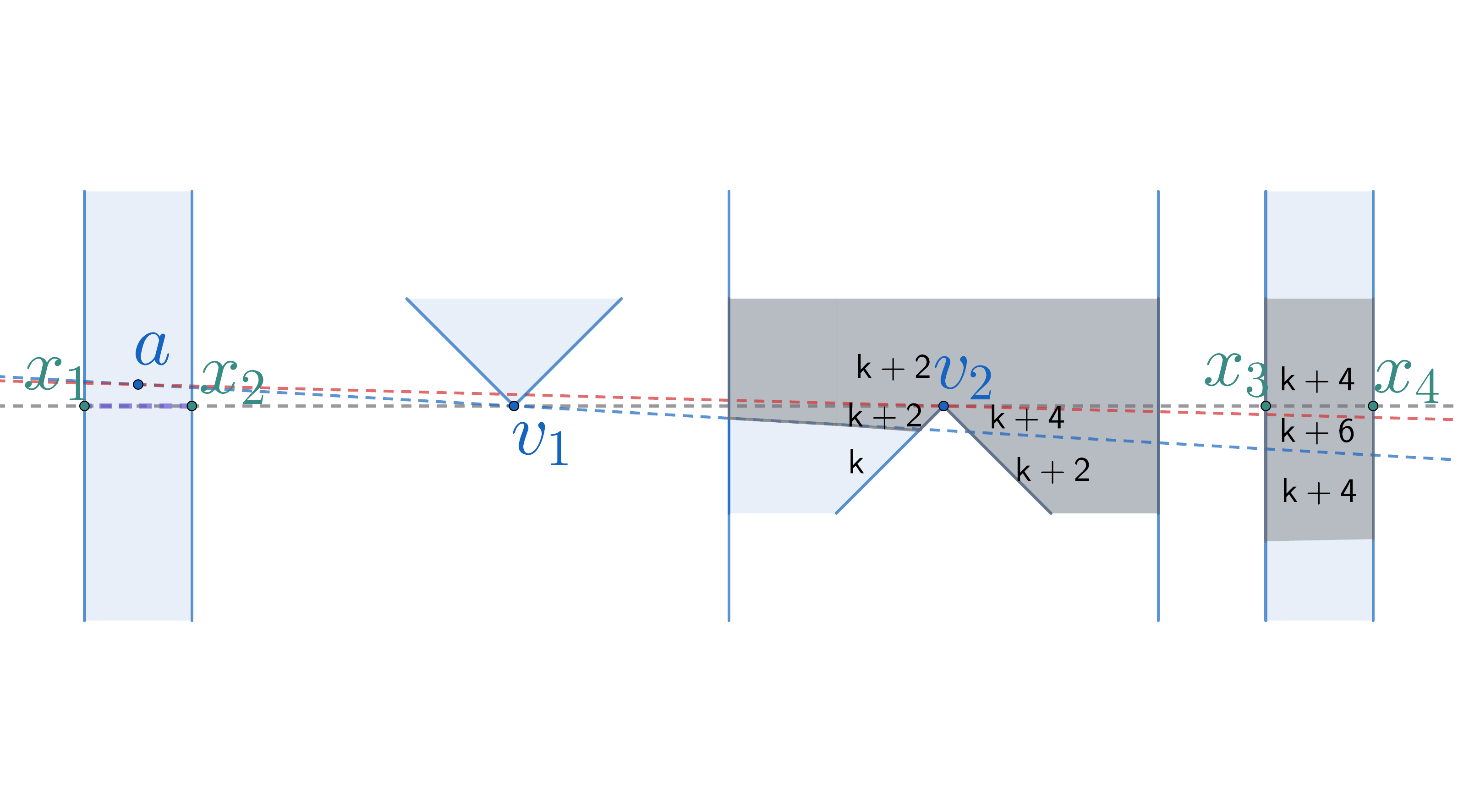}
\caption{Above $\ell_{g}$}\label{fig:CRS-genericA2}
\end{subfigure}
\begin{subfigure}[b]{.49\linewidth}
\includegraphics[width=\linewidth]{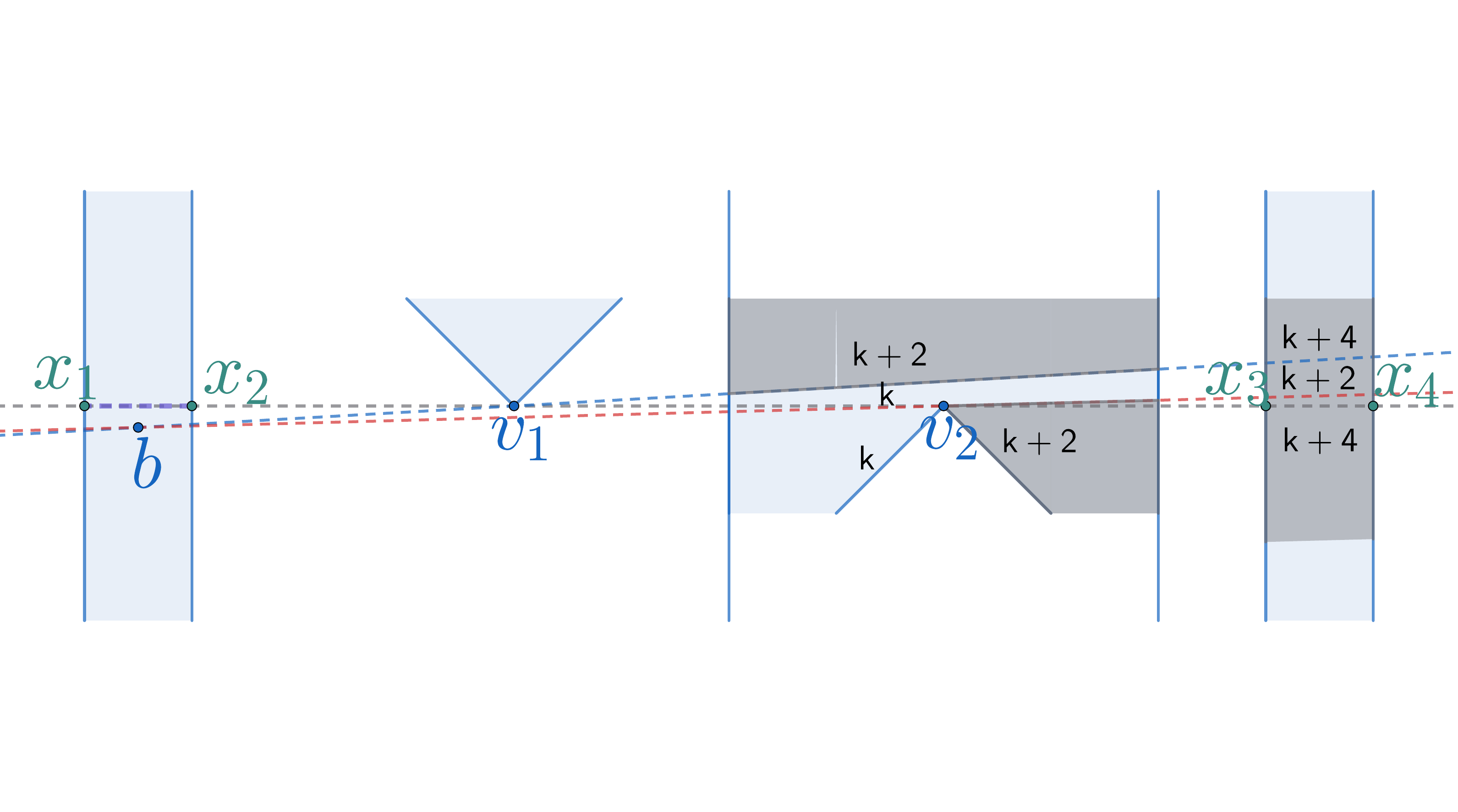}
\caption{Below $\ell_{g}$}\label{fig:CRS-genericB2}
\end{subfigure}

\caption{CRO; $Z=k$ (Merge/Split) and $W=k+2$.}
\label{fig:CRS-generic2}
\end{figure}

 \begin{figure}[H]
\centering
\begin{subfigure}[b]{.49\linewidth}
\includegraphics[width=\linewidth]{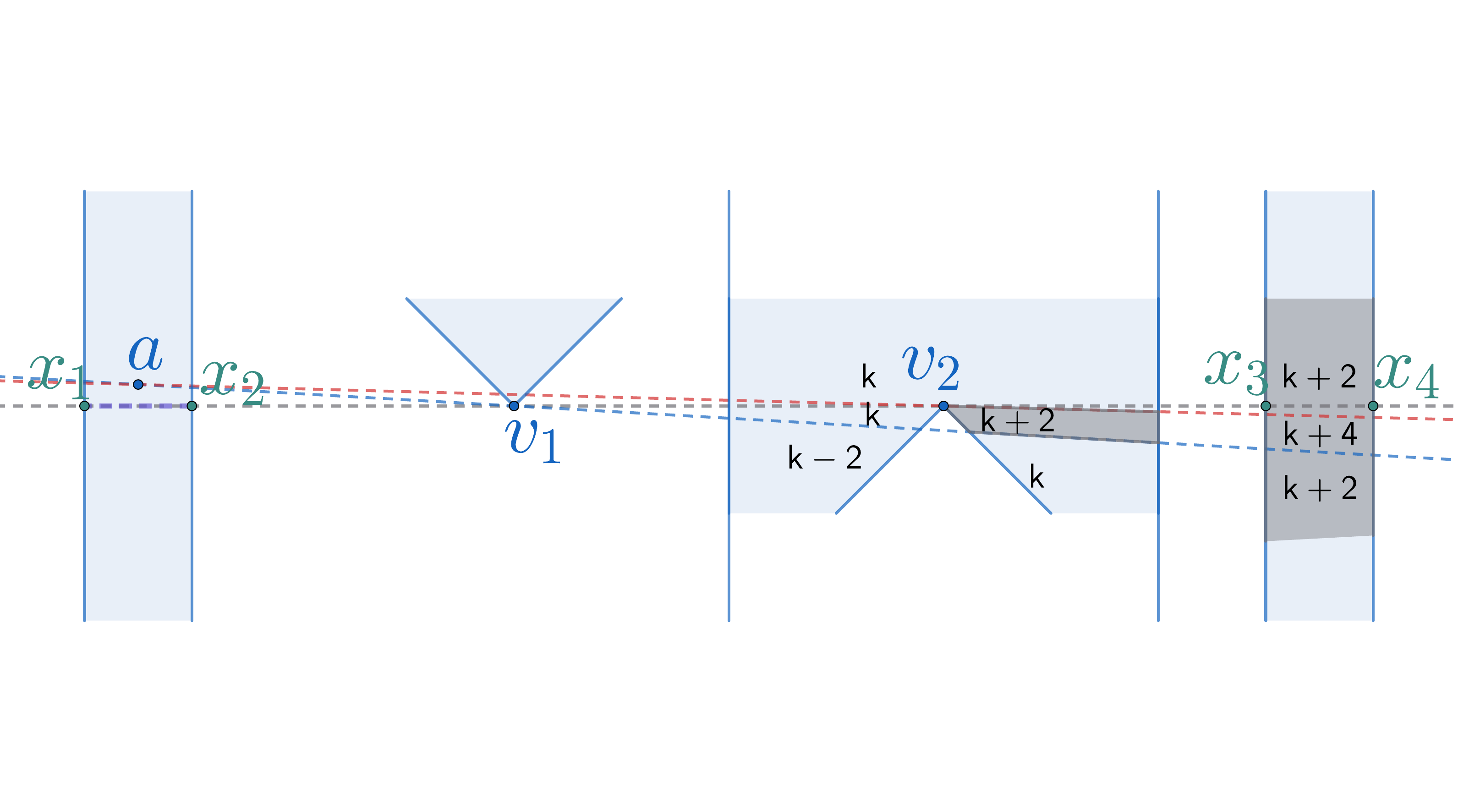}
\caption{Above $\ell_{g}$}\label{fig:CRS-genericA4}
\end{subfigure}
\begin{subfigure}[b]{.49\linewidth}
\includegraphics[width=\linewidth]{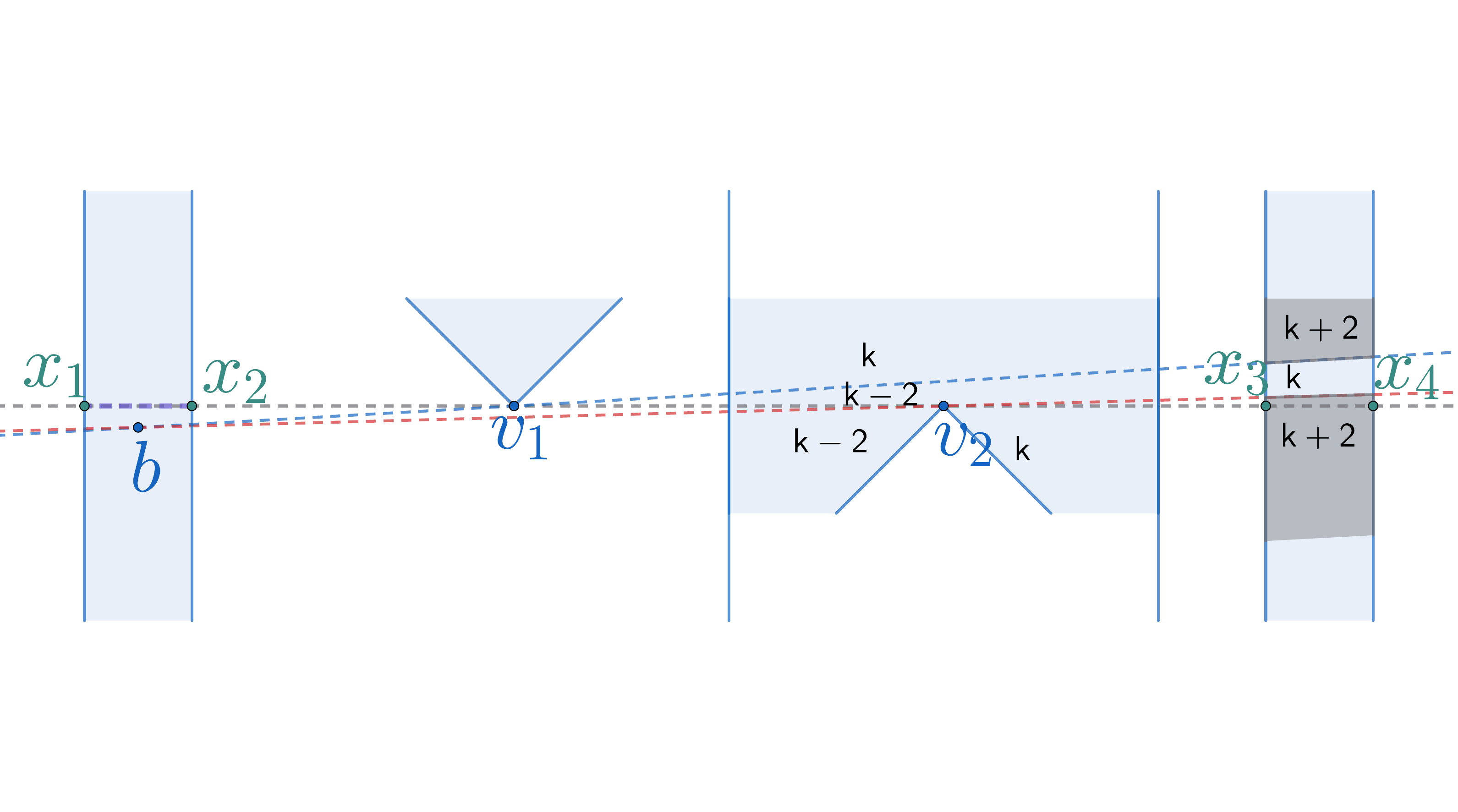}
\caption{Below $\ell_{g}$}\label{fig:CRS-genericB4}
\end{subfigure}

\caption{CRO; $Z = k - 2$ (Appear/Disappear), $W = k$ (Merge/Split).}
\label{fig:CRS-generic4}
\end{figure}

 \begin{figure}[H]
\centering
\begin{subfigure}[b]{.49\linewidth}
\includegraphics[width=\linewidth]{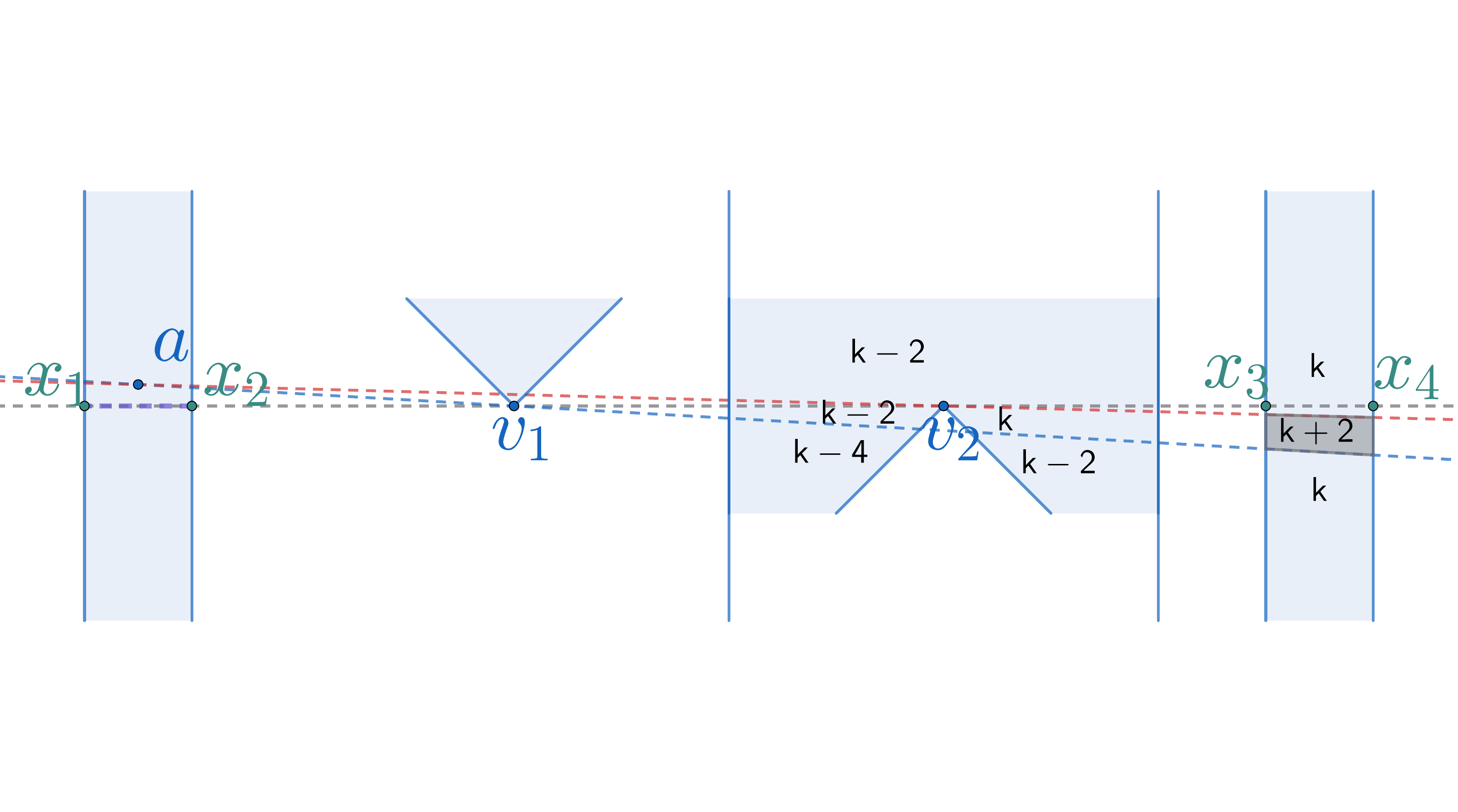}
\caption{Above $\ell_{g}$}\label{fig:CRS-genericA6}
\end{subfigure}
\begin{subfigure}[b]{.49\linewidth}
\includegraphics[width=\linewidth]{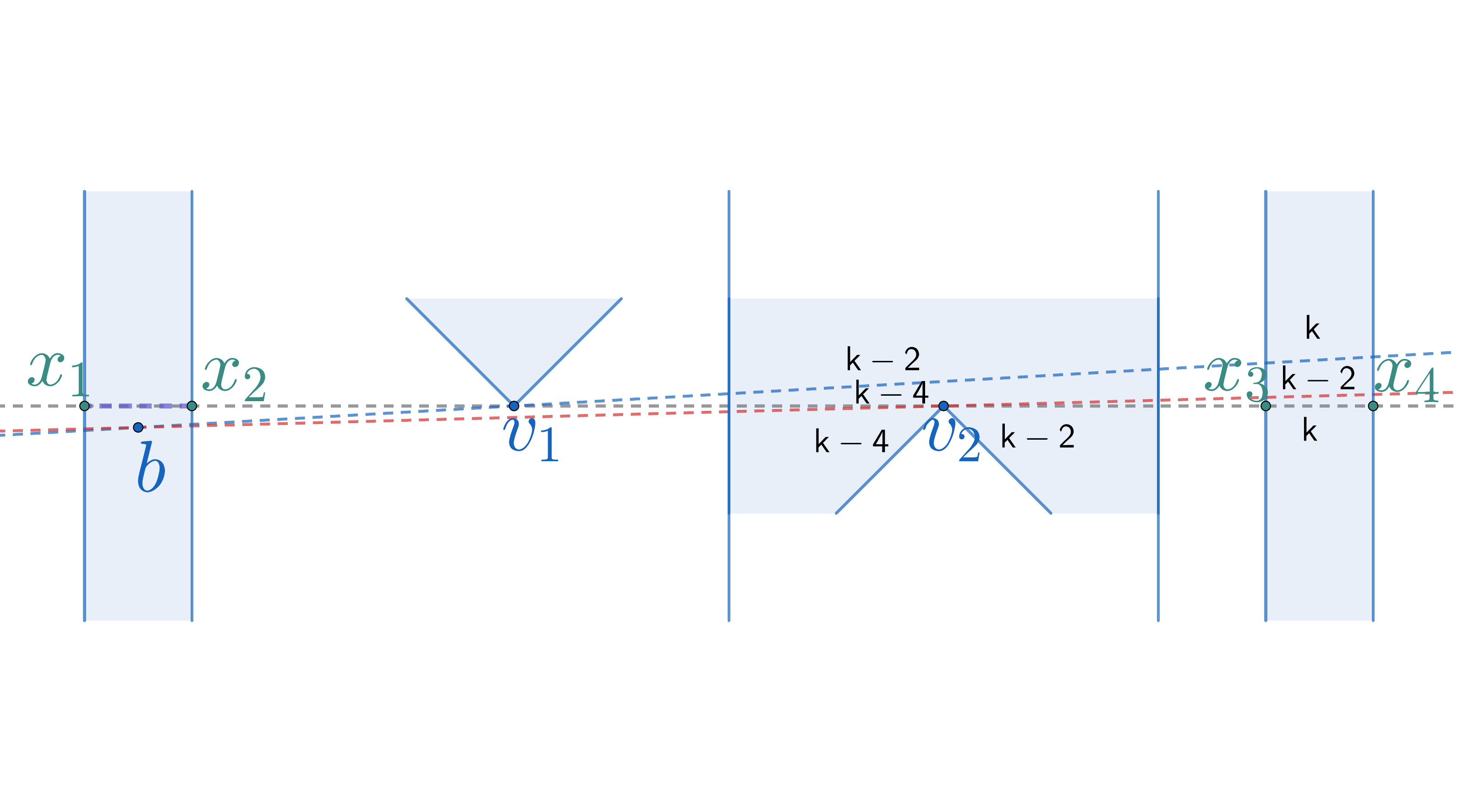}
\caption{Below $\ell_{g}$}\label{fig:CRS-genericB6}
\end{subfigure}

\caption{CRO; $Z = k - 4$, $W = k - 2$ (Appear/Disappear).}
\label{fig:CRS-generic6}
\end{figure}

 \begin{figure}[H]
\centering
\begin{subfigure}[b]{.49\linewidth}
\includegraphics[width=\linewidth]{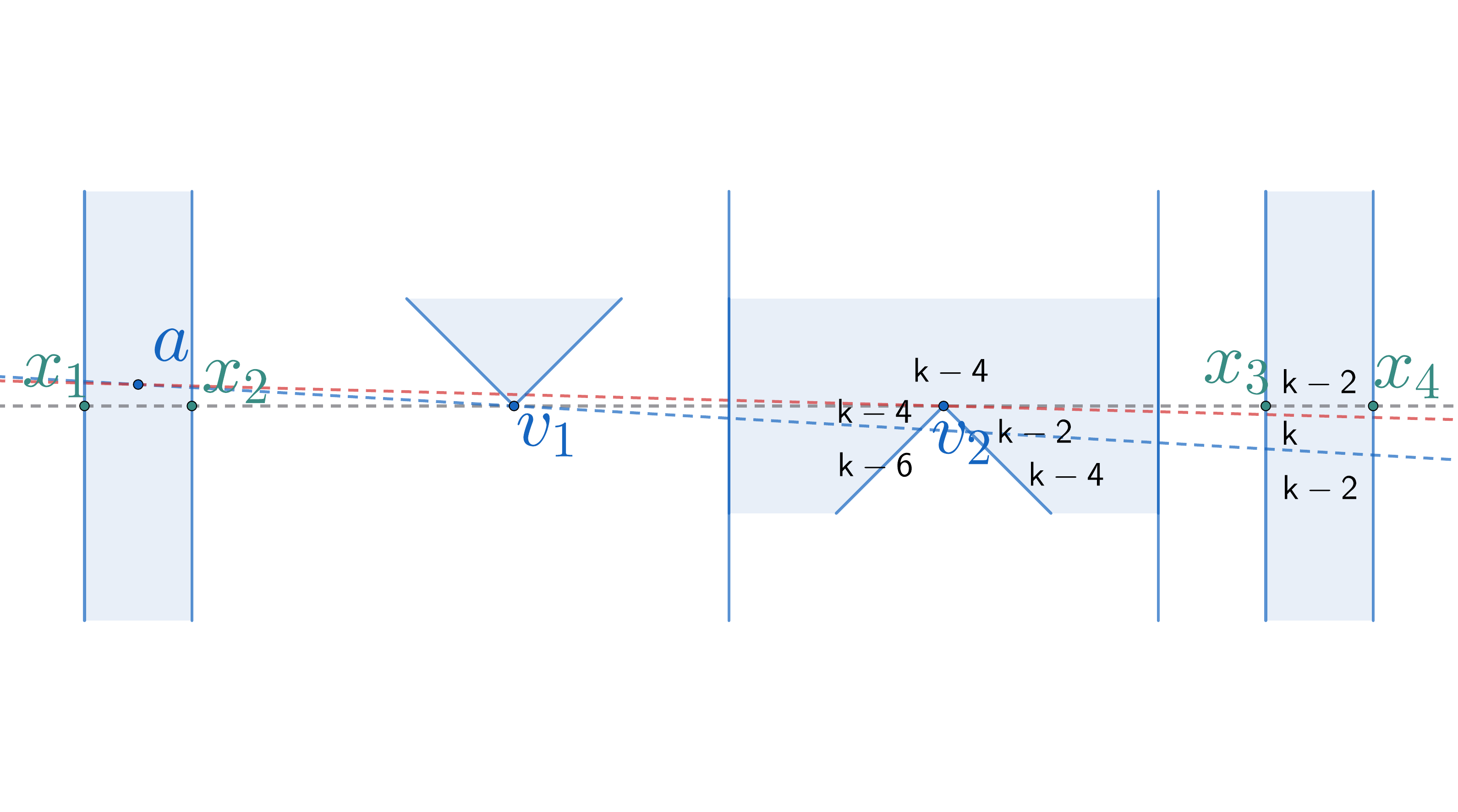}
\caption{Above $\ell_{g}$}\label{fig:CRS-genericA8}
\end{subfigure}
\begin{subfigure}[b]{.49\linewidth}
\includegraphics[width=\linewidth]{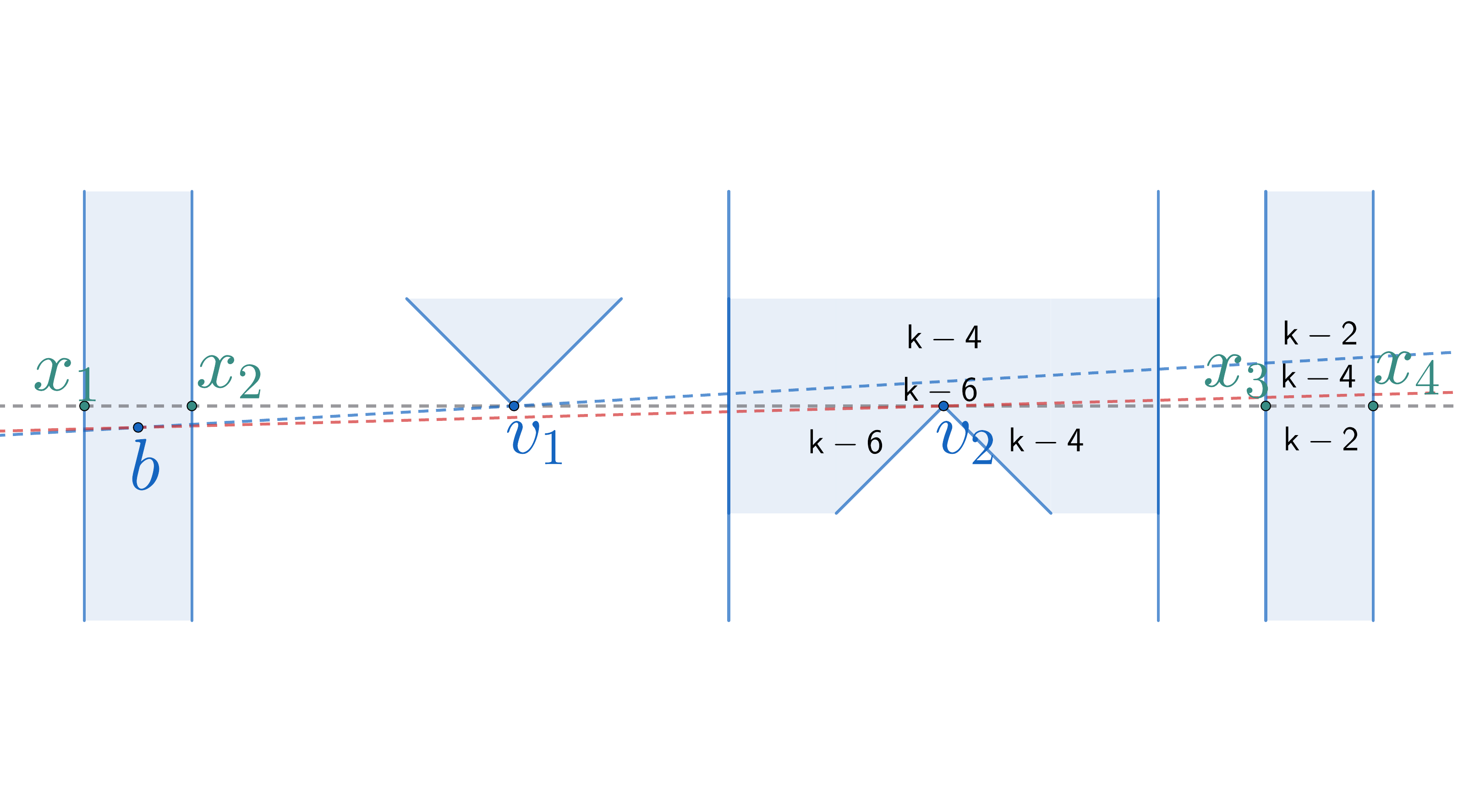}
\caption{Below $\ell_{g}$}\label{fig:CRS-genericB8}
\end{subfigure}
\caption{CRO; $Z = k - 6$, $W = k - 4$.}
\label{fig:CRS-generic8}
\end{figure}

\section{CRS}
\label{appendix:CRS}
\subsection{Convex Reflex Same (CRS)}
\begin{lemma}
\label{lemma:CRS}
    When $Z = k$, there is a merge/split event at $v_{2}$(Figure~\ref{fig:CRO-generic2}). No event occurs for all other $Z$ and $W$. 
\end{lemma}
\begin{proof}
    See Figure ~\ref{fig:CRO-generic0} - Figure~\ref{fig:CRO-generic6}

For $Z \geq k+4$, $v_{2}$ and its surroundings are in shadow. For $W \geq k+6$, $x_{3}x_{4}$ and its surroundings are in shadow. 

For $Z \leq k - 8$, $v_{2}$ and its surroundings are entirely visible. For $W \leq k - 6$, $x_{3}x_{4}$ and its surroundings are entirely visible. 
 \end{proof}
\begin{figure}[H]
\centering
\begin{subfigure}[b]{.49\linewidth}
\includegraphics[width=\linewidth]{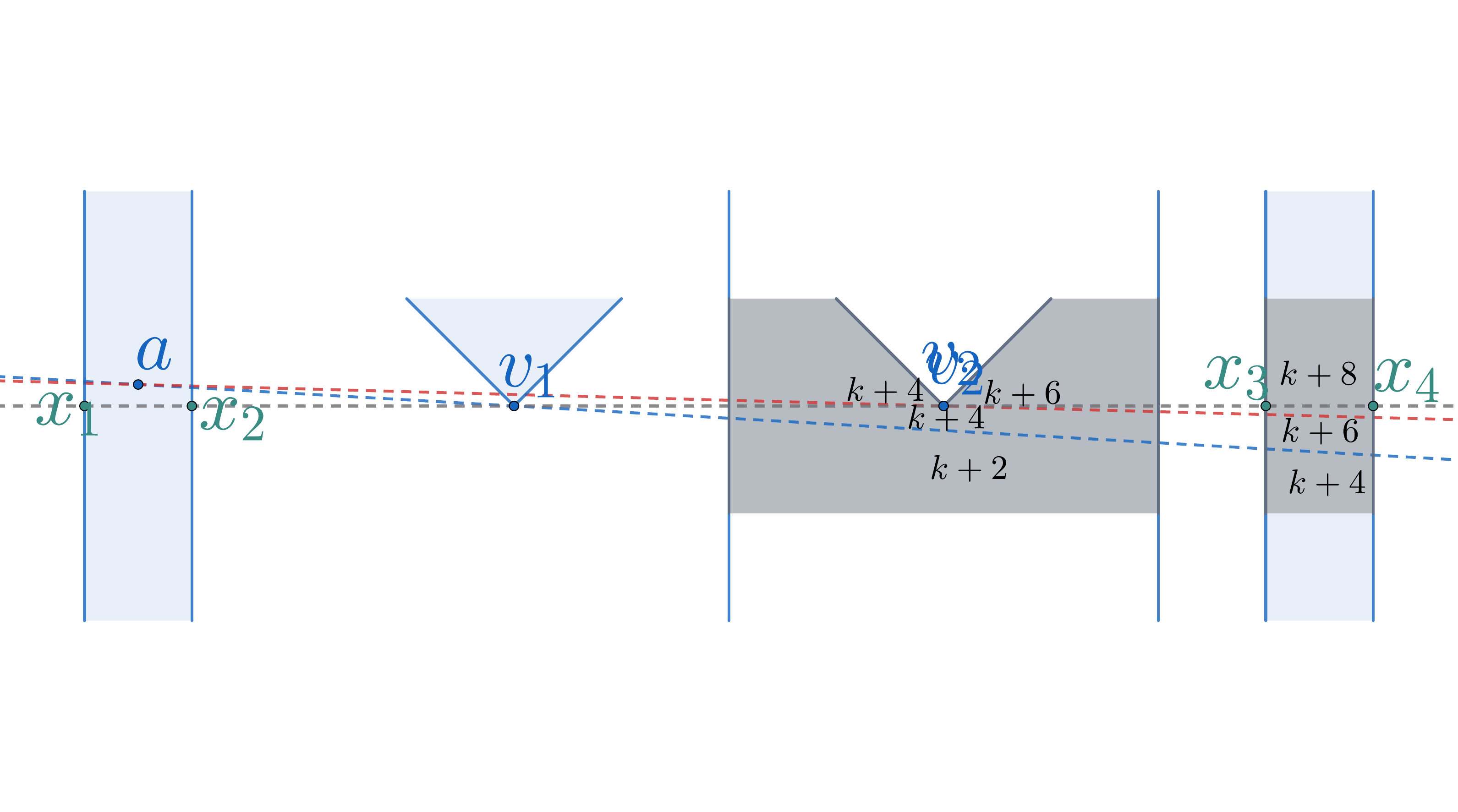}
\caption{Above $\ell_{g}$}\label{fig:CRO-genericA0}
\end{subfigure}
\begin{subfigure}[b]{.49\linewidth}
\includegraphics[width=\linewidth]{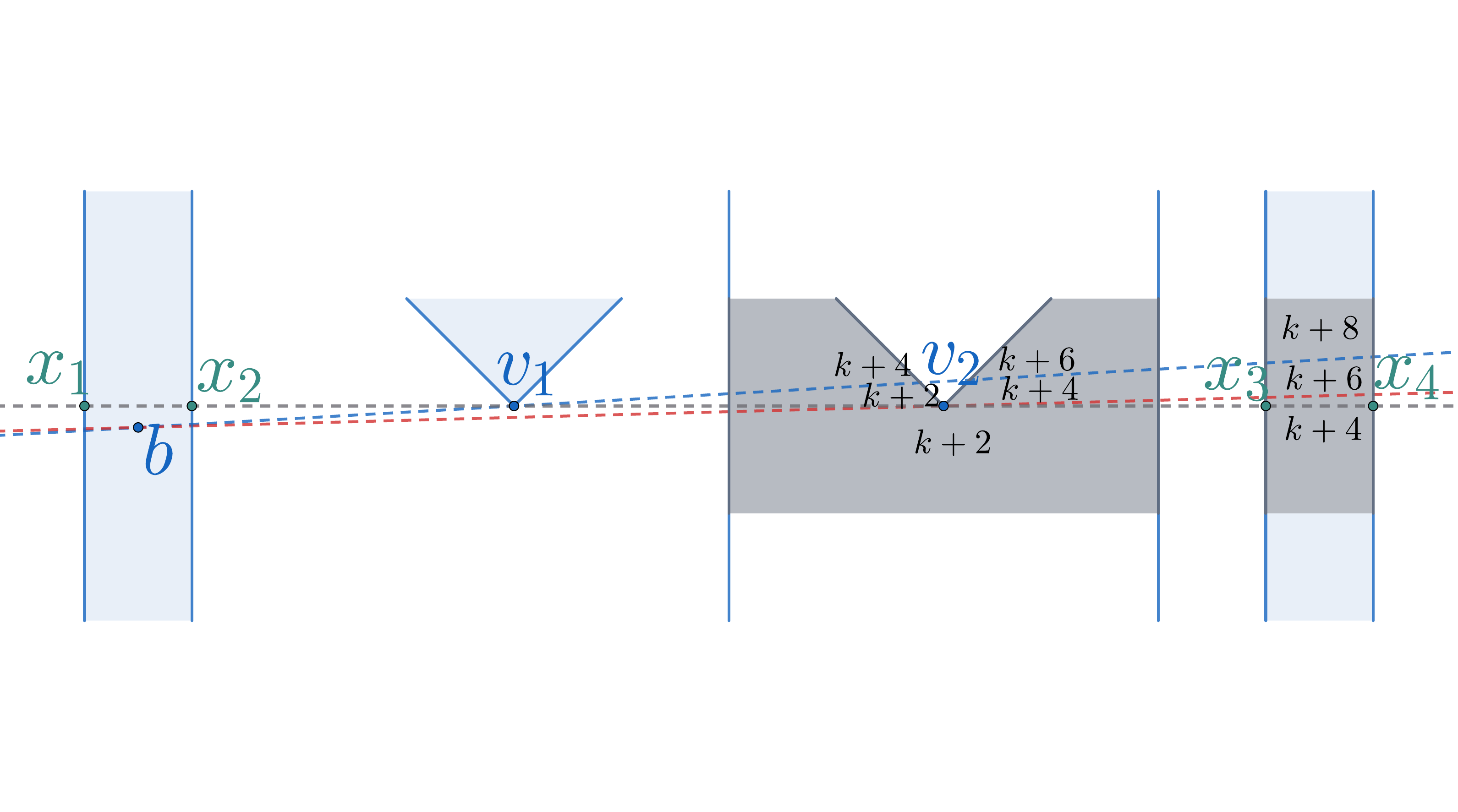}
\caption{Below $\ell_{g}$}\label{fig:CRO-genericB0}
\end{subfigure}

\caption{CRS; $Z = k+2$, $W = k+4$.}
\label{fig:CRO-generic0}
\end{figure}

\begin{figure}[H]
\centering
\begin{subfigure}[b]{.49\linewidth}
\includegraphics[width=\linewidth]{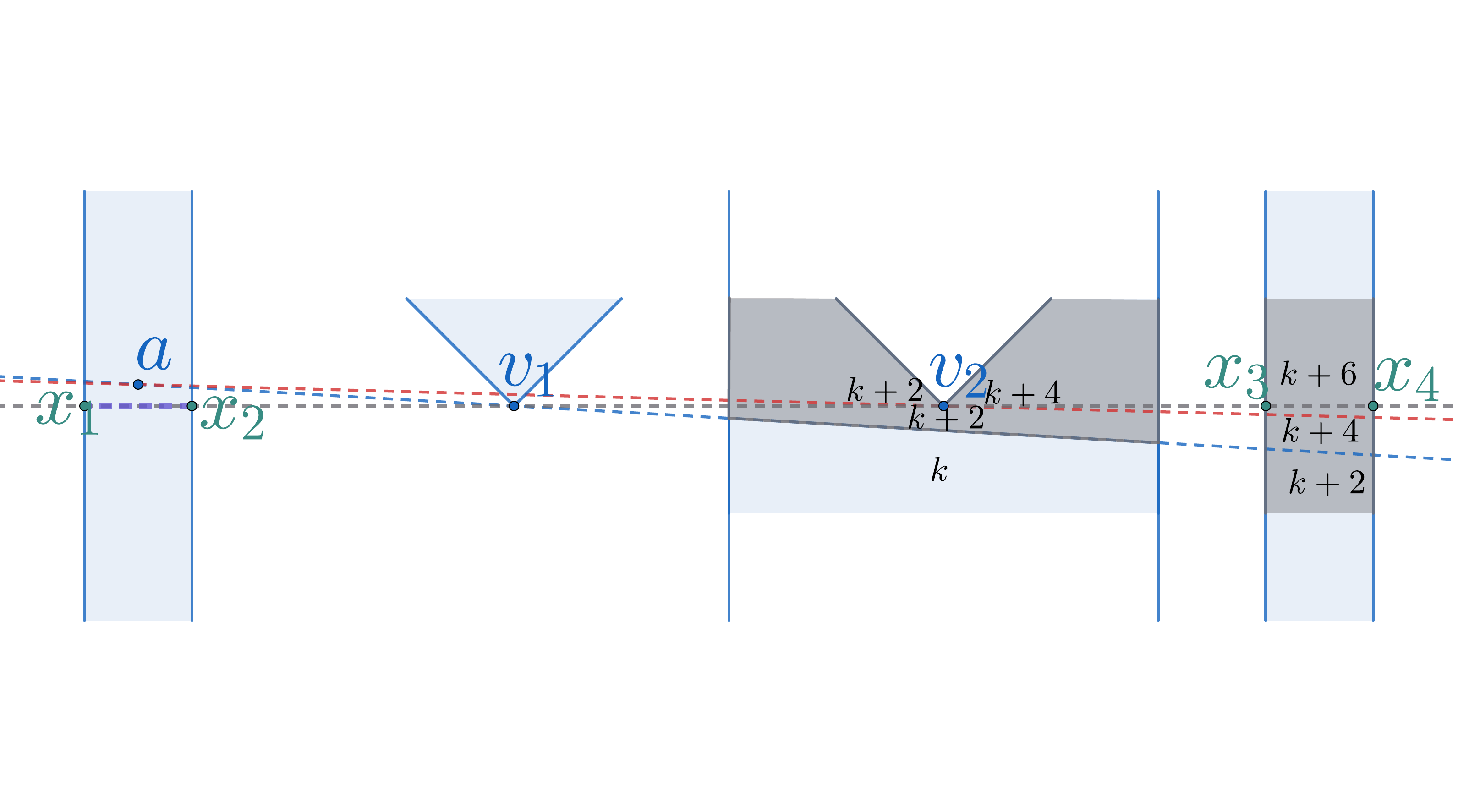}
\caption{Above $\ell_{g}$}\label{fig:CRO-genericA2}
\end{subfigure}
\begin{subfigure}[b]{.49\linewidth}
\includegraphics[width=\linewidth]{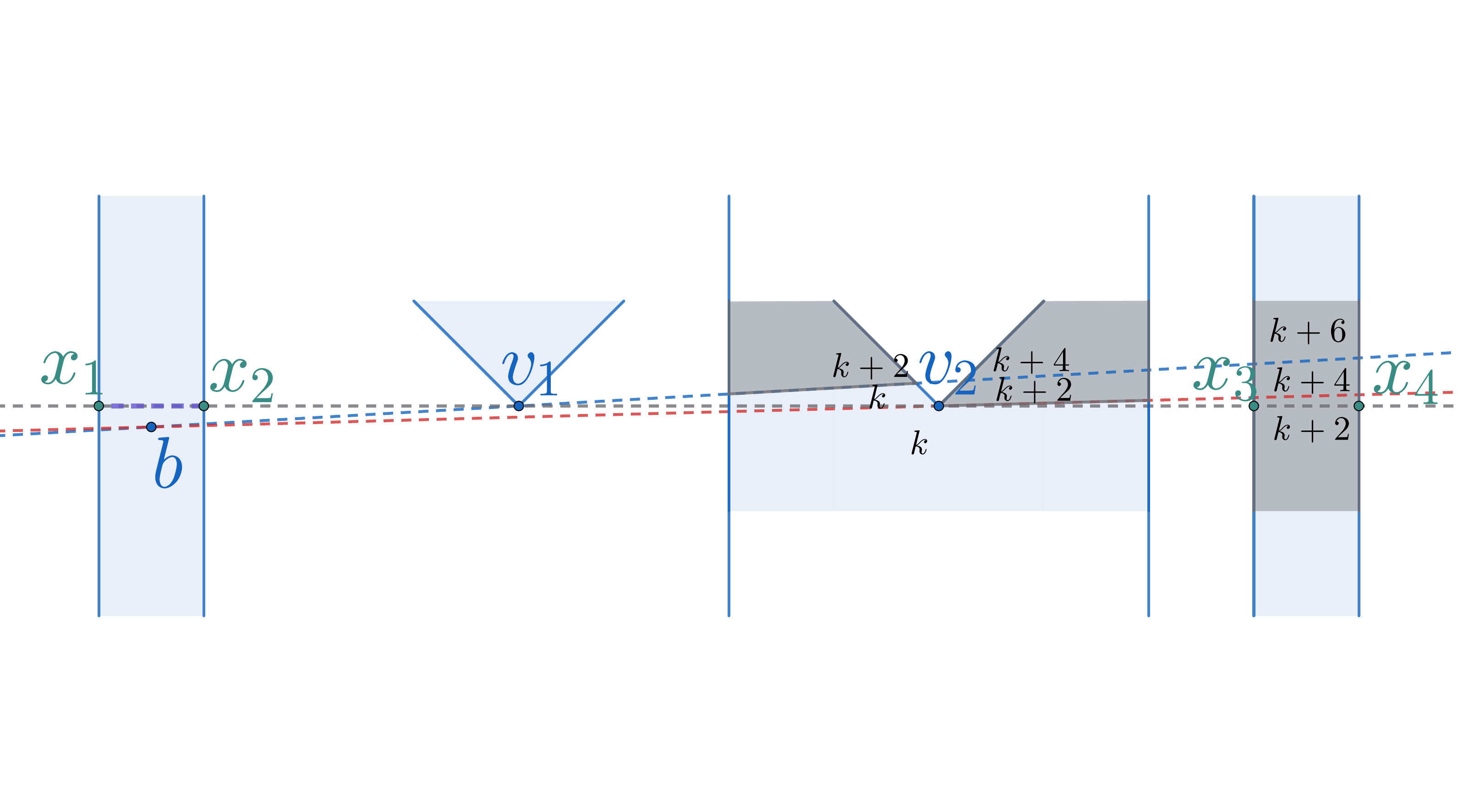}
\caption{Below $\ell_{g}$}\label{fig:CRO-genericB2}
\end{subfigure}

\caption{CRS; $Z=k$ (Merge/Split),  $W=k+2$.}
\label{fig:CRO-generic2}
\end{figure}

 \begin{figure}[H]
\centering
\begin{subfigure}[b]{.49\linewidth}
\includegraphics[width=\linewidth]{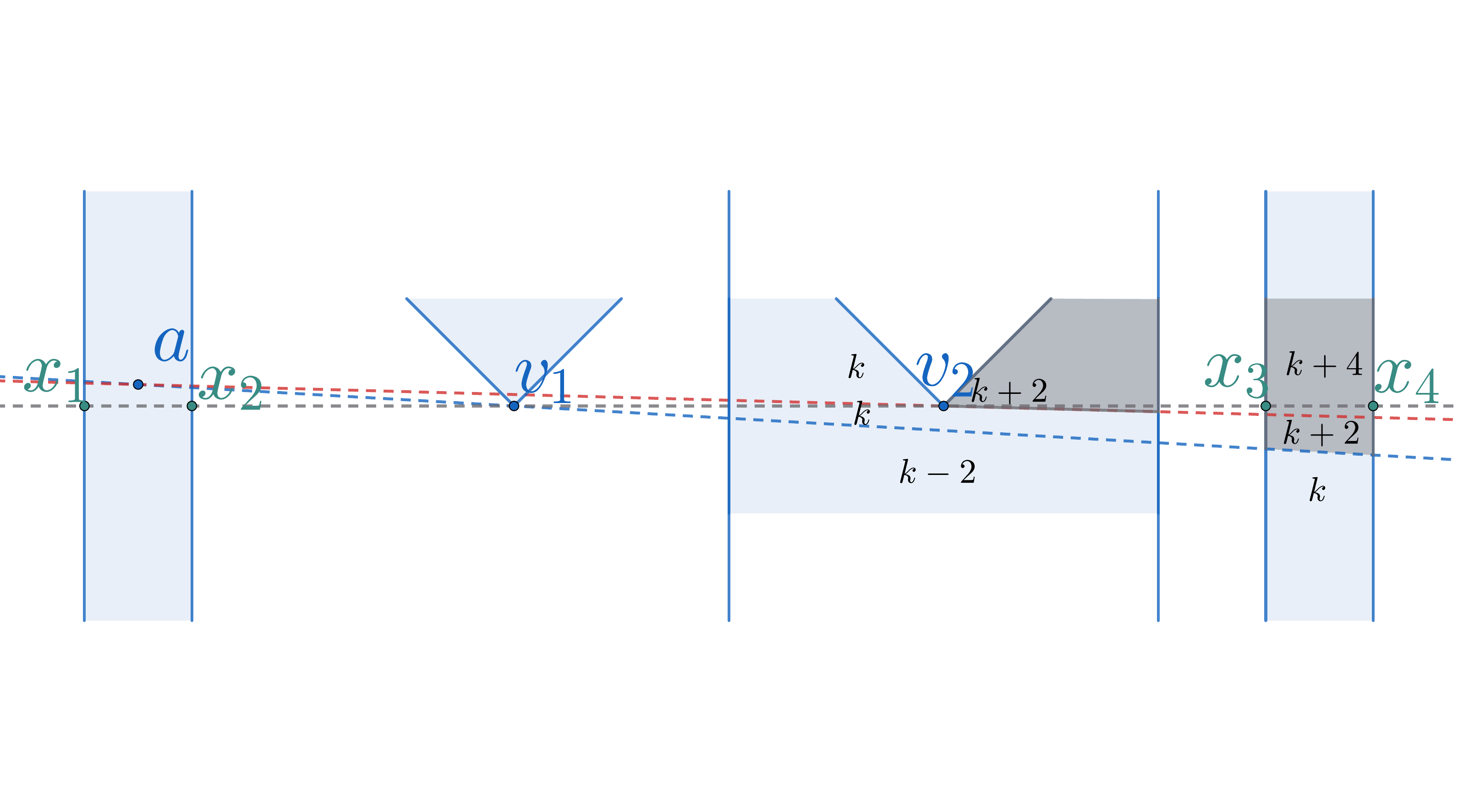}
\caption{Above $\ell_{g}$}\label{fig:CRO-genericA4}
\end{subfigure}
\begin{subfigure}[b]{.49\linewidth}
\includegraphics[width=\linewidth]{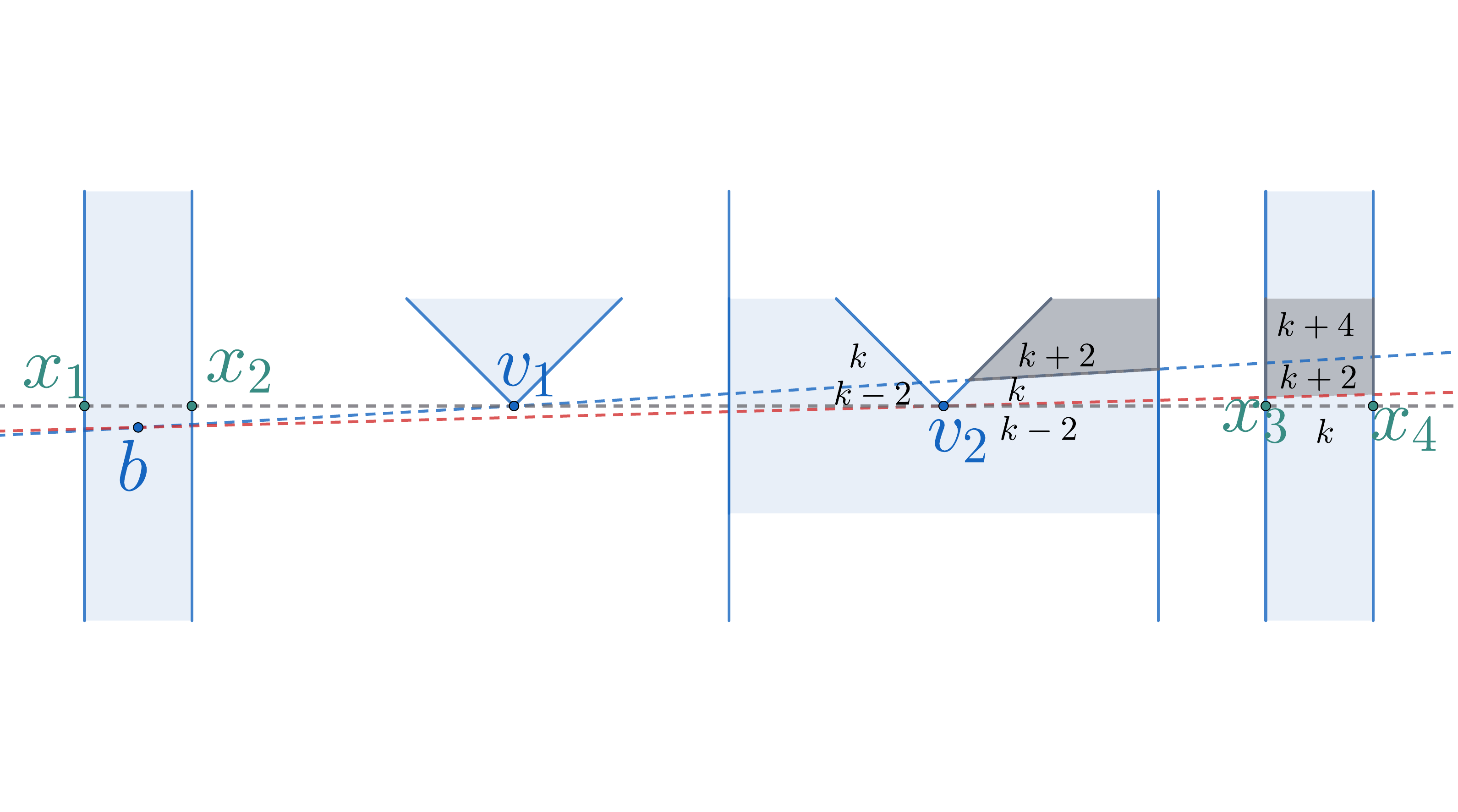}
\caption{Below $\ell_{g}$}\label{fig:CRO-genericB4}
\end{subfigure}

\caption{CRS; $Z = k - 2$; $W = k$.}
\label{fig:CRO-generic4}
\end{figure}

 \begin{figure}[H]
\centering
\begin{subfigure}[b]{.49\linewidth}
\includegraphics[width=\linewidth]{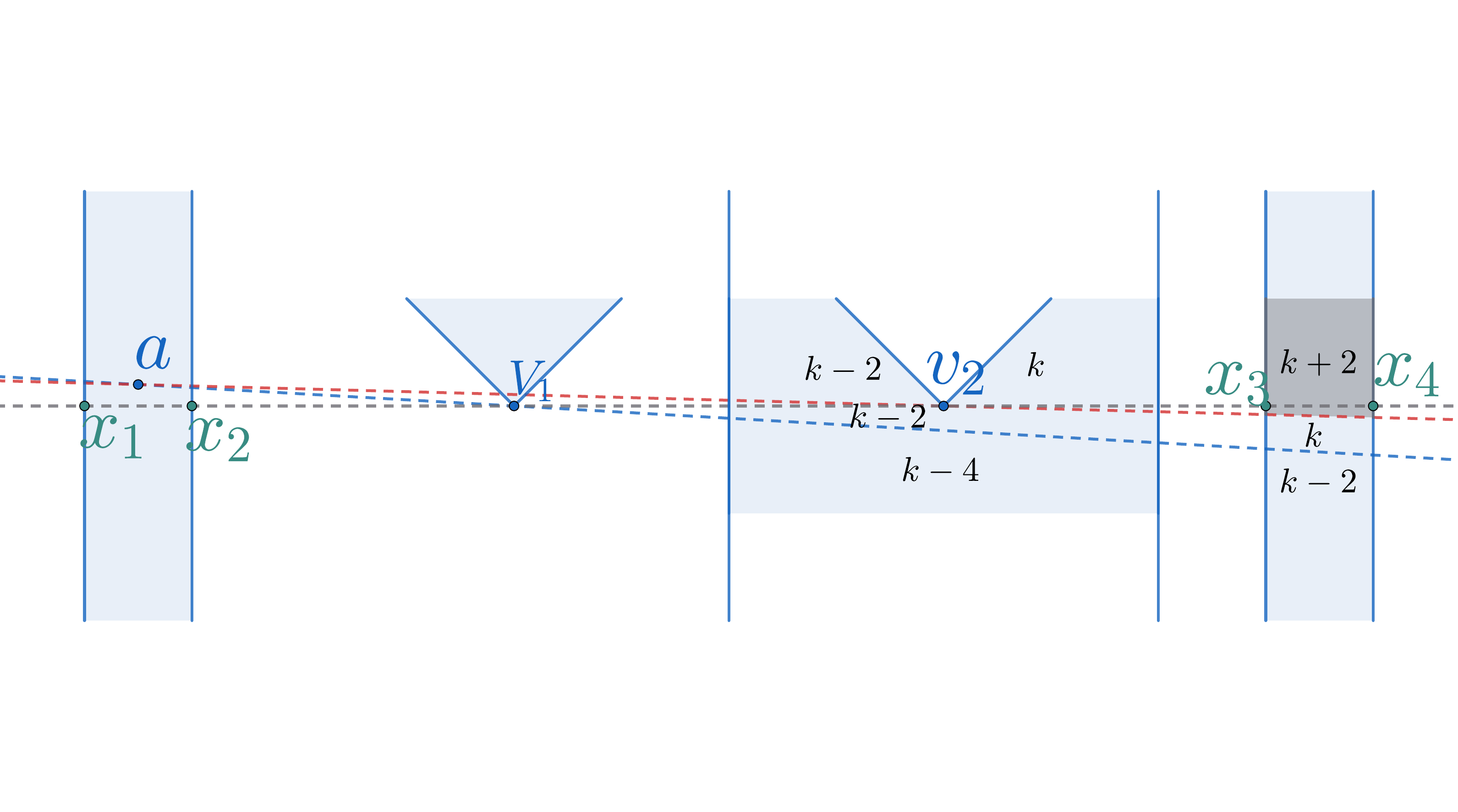}
\caption{Above $\ell_{g}$ }\label{fig:CRO-genericA6}
\end{subfigure}
\begin{subfigure}[b]{.49\linewidth}
\includegraphics[width=\linewidth]{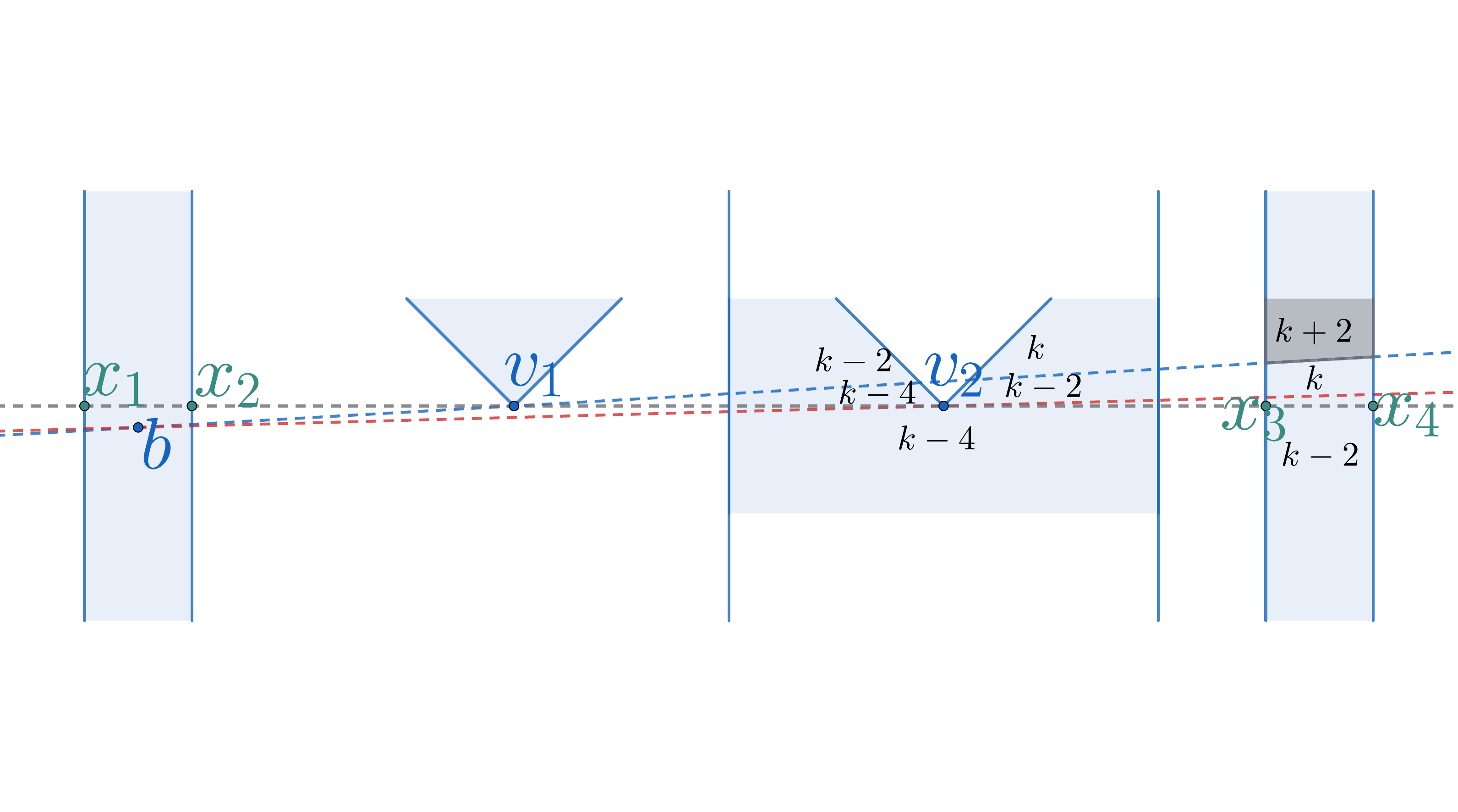}
\caption{Below $\ell_{g}$ }\label{fig:CRO-genericB6}
\end{subfigure}

\caption{CRS; $Z = k - 4$, $W = k - 2$.}
\label{fig:CRO-generic6}
\end{figure}

\begin{figure}[H]
\centering
\begin{subfigure}[b]{.49\linewidth}
\includegraphics[width=\linewidth]{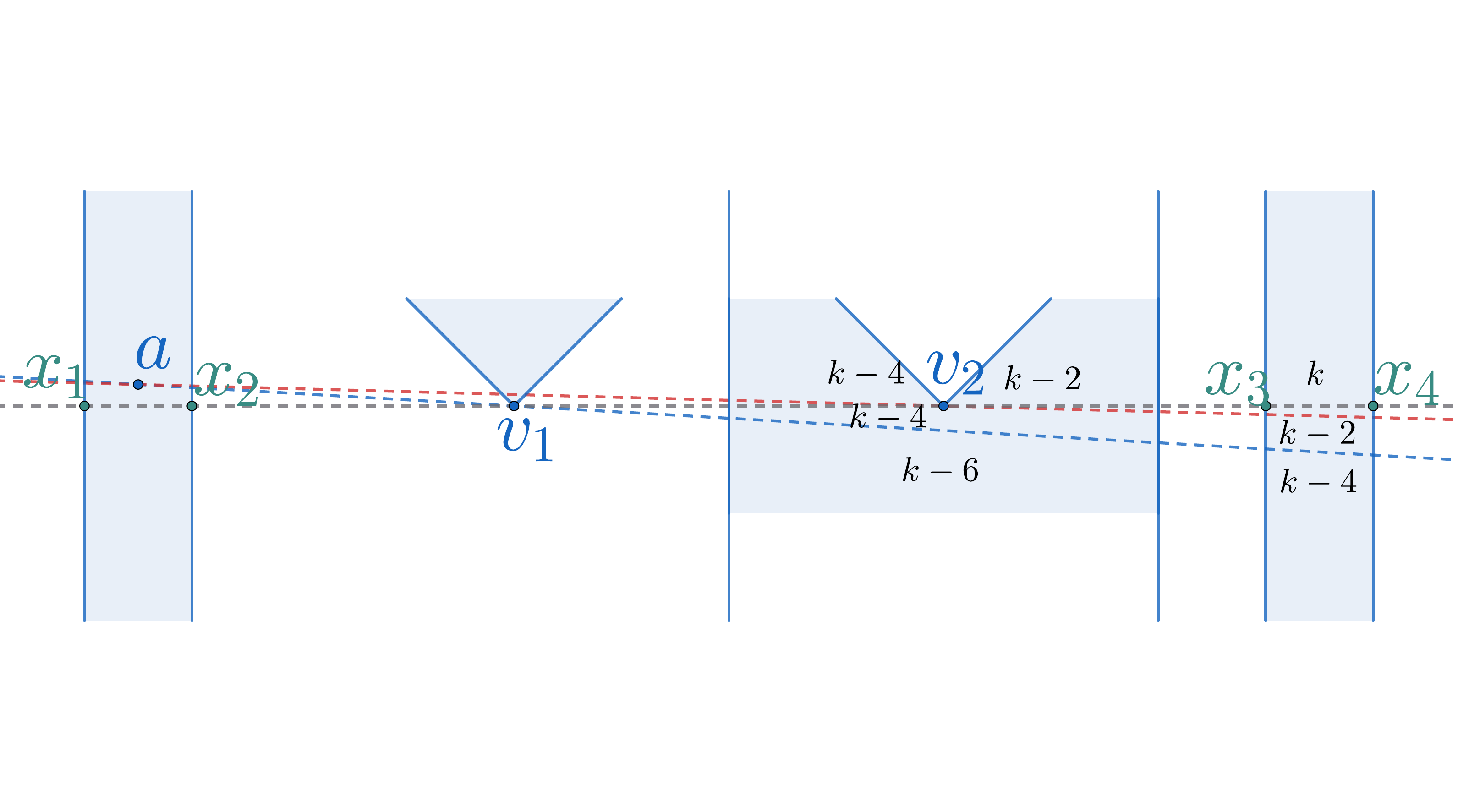}
\caption{Above $\ell_{g}$ }\label{fig:CRO-genericA8}
\end{subfigure}
\begin{subfigure}[b]{.49\linewidth}
\includegraphics[width=\linewidth]{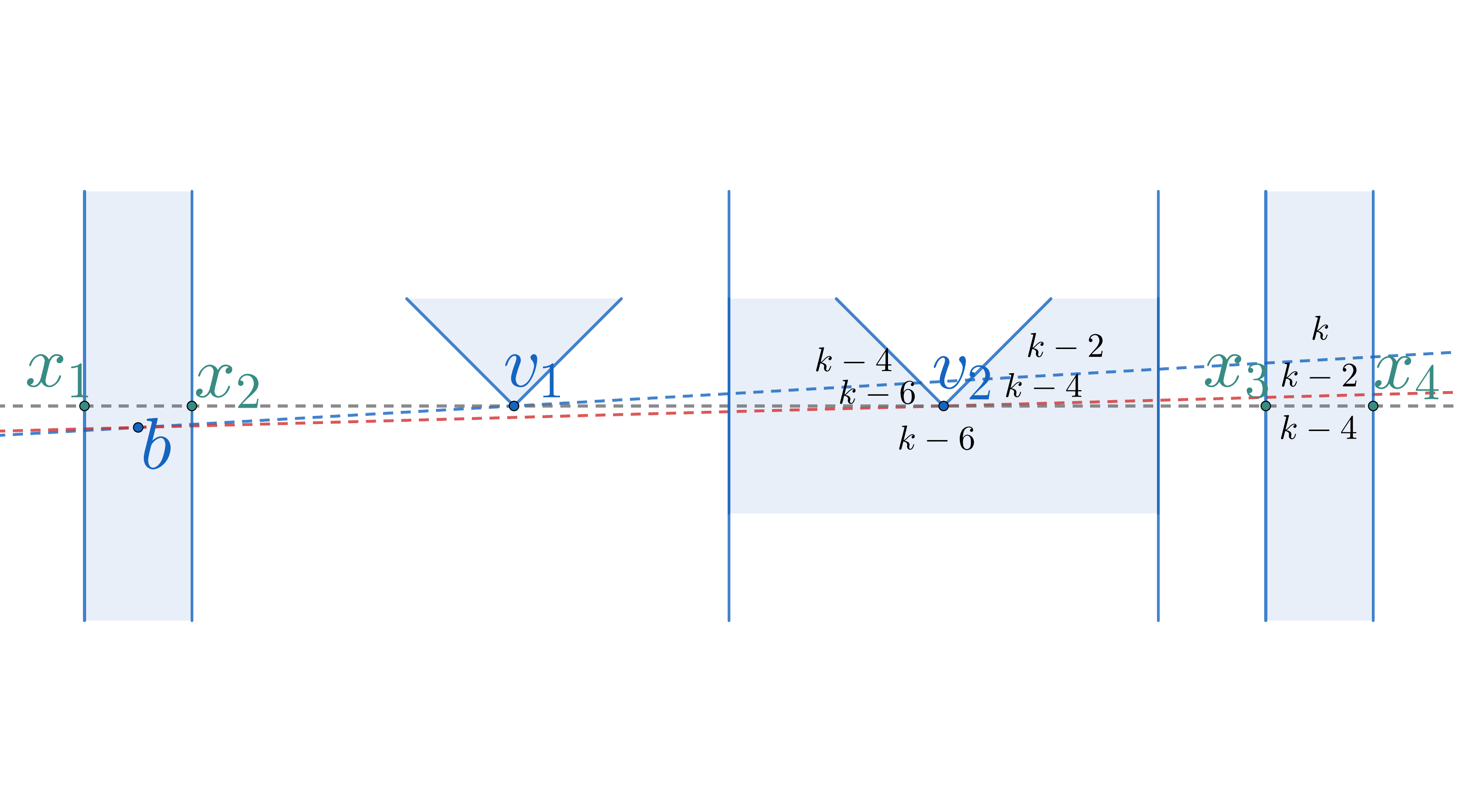}
\caption{Below $\ell_{g}$ }\label{fig:CRO-genericB8}
\end{subfigure}

\caption{CRS; $Z = k - 6$, $W = k - 4$.}
\label{fig:CRO-generic8}
\end{figure}


\section{RCS}
\label{appendix:RCS}
\subsection{Reflex Convex Same (RCS)}
\begin{lemma}
\label{lemma:RCS}
No event occurs for RCS for any $Z$ or $W$. 

\end{lemma}
\begin{proof}
    See Figure~\ref{fig:RCS-generic-0} to Figure~\ref{fig:RCS-generic-8}. Looking at the counts of when the agent is at $a$ and when the agent is at $b$, it can be seen that no event occurs for any $Z$ or $W$. 

    For $Z \geq k + 3$, $v_{2}$ and its surroundings are entirely in shadow. For $W \geq k + 4$, $x_{3}x_{4}$ and its surroundings are entirely in shadow.

For $Z \leq k - 9$, $v_{2}$ and its surroundings are entirely visible. For $W \leq k - 8$, $x_{3}x_{4}$ and its surroundings are entirely visible.
\end{proof}

 \begin{figure}[H]
\centering
\begin{subfigure}[b]{.49\linewidth}
\includegraphics[width=\linewidth]{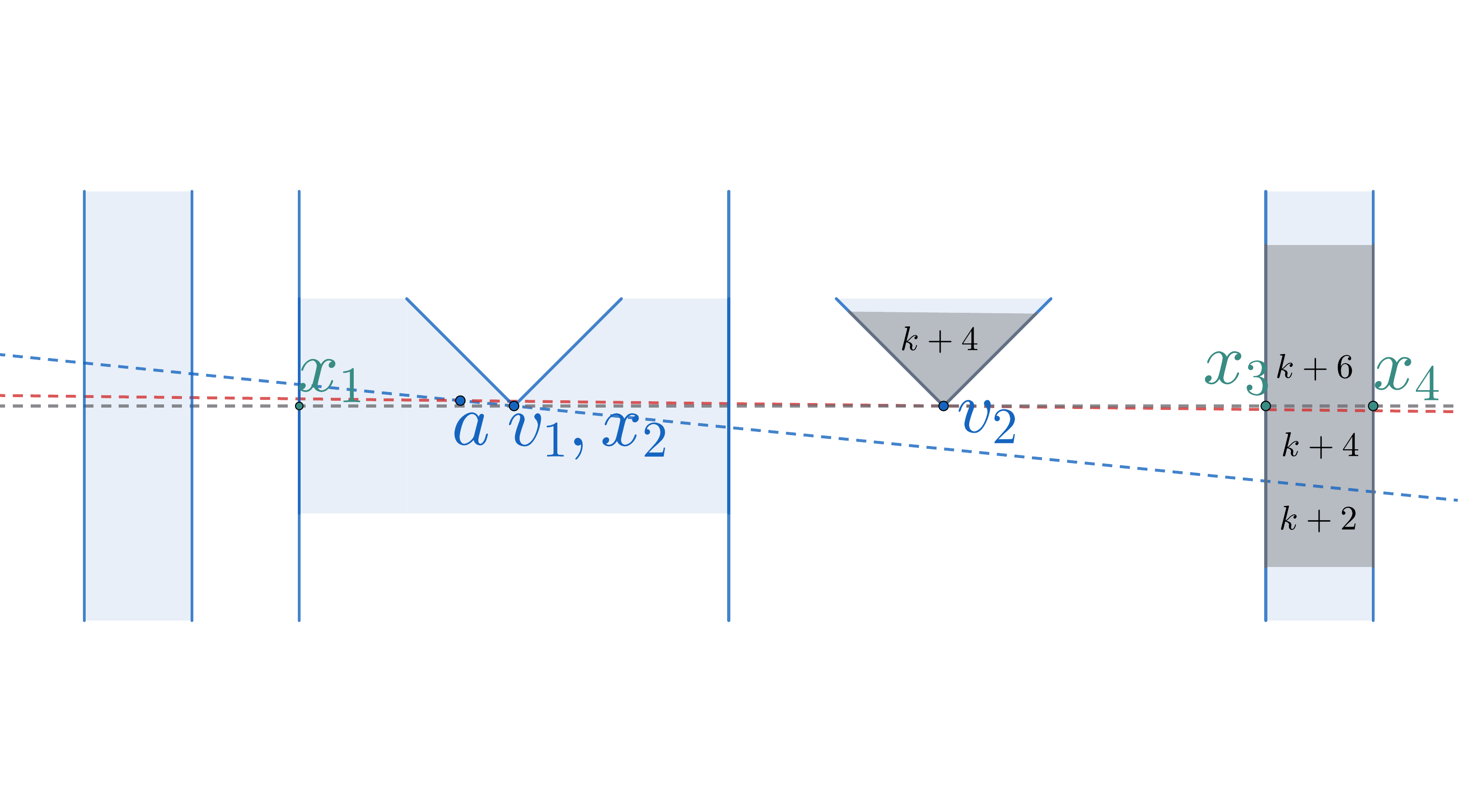}
\caption{Above $\ell_{g}$}\label{fig:RCS-genericA0}
\end{subfigure}
\begin{subfigure}[b]{.49\linewidth}
\includegraphics[width=\linewidth]{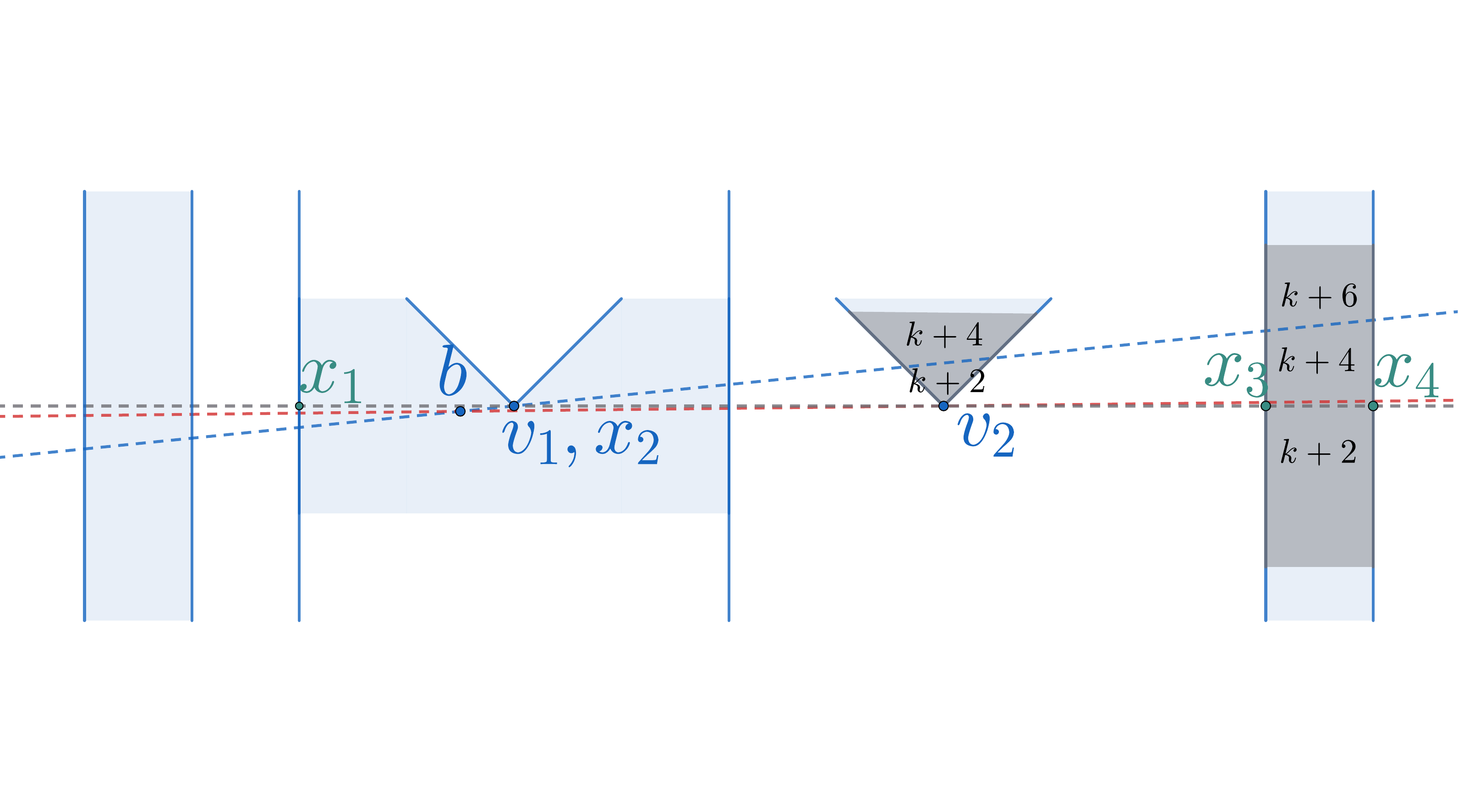}
\caption{Below $\ell_{g}$}\label{fig:RCS-genericB0}
\end{subfigure}

\caption{RCS; $Z = k + 1$, $W = k + 2$.}
\label{fig:RCS-generic-0}
\end{figure}

 \begin{figure}[H]
\centering
\begin{subfigure}[b]{.49\linewidth}
\includegraphics[width=\linewidth]{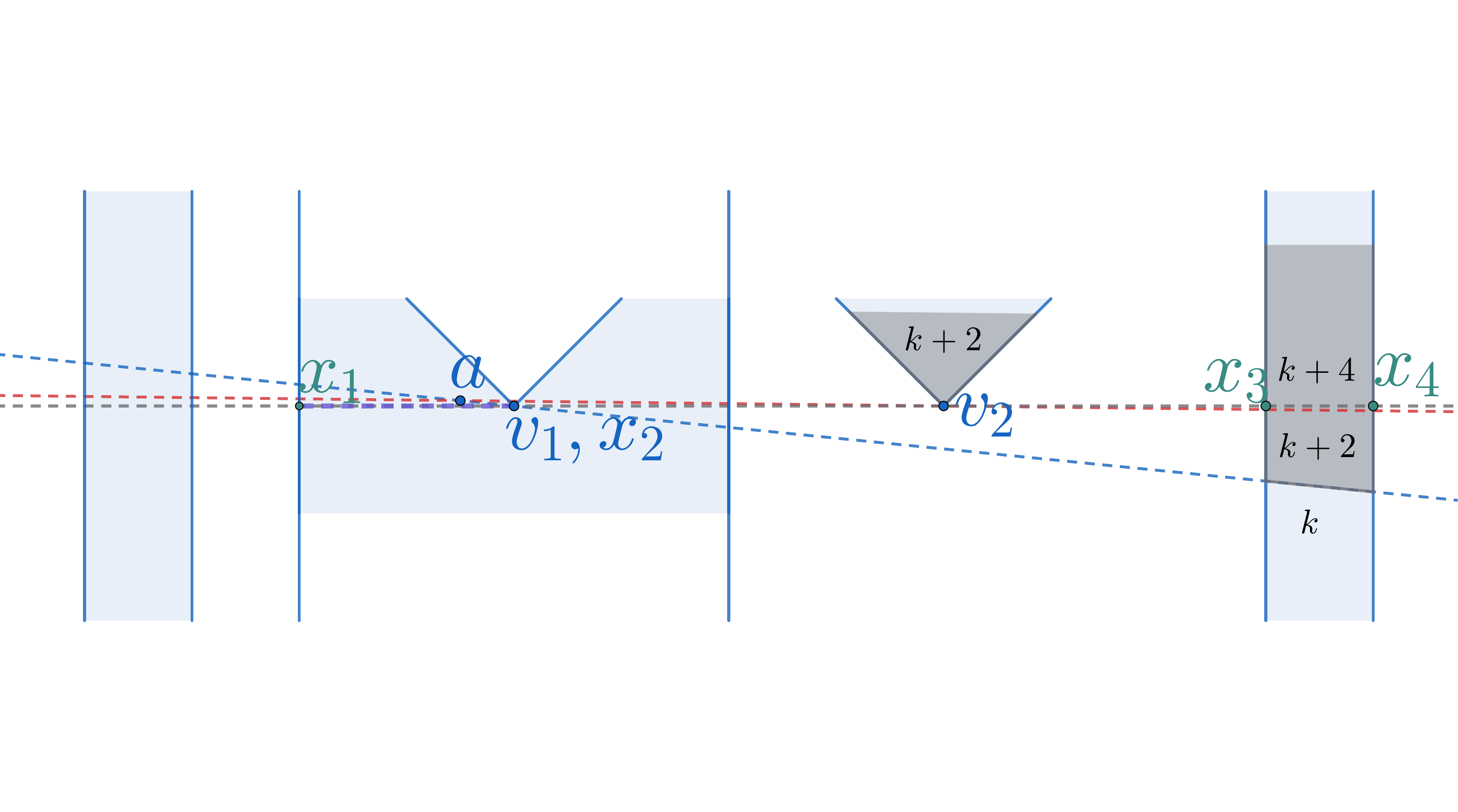}
\caption{Above $\ell_{g}$}\label{fig:RCS-genericA2}
\end{subfigure}
\begin{subfigure}[b]{.49\linewidth}
\includegraphics[width=\linewidth]{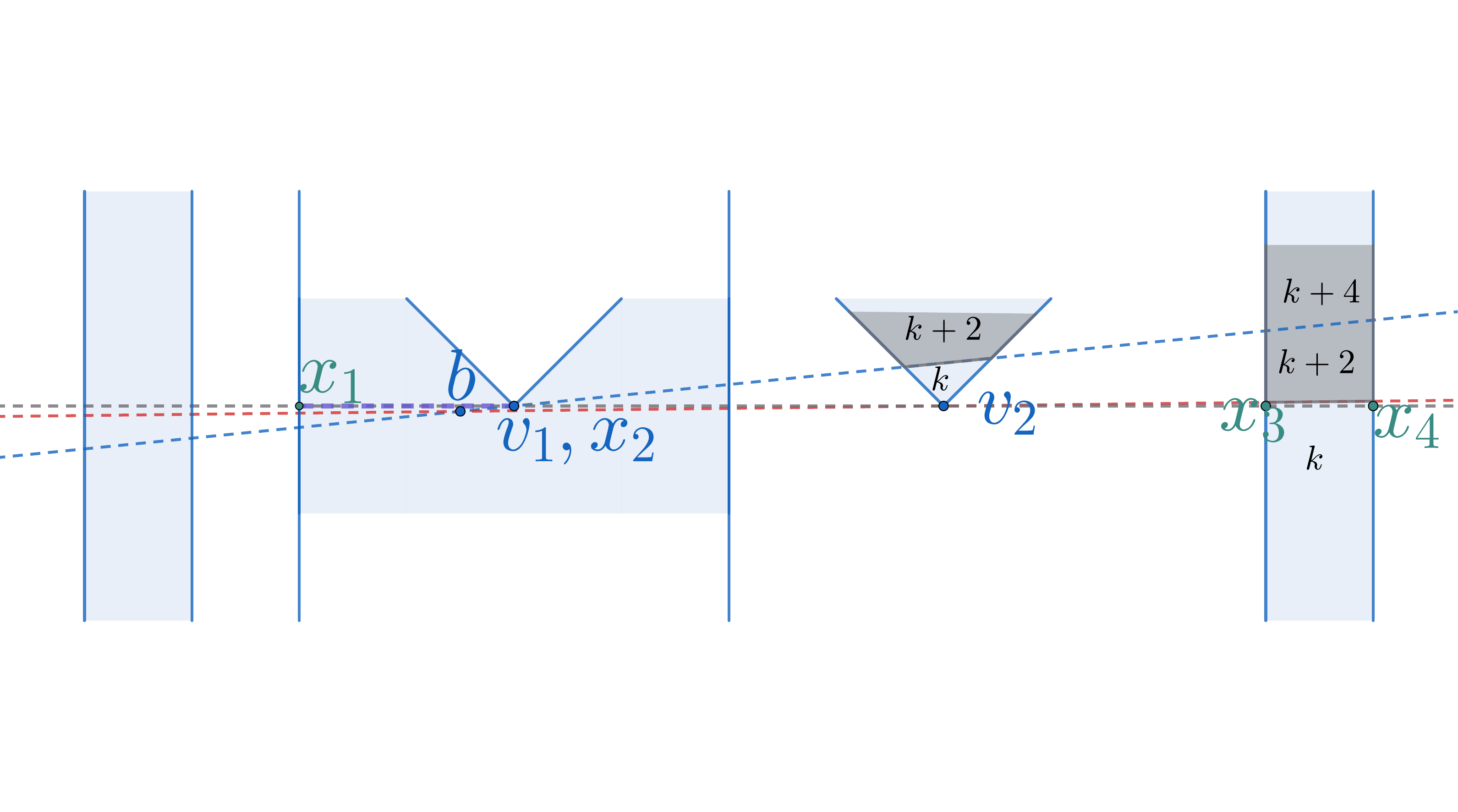}
\caption{Below $\ell_{g}$}\label{fig:RCS-genericB2}
\end{subfigure}

\caption{RCS; $Z = k - 1$, $W = k$.}
\label{fig:RCS-generic-2}
\end{figure}

 \begin{figure}[H]
\centering
\begin{subfigure}[b]{.49\linewidth}
\includegraphics[width=\linewidth]{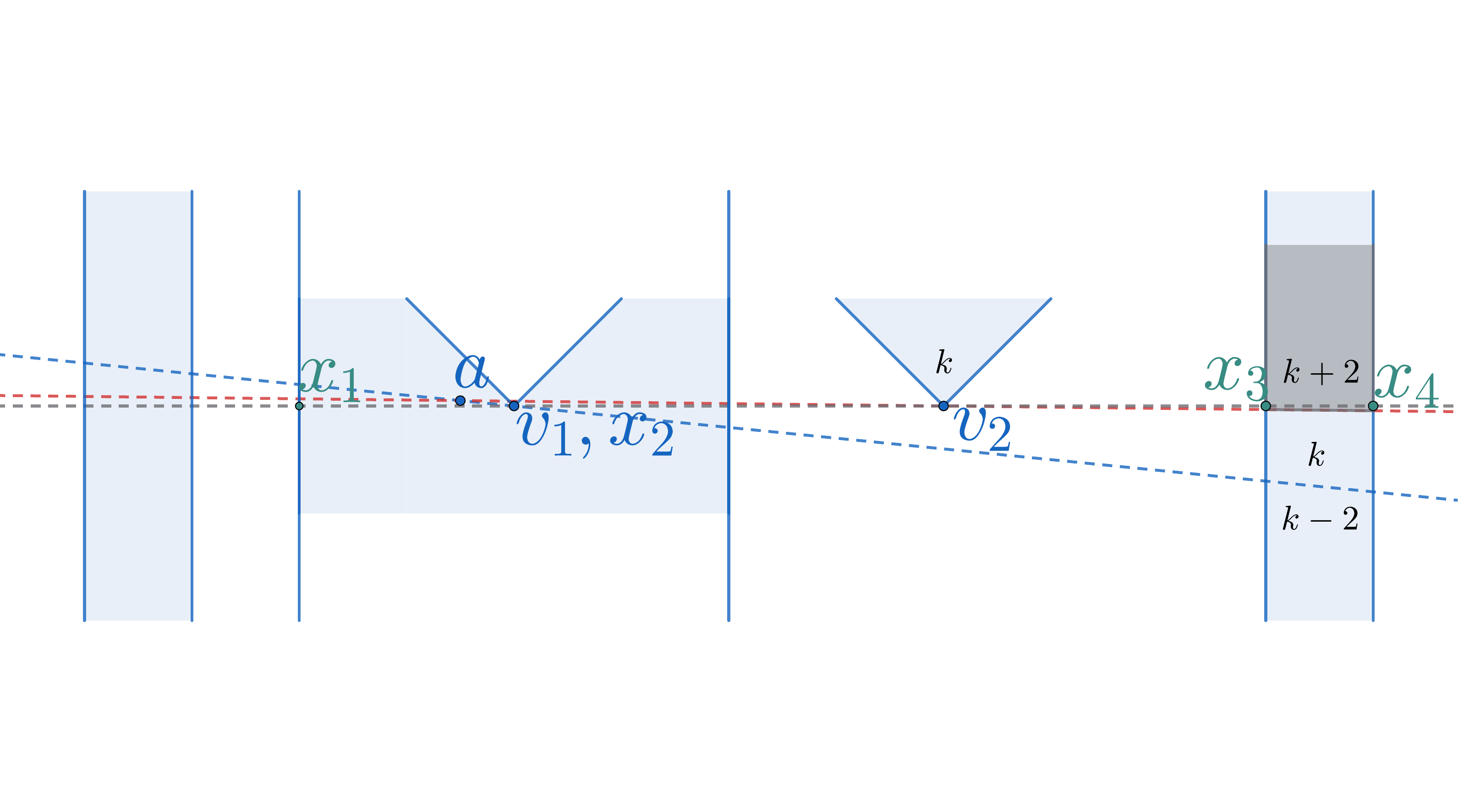}
\caption{Above $\ell_{g}$}\label{fig:RCS-genericA4}
\end{subfigure}
\begin{subfigure}[b]{.49\linewidth}
\includegraphics[width=\linewidth]{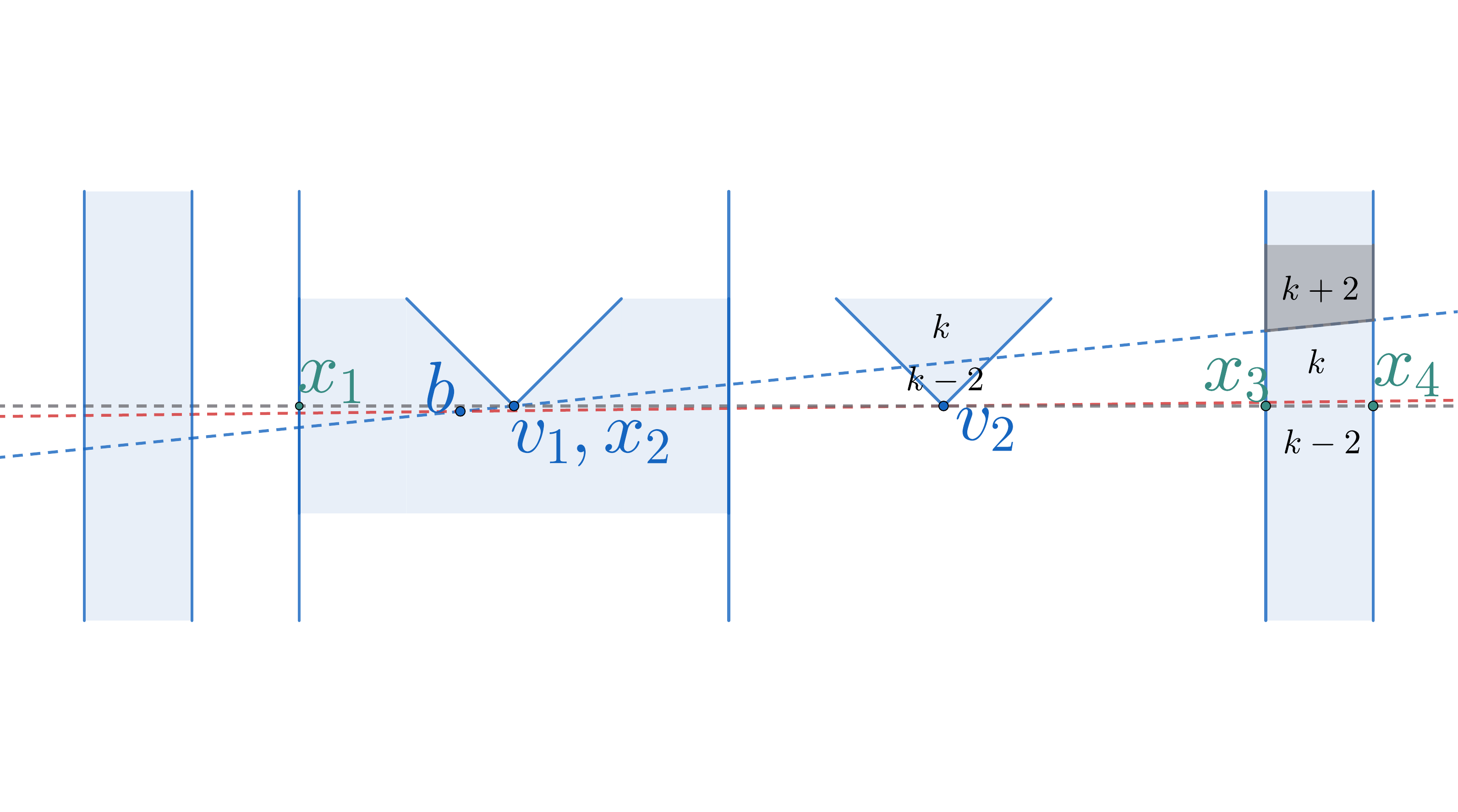}
\caption{Below $\ell_{g}$}\label{fig:RCS-genericB4}
\end{subfigure}

\caption{RCS; $Z = k - 3$, $W = k - 2$.}
\label{fig:RCS-generic-4}
\end{figure}

 \begin{figure}[H]
\centering
\begin{subfigure}[b]{.49\linewidth}
\includegraphics[width=\linewidth]{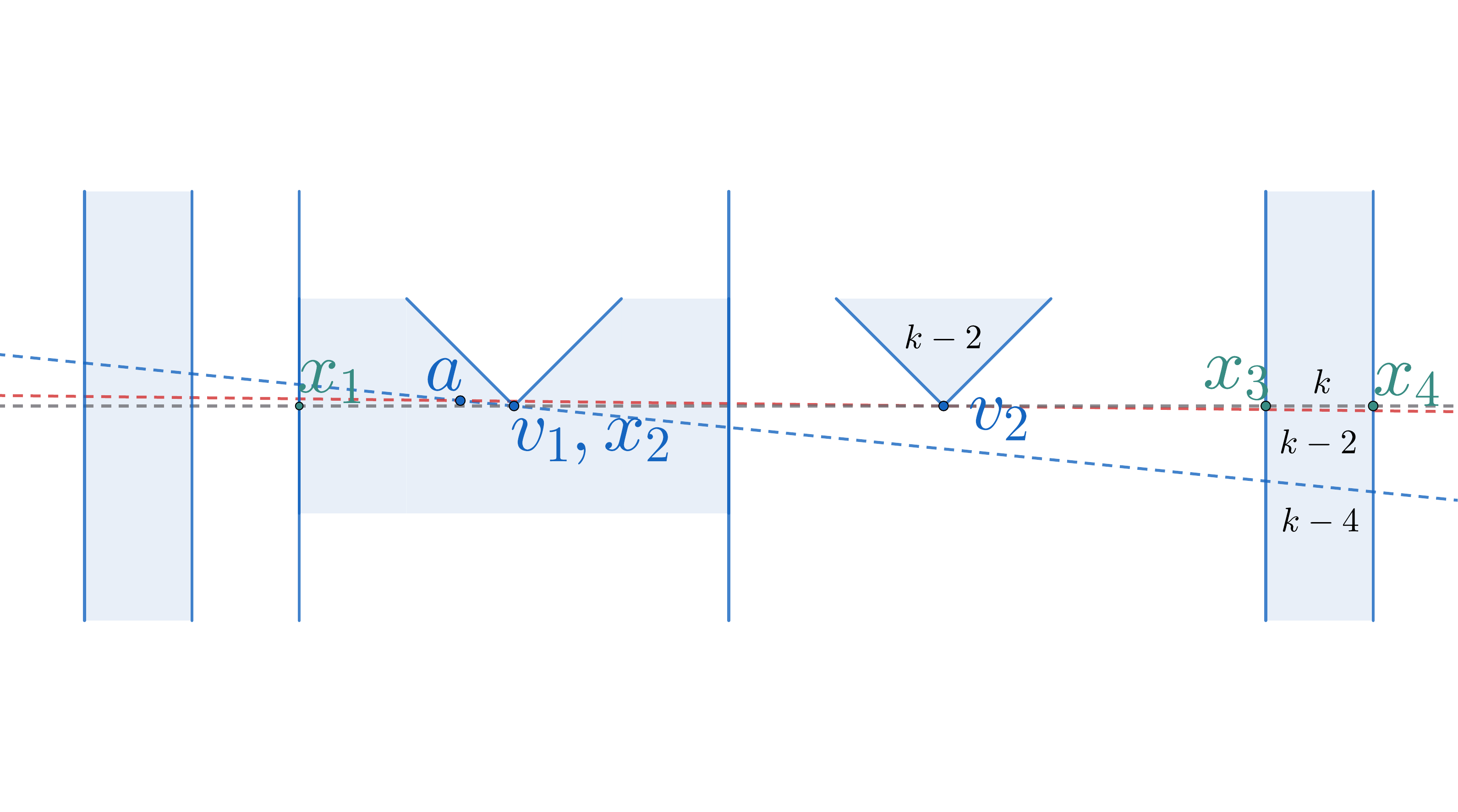}
\caption{Above $\ell_{g}$}\label{fig:RCS-genericA6}
\end{subfigure}
\begin{subfigure}[b]{.49\linewidth}
\includegraphics[width=\linewidth]{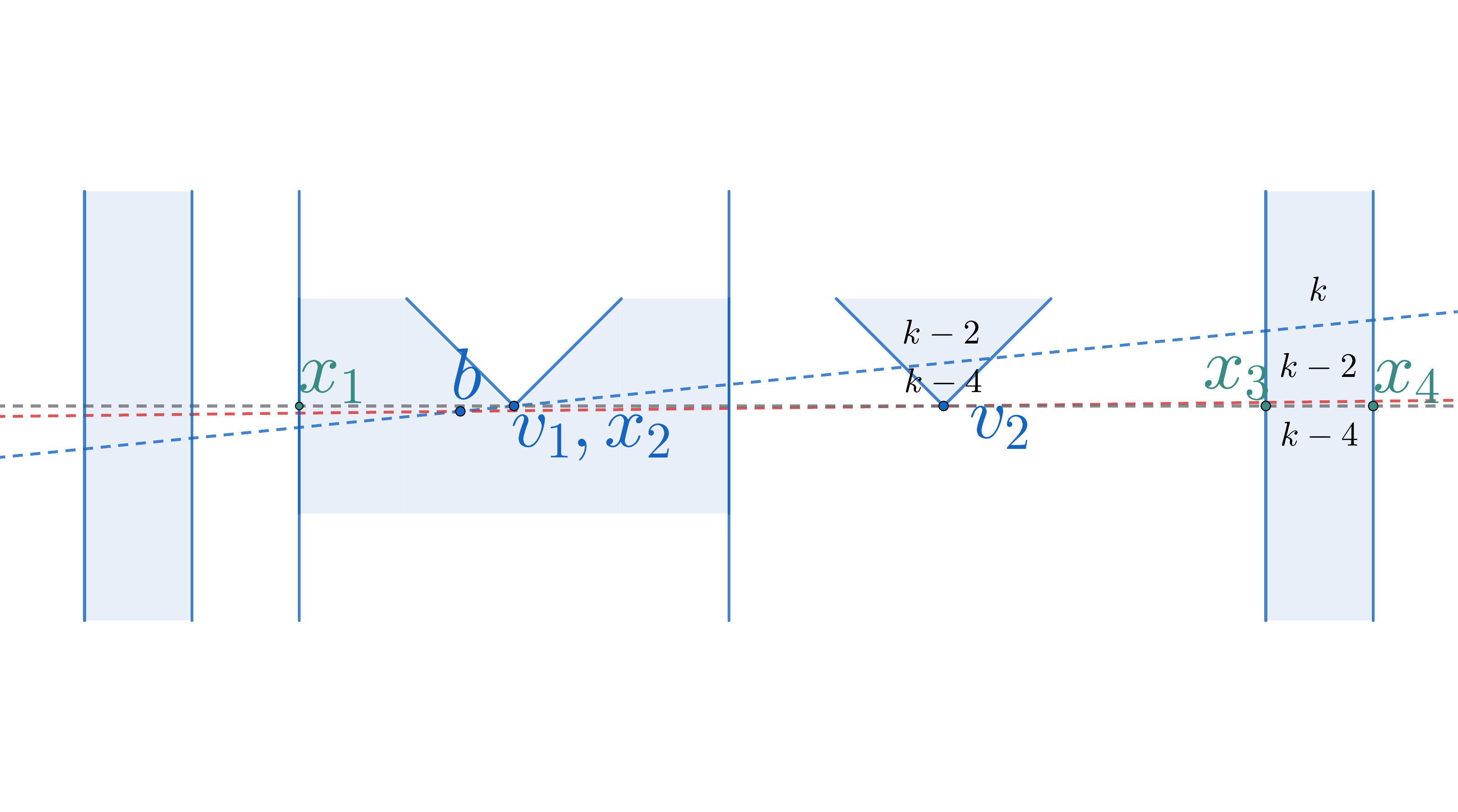}
\caption{Below $\ell_{g}$}\label{fig:RCS-genericB6}
\end{subfigure}

\caption{RCS; $Z = k - 5$, $W = k - 4$.}
\label{fig:RCS-generic-6}
\end{figure}

 \begin{figure}[H]
\centering
\begin{subfigure}[b]{.49\linewidth}
\includegraphics[width=\linewidth]{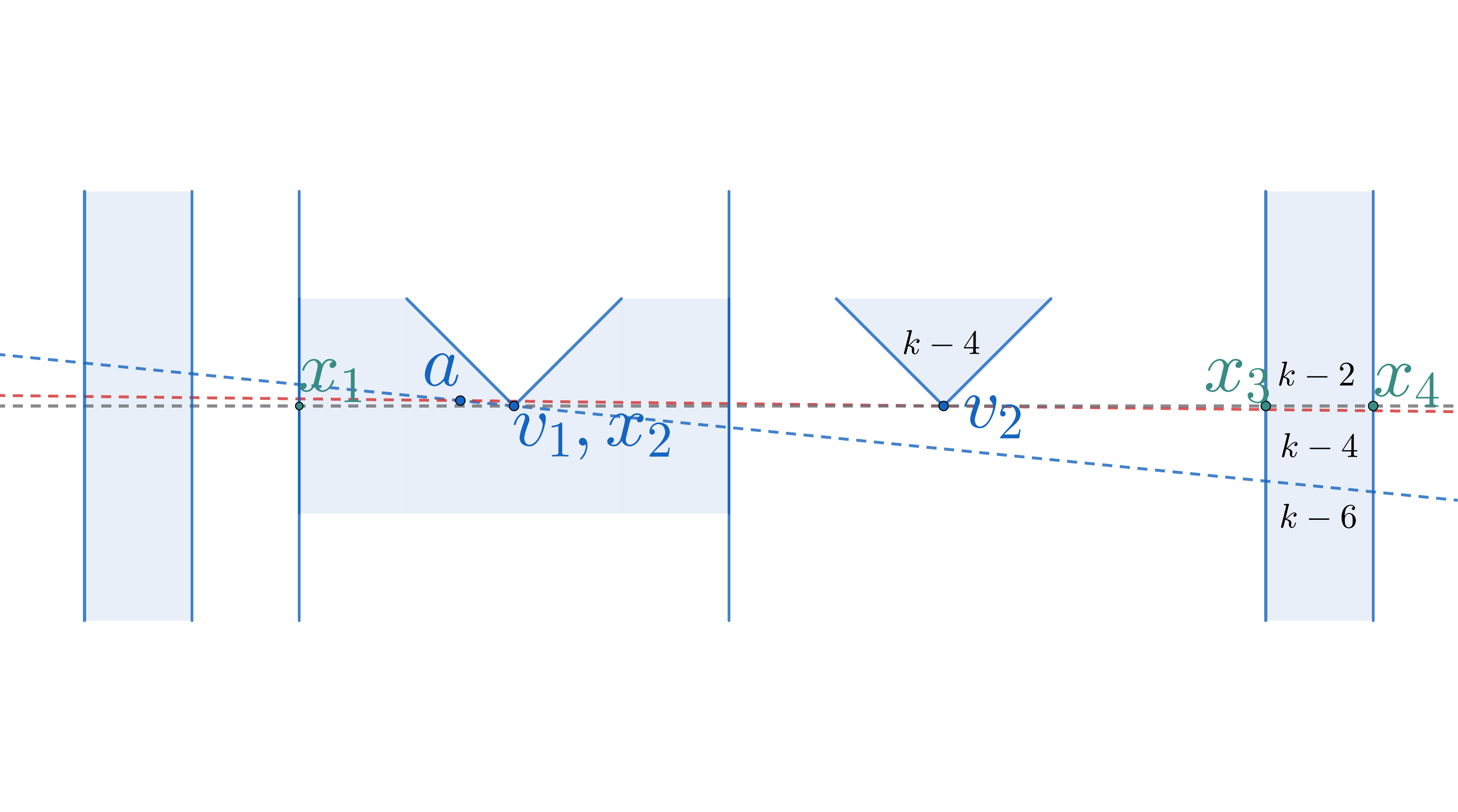}
\caption{Above $\ell_{g}$}\label{fig:RCS-genericA8}
\end{subfigure}
\begin{subfigure}[b]{.49\linewidth}
\includegraphics[width=\linewidth]{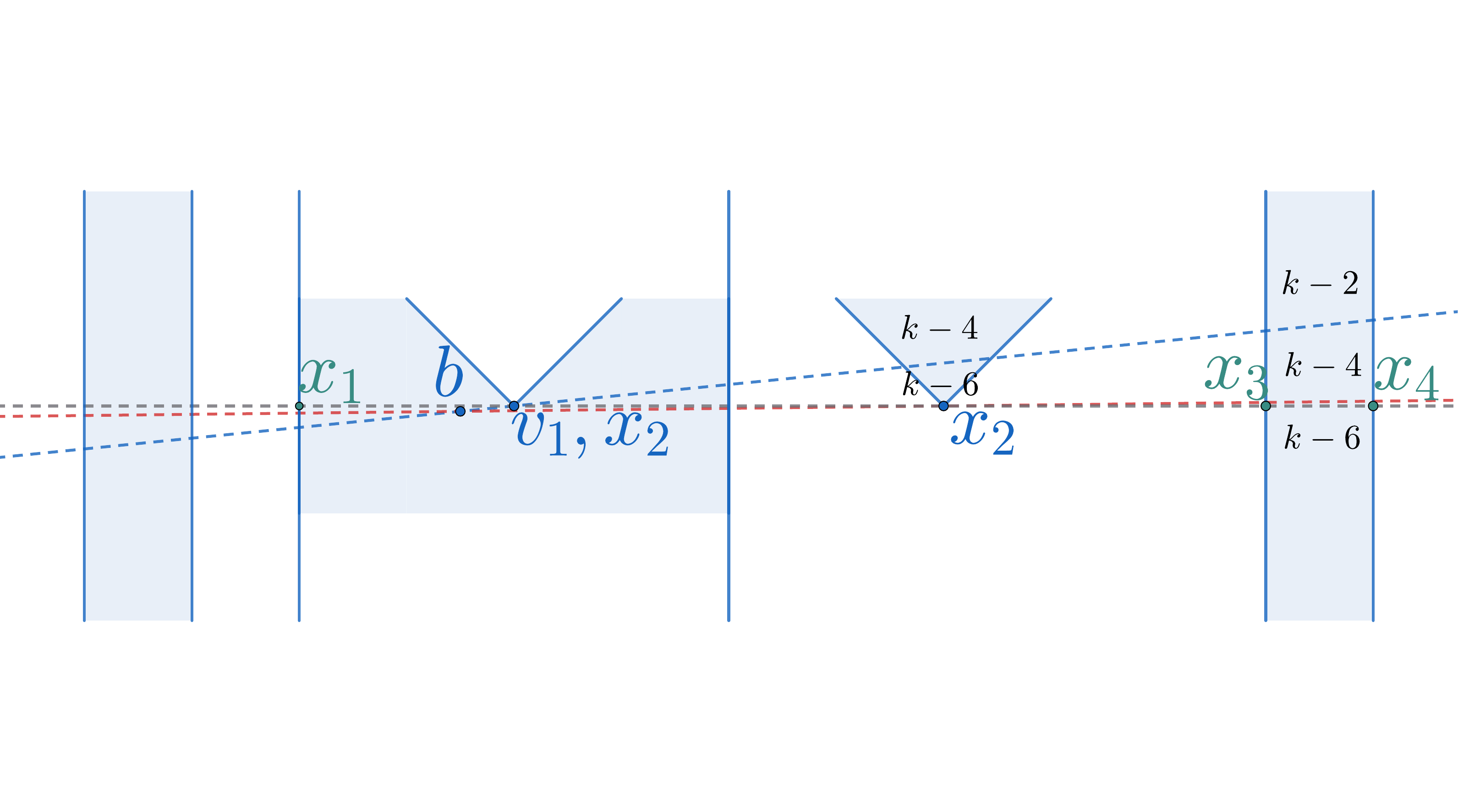}
\caption{Below $\ell_{g}$}\label{fig:RCS-genericB8}
\end{subfigure}

\caption{RCS; $Z = k - 7$, $W = k - 6$.}
\label{fig:RCS-generic-8}
\end{figure}

\section{RCO}
\label{appendix:RCO}
\subsection{Reflex Convex Opposite (RCO)}
\begin{lemma} 
\label{lemma:RCO}
When $Z = k - 1$ (Figure~\ref{fig:RCO-generic-2}), an appear/disappear event occurs at $v_{2}$. When $W = k$, a merge/split event occurs at $x_{3}x_{4}$ (Figure~\ref{fig:RCO-generic-2}). When $W=k-2$, an appear/disappear event occurs at $x_{3}x_{4}$ (Figure~\ref{fig:RCO-generic-4}). No event occurs for any other $Z$ or $W$. 
\\

\end{lemma}
\begin{proof}
    See Figure~\ref{fig:RCO-generic-0} to Figure~\ref{fig:RCO-generic-8}.

    For $Z \geq k + 3$, the region near $v_{2}$ is entirely in shadow for $W \geq k + 4$, the region surrounding $x_{3}x_{4}$ is in shadow.
    
    For $Z \leq k - 9$, the region surrounding $v_{2}$ is entirely visible. For $W \leq k - 8$, the region surrounding $x_{3}x_{4}$ is entirely visible.
\end{proof}

 \begin{figure}[H]
\centering
\begin{subfigure}[b]{.49\linewidth}
\includegraphics[width=\linewidth]{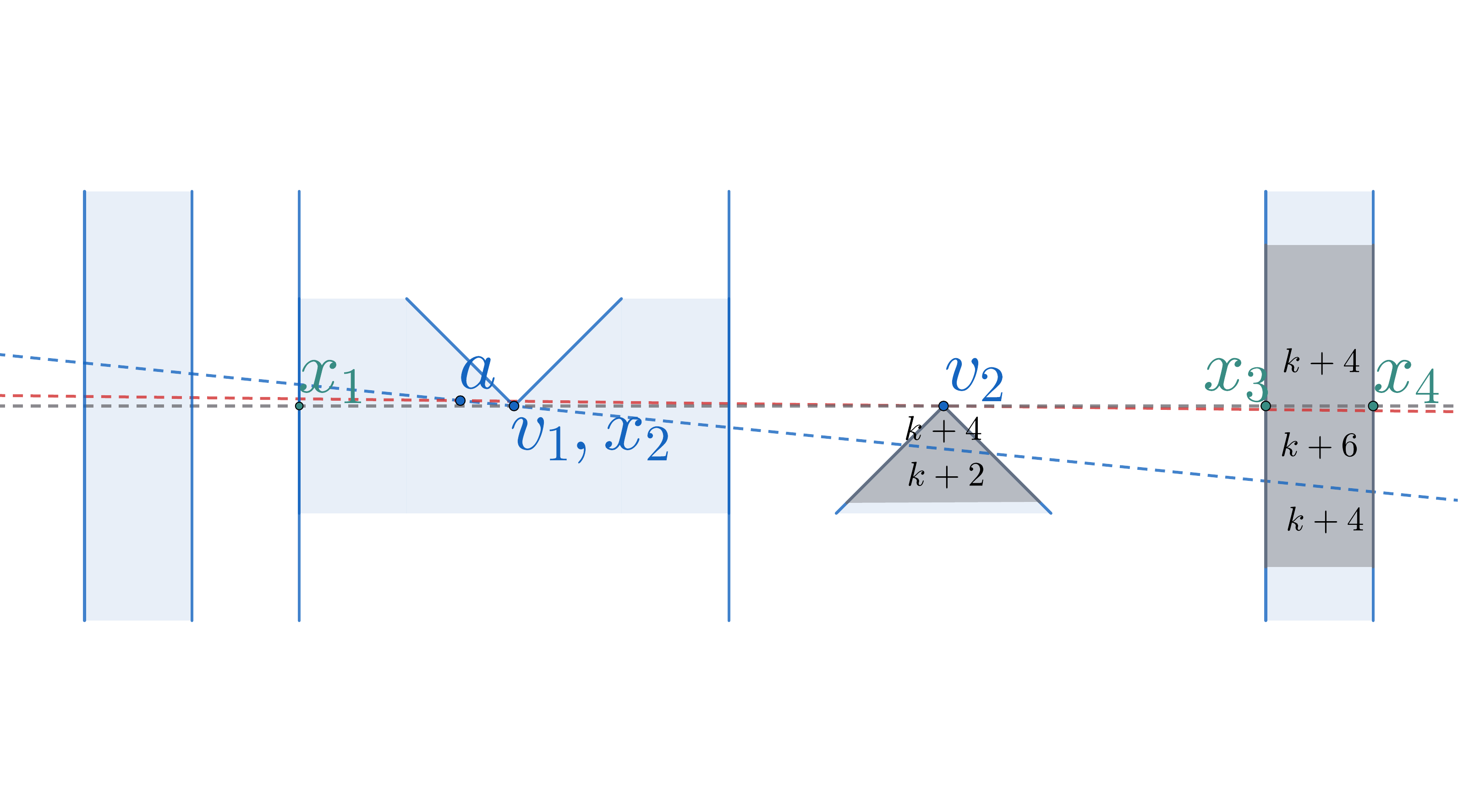}
\caption{Above $\ell_{g}$}\label{fig:RCO-genericA0}
\end{subfigure}
\begin{subfigure}[b]{.49\linewidth}
\includegraphics[width=\linewidth]{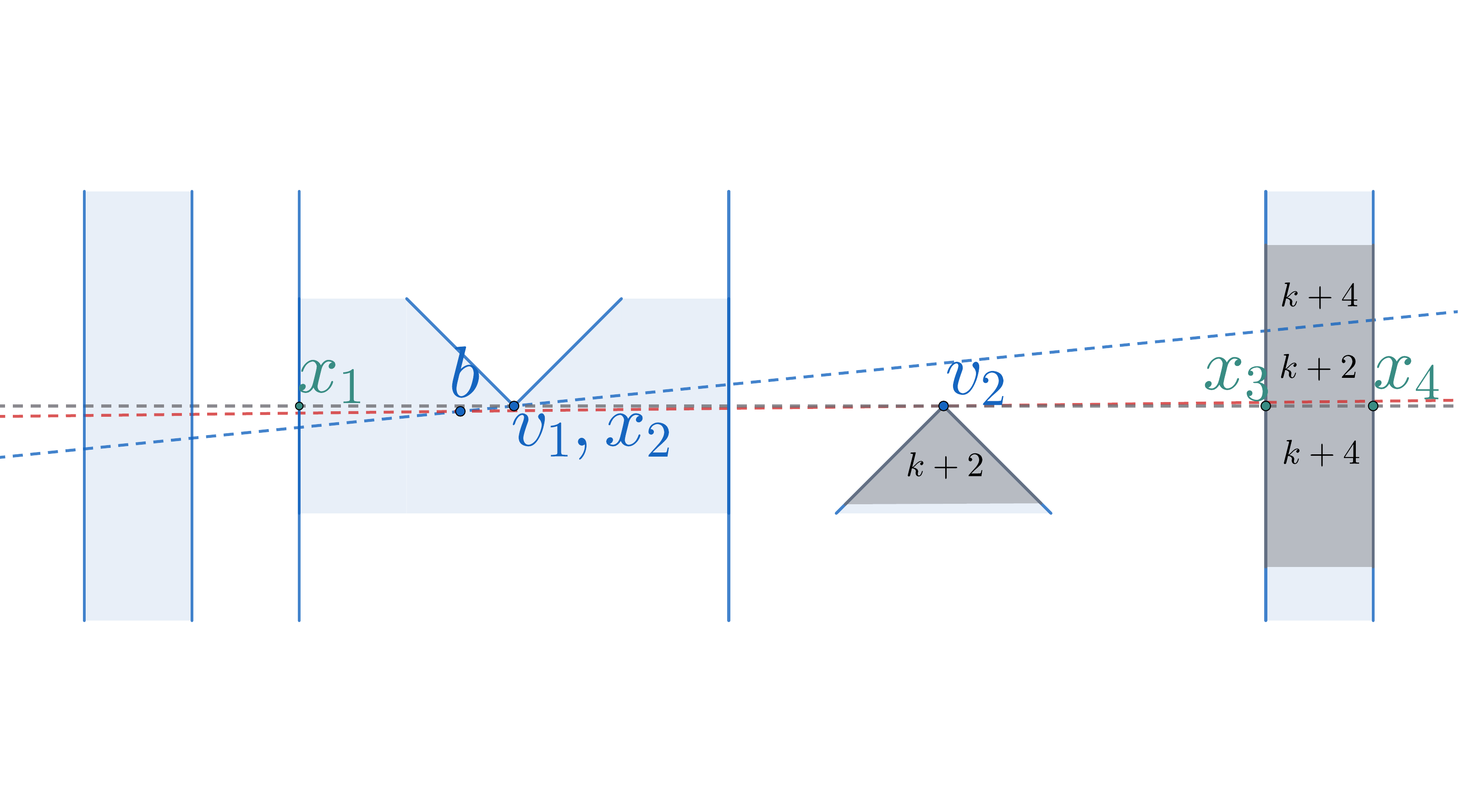}
\caption{Below $\ell_{g}$}\label{fig:RCO-genericB0}
\end{subfigure}

\caption{RCO; $Z = k + 1$, $W = k + 2$.}
\label{fig:RCO-generic-0}
\end{figure}

 \begin{figure}[H]
\centering
\begin{subfigure}[b]{.49\linewidth}
\includegraphics[width=\linewidth]{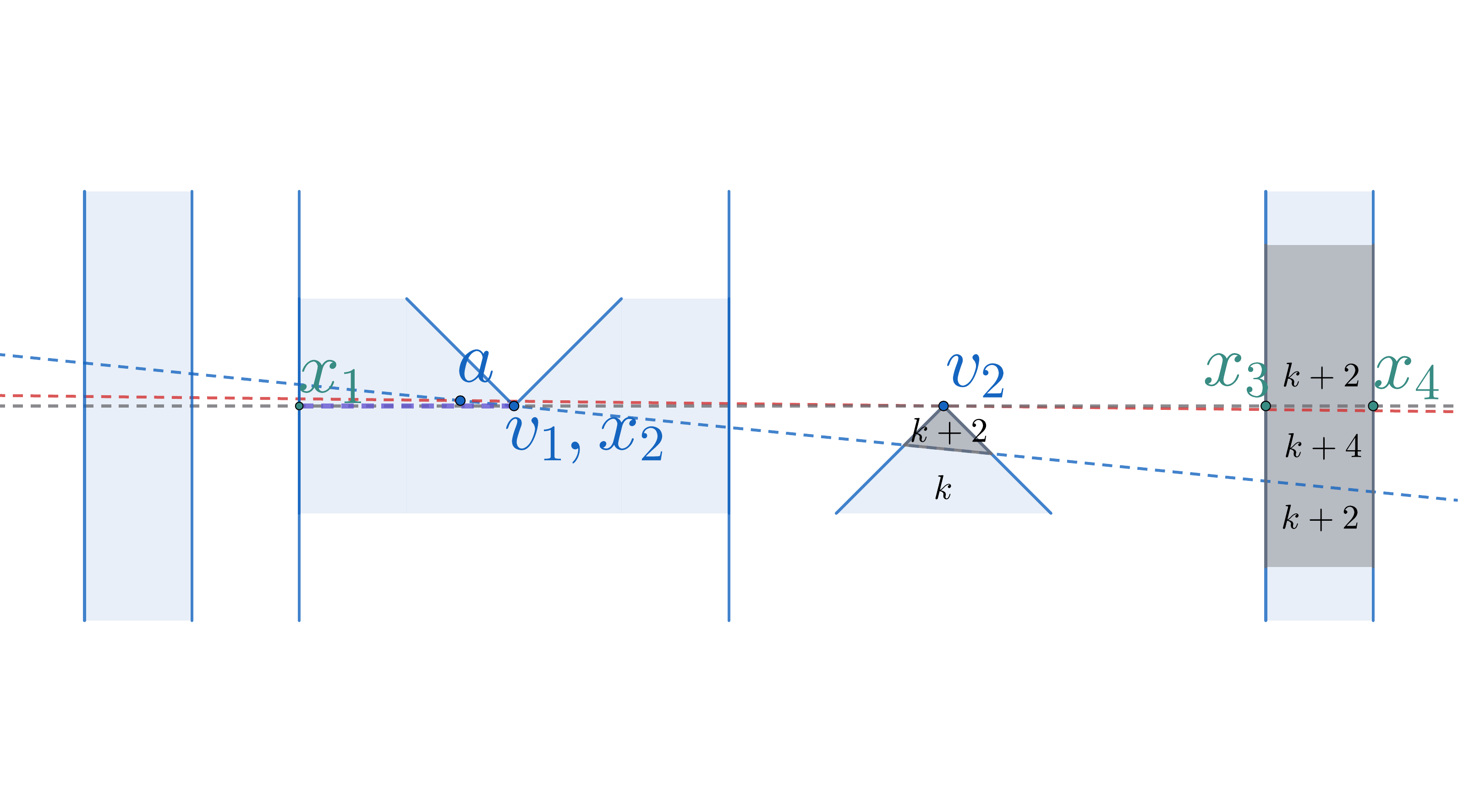}
\caption{Above $\ell_{g}$}\label{fig:RCO-genericA2}
\end{subfigure}
\begin{subfigure}[b]{.49\linewidth}
\includegraphics[width=\linewidth]{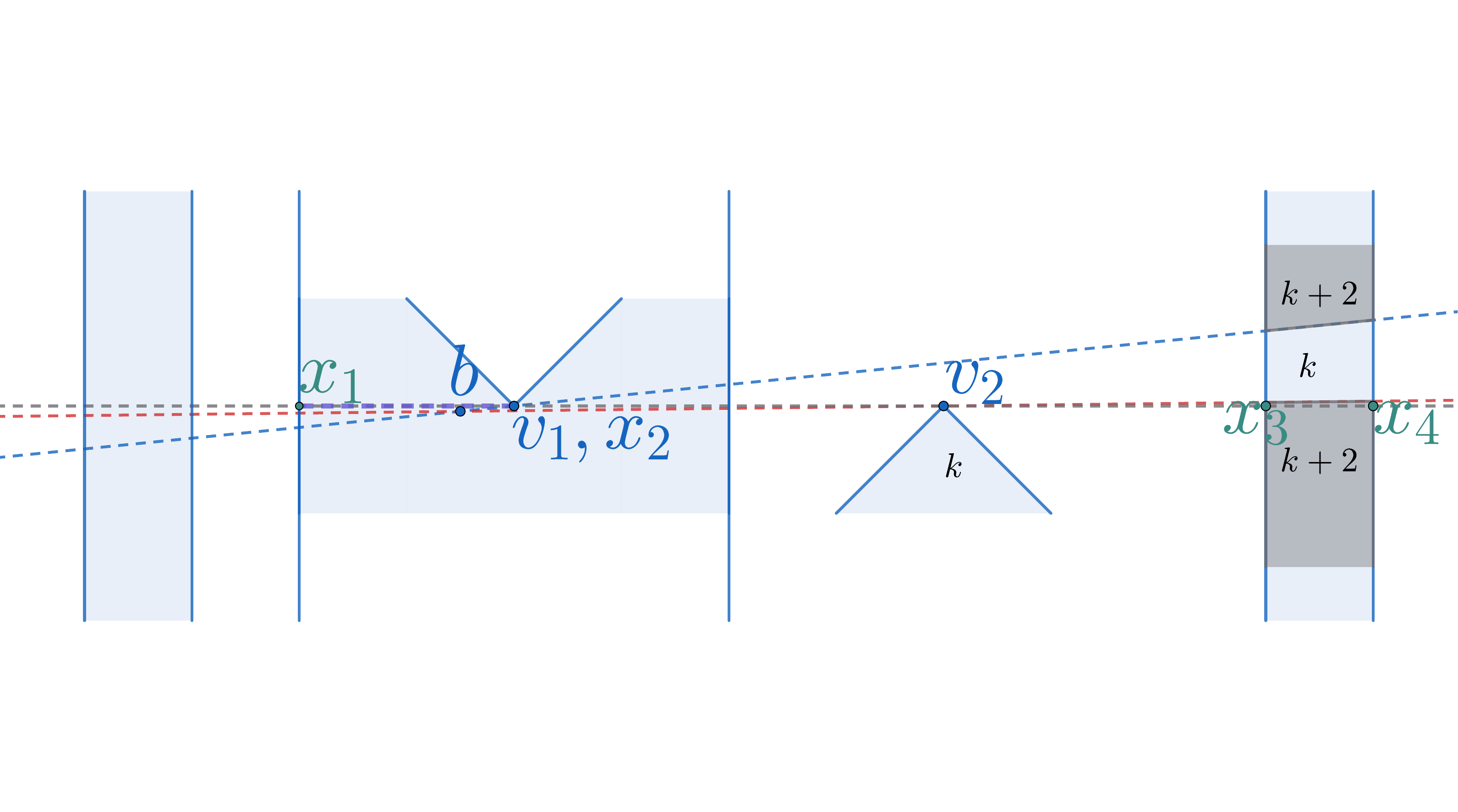}
\caption{Below $\ell_{g}$}\label{fig:RCO-genericB2}
\end{subfigure}

\caption{RCO; $Z = k - 1$, $W = k$ (Merge/Split).}
\label{fig:RCO-generic-2}
\end{figure}

 \begin{figure}[H]
\centering
\begin{subfigure}[b]{.49\linewidth}
\includegraphics[width=\linewidth]{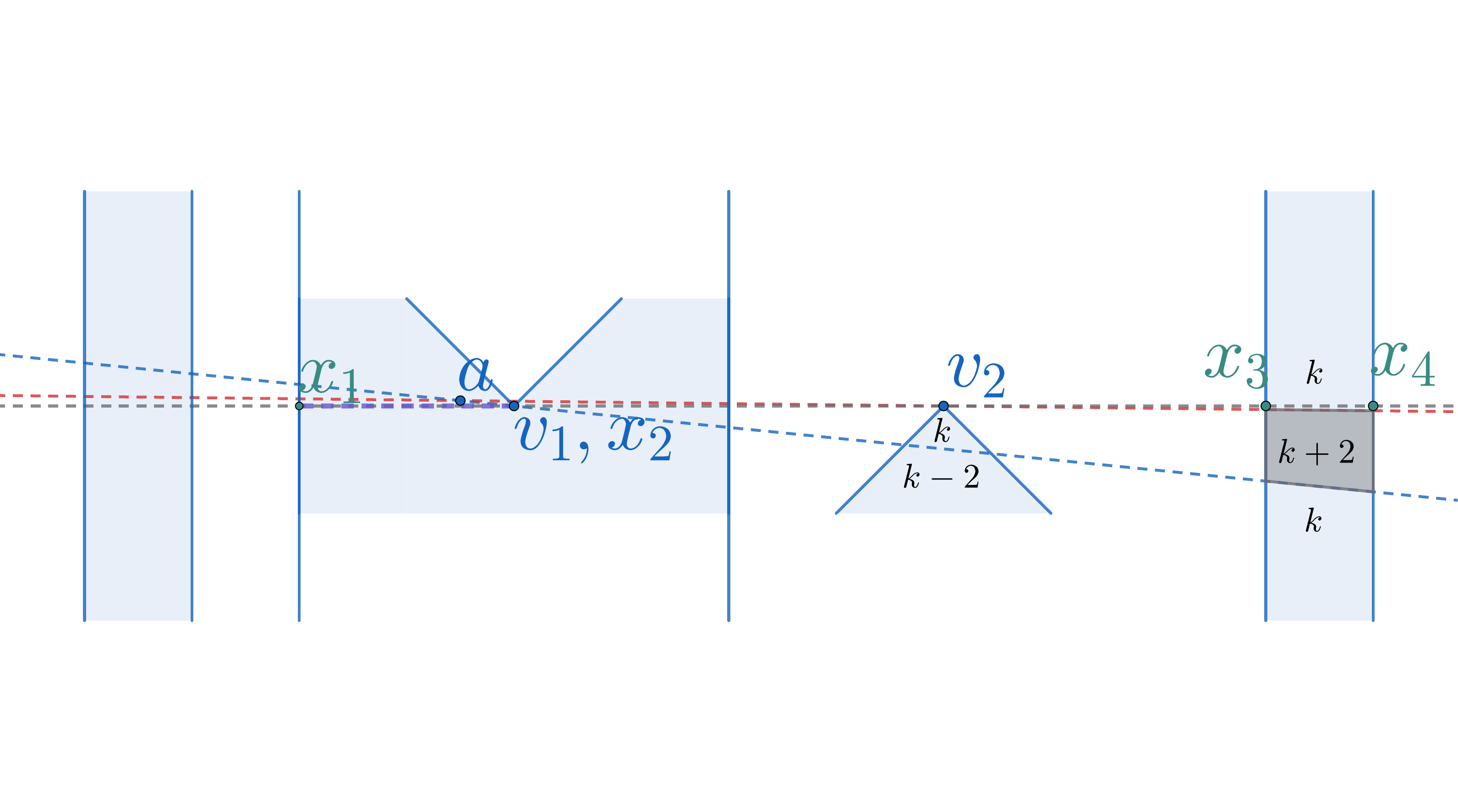}
\caption{Above $\ell_{g}$}\label{fig:RCO-genericA4}
\end{subfigure}
\begin{subfigure}[b]{.49\linewidth}
\includegraphics[width=\linewidth]{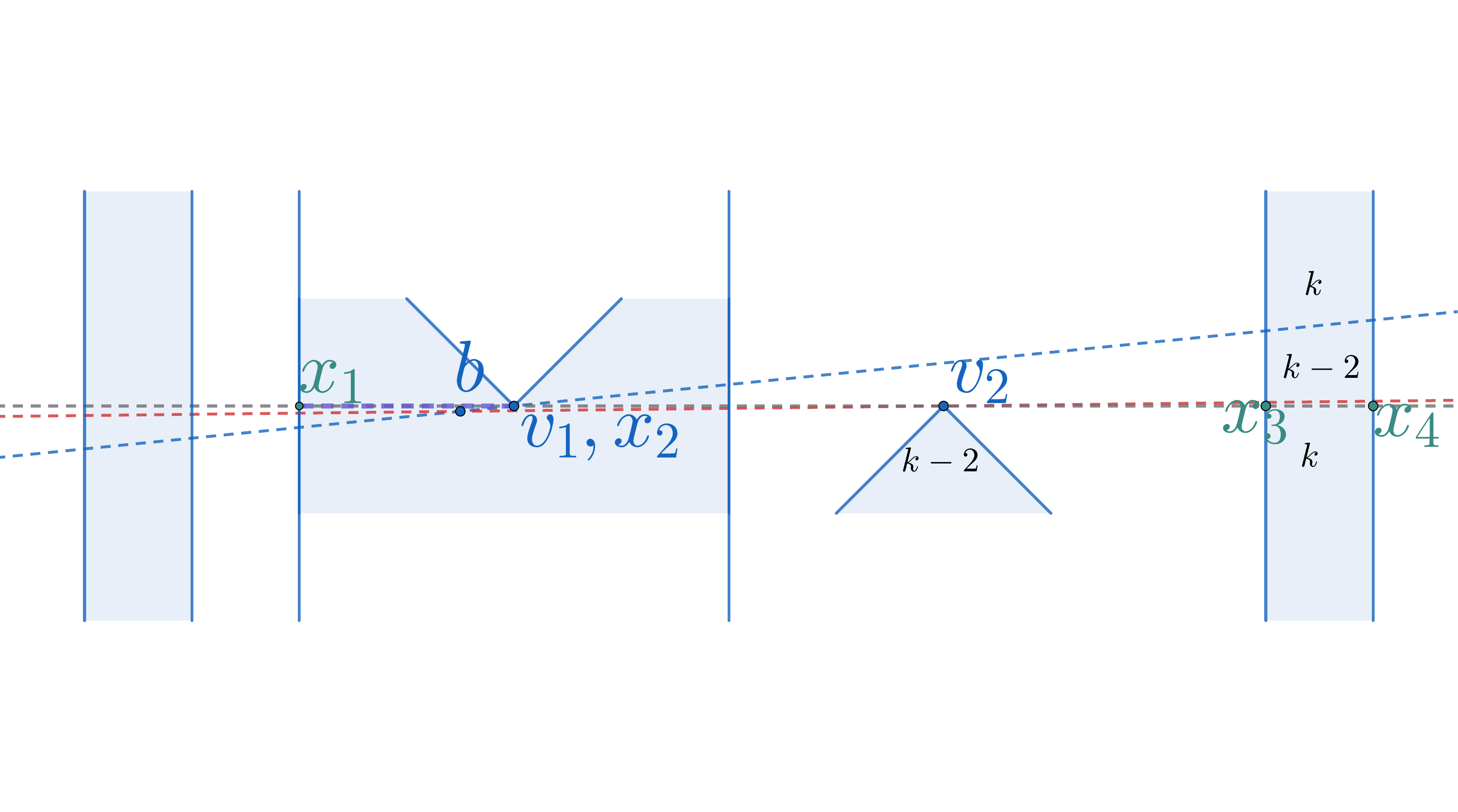}
\caption{Below $\ell_{g}$}\label{fig:RCO-genericB4}
\end{subfigure}

\caption{RCO; $Z = k - 3$, $W = k - 2$ (Appear/Disappear).}
\label{fig:RCO-generic-4}
\end{figure}

 \begin{figure}[H]
\centering
\begin{subfigure}[b]{.49\linewidth}
\includegraphics[width=\linewidth]{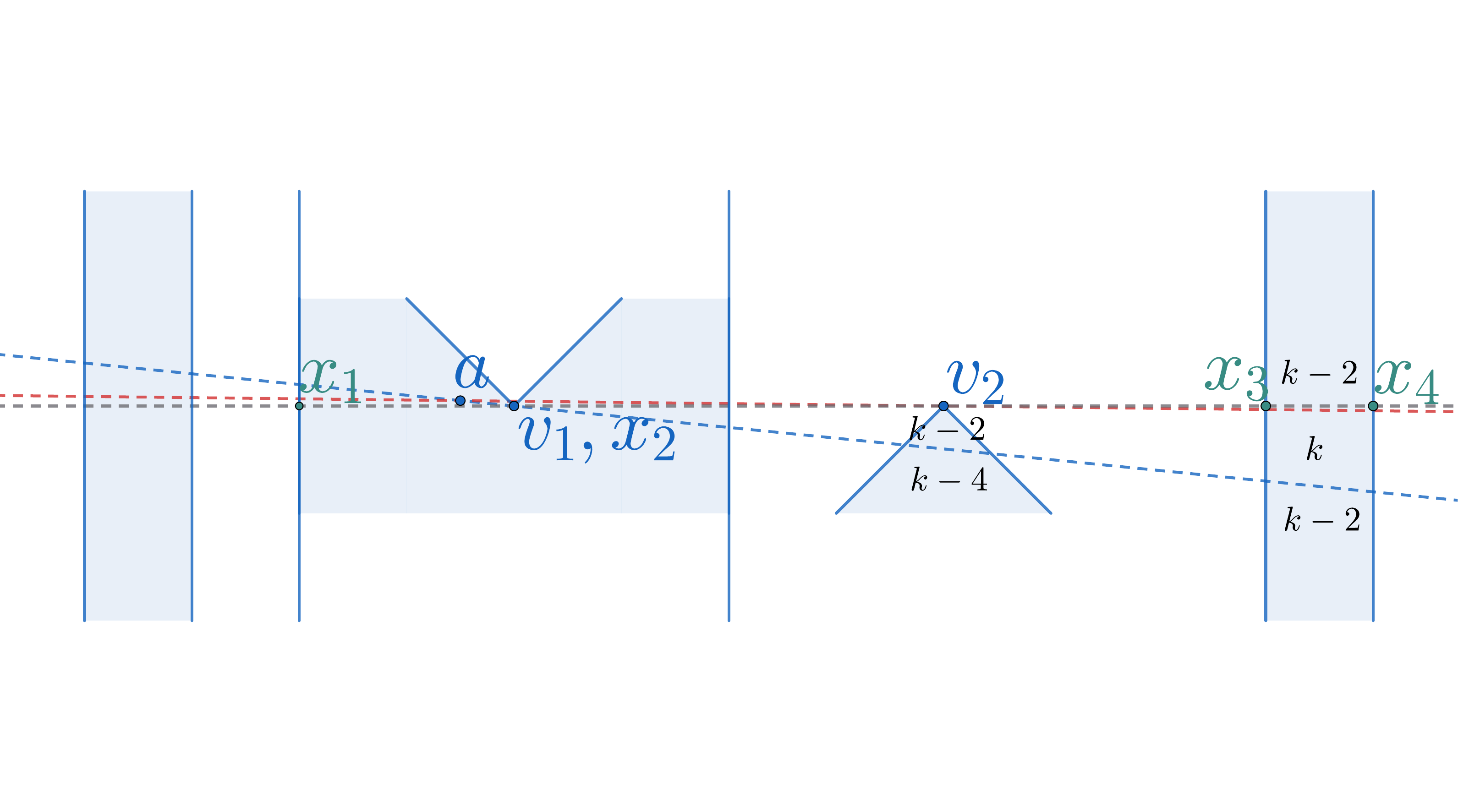}
\caption{Above $\ell_{g}$}\label{fig:RCO-genericA6}
\end{subfigure}
\begin{subfigure}[b]{.49\linewidth}
\includegraphics[width=\linewidth]{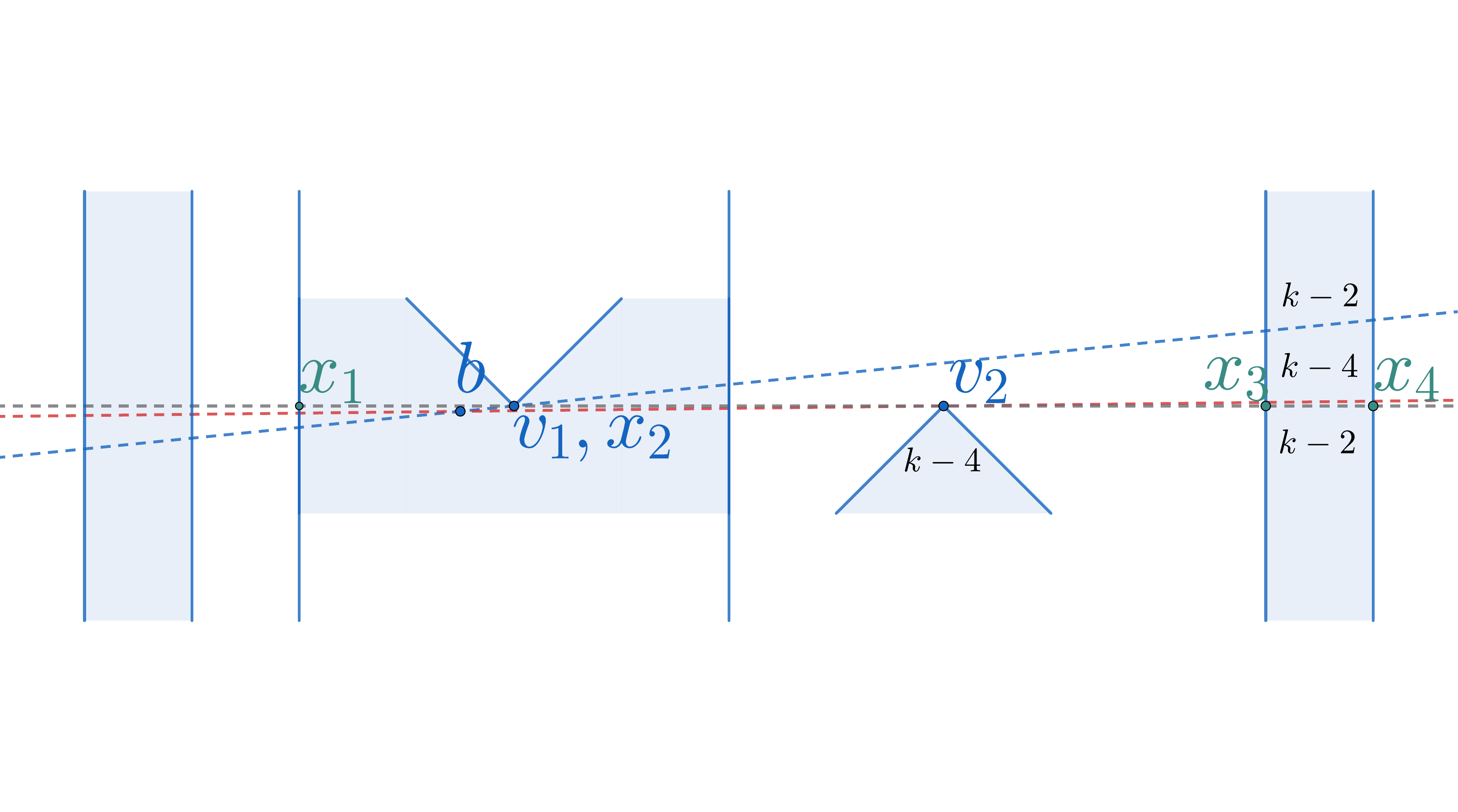}
\caption{Below $\ell_{g}$}\label{fig:RCO-genericB6}
\end{subfigure}

\caption{RCO; $Z = k - 5$, $W = k - 4$.}
\label{fig:RCO-generic-6}
\end{figure}

 \begin{figure}[H]
\centering
\begin{subfigure}[b]{.49\linewidth}
\includegraphics[width=\linewidth]{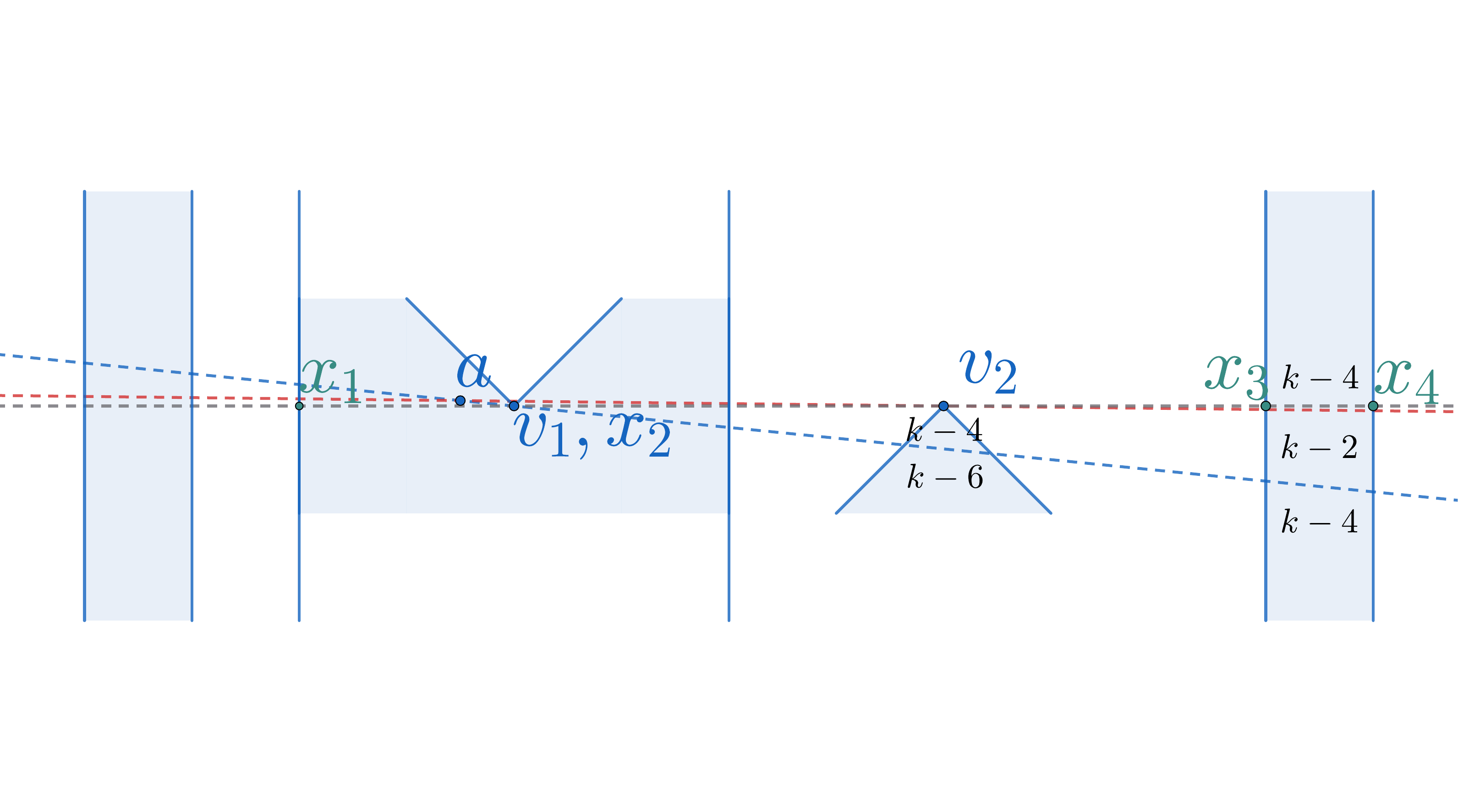}
\caption{Above $\ell_{g}$}\label{fig:RCO-genericA8}
\end{subfigure}
\begin{subfigure}[b]{.49\linewidth}
\includegraphics[width=\linewidth]{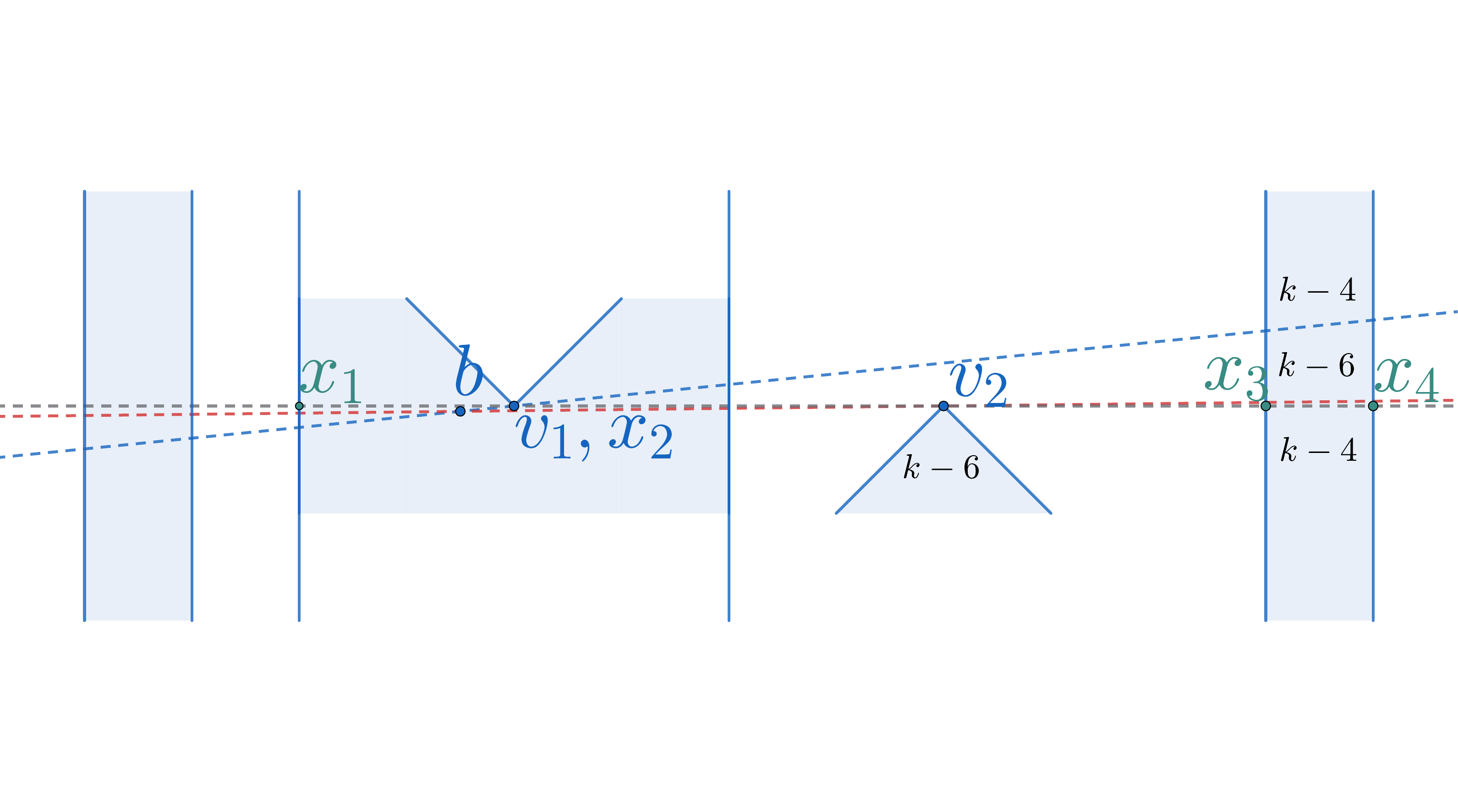}
\caption{Below $\ell_{g}$}\label{fig:RCO-genericB8}
\end{subfigure}

\caption{RCO; $Z = k - 7$, $W = k - 6$.}
\label{fig:RCO-generic-8}
\end{figure}

\section{RRS}
\label{appendix:RRS}
\subsection{Reflex Reflex Same (RRS)}
\begin{lemma} 
\label{lemma:RRS}
When $Z = k$ (Figure~\ref{fig:RRS-generic-0}), a merge/split event occurs at $v_{2}$. No event occurs for any other $Z$ or $W$.

\end{lemma}
\begin{proof}
    See Figure~\ref{fig:RRS-generic-0} to Figure~\ref{fig:RRS-generic-8}.

For $Z \geq k + 2$, $v_{2}$ and its surroundings are entirely in shadow. For $W \geq k + 4$, $x_{3}x_{4}$ and its surroundings are entirely in shadow.

For $Z \leq k - 10$, $v_{2}$ and its surroundings are entirely in shadow. For $W \leq k - 8$, $x_{3}x_{4}$ and its surroundings are entirely in shadow. 
\end{proof}

 \begin{figure}[H]
\centering
\begin{subfigure}[b]{.49\linewidth}
\includegraphics[width=\linewidth]{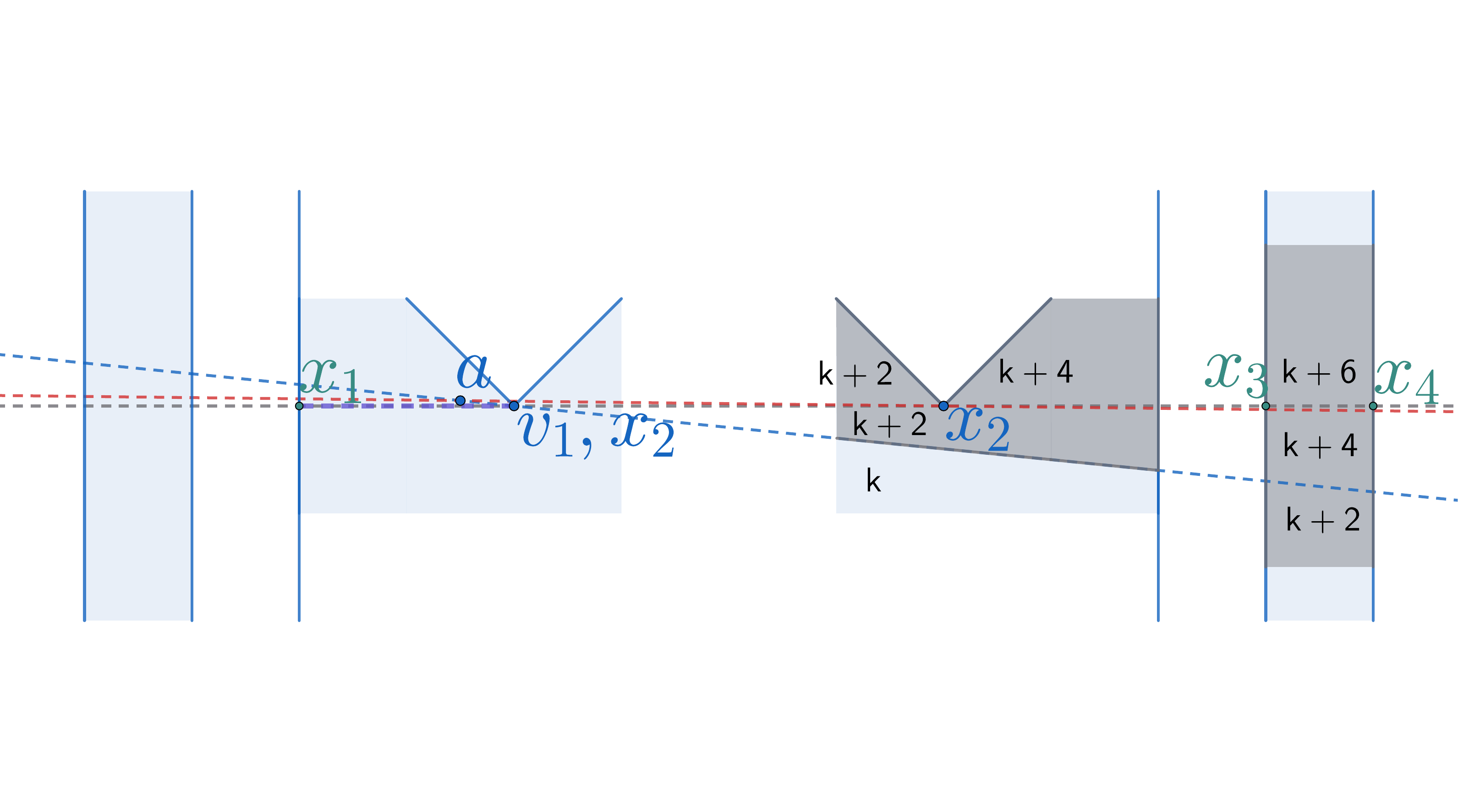}
\caption{Above $\ell_{g}$}\label{fig:RRS-genericA0}
\end{subfigure}
\begin{subfigure}[b]{.49\linewidth}
\includegraphics[width=\linewidth]{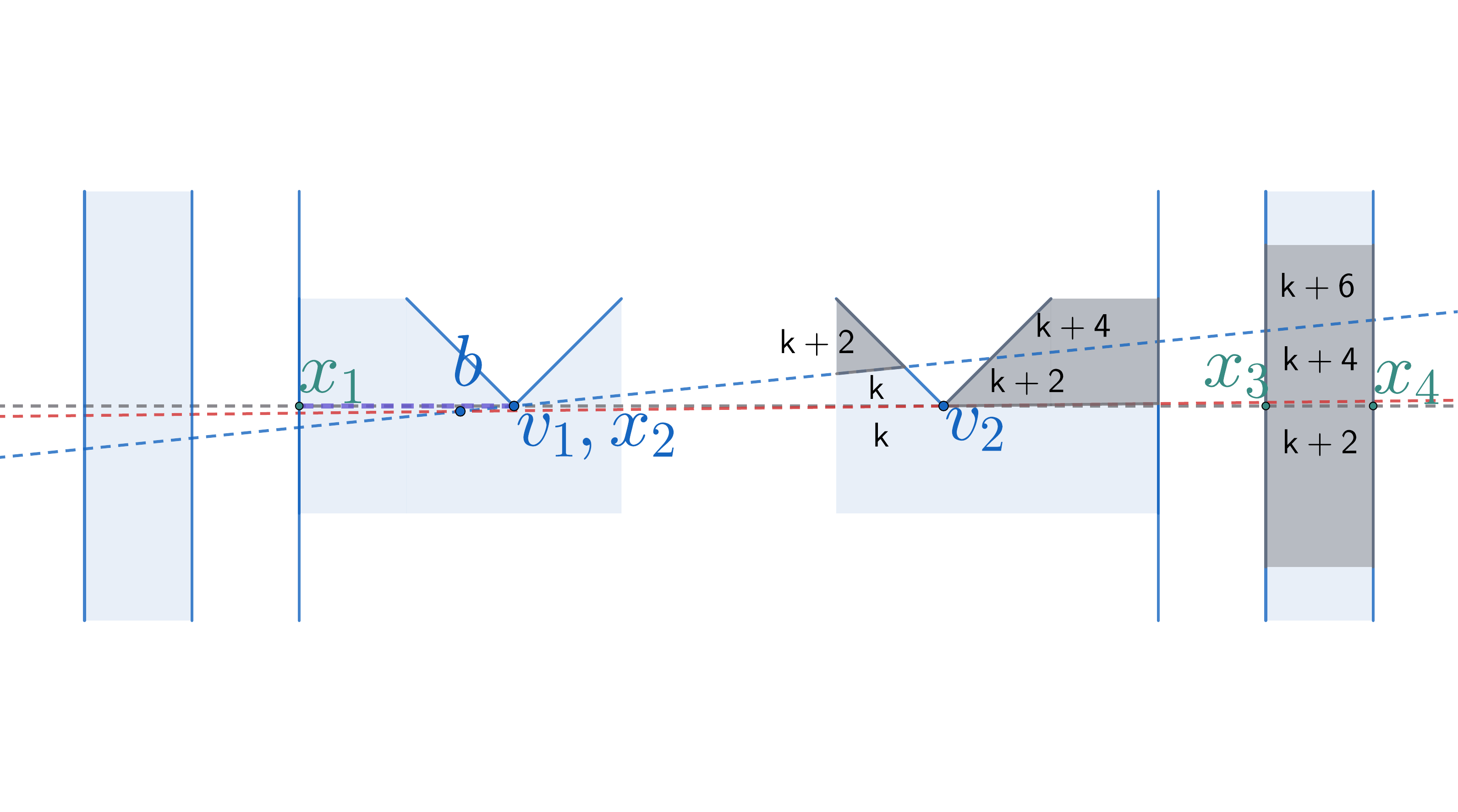}
\caption{Below $\ell_{g}$}\label{fig:RRS-genericB0}
\end{subfigure}

\caption{RRS; $Z = k$ (Merge/Split), $W = k + 2$.}
\label{fig:RRS-generic-0}
\end{figure}

 \begin{figure}[H]
\centering
\begin{subfigure}[b]{.49\linewidth}
\includegraphics[width=\linewidth]{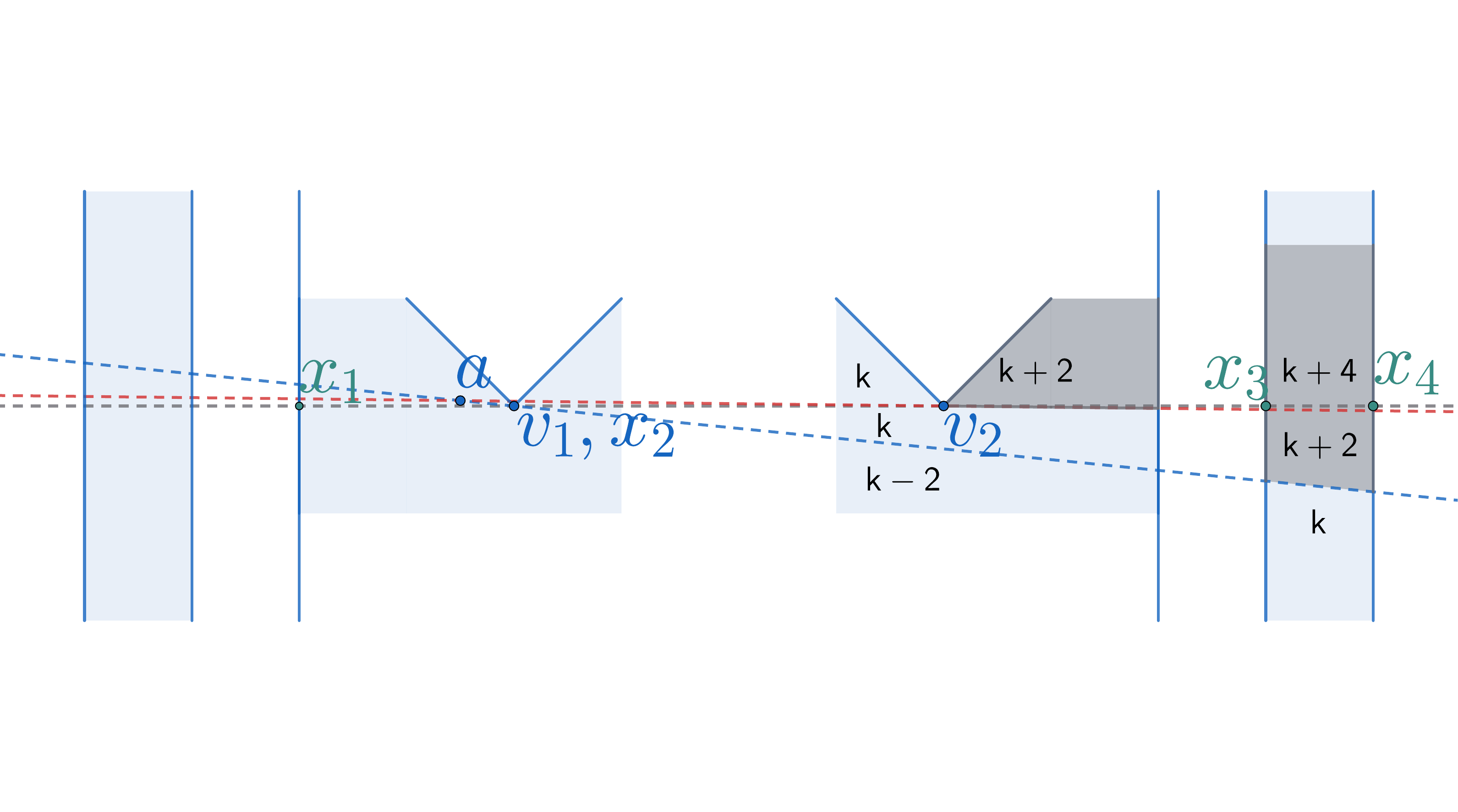}
\caption{Above $\ell_{g}$}\label{fig:RRS-genericA2}
\end{subfigure}
\begin{subfigure}[b]{.49\linewidth}
\includegraphics[width=\linewidth]{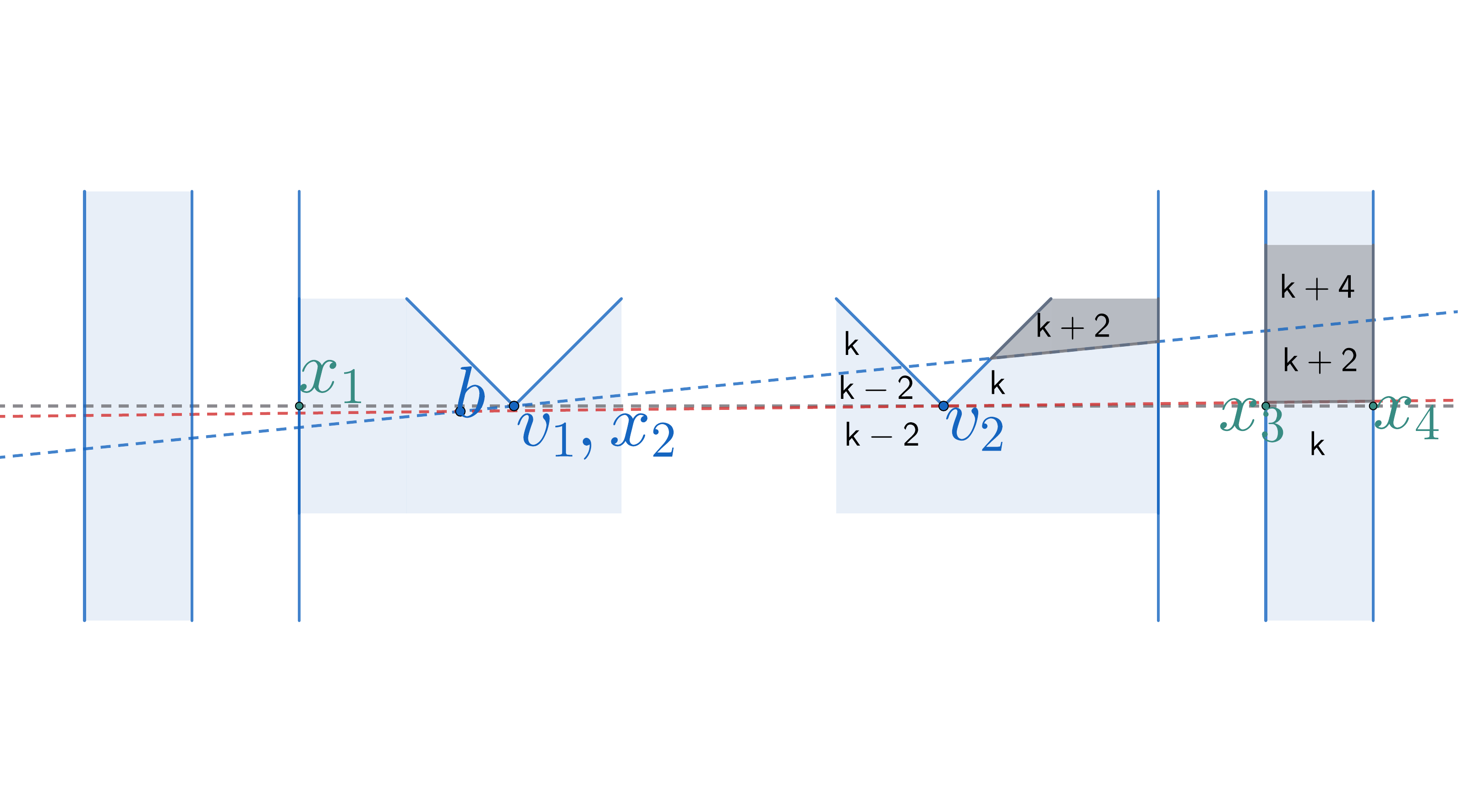}
\caption{Below $\ell_{g}$}\label{fig:RRS-genericB2}
\end{subfigure}

\caption{RRS; $Z = k - 2$, $W = k$.}
\label{fig:RRS-generic-2}
\end{figure}

 \begin{figure}[H]
\centering
\begin{subfigure}[b]{.49\linewidth}
\includegraphics[width=\linewidth]{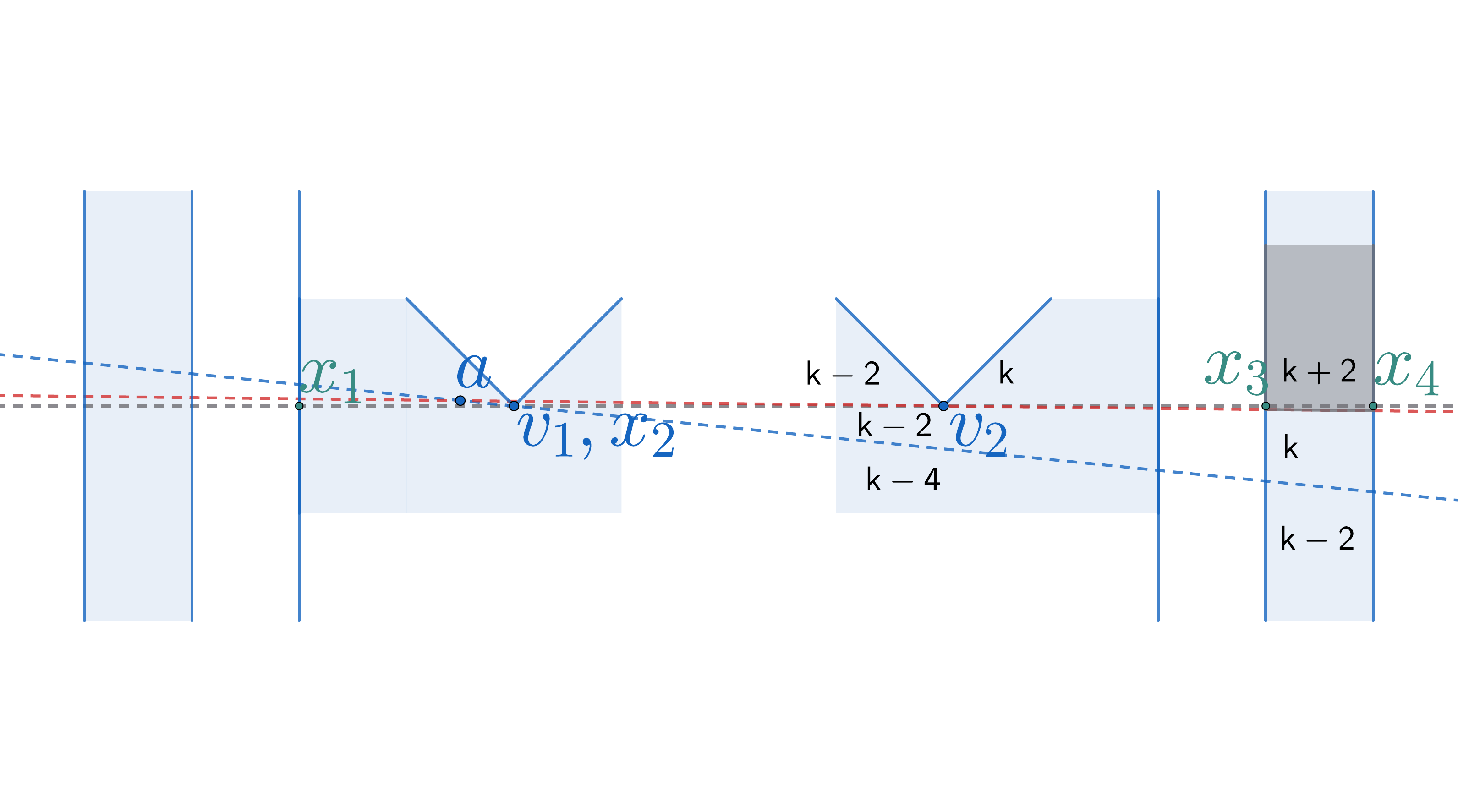}
\caption{Above $\ell_{g}$}\label{fig:RRS-genericA4}
\end{subfigure}
\begin{subfigure}[b]{.49\linewidth}
\includegraphics[width=\linewidth]{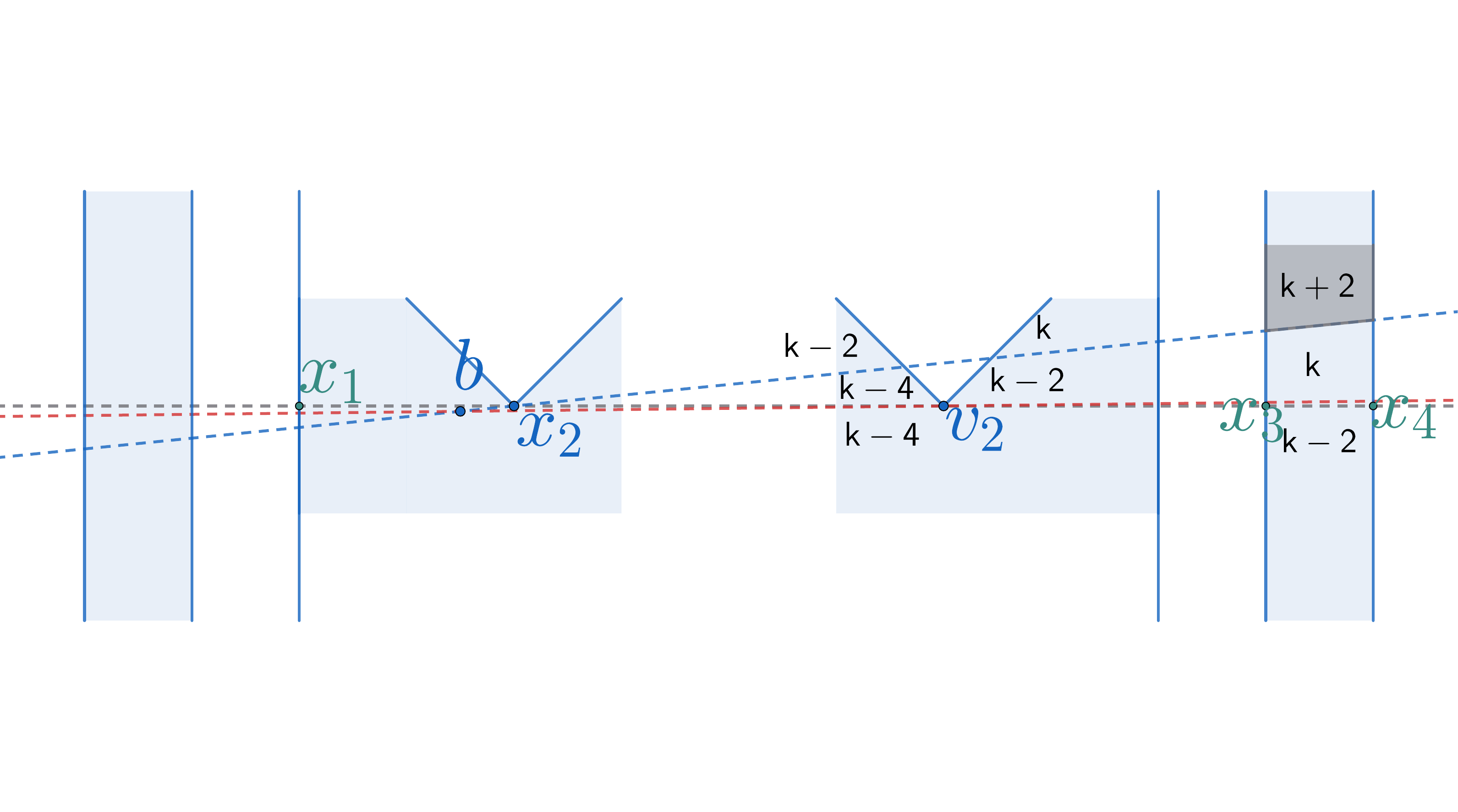}
\caption{Below $\ell_{g}$}\label{fig:RRS-genericB4}
\end{subfigure}

\caption{RRS; $Z = k - 4$, $W = k - 2$.}
\label{fig:RRS-generic-4}
\end{figure}

 \begin{figure}[H]
\centering
\begin{subfigure}[b]{.49\linewidth}
\includegraphics[width=\linewidth]{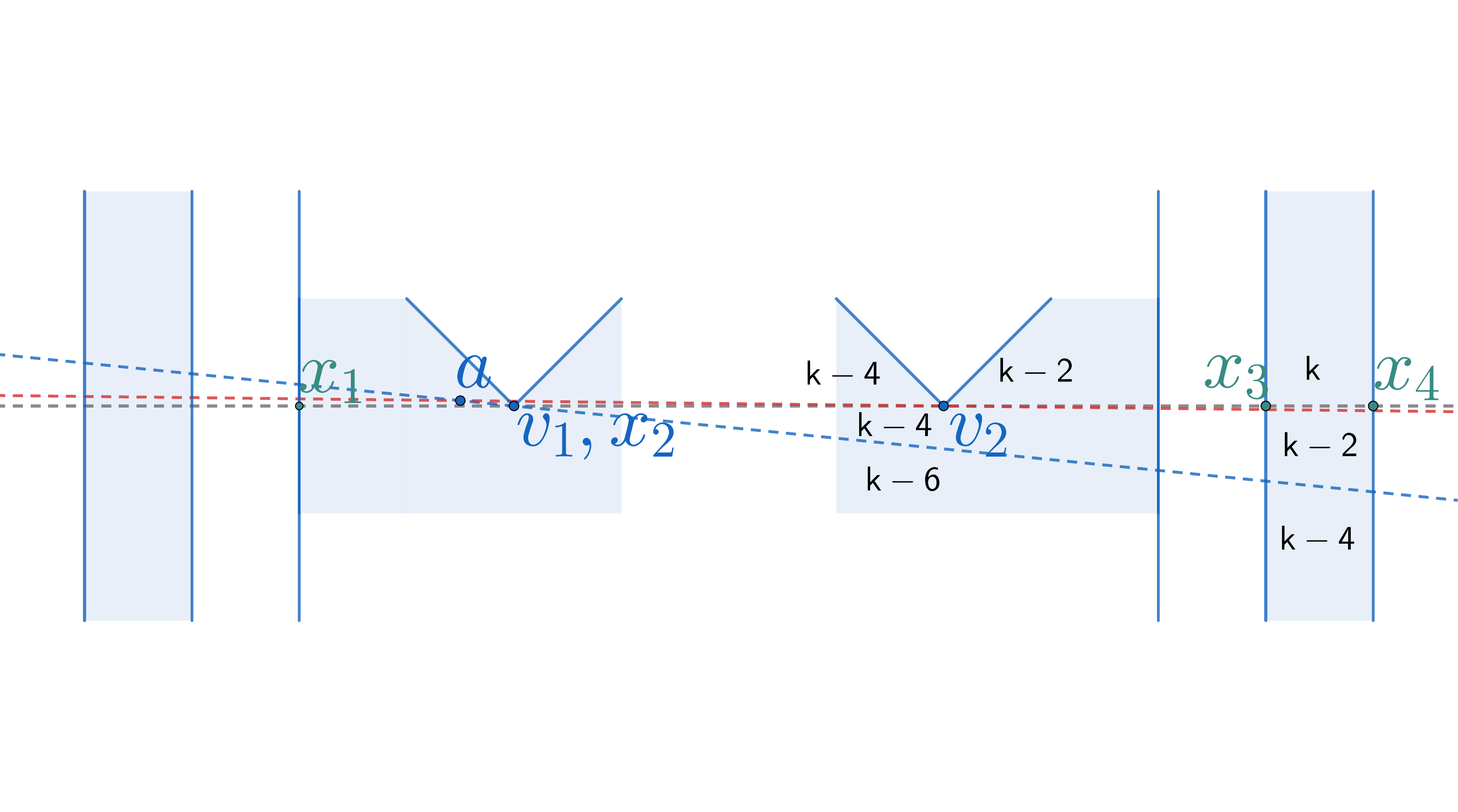}
\caption{Above $\ell_{g}$}\label{fig:RRS-genericA6}
\end{subfigure}
\begin{subfigure}[b]{.49\linewidth}
\includegraphics[width=\linewidth]{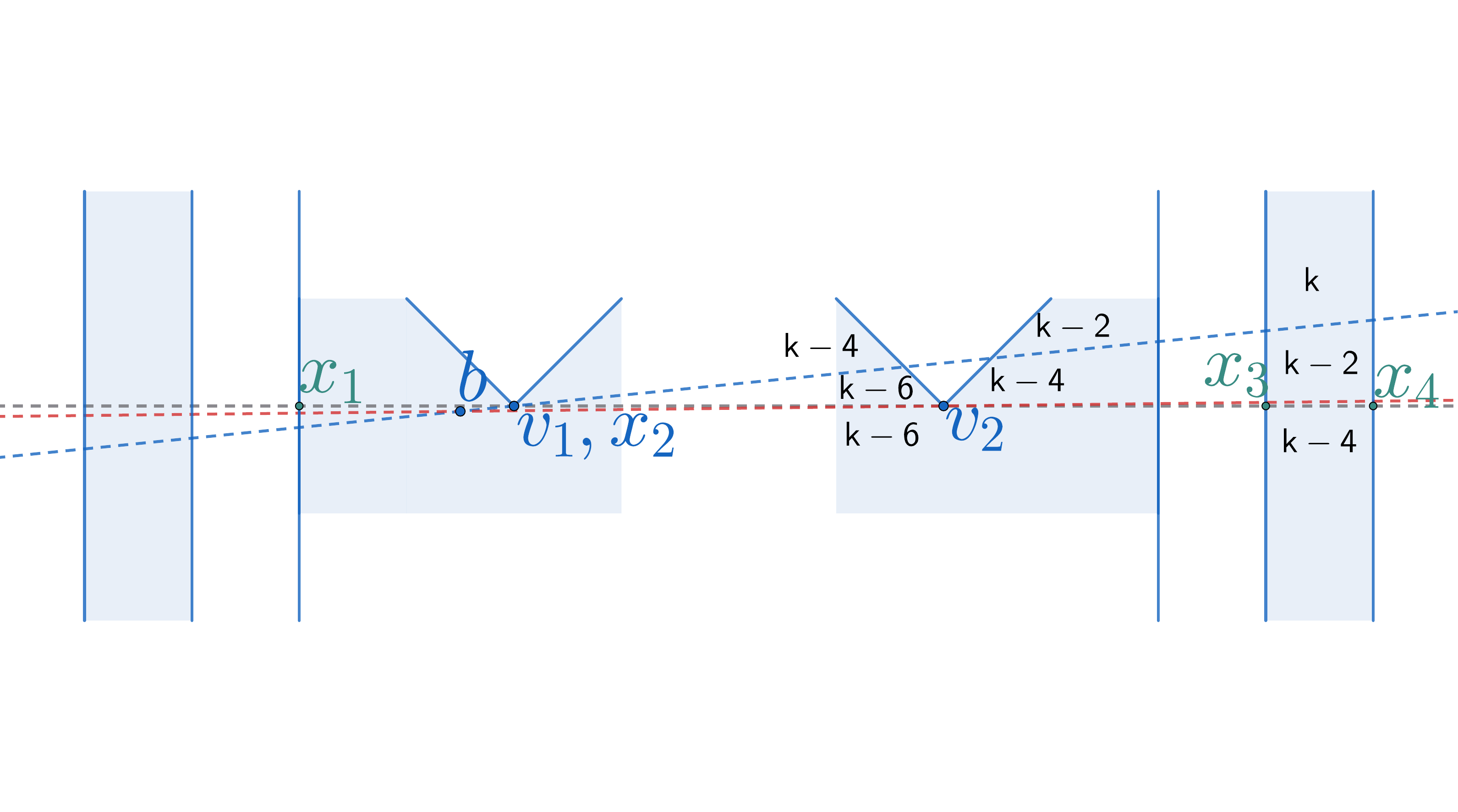}
\caption{Below $\ell_{g}$}\label{fig:RRS-genericB6}
\end{subfigure}

\caption{RRS; $Z = k - 6$, $W = k - 4$.}
\label{fig:RRS-generic-6}
\end{figure}

 \begin{figure}[H]
\centering
\begin{subfigure}[b]{.49\linewidth}
\includegraphics[width=\linewidth]{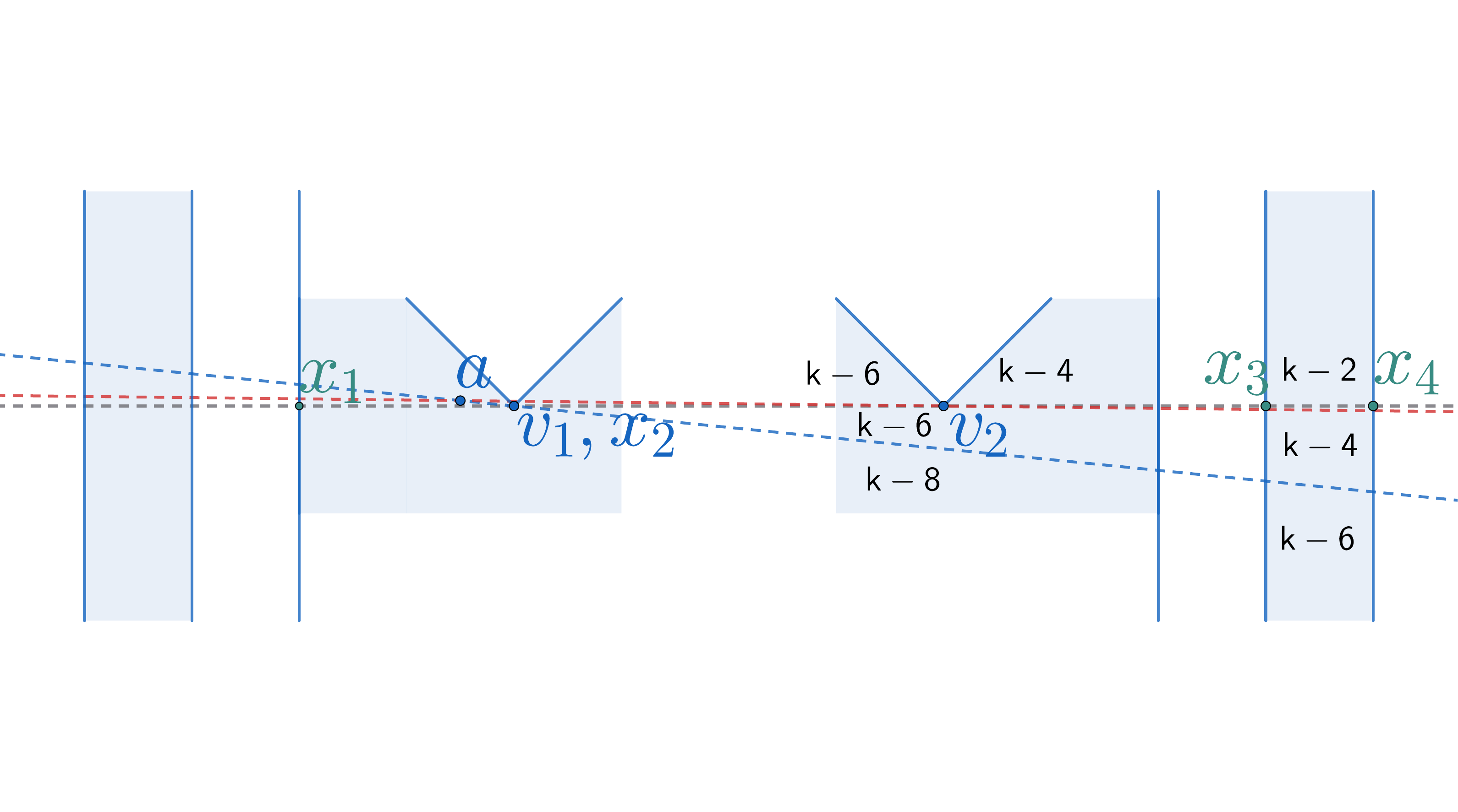}
\caption{Above $\ell_{g}$}\label{fig:RRS-genericA8}
\end{subfigure}
\begin{subfigure}[b]{.49\linewidth}
\includegraphics[width=\linewidth]{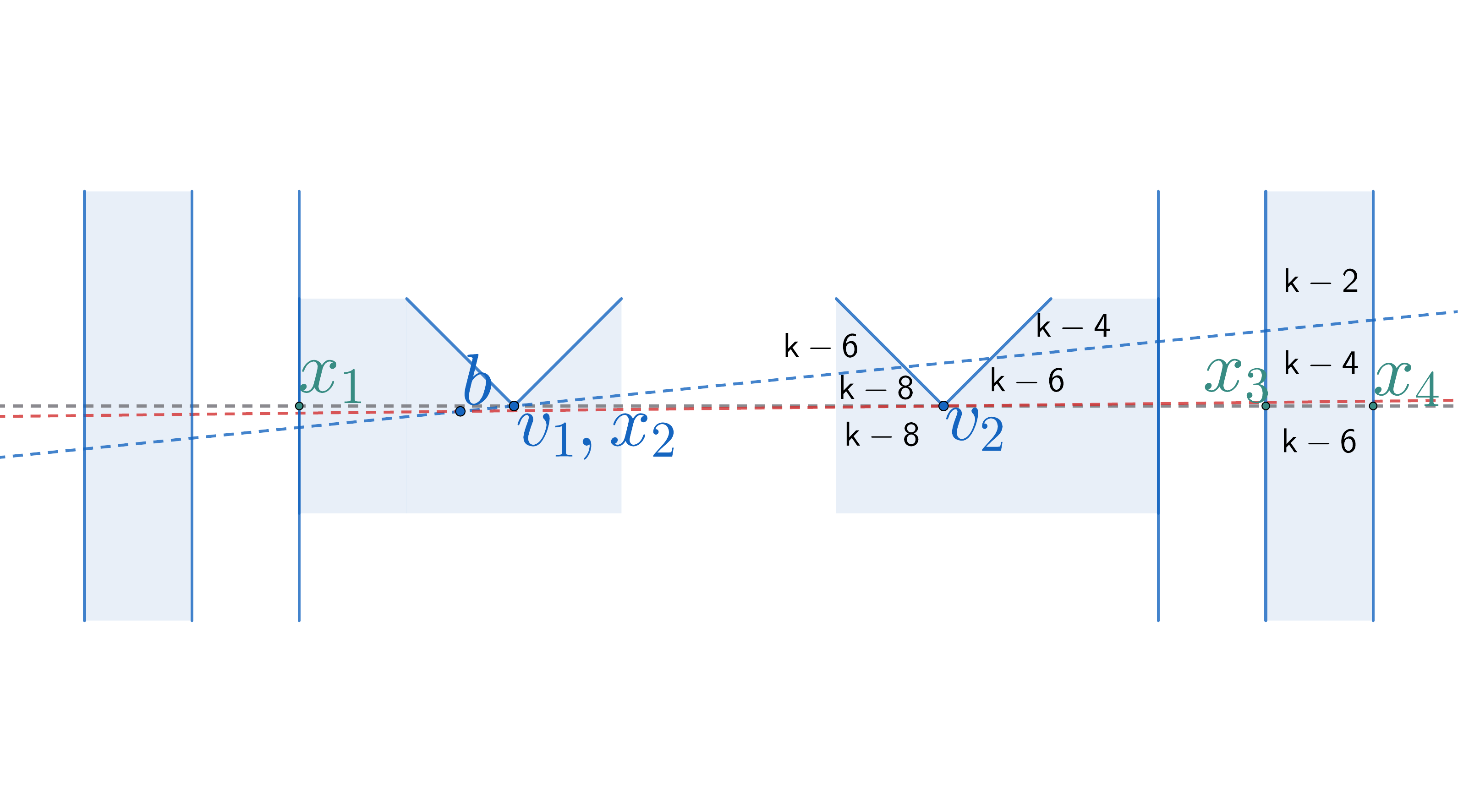}
\caption{Below $\ell_{g}$}\label{fig:RRS-genericB8}
\end{subfigure}

\caption{RRS; $Z = k - 8$, $W = k - 6$.}
\label{fig:RRS-generic-8}
\end{figure}

\section{RRO}
\label{appendix:RRO}
\subsection{Reflex Reflex Opposite (RRO)}
\begin{lemma} 
\label{lemma:RRO}
When $Z = k$ (Figure~\ref{fig:RRO-generic-0}), a merge/split event occurs at $v_{2}$. If $Z = k - 2$, an appear/disappear event occurs at $v_{2}$ (Figure~\ref{fig:RRO-generic-2}). Also, if $W = k$, a merge/split event occurs at $x_{3}x_{4}$ (Figure~\ref{fig:RRO-generic-2}). Additionally, if $W = k - 2$, an appear/disappear event occurs at $x_{3}x_{4}$ (Figure~\ref{fig:RRO-generic-4}). No event occurs for any other $Z$ or $W$. 

\end{lemma}
\begin{proof}
    See Figure~\ref{fig:RRO-generic-0} to Figure~\ref{fig:RRO-generic-8}.

    For $Z \geq k + 2$, $v_{2}$ and its surroundings are entirely in shadow. For $W \geq k + 4$, $x_{3}x_{4}$ and its surroundings are entirely in shadow. 

    For $Z \leq k - 10$, $v_{2}$ and its surroundings are entirely visible. For $W \leq k - 8$, $x_{3}x_{4}$ and its surroundings are entirely visible. 


\end{proof}

 \begin{figure}[H]
\centering
\begin{subfigure}[b]{.49\linewidth}
\includegraphics[width=\linewidth]{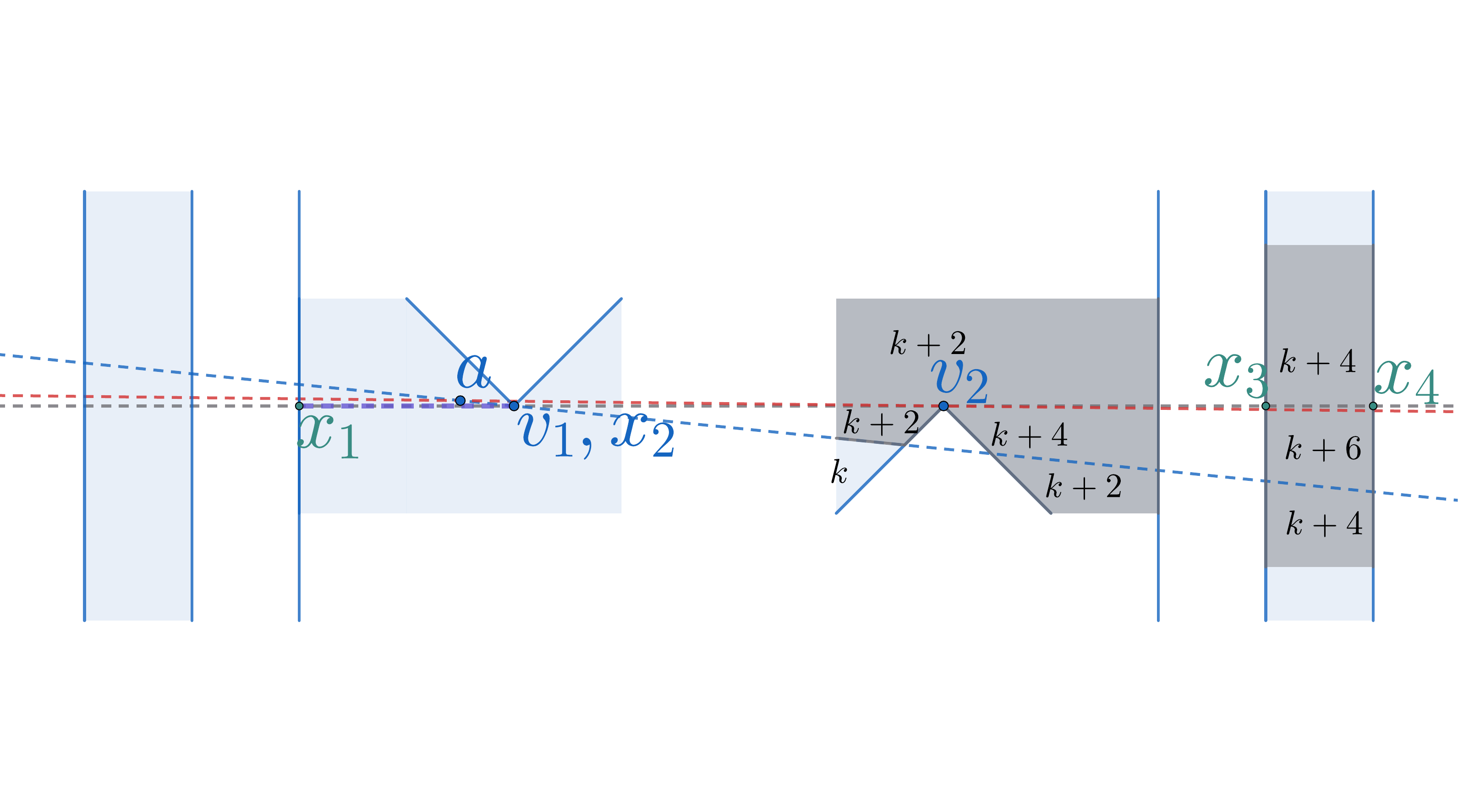}
\caption{Above $\ell_{g}$}\label{fig:RRO-genericA0}
\end{subfigure}
\begin{subfigure}[b]{.49\linewidth}
\includegraphics[width=\linewidth]{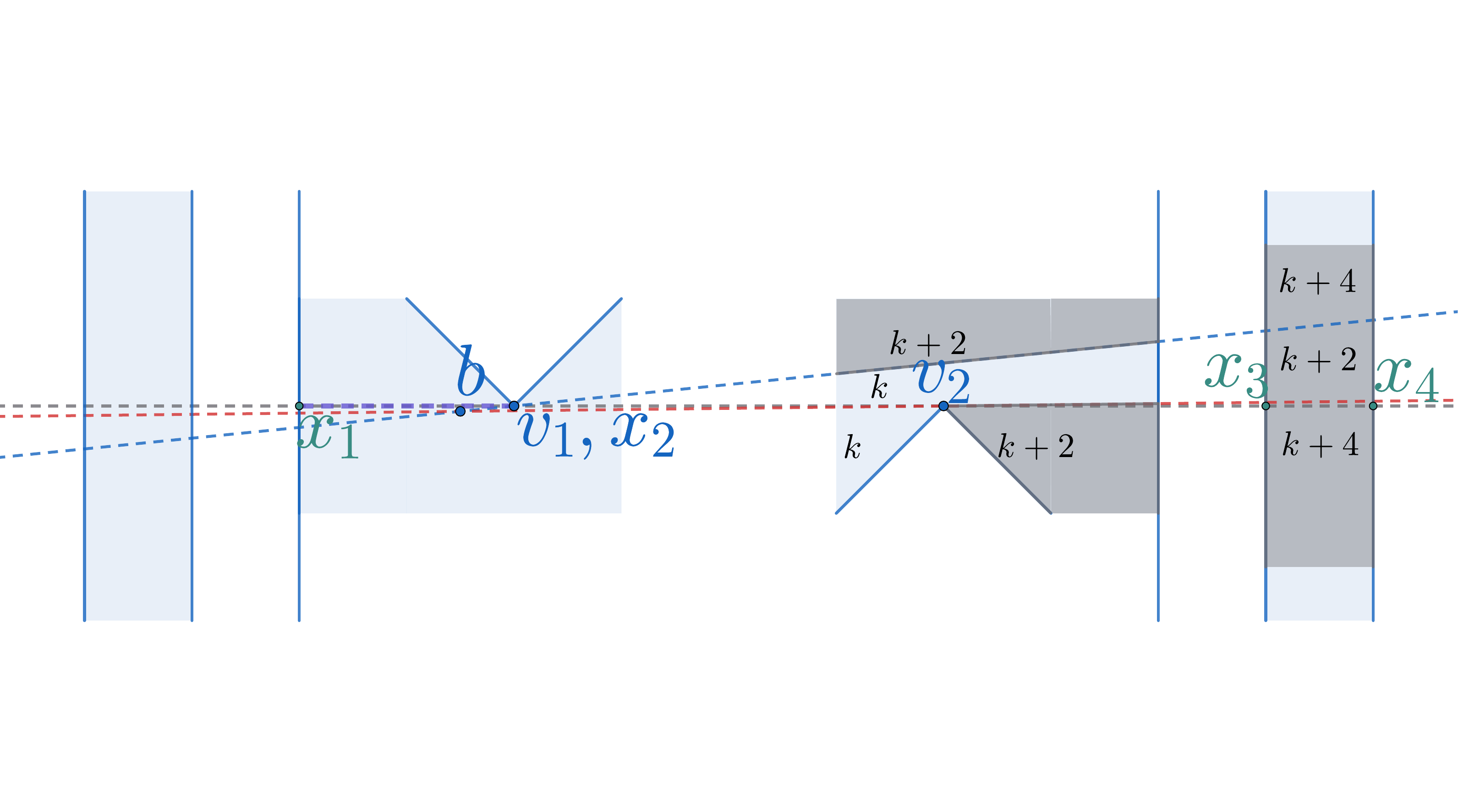}
\caption{Below $\ell_{g}$}\label{fig:RRO-genericB0}
\end{subfigure}

\caption{RRO; $Z = k$ (Merge/Split), $W = k + 2$.}
\label{fig:RRO-generic-0}
\end{figure}

 \begin{figure}[H]
\centering
\begin{subfigure}[b]{.49\linewidth}
\includegraphics[width=\linewidth]{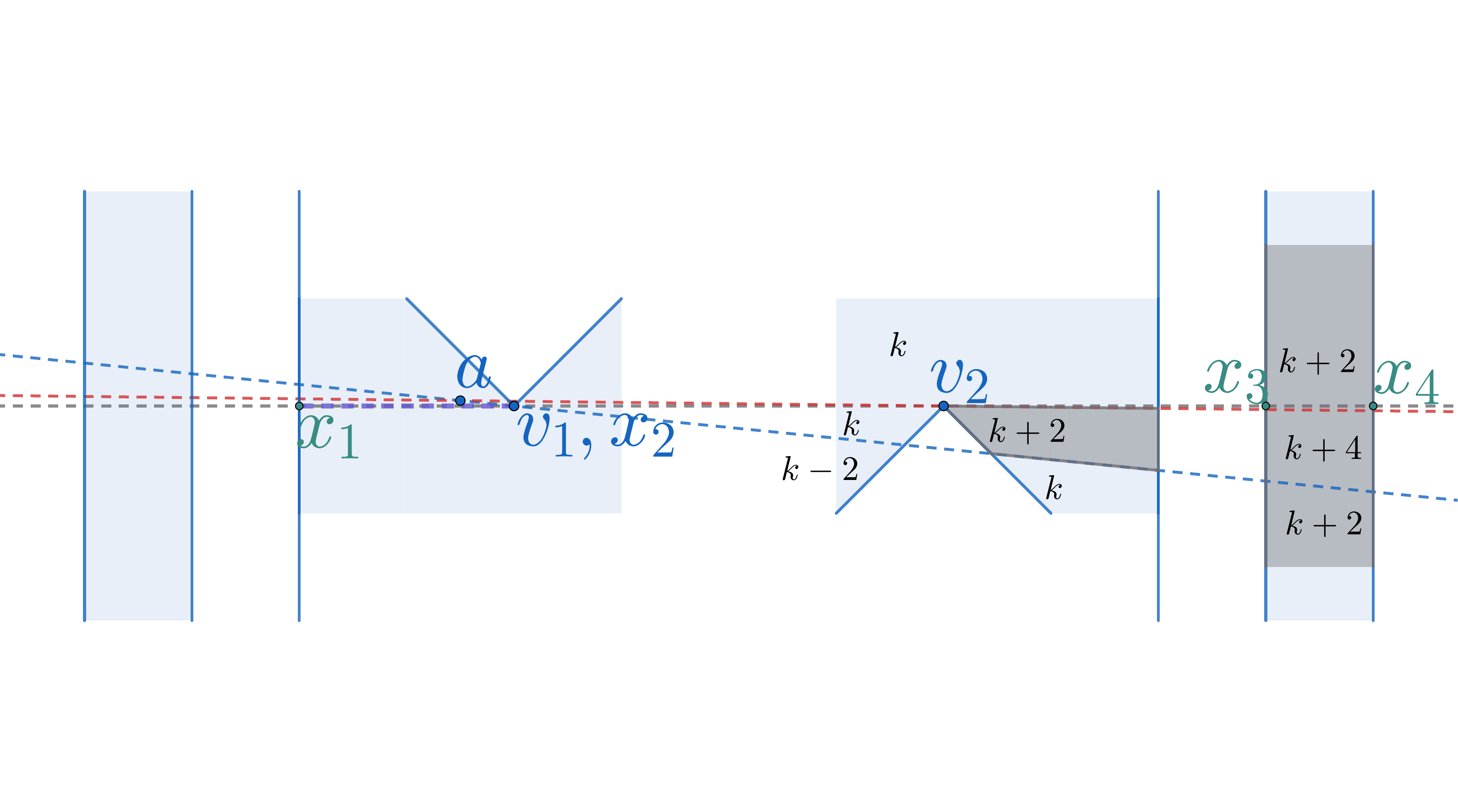}
\caption{Above $\ell_{g}$}\label{fig:RRO-genericA2}
\end{subfigure}
\begin{subfigure}[b]{.49\linewidth}
\includegraphics[width=\linewidth]{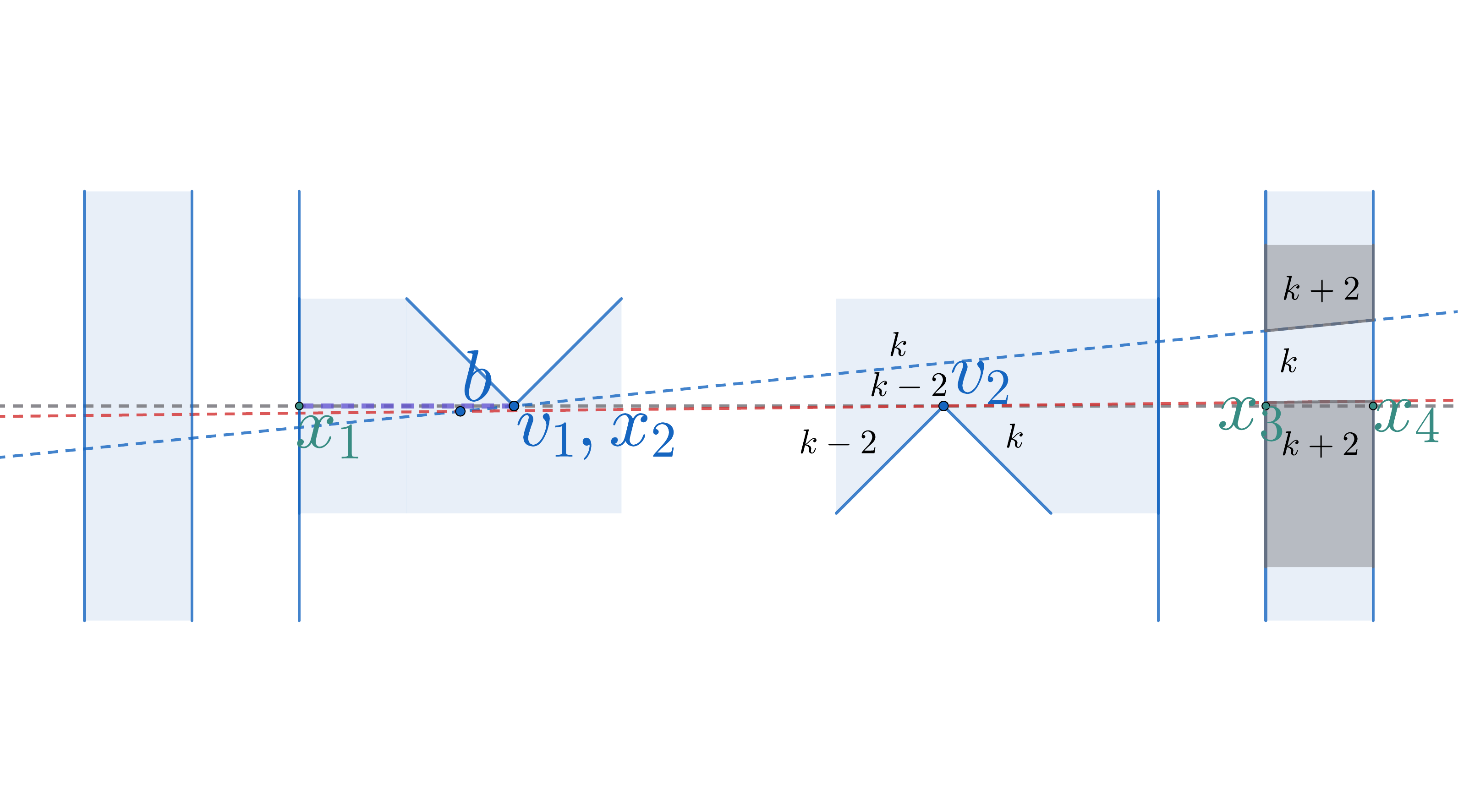}
\caption{Below $\ell_{g}$}\label{fig:RRO-genericB2}
\end{subfigure}

\caption{RRO; $Z = k - 2$ (Appear/Disappear), $W = k$ (Merge/Split).}
\label{fig:RRO-generic-2}
\end{figure}

 \begin{figure}[H]
\centering
\begin{subfigure}[b]{.49\linewidth}
\includegraphics[width=\linewidth]{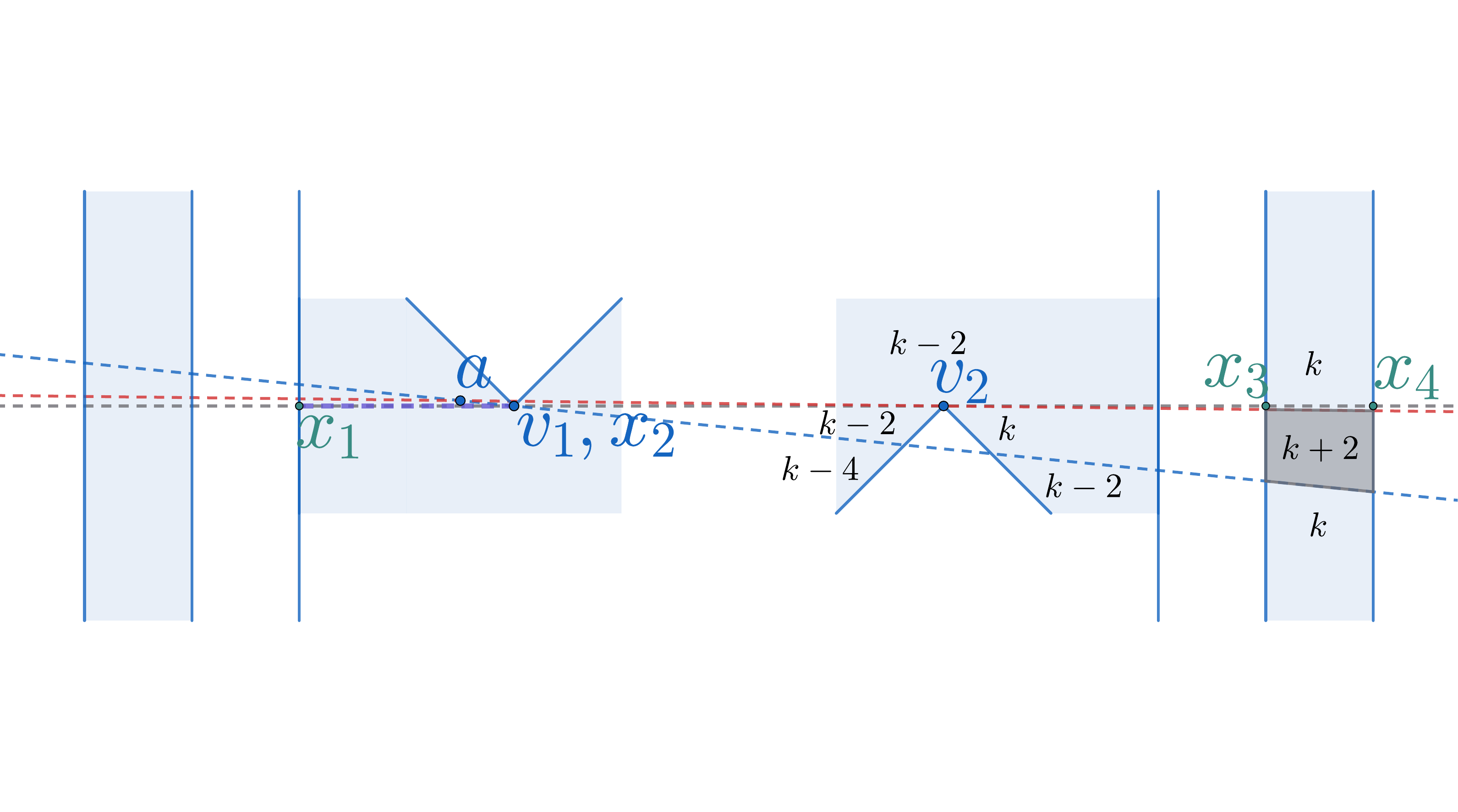}
\caption{Above $\ell_{g}$}\label{fig:RRO-genericA4}
\end{subfigure}
\begin{subfigure}[b]{.49\linewidth}
\includegraphics[width=\linewidth]{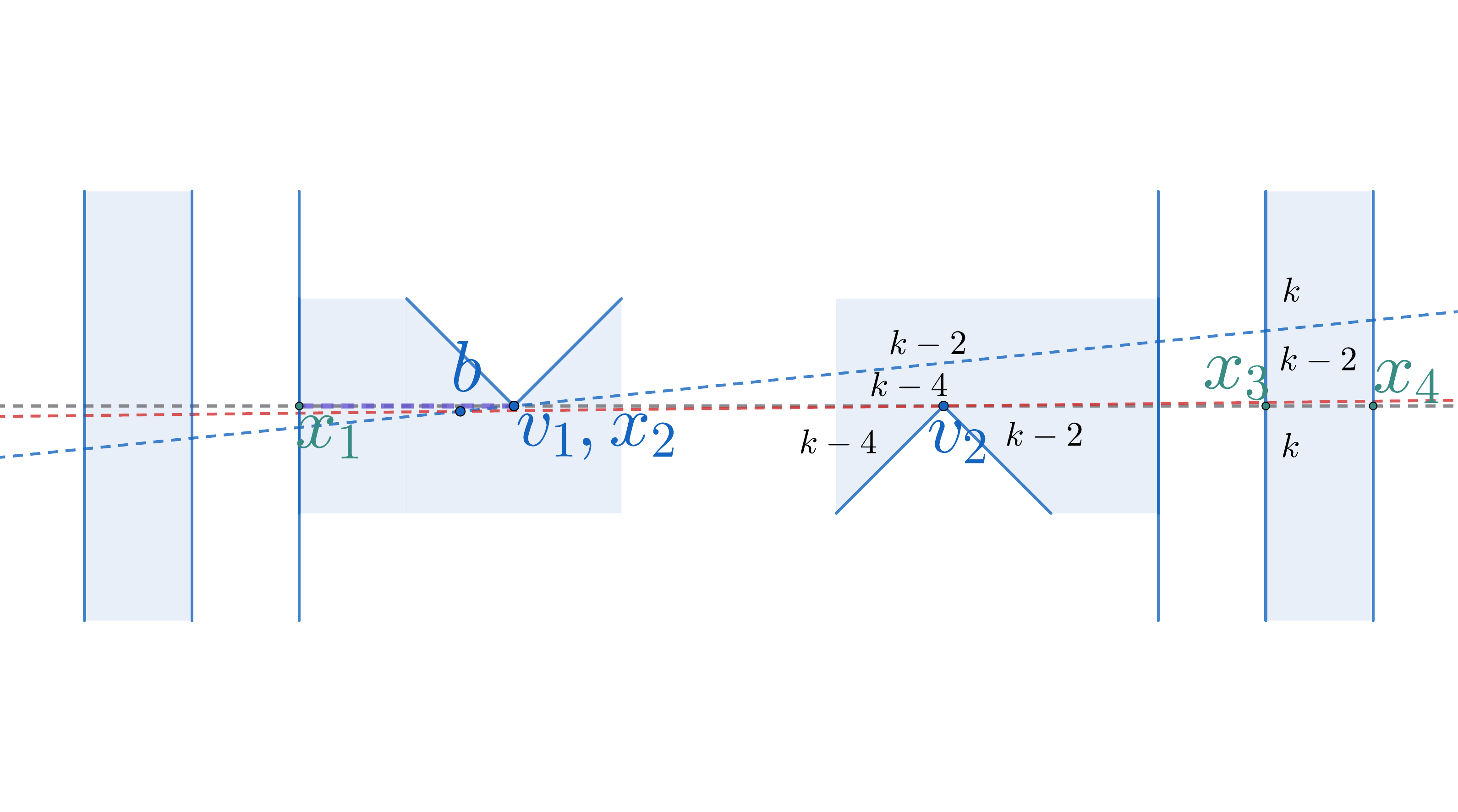}
\caption{Below $\ell_{g}$}\label{fig:RRO-genericB4}
\end{subfigure}

\caption{RRO; $Z = k - 4$, $W = k - 2$ (Appear/Disappear).}
\label{fig:RRO-generic-4}
\end{figure}

 \begin{figure}[H]
\centering
\begin{subfigure}[b]{.49\linewidth}
\includegraphics[width=\linewidth]{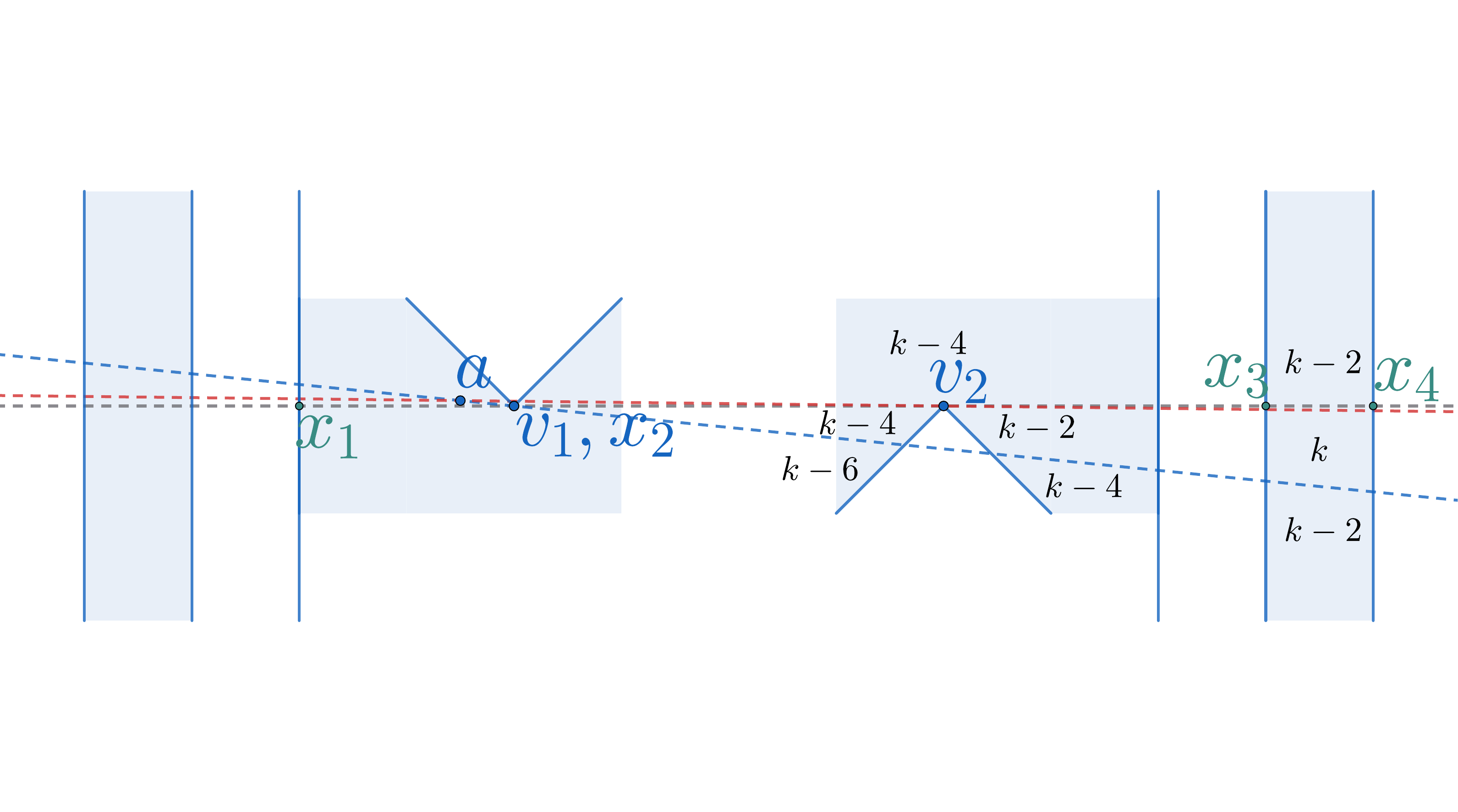}
\caption{Above $\ell_{g}$}\label{fig:RRO-genericA6}
\end{subfigure}
\begin{subfigure}[b]{.49\linewidth}
\includegraphics[width=\linewidth]{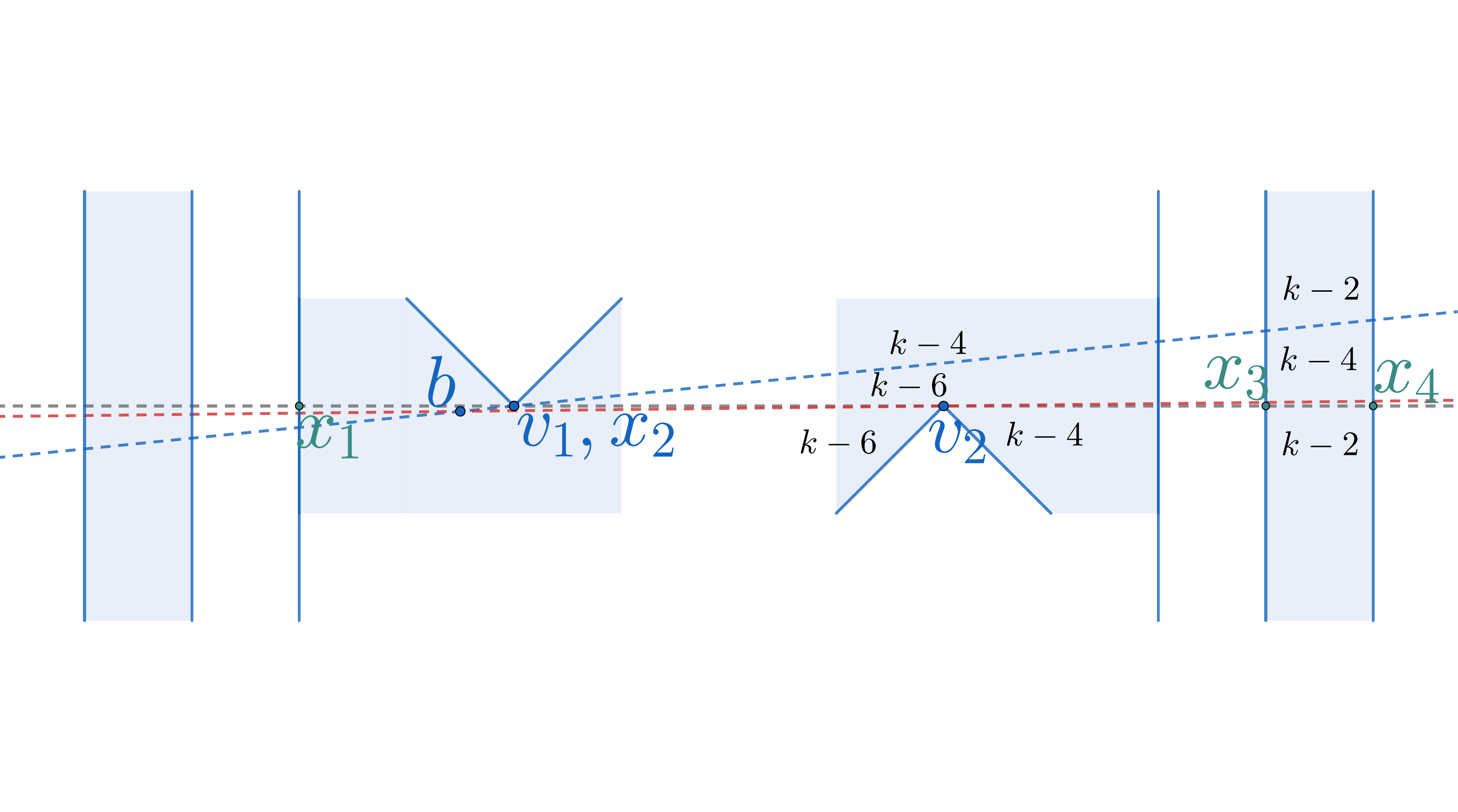}
\caption{Below $\ell_{g}$}\label{fig:RRO-genericB6}
\end{subfigure}

\caption{RRO; $Z = k - 6$, $W = k - 4$.}
\label{fig:RRO-generic-6}
\end{figure}

 \begin{figure}[H]
\centering
\begin{subfigure}[b]{.49\linewidth}
\includegraphics[width=\linewidth]{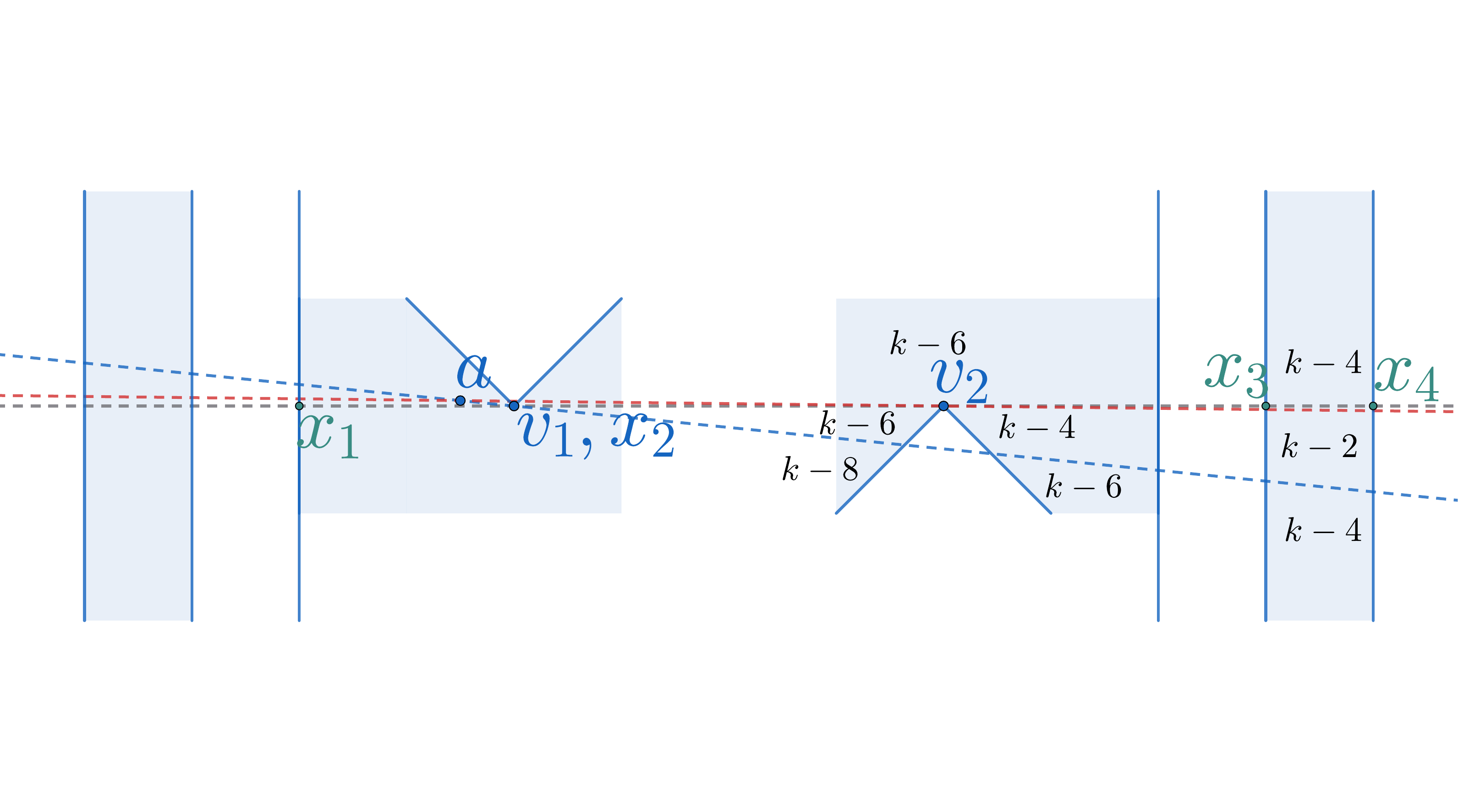}
\caption{Above $\ell_{g}$}\label{fig:RRO-genericA8}
\end{subfigure}
\begin{subfigure}[b]{.49\linewidth}
\includegraphics[width=\linewidth]{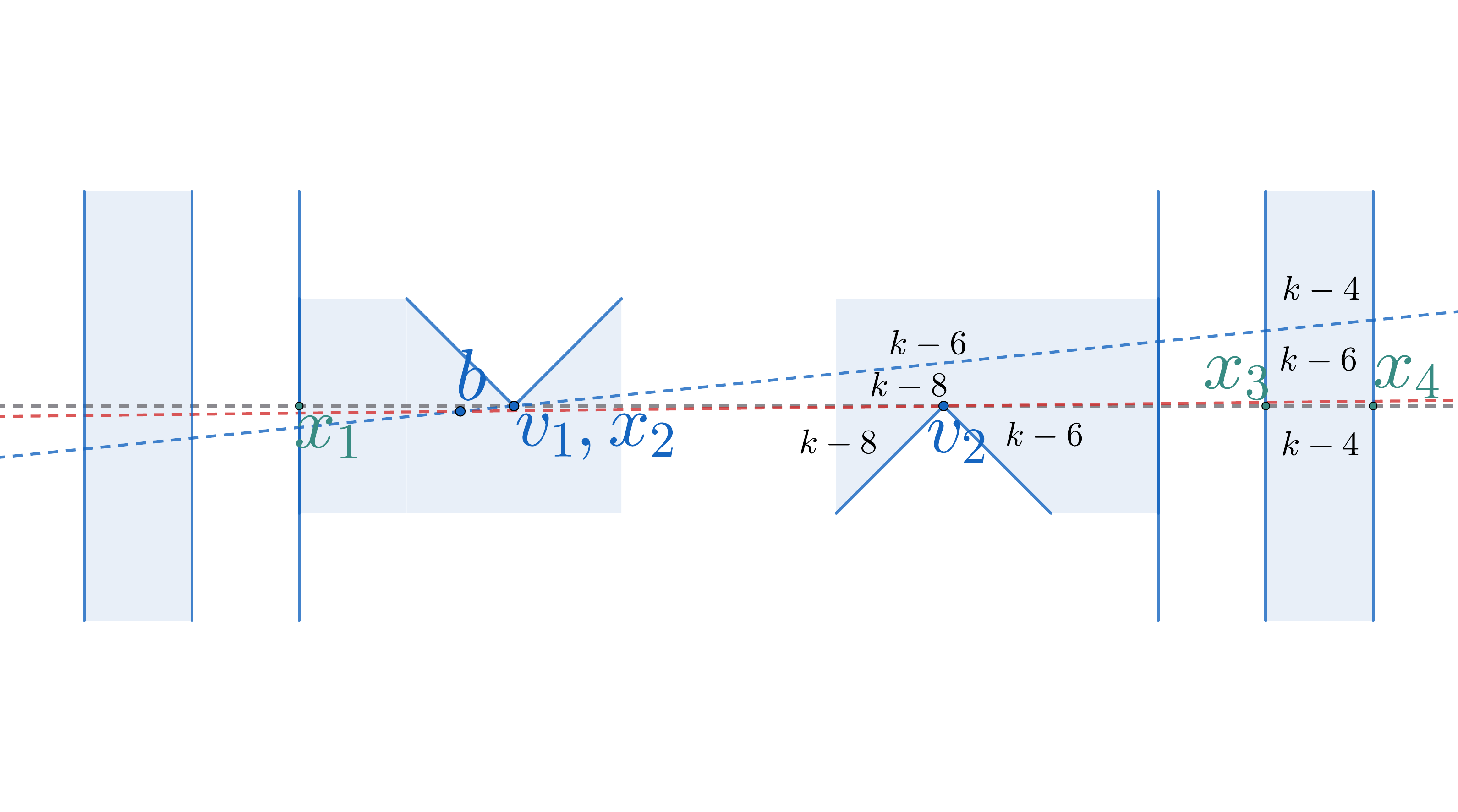}
\caption{Below $\ell_{g}$}\label{fig:RRO-genericB8}
\end{subfigure}

\caption{RRO; $Z = k - 8$, $W = k - 6$.}
\label{fig:RRO-generic-8}
\end{figure}

\section{RC-SC}
\label{appendix:RC-SC}
\subsection{Reflex Convex Special Case}
\begin{lemma} 
\label{lemma:RC-SC}
 When $Z=k$, there is an appear/disappear event at $v_{2}$. When $W = k-1$, an appear/disappear event occurs at $x_{3}x_{4}$. No event occurs for any other $Z$ or $W$. 

\end{lemma}
\begin{proof}
    See Figure~\ref{fig:RC-SpecialCase-0} to Figure~\ref{fig:RC-SpecialCase-4}.

    For $Z \geq k + 2$, $v_{2}$ and its surroundings are entirely in shadow. For $W \geq k + 3$, $x_{3}x_{4}$ and its surroundings are entirely in shadow.

    For $Z \leq k - 6$, $v_{2}$ and its surroundings are entirely visible. For $W \leq k - 5$, $x_{3}x_{4}$ and its surroundings are entirely visible. 
\end{proof}

 \begin{figure}[H]
\centering
\begin{subfigure}[b]{.49\linewidth}
\includegraphics[width=\linewidth]{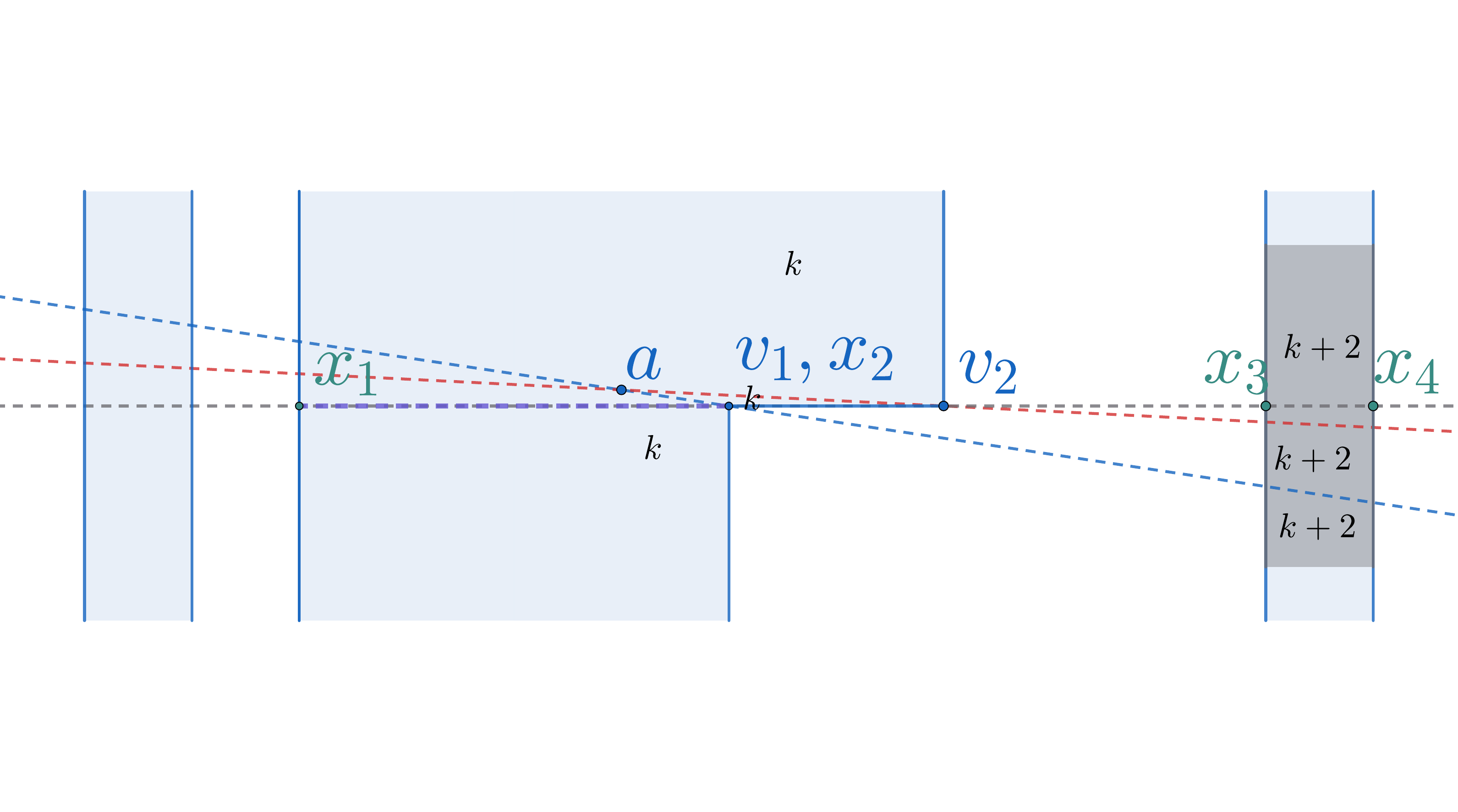}
\caption{Above $\ell_{g}$}\label{fig:RC-SpecialCase-A0}
\end{subfigure}
\begin{subfigure}[b]{.49\linewidth}
\includegraphics[width=\linewidth]{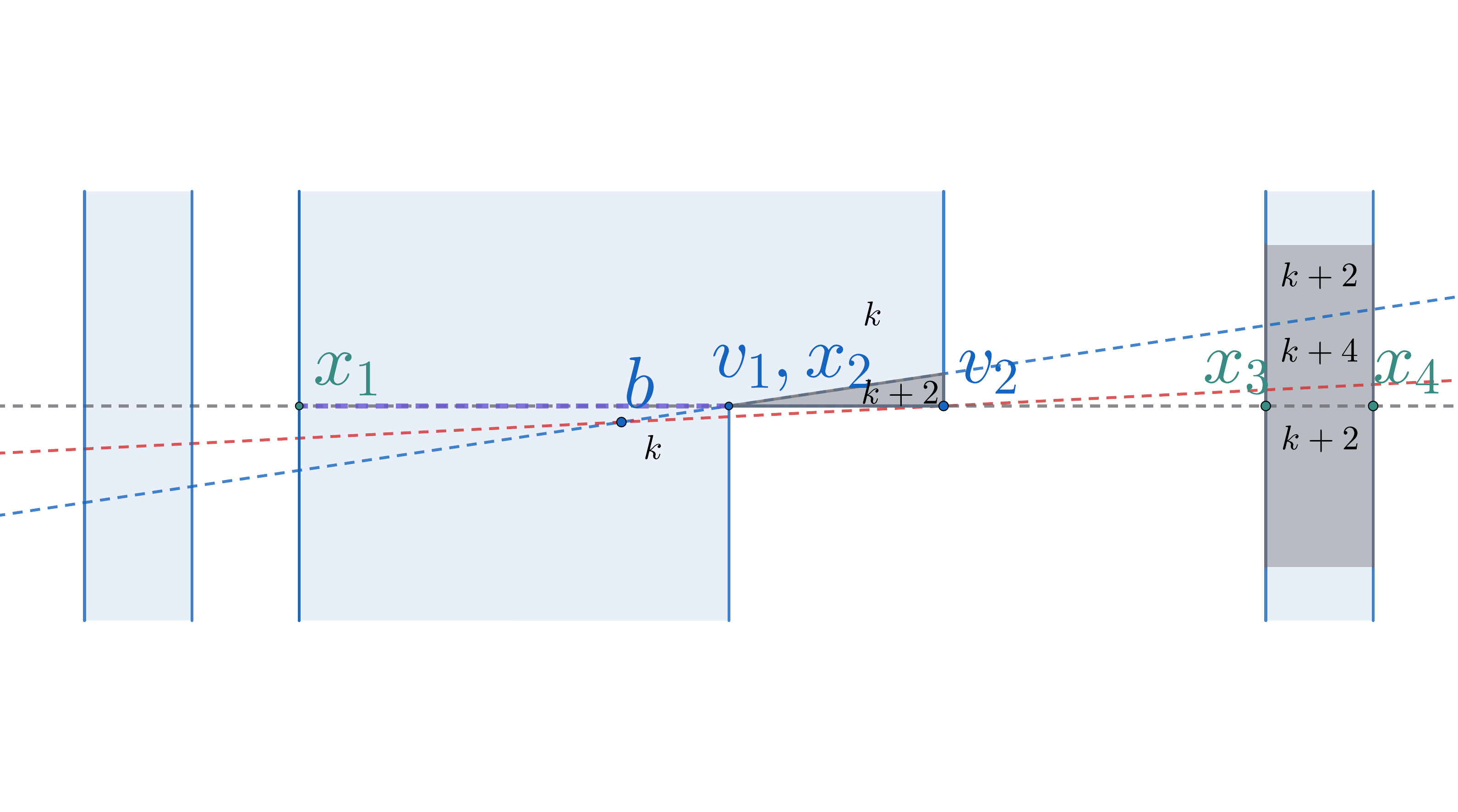}
\caption{Below $\ell_{g}$}\label{fig:RC-SpecialCase-B0}
\end{subfigure}

\caption{RC-SC; $Z = k$ (Disappear/Appear), $W = k + 1$.}
\label{fig:RC-SpecialCase-0}
\end{figure}

 \begin{figure}[H]
\centering
\begin{subfigure}[b]{.49\linewidth}
\includegraphics[width=\linewidth]{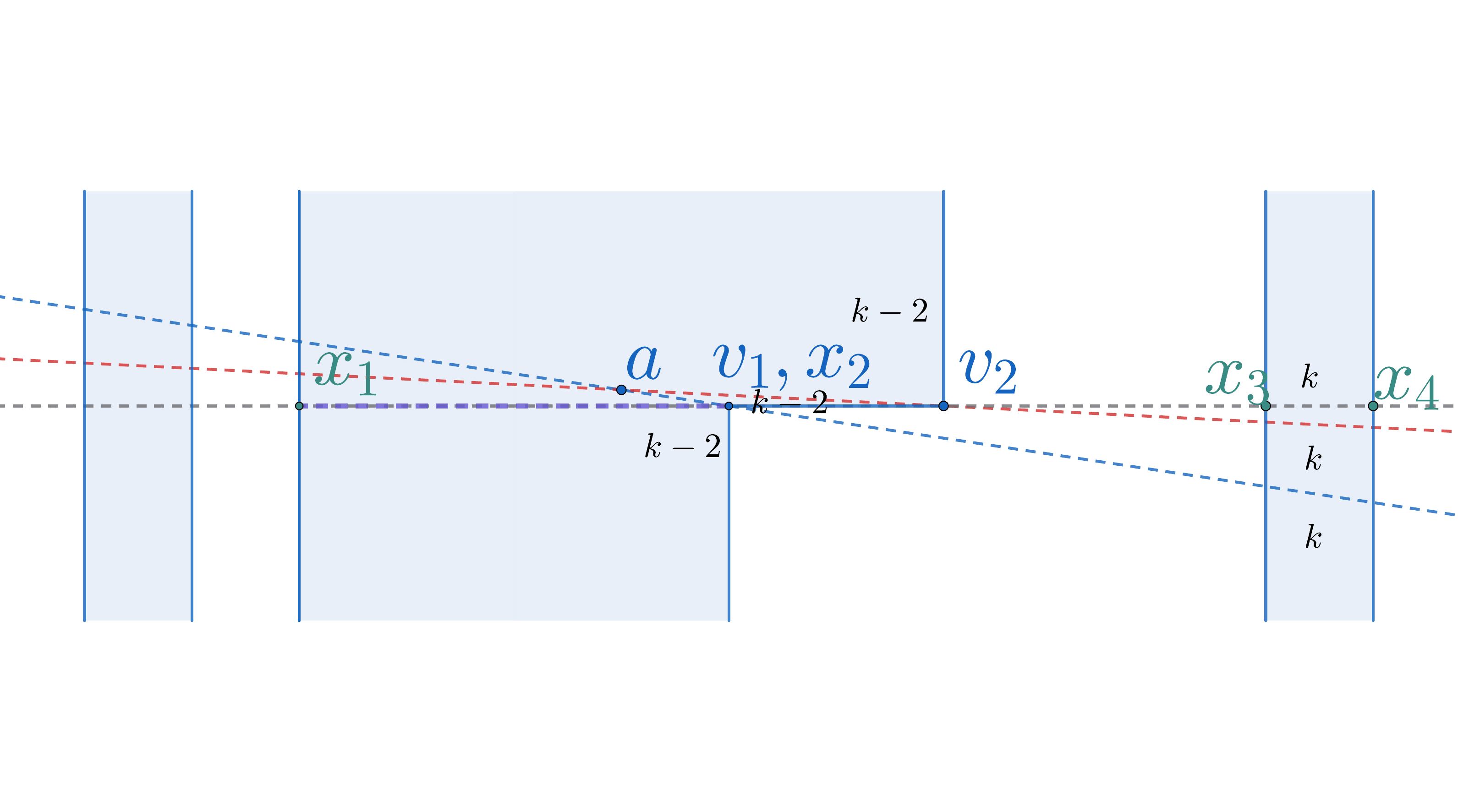}
\caption{Above $\ell_{g}$}\label{fig:RC-SpecialCase-A2}
\end{subfigure}
\begin{subfigure}[b]{.49\linewidth}
\includegraphics[width=\linewidth]{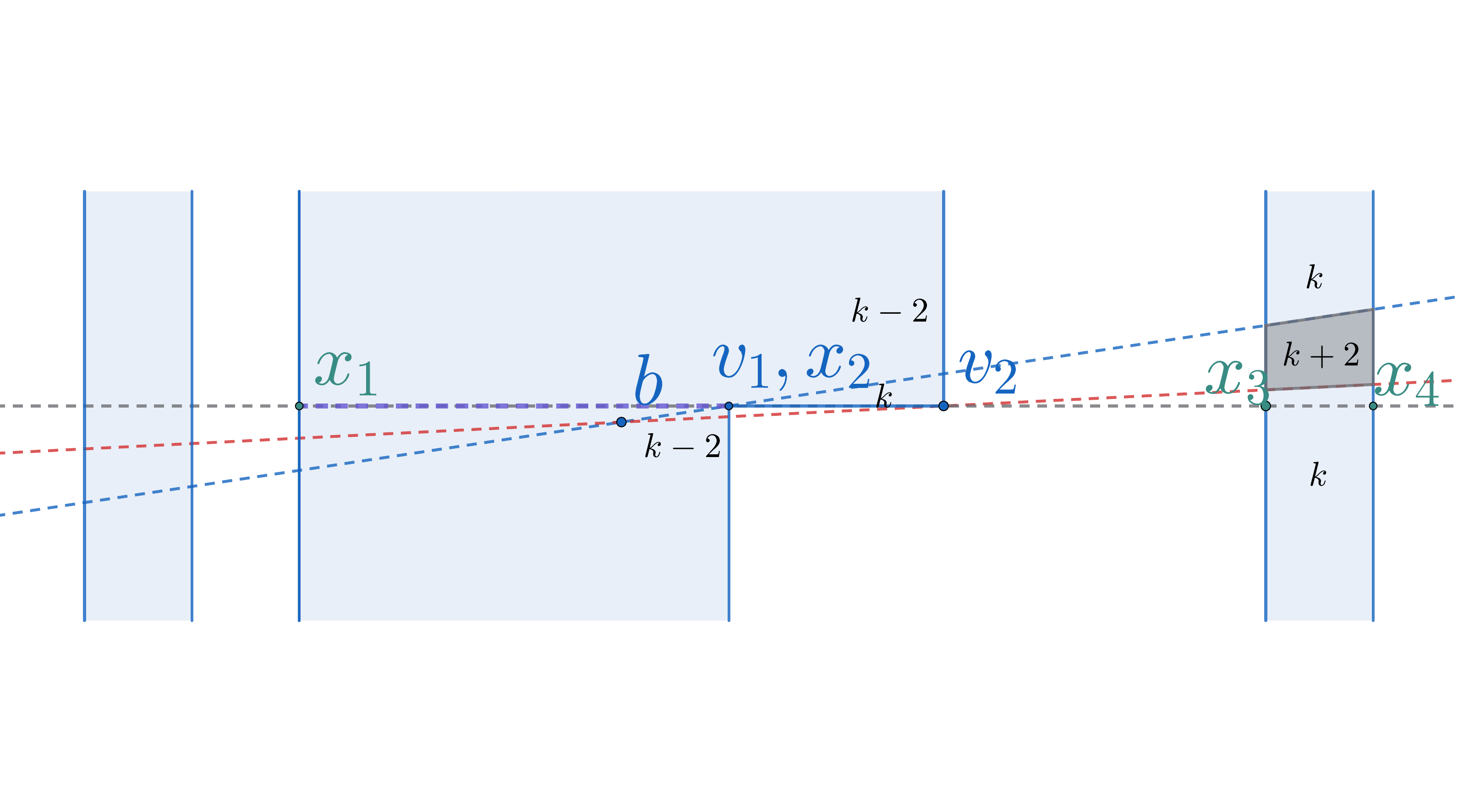}
\caption{Below $\ell_{g}$}\label{fig:RC-SpecialCase-B2}
\end{subfigure}

\caption{RC-SC; $Z = k - 2$, $W = k - 1$ (Disappear/Appear).}
\label{fig:RC-SpecialCase-2}
\end{figure}

  \begin{figure}[H]
\centering
\begin{subfigure}[b]{.49\linewidth}
\includegraphics[width=\linewidth]{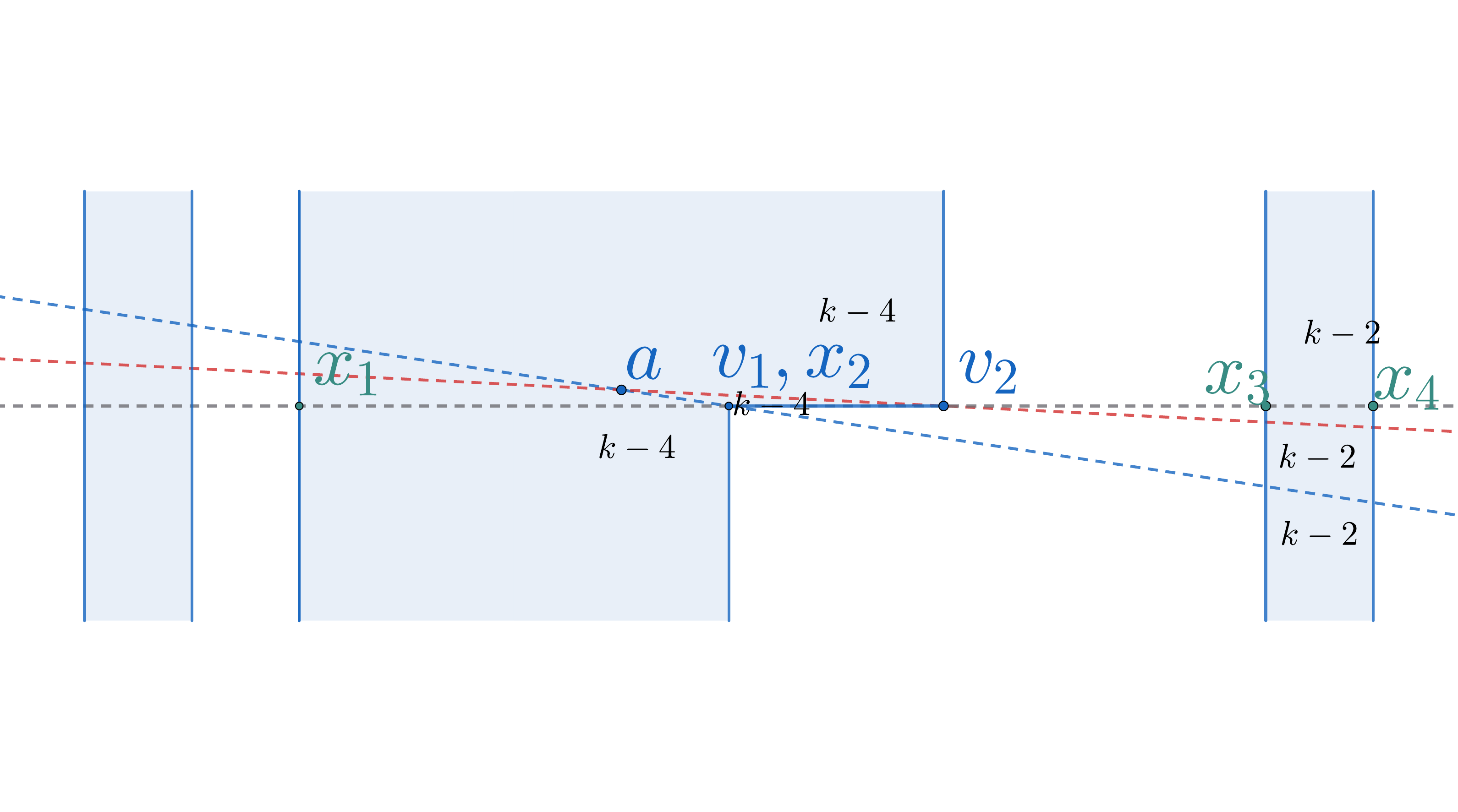}
\caption{Above $\ell_{g}$}\label{fig:RC-SpecialCase-A4}
\end{subfigure}
\begin{subfigure}[b]{.49\linewidth}
\includegraphics[width=\linewidth]{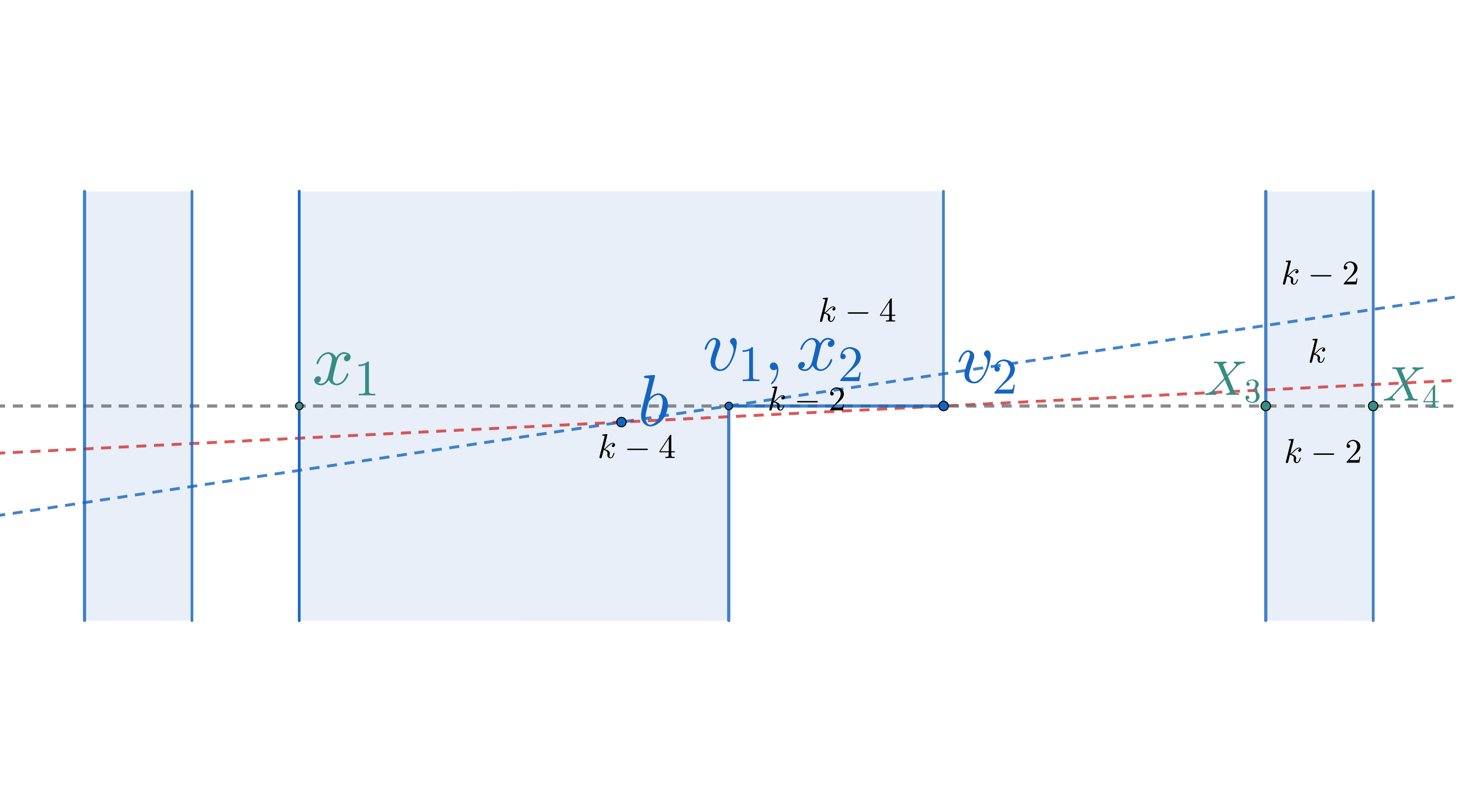}
\caption{Below $\ell_{g}$}\label{fig:RC-SpecialCase-B4}
\end{subfigure}

\caption{RC-SC; $Z = k - 4$, $W = k - 3$.}
\label{fig:RC-SpecialCase-4}
\end{figure}

\section{RR-SC}
\label{appendix:RR-SC}
\subsection{Reflex Reflex Special Case}
\begin{lemma} 
\label{lemma:RR-SC}
  No event occurs for any $Z$ or $W$. 

\end{lemma}
\begin{proof}
    See Figure~\ref{fig:RR-SpecialCase-0} to Figure~\ref{fig:RR-SpecialCase-4}.

    For $Z \geq k + 2$, $v_{2}$ and its surroundings are entirely in shadow. For $W \geq k + 4$, $x_{3}x_{4}$ and its surroundings are entirely in shadow.

    For $Z \leq k - 6$, $v_{2}$ and its surroundings are entirely visible. For $W \leq k - 4$, $x_{3}x_{4}$ and its surroundings are entirely visible.
\end{proof}

 \begin{figure}[H]
\centering
\begin{subfigure}[b]{.49\linewidth}
\includegraphics[width=\linewidth]{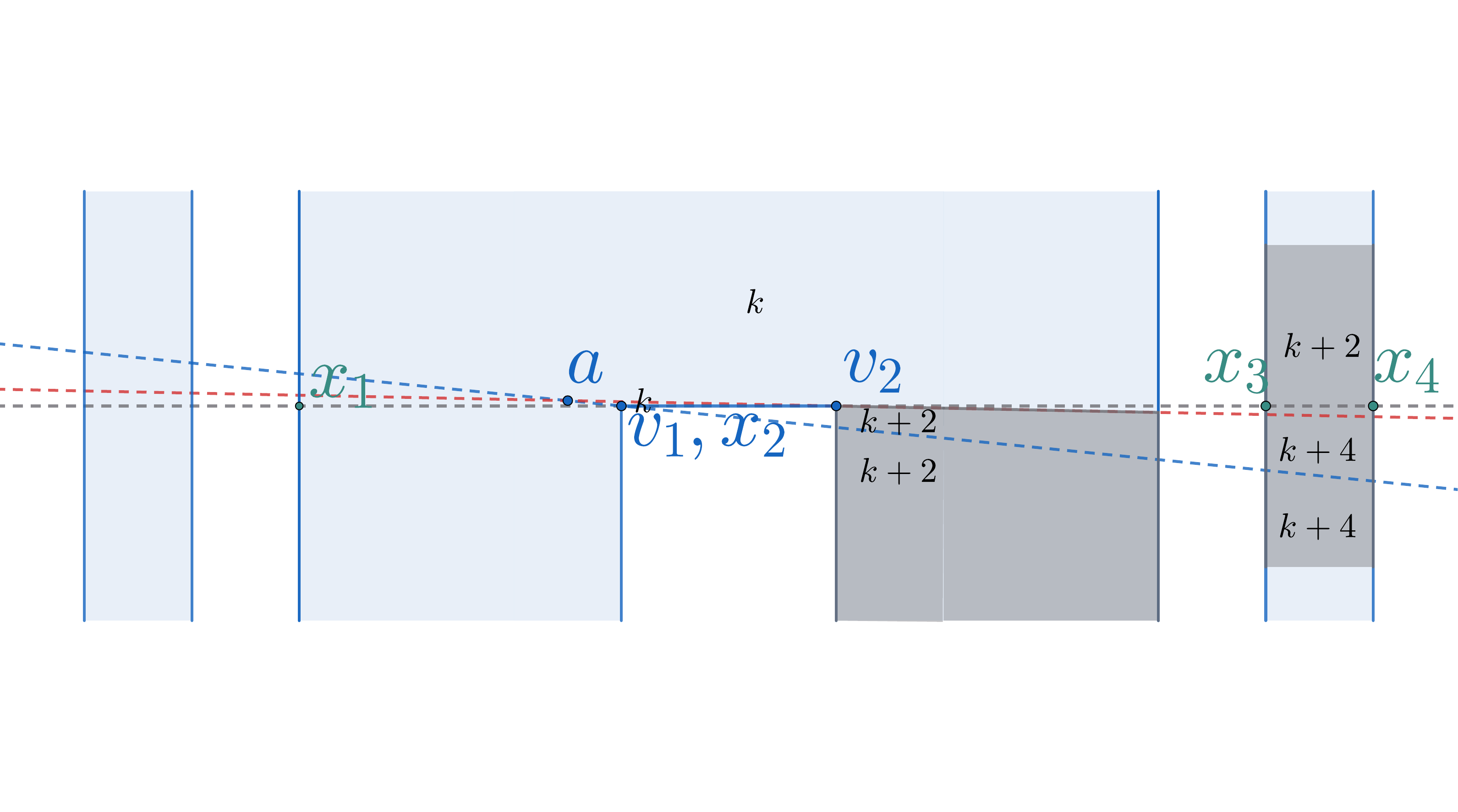}
\caption{Above $\ell_{g}$}\label{fig:RR-SpecialCase-A0}
\end{subfigure}
\begin{subfigure}[b]{.49\linewidth}
\includegraphics[width=\linewidth]{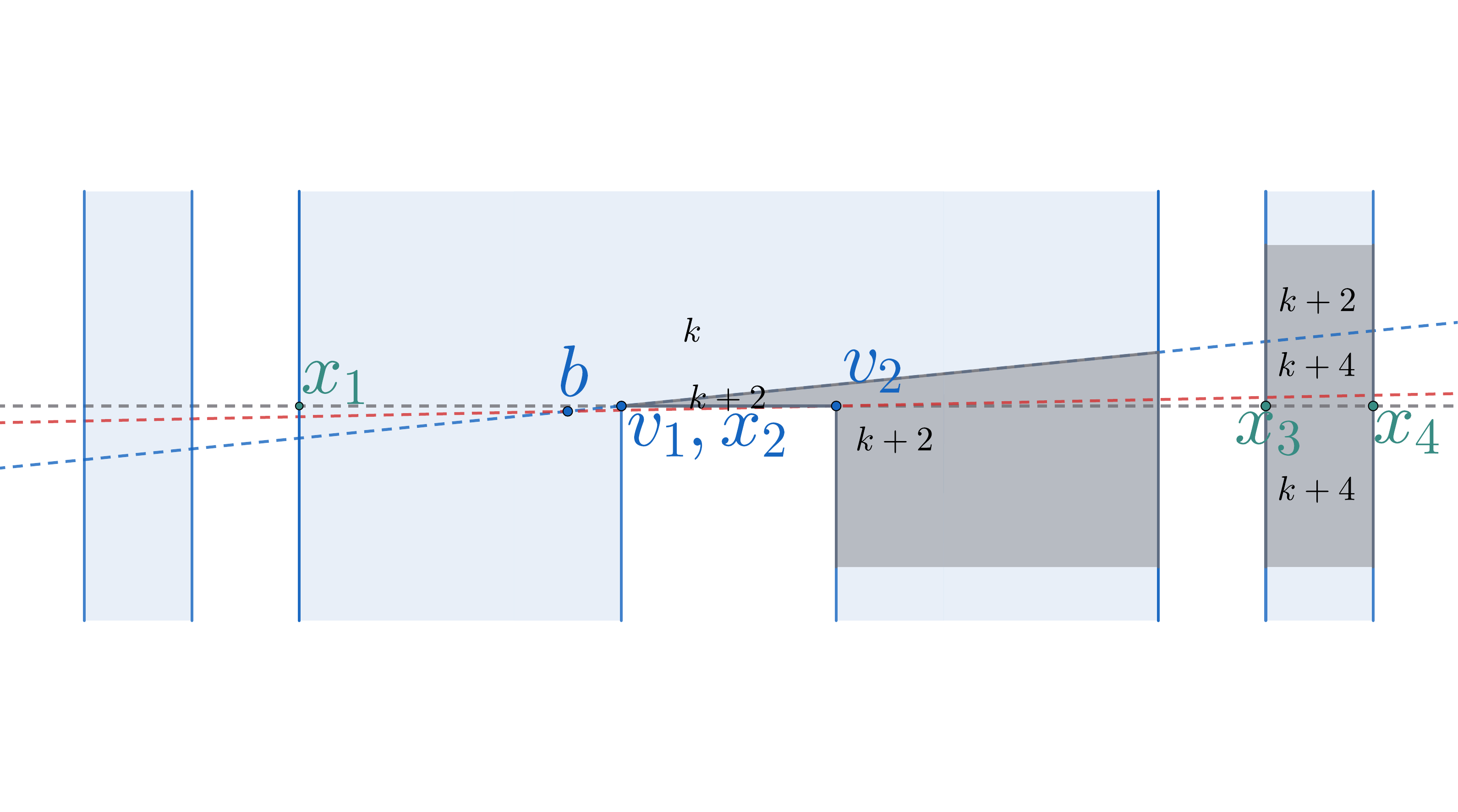}
\caption{Below $\ell_{g}$}\label{fig:RR-SpecialCase-B0}
\end{subfigure}

\caption{RR-SC; $Z = k$, $W = k + 2$.}
\label{fig:RR-SpecialCase-0}
\end{figure}

 \begin{figure}[H]
\centering
\begin{subfigure}[b]{.49\linewidth}
\includegraphics[width=\linewidth]{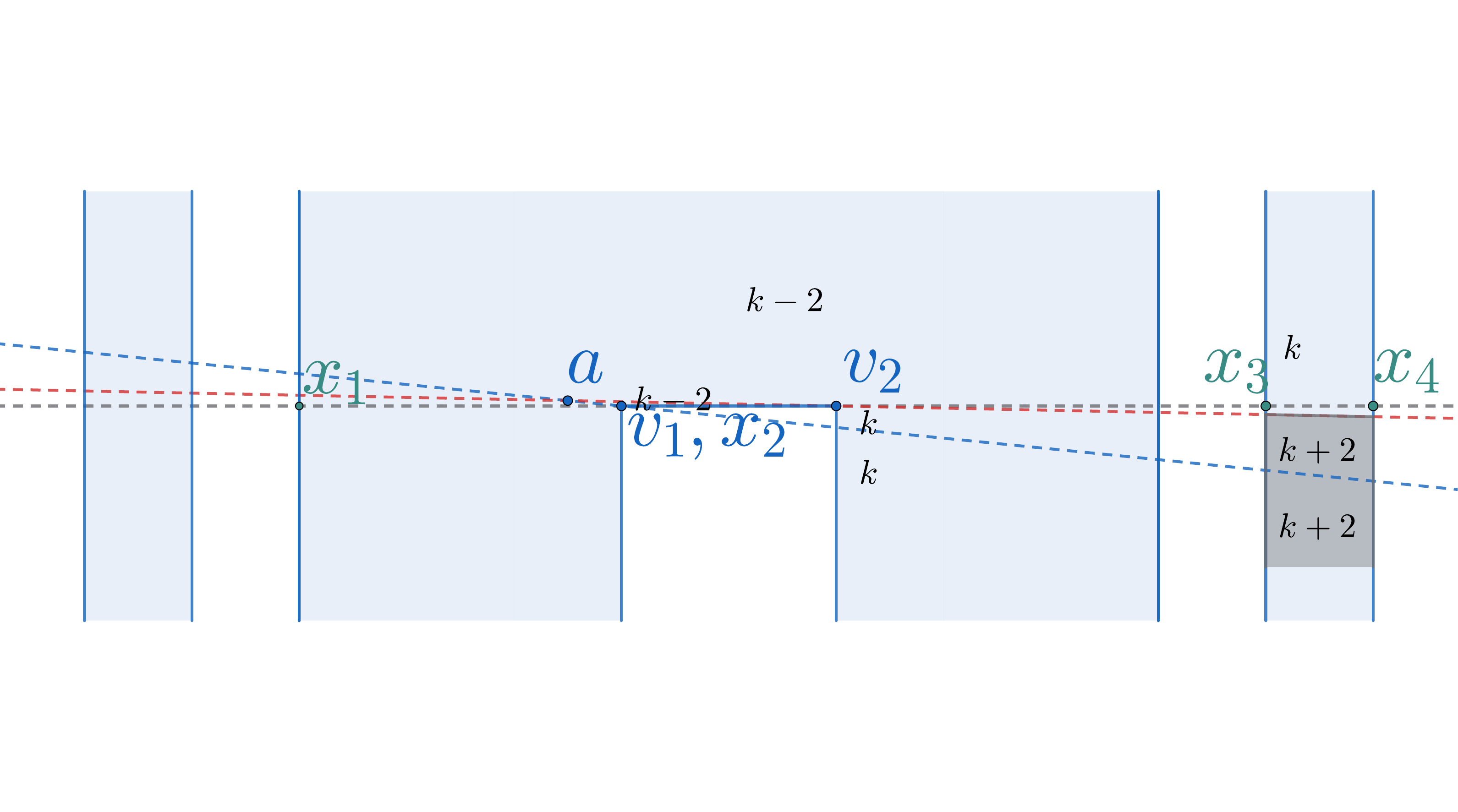}
\caption{Above $\ell_{g}$}\label{fig:RR-SpecialCase-A2}
\end{subfigure}
\begin{subfigure}[b]{.49\linewidth}
\includegraphics[width=\linewidth]{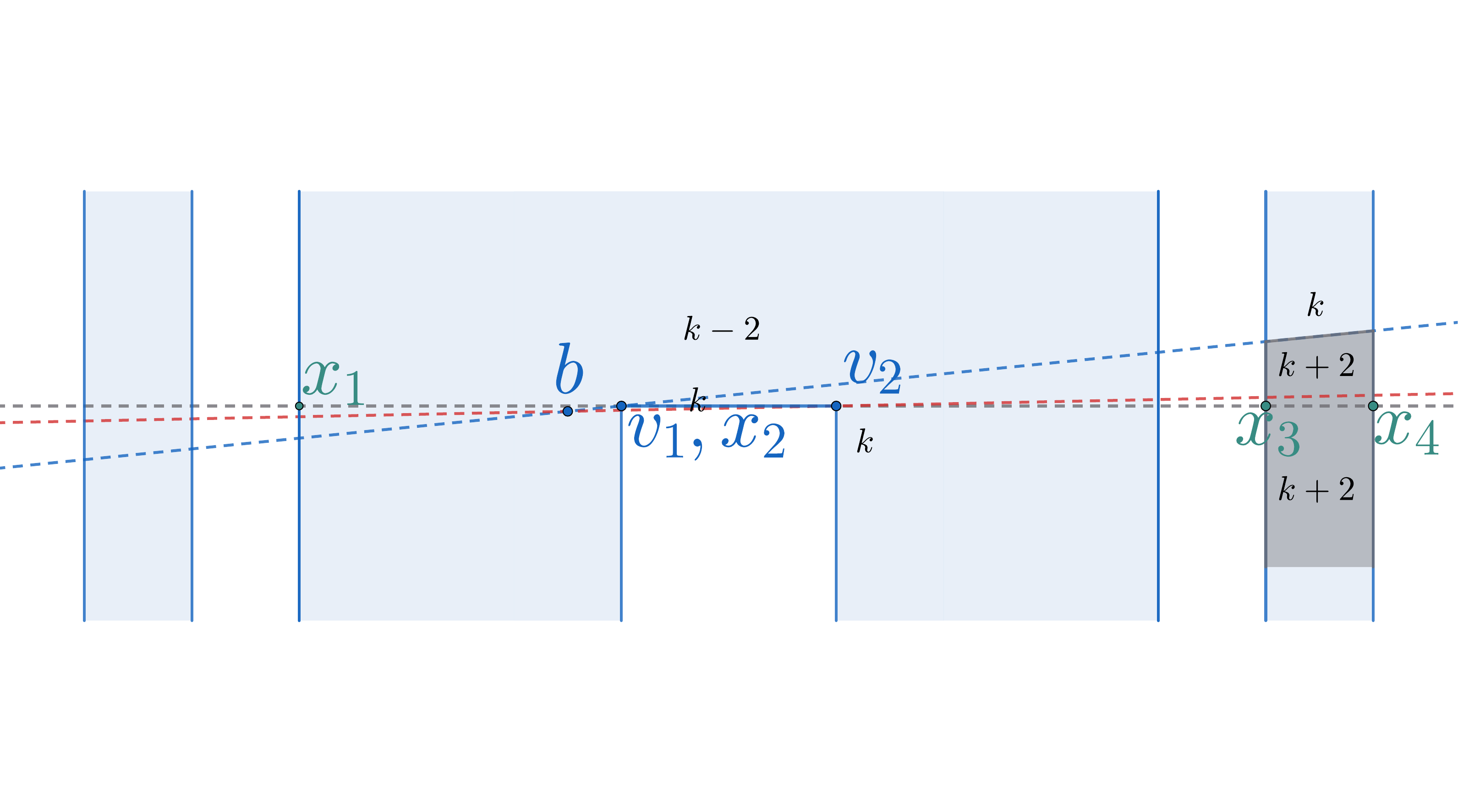}
\caption{Below $\ell_{g}$}\label{fig:RR-SpecialCase-B2}
\end{subfigure}

\caption{RR-SC; $Z = k - 2$, $W = k$.}
\label{fig:RR-SpecialCase-2}
\end{figure}

  \begin{figure}[H]
\centering
\begin{subfigure}[b]{.49\linewidth}
\includegraphics[width=\linewidth]{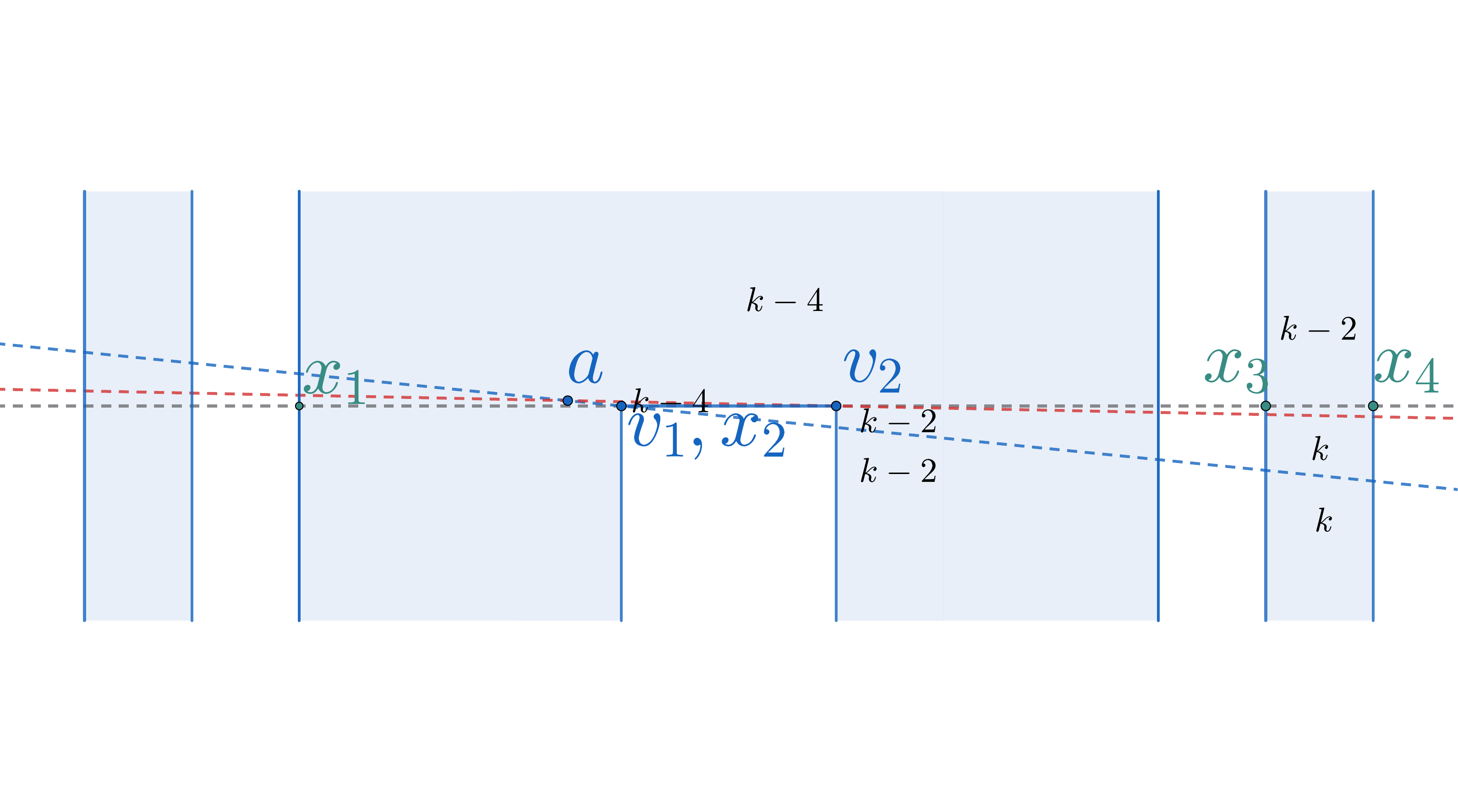}
\caption{Above $\ell_{g}$}\label{fig:RR-SpecialCase-A4}
\end{subfigure}
\begin{subfigure}[b]{.49\linewidth}
\includegraphics[width=\linewidth]{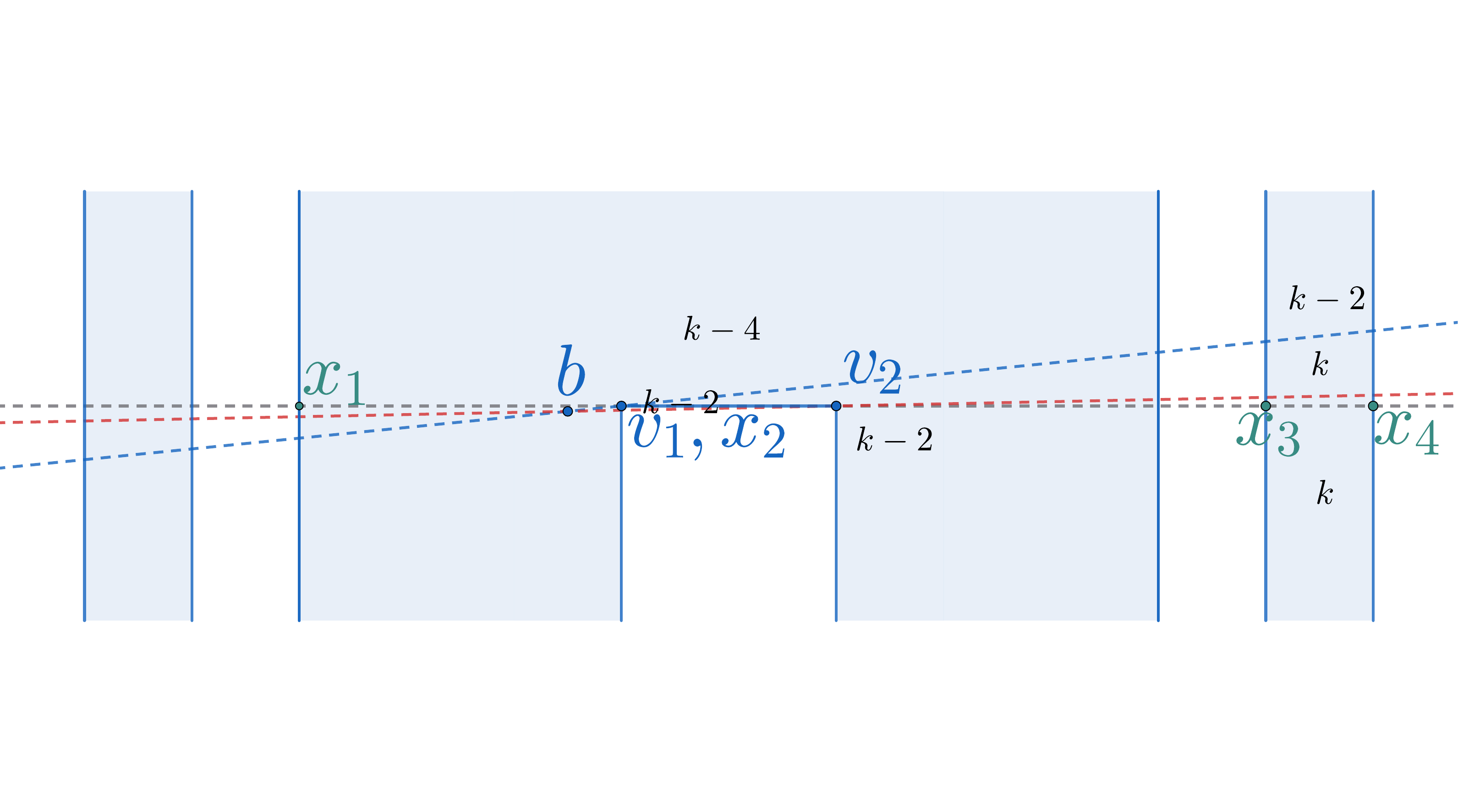}
\caption{Below $\ell_{g}$}\label{fig:RR-SpecialCase-B4}
\end{subfigure}

\caption{RR-SC; $Z = k - 4$, $W = k - 2$.}
\label{fig:RR-SpecialCase-4}
\end{figure}

\section{CR-SC}
\label{appendix:CR-SC}
\subsection{Convex Reflex Special Case}
\begin{lemma} 
\label{lemma:CR-SC}
  When $Z = k - 1$, an appear/disappear event occurs at $v_{2}$. If $W = k - 1$, an appear/disappear event occurs at $x_{3}X_{4}$. No event occurs for any other $Z$ or $W$. 

\end{lemma}
\begin{proof}
    See Figure~\ref{fig:CR-SpecialCase-2} to Figure~\ref{fig:CR-SpecialCase-4}. 

    For $Z \geq k + 1$, $v_{2}$ and its surroundings are entirely in shadow. For $W \geq k + 3$, $x_{3}x_{4}$ and its surroundings are entirely in shadow.

    For $Z \leq k - 5$, $v_{2}$ and its surroundings are entirely visible. For $W \leq k - 3$, $x_{3}x_{4}$ and its surroundings are entirely visible. 

\end{proof}



 \begin{figure}[H]
\centering
\begin{subfigure}[b]{.49\linewidth}
\includegraphics[width=\linewidth]{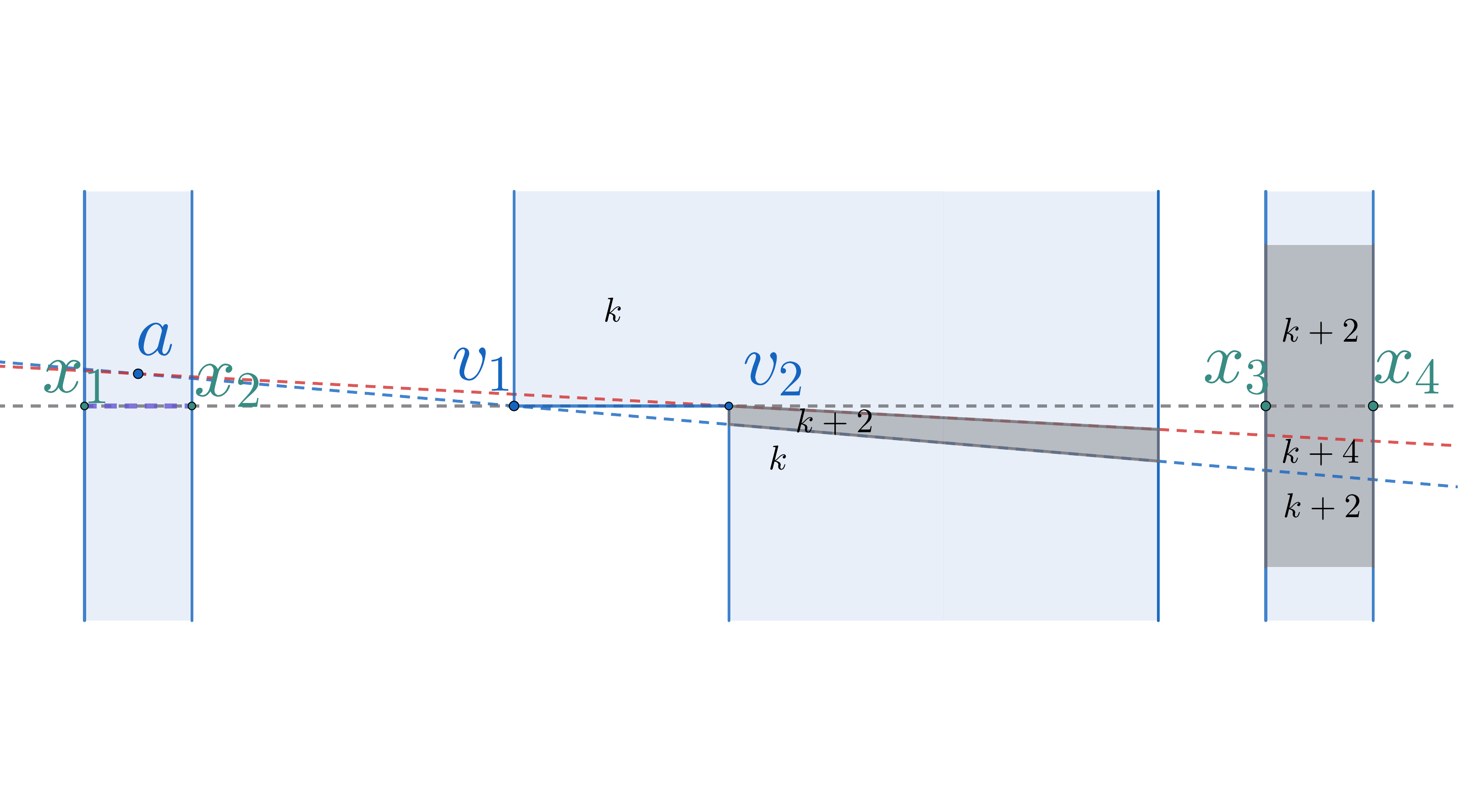}
\caption{Above $\ell_{g}$}\label{fig:CR-SpecialCase-A2}
\end{subfigure}
\begin{subfigure}[b]{.49\linewidth}
\includegraphics[width=\linewidth]{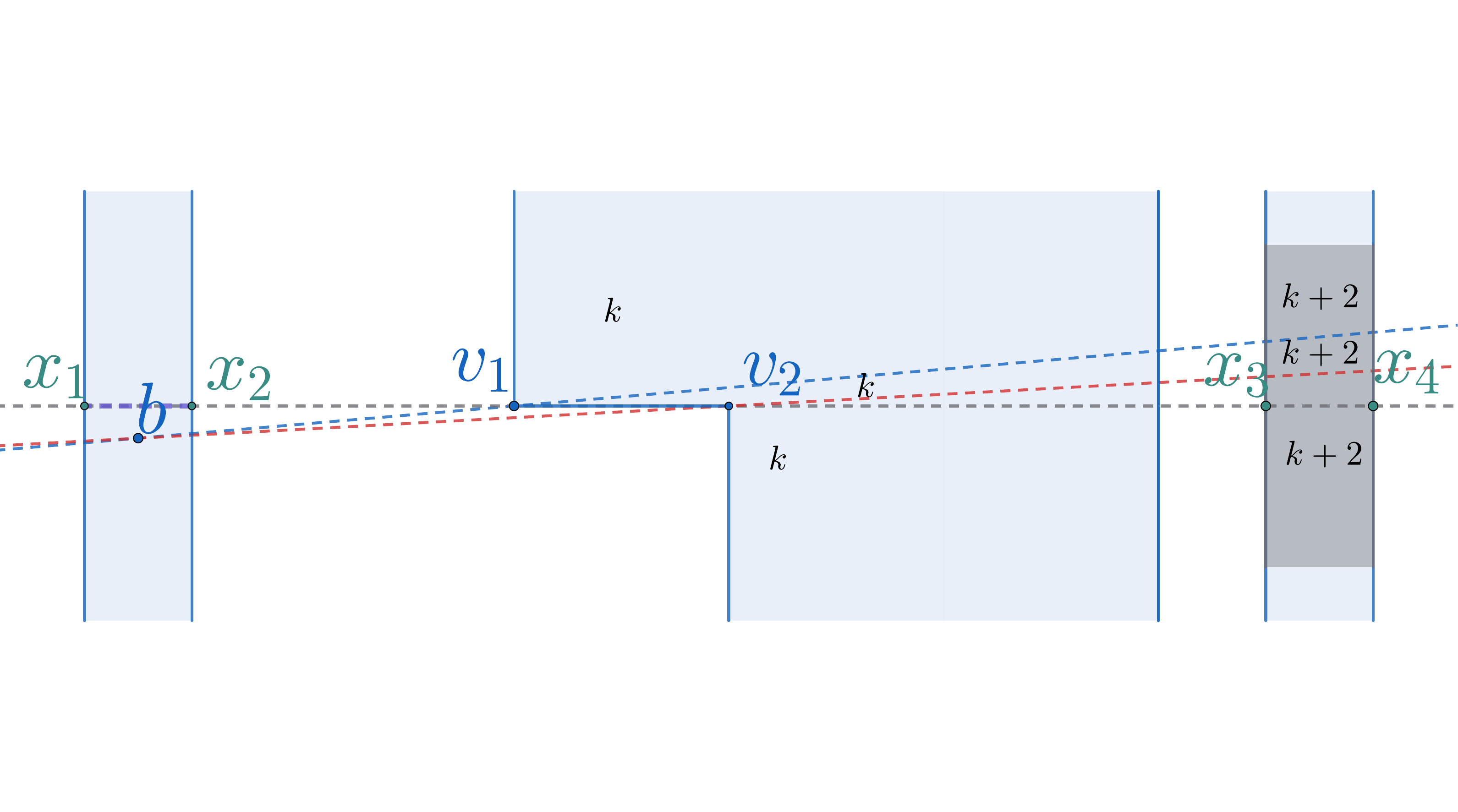}
\caption{Below $\ell_{g}$}\label{fig:CR-SpecialCase-B2}
\end{subfigure}

\caption{CR-SC; $Z = k - 1$ (Appear/Disappear), $W = k + 1$.}
\label{fig:CR-SpecialCase-2}
\end{figure}

  \begin{figure}[H]
\centering
\begin{subfigure}[b]{.49\linewidth}
\includegraphics[width=\linewidth]{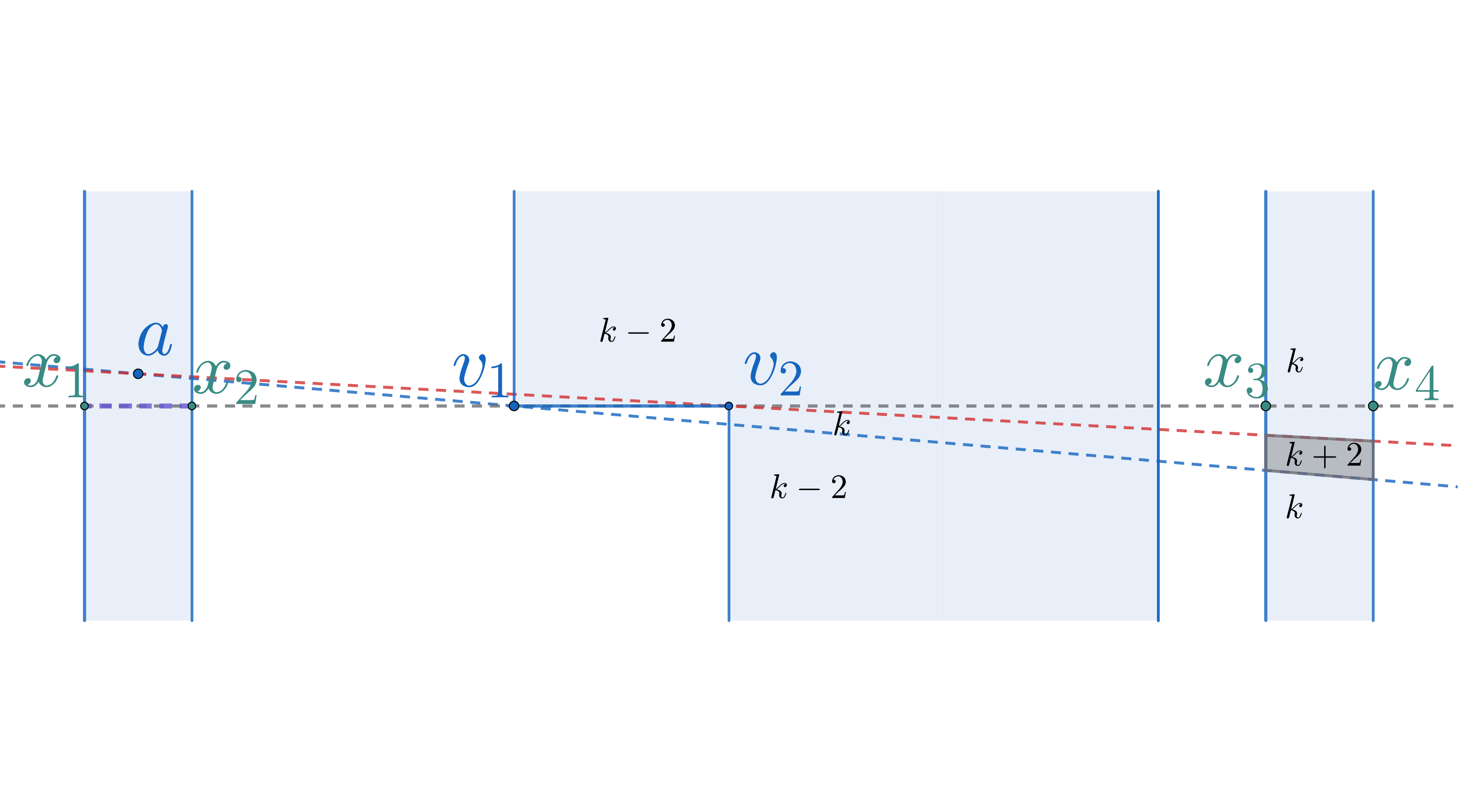}
\caption{Above $\ell_{g}$}\label{fig:CR-SpecialCase-A4}
\end{subfigure}
\begin{subfigure}[b]{.49\linewidth}
\includegraphics[width=\linewidth]{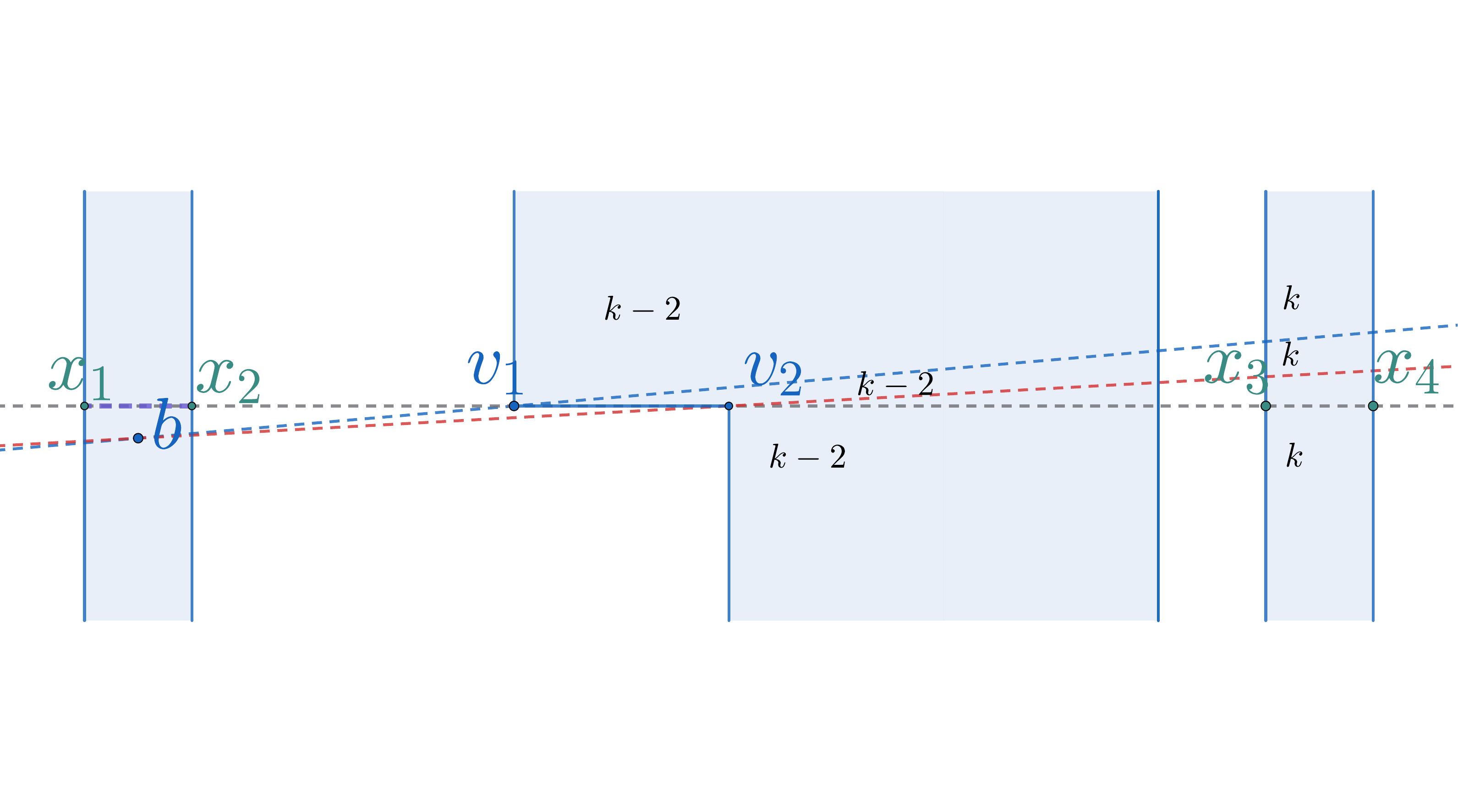}
\caption{Below $\ell_{g}$}\label{fig:CR-SpecialCase-B4}
\end{subfigure}

\caption{CR-SC; $Z = k - 3$, $W = k - 1$ (Appear/Disappear).}
\label{fig:CR-SpecialCase-4}
\end{figure}

\section{CC-SC}
\label{appendix:CC-SC}
\subsection{Convex Convex Special Case}
\begin{lemma} 
\label{lemma:CC-SC}
   No event occurs for any $Z$ or $W$. 

\end{lemma}

\begin{proof}
    See Figure~\ref{fig:CC-SpecialCase-2}. 
    
    For $Z \geq k + 1$, $v_{2}$ and its surroundings are entirely in shadow. For $W \geq k + 2$, $x_{3}x_{4}$ and its surroundings are entirely in shadow.

    For $Z \leq k - 3$, $v_{2}$ and its surroundings are entirely visible. For $W \leq k - 2$, $x_{3}x_{4}$ and its surroundings are entirely visible. 
\end{proof}

  \begin{figure}[H]
\centering
\begin{subfigure}[b]{.49\linewidth}
\includegraphics[width=\linewidth]{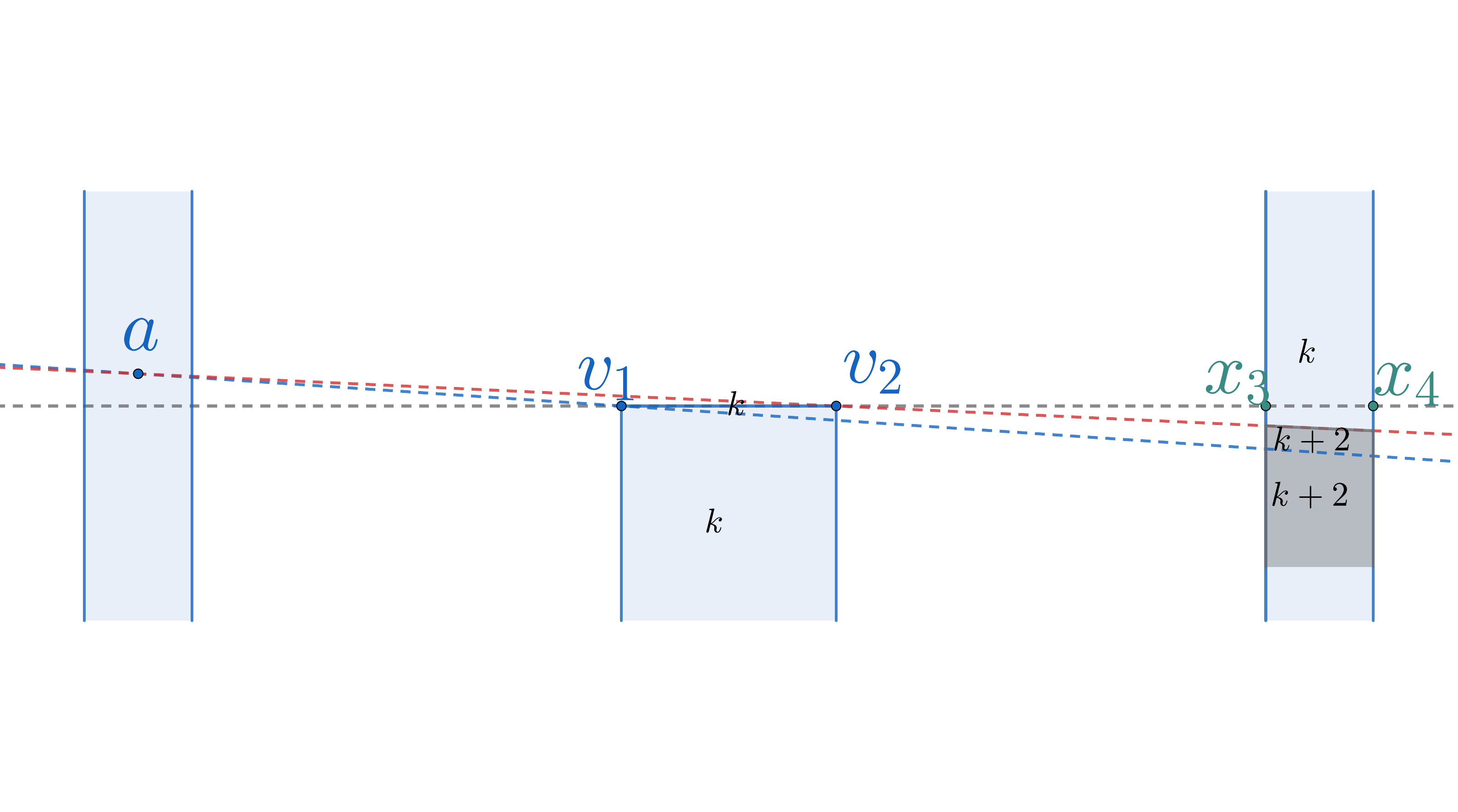}
\caption{Above $\ell_{g}$}\label{fig:CC-SpecialCase-A2}
\end{subfigure}
\begin{subfigure}[b]{.49\linewidth}
\includegraphics[width=\linewidth]{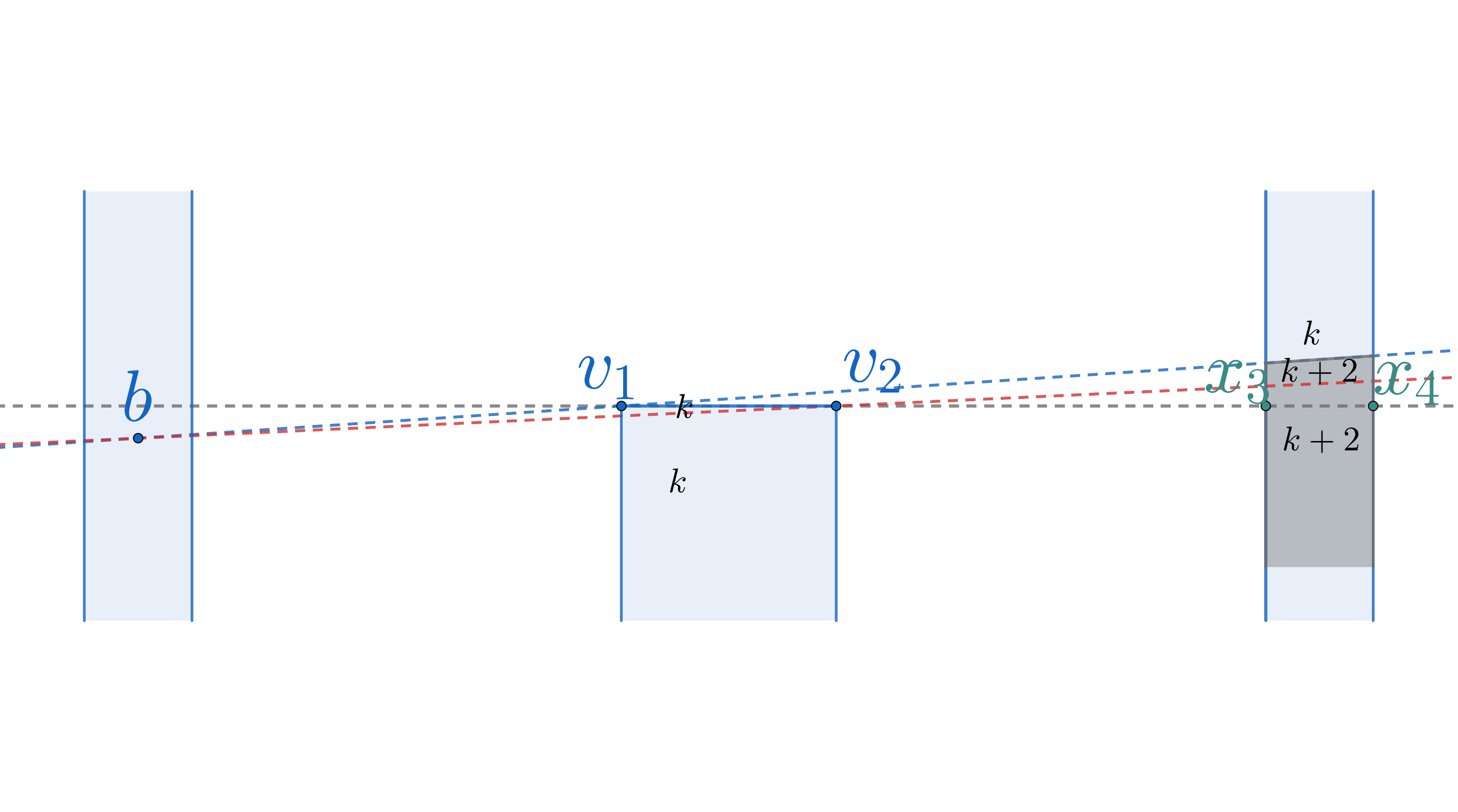}
\caption{Below $\ell_{g}$}\label{fig:CC-SpecialCase-B2}
\end{subfigure}

\caption{CC-SC; $Z = k - 1$, $W = k$}
\label{fig:CC-SpecialCase-2}
\end{figure}

\begin{table*}[t]

\centering
\caption{Summary of Combinatorial Changes, Conditions, and Associated Figures}
\label{tab:combinatorial_changes}
\begin{tabular}{lllll}
\hline
\textbf{Case} & \textbf{Lemma} & \textbf{Condition} & \textbf{Event Type \& Location} & \textbf{Figure} \\ \hline
\textbf{CCO} & Lemma~\ref{lemma:CCO-main} & $Z = k - 1$ & Appear/disappear event at $v_{2}$ & Figures~\ref{fig:CCO-genericA2}, \ref{fig:CCO-genericB2} \\
             &           & $W = k$     & Merge/split event at $x_{3}x_{4}$ & Figures~\ref{fig:CCO-genericA2}, \ref{fig:CCO-genericB2} \\
             &           & $W = k - 2$ & Appear/disappear event at $x_{3}x_{4}$ & Figures~\ref{fig:CCO-genericA4}, \ref{fig:CCO-genericB4} \\
             &           & Other $Z, W$ & No event & -- \\ \hline
\textbf{CRO} & Lemma~\ref{lemma:CRO} & $Z = k$ & Merge/split event at $v_{2}$ & Figure~\ref{fig:CRS-generic2} \\
             &           & $Z = k - 2$ & Appear/disappear event at $v_{2}$ & Figure~\ref{fig:CRS-generic4} \\
             &           & $W = k$ & Merge/split event at $x_{3}x_{4}$ & Figure~\ref{fig:CRS-generic4} \\
             &           & $W = k - 2$ & Appear/disappear event at $x_{3}x_{4}$ & Figure~\ref{fig:CRS-generic6}  \\ \hline
\textbf{CRS} & Lemma~\ref{lemma:CRS} & $Z = k$ & Merge/split event at $v_{2}$ & Figure~\ref{fig:CRO-generic2} \\
             &           & Other $Z, W$ & No event & -- \\ \hline
\textbf{RCO} & Lemma~\ref{lemma:RCO}& $Z = k - 1$ & Appear/disappear event at $v_{2}$ & Figure~\ref{fig:RCO-generic-2} \\
             &           & $W = k$ & Merge/split event at $x_{3}x_{4}$ & Figure~\ref{fig:RCO-generic-2} \\
             &           & $W = k - 2$ & Appear/disappear event at $x_{3}x_{4}$ & Figure~\ref{fig:RCO-generic-4} \\
             &           & Other $Z, W$ & No event & -- \\ \hline
\textbf{RRS} & Lemma~\ref{lemma:RRS} & $Z = k$ & Merge/split event at $v_{2}$ & Figure~\ref{fig:RRS-generic-0} \\
             &           & Other $Z, W$ & No event & -- \\ \hline
\textbf{RRO} & Lemma~\ref{lemma:RRO} & $Z = k$ & Merge/split event at $v_{2}$ & Figure~\ref{fig:RRO-generic-0} \\
             &           & $Z = k - 2$ & Appear/disappear event at $v_{2}$ & Figure~\ref{fig:RRO-generic-2} \\
             &           & $W = k$ & Merge/split event at $x_{3}x_{4}$ & Figure~\ref{fig:RRO-generic-2} \\
             &           & $W = k - 2$ & Appear/disappear event at $x_{3}x_{4}$ & Figure~\ref{fig:RRO-generic-4} \\
             &           & Other $Z, W$ & No event & -- \\ \hline
\textbf{RC-SC} & Lemma~\ref{lemma:RC-SC} & $Z = k$ & Appear/disappear event at $v_{2}$ & Figure~\ref{fig:RC-SpecialCase-0} \\
               &           & $W = k - 1$ & Appear/disappear event at $x_{3}x_{4}$ & Figure~\ref{fig:RC-SpecialCase-2} \\
               &           & Other $Z, W$ & No event & -- \\ \hline
\textbf{CR-SC} & Lemma~\ref{lemma:CR-SC} & $Z = k - 1$ & Appear/disappear event at $v_{2}$ & Figure~\ref{fig:CR-SpecialCase-2} \\
               &           & $W = k - 1$ & Appear/disappear event at $x_{3}x_{4}$ & Figure~\ref{fig:CR-SpecialCase-4} \\
               &           & Other $Z, W$ & No event & -- \\ \hline
\textbf{CCS}, \textbf{RCS}, \textbf{RR-SC}, \textbf{CC-SC} & 
\begin{tabular}[c]{@{}l@{}}
    Lemma~\ref{lemma:CCS},\\ 
    Lemma~\ref{lemma:RCS},\\ 
    Lemma~\ref{lemma:RR-SC},\\ 
    Lemma~\ref{lemma:CC-SC}
\end{tabular} & Any $Z, W$ & No event & -- \\ \hline
\end{tabular}
\label{table:all-cases}
\end{table*}

\end{document}